\documentclass[citeautoscript,floatfix,aps,prd,twocolumn,superscriptaddress,longbibliography]{revtex4-2}  %
\usepackage{graphicx}  %
\usepackage{dcolumn}   %
\usepackage{bm}        %
\usepackage{amssymb}   %
\usepackage{amsmath}
\usepackage{amsthm}
\theoremstyle{definition}
\newtheorem{proposition}{Proposition}[section]
\newtheorem{conjecture}{Conjecture}[section]
\newtheorem{assumption}{Assumption}[section]
\newtheorem{definition}{Definition}[section]
\newtheorem{theorem}{Theorem}[section]
\newtheorem{lemma}[theorem]{Lemma}
\newtheorem{corollary}[theorem]{Corollary}
\newtheorem{fact}{Fact}[section]
\newtheorem{remark}{Remark}[section]
\newtheorem{example}{Example}[section]
\usepackage{thmtools}
\usepackage{thm-restate}

\usepackage{multirow}
\usepackage{braket}
\usepackage{amsfonts}
\usepackage{dsfont}
\usepackage{hyperref}
\usepackage{cleveref}
\usepackage{xcolor}
\usepackage{simplewick}

\usepackage[export]{adjustbox}
\usepackage{mathtools}
\DeclarePairedDelimiter\floor{\lfloor}{\rfloor}
\newcommand{\Rmm}{\rotatebox[origin=c]{180}{$R$}}
\newcommand{\Rmp}{\rotatebox[origin=c]{90}{$R$}}

\newcommand{\id}{\mathrm{id}}

\newcommand{\C}{\mathbb{C}}
\newcommand{\Z}{\mathbb{Z}}

\newcommand{\calA}{\mathcal{A}}
\newcommand{\calC}{\mathcal{C}}

\newcommand{\calR}{\mathcal{R}}

\newcommand{\FA}{\mathfrak{A}}
\newcommand{\frakF}{\mathfrak{F}}
\newcommand{\End}{\mathrm{End}}
\newcommand{\Rep}{\mathrm{Rep}}
\newcommand{\Hom}{\mathrm{Hom}}
\newcommand{\Tr}{\mathrm{Tr}}

\newcommand{\Hil}{\mathcal{H}}

\newcommand{\T}{\mathcal{T}}
\newcommand{\oA}{o}
\newcommand{\oB}{s}

\newcommand{\mpwf}[3]{\Ket{{#1}_{#2}^{#3}}} %
\usepackage{pifont}%
\newcommand{\cmark}{\ding{51}}%
\newcommand{\xmark}{\ding{55}}%
\newcommand{\Yes}{{\color{blue}\cmark}}
\newcommand{\No}{{\color{red}\xmark}}
\newcommand{\Dsw}{\mathcal{D}_{\swarrow}}
\newcommand{\Dse}{\mathcal{D}_{\searrow}}
\newcommand{\DC}{\mathcal{D}'}
\newcommand{\DCsw}{\mathcal{D}'_{\swarrow}}
\newcommand{\DCse}{\mathcal{D}'_{\searrow}}
\newcommand{\zetapsi}{\psi}
\newcommand{\thetaR}{\theta}

\definecolor{darkorange}{RGB}{200,100,0}

\usepackage{tikz}
\usepackage{tensor}
\usepackage{tikz-network}
\usepackage{paratensor}
\usepackage{ulem}
\usepackage{color}
\definecolor{mypurple}{RGB}{255,0,255}

\makeatletter
\newcommand{\customlabel}[2]{%
	\protected@write \@auxout {}{\string \newlabel {#1}{{#2}{\thepage}{#2}{#1}{}} }%
	\hypertarget{#1}{#2}
}
\makeatother
\begin{document}

\title{On $R$-parastatistics I: Foundation}
	\author{Zhiyuan Wang}
    \affiliation{Perimeter Institute for Theoretical Physics, 31 Caroline St N, Waterloo, Ontario N2L-2Y5, Canada}
	\author{Kaden R.~A. Hazzard}
	\affiliation{Department of Physics and Astronomy, Rice University, Houston, Texas 77005,
		USA}
	\affiliation{Smalley-Curl Institute, Rice University, Houston, Texas 77005, USA}
	\date{\today}
\begin{abstract}
Parastatistics is an exotic %
type of exchange statistics beyond fermions and bosons. Paraparticles transform in higher dimensional representations of the
exchange symmetry group, 
analogous to non-Abelian anyons, %
yet consistently defined in any spatial dimension. %
Although paraparticles have long been proposed~\cite{Green1952},
they were widely believed to be physically equivalent to fermions or bosons~\cite{doplicher1971local, *doplicher1974local,doplicherFieldsStatisticsNonabelian1972}. Nevertheless, a recent paper~\cite{wang2023para} proposed a different theory, called $R$-parastatistics, 
and demonstrated that nontrivial $R$-paraparticles can emerge as quasiparticles in condensed matter systems, and are observably distinct from both fermions and bosons. This paper develops %
the theoretical foundation and several extensions of $R$-parastatistics, with particular emphasis on its observable consequences. 
Central to this paper is a general theory of local observables extending the basic family introduced in Ref.~\cite{wang2023para}. 
First, we define local observables that distinguish  particle types, providing the basis for mutual parastatistics and a systematic discussion of local indistinguishability of $R$-paraparticles. 
Second, we formulate local observables at special point defects that probe the internal indices of $R$-paraparticles, giving a firmer theoretical footing for %
observing $R$-parastatistics and for the proposed applications in quantum information~\cite{wang2024parastatistics,wang2025secret}. 
Third, we introduce local observables that create or annihilate particle-antiparticle pairs, providing a key ingredient for the $R$-paraparticle analog of  Bogoliubov-de Gennes mean-field theory of superconductors and of relativistic quantum field theory. %
We further introduce generalized hidden symmetries that act on internal indices of $R$-paraparticles while preserving the local observable algebra, providing a basis for proving local indistinguishability and for connecting to a categorical description of $R$-paraparticles. 
On the mathematical side, we 
prove a classification theorem on unitary $R$-matrices admitting local pair creation. %
This work sets a solid theoretical foundation for understanding the fundamental physical properties of $R$-paraparticles and %
pave the way for finding them in nature~(in particular as emergent quasiparticles in condensed matter systems, and as elementary particles in relativistic quantum field theories),  %
and the relation to previous no-go theorems~\cite{doplicher1971local, *doplicher1974local,doplicherFieldsStatisticsNonabelian1972,Buchholz1982}, all of which will be completed in the second part of this series.
\end{abstract}
	
	\maketitle
	\tableofcontents
\section{Introduction}\label{sec:intro}
\subsection{Background}
In 1953, Green~\cite{Green1952} proposed a generalized method of field quantization that go beyond the canonical quantization of fermionic and bosonic fields~(see App.~\ref{app:Green-theory} for a brief summary of Green's theory). 
The resulting quantum field theory describes an exotic type of particle exchange statistics, 
now called parastatistics, where identical particles transform in higher dimensional representations of the symmetric group $S_N$ under exchange~\cite{hartleQuantumMechanicsParaparticles1969}, in constrast to the one-dimensional representations realized by fermions and bosons. 
This theory was subsequently studied in detail~\cite{Araki1961,greenbergSpinUnitarySpinIndependence1964,Greenberg1965,LANDSHOFF196772,druhl1970parastatistics,Taylor1970a,Taylor1970b,tolstoyOnceMoreParastatistics2014,Baker2015Conventionality}, 
and also more generally and  rigorously~\cite{doplicher1971local, *doplicher1974local,doplicherFieldsStatisticsNonabelian1972,Buchholz1982, doplicher1990} %
within the framework of algebraic quantum field theory~\cite{WightmanPCTbook, haag2012book,halvorson2006algebraic}.
While these works did not rule out parastatistics as a theoretical possibility, they led to the conclusion that under natural physical assumptions~(including locality, unitarity, special relativity, along with a few more technical ones), paraparticles are 
 physically indistinguishable from ordinary fermions and bosons carrying some extra degrees of freedom such as color or flavor, in the sense that 
any physical prediction produced by previous paraparticle theories~(in particular, Green's theory) can also be produced by some equivalent theories of fermions and bosons. This is now called the ``conventionality argument'' ~\cite{Baker2015Conventionality}, which
prevents us from recognizing parastatistics as a distinctive type of exchange statistics, for its failing to produce experimental predictions that separate it from conventional particle statistics.

Nevertheless, a recent work~\cite{wang2023para} challenged this conventionality argument by introducing a different second-quantized theory of paraparticles, called $R$-parastatistics, based on bilinear commutation relations involving an $R$-matrix rather than Green’s trilinear commutation relations. 
Importantly, $R$-paraparticles described by this theory generally display nontrivial exclusion and exchange statistics that in general cannot be reduced to 
ordinary fermions or bosons carrying an extra degree of
freedom. 
The theory of $R$-paraparticles is compatible with spatial locality and unitary time evolution, and does not contradict the earlier no-go theorems~\cite{doplicher1971local,doplicher1974local}, because it evades some of their restrictive assumptions~(details will be discussed in part II of this series).
Ref.~\cite{wang2023para} further constructed exactly solvable quantum spin models in one and two dimensions~(later generalized to three dimensions in the supplemental material of Ref.~\cite{wang2024parastatistics}) in which such $R$-paraparticles emerge as quasiparticle excitations, and showed that the distinctive features of $R$-parastatistics can, in principle, be physically observed in condensed matter systems. %
Refs.~\cite{wang2024parastatistics,wang2025secret} sharpened this physical distinction from a quantum-information viewpoint, by formulating a secret-communication challenge game that cleanly separates $R$-paraparticles from ordinary fermions and bosons.

However, because Ref.~\cite{wang2023para} is necessarily brief, it did not lay down the full theoretical foundation %
of $R$-parastatistics. Even taking into account the supplemental information of Ref.~\cite{wang2023para}, several fundamental aspects of $R$-paraparticles are missing. First, Ref.~\cite{wang2023para} lacks a serious discussion on the (local) indistinguishability of $R$-paraparticles,
a fundamental physical property that qualifies them to be called identical particles in the first place, and is also crucial for the proposed applications in quantum information~\cite{wang2024parastatistics,wang2025secret}.  It is also unclear in general what physical observables should be measured in order to experimentally observe the distinctive features of $R$-parastatistics, although this is known in the specific models constructed in Ref.~\cite{wang2023para}. It also misses the discussion of antiparticles, which is fundamental for the construction of relativistic quantum field theories of $R$-paraparticles~(which is still missing to date, but will be completed in the second part of this series). 

In the following we first provide a non-technical introduction to parastatistics in the first quantization formulation and discuss the important concept of local indistinguishability, and then in Sec.~\ref{sec:intro-goal} we state the main goals of this paper and address the aforementioned open questions.

\subsection{What is parastatistics}\label{sec:what_is_para}
Below we give a non-technical introduction to parastatistics in the first quantization formulation, which provides a clean conceptual explanation of how parastatistics generalizes ordinary fermions and bosons, in a way distinct from non-Abelian anyons, and we also clarify the fundamental issue of (local) indistinguishability in this context. 

We begin by reviewing the standard textbook treatment of identical particle statistics that leads to the fermion-boson dichotomy. In quantum mechanics, each multiparticle quantum state is described by a ket $\mpwf{\Psi}{x_1,x_2,\ldots, x_n}{}$, where $x_1,x_2,\ldots, x_n\in \mathbb{R}^d$ are particle coordinates in a $d$ dimensional space. The particles are identical, meaning that when we exchange the positions of any two of them~(say $x_1,x_2$), the resulting state $\mpwf{\Psi}{x_2,x_1,\ldots, x_n}{}$ must represent the same physical state, and therefore can change by at most a constant factor  
\begin{equation}\label{eq:wavefuntion_exchange}
	\mpwf{\Psi}{x_2,x_1,\ldots, x_n}{}=c \mpwf{\Psi}{x_1,x_2,\ldots, x_n}{}.
\end{equation} 
If we do a second exchange, we have 
\begin{eqnarray}\label{eq:wavefuntion_exchange_2nd}
	\mpwf{\Psi}{x_1,x_2,\ldots, x_n}{}&=&c \mpwf{\Psi}{x_2,x_1,\ldots, x_n}{}\nonumber\\
	&=&c^2 \mpwf{\Psi}{x_1,x_2,\ldots, x_n}{},
\end{eqnarray}
leading to $c^2=1$. This provides exactly two possibilities, bosons~($c=1$) and fermions~($c=-1$). %

Despite being simple and convincing, there are two important exceptions to the fermion/boson dichotomy. The first is anyons  in two spatial dimension~(2D)~\cite{Leinaas1977,Wilczek1982Magnetic,*Wilczek1982Quantum,Wilczek1990book,Nayak2008NAAnyons,STERN2008204}, where we have to distinguish the two inequivalent ways of exchanging two anyons in 2+1D space time. 
Here we use $\hat{E}^{\text{ccw}}_{12}$ to denote the physical operation of braiding anyons 1 and 2 in the counterclockwise direction, and similarly  $\hat{E}^{\text{cw}}_{12}$ for the clockwise direction,  
and Eq.~\eqref{eq:wavefuntion_exchange} should be replaced by 
\begin{eqnarray}\label{eq:wavefuntion_exchange_anyon}
\hat{E}^{\text{ccw}}_{12}\mpwf{\Psi}{x_1,x_2,\ldots, x_n}{}&=&c \mpwf{\Psi}{x_1,x_2,\ldots, x_n}{},\nonumber\\
	\hat{E}^{\text{cw}}_{12}\mpwf{\Psi}{x_1,x_2,\ldots, x_n}{}&=&c^* \mpwf{\Psi}{x_1,x_2,\ldots, x_n}{},
\end{eqnarray} 
where $c$ can be any U$(1)$ phase factor~(hence the name ``anyons''), since in $2+1$D spacetime, $(\hat{E}^{\text{ccw}}_{12})^2$ is inequivalent to the trivial operation. Note that the second line of  Eq.~\eqref{eq:wavefuntion_exchange_anyon} is obtained from the first line due to the relation $ \hat{E}^{\text{cw}}_{12}=(\hat{E}^{\text{ccw}}_{12})^{-1}$. %

The second exception to the textbook argument is parastatistics~\cite{Green1952,Araki1961,Greenberg1965,LANDSHOFF196772,druhl1970parastatistics,Taylor1970b}, which can be consistently defined in any spatial dimension. The way this evades the above  argument is that the multiparticle state can carry extra indices that transform nontrivially during an exchange. Consider an $n$-particle state $\mpwf{\Psi}{x_1,x_2,\ldots, x_n}{I}$, where $I$ is a collection of extra indices corresponding to some internal degrees of freedom inaccessible to local measurements~(we will define precisely what this means in Sec.~\ref{sec:indistcrit}). Under an exchange between particles $j$ and $j+1$~\footnote{Note that we only need to specify the behavior of the state under exchange of particles with adjacent labels, since exchange of particles with nonadjacent labels can always be decomposed into a series of adjacent exchanges. For example, under the exchange of particles $1$ and $3$, the state should multiply by the matrix $R_{1}R_{2}R_{1}$.}, the state may undergo a matrix transformation
\begin{equation}\label{eq:wavefuntion_exchange_para}
	\mpwf{\Psi}{\{x_i\}^n_{i=1}}{I}|_{x_j\leftrightarrow x_{j+1}}=\sum_J \mpwf{\Psi}{\{x_i\}^n_{i=1}}{J}(R_{j,j+1})^J_{I},
\end{equation} 
for $j=1,\ldots,n-1$, where the summation is over all possible values of $J$. Similar to  the $c^2=1$ constraint for Eq.~\eqref{eq:wavefuntion_exchange}, the matrices $(R_{j,j+1})^I_{J}$ have to satisfy some algebraic constraints to guarantee consistency of Eq.~\eqref{eq:wavefuntion_exchange_para}. Here for simplicity we derive these consistency relations for the case $N=3$, but the generalization to arbitrary number of particles is straightforward.  Eq.~\eqref{eq:wavefuntion_exchange} is generalized to
\begin{eqnarray}\label{eq:involutive_condition}
	\mpwf{\Psi}{x_1,x_2,x_3}{I}&=& \sum_J\mpwf{\Psi}{x_2,x_1,x_3}{J}( R_{12})^J_I\nonumber\\
	&=& \sum_{J,K}\mpwf{\Psi}{x_1,x_2,x_3}{K}(R_{12})^K_J(R_{12})^J_I,
\end{eqnarray}
therefore $R_{12}^2=\mathds{1}$, the identity matrix, and similarly for $R_{23}$.  A second set of consistency conditions is obtained from the equivalence of two different ways of swapping $1,2,3$ to $3,2,1$, as shown in Fig.~\ref{fig:YBEgraphical}: on one side, we have
\begin{eqnarray}\label{eq:YBE1}
	\mpwf{\Psi}{x_1,x_2,x_3}{I}&=&  \sum_J \mpwf{\Psi}{x_2,x_1,x_3}{J}(R_{12})^J_I\\
	&=& \sum_J\mpwf{\Psi}{x_2,x_3,x_1}{J}(R_{23} R_{12})^J_I\nonumber\\
	&=&  \sum_J\mpwf{\Psi}{x_3,x_2,x_1}{J}(R_{12}R_{23}R_{12})^J_I,\nonumber
\end{eqnarray}
while on the other side, we have
\begin{eqnarray}\label{eq:YBE2}
	\mpwf{\Psi}{x_1,x_2,x_3}{I}&=& \sum_J \mpwf{\Psi}{x_1,x_3,x_2}{J}(R_{23})^J_I\\
	&=& \sum_J\mpwf{\Psi}{x_3,x_1,x_2}{J}(R_{12}R_{23})^J_I\nonumber\\
	&=& \sum_J\mpwf{\Psi}{x_3,x_2,x_1}{J}(R_{23}R_{12}R_{23})^J_I.\nonumber
\end{eqnarray}
Therefore, to guarantee consistency, we require $R_{12}R_{23}R_{12}=R_{23}R_{12}R_{23}$. Generalizing to $N$-particles, we have
\begin{eqnarray}\label{eq:YBE_SN}
	R_{j,j+1}^2&=&\mathds{1},~~~~ 1\leq j\leq N-1,\\
	R_{j-1,j}R_{j,j+1}R_{j-1,j}&=&R_{j,j+1}R_{j-1,j}R_{j,j+1},  \nonumber\\ %
	R_{i,i+1}R_{j,j+1}&=&R_{j,j+1}R_{i,i+1},\quad |i-j|\geq 2,\nonumber
\end{eqnarray}
where the last one is due to the commutativity of the swaps $x_{i}\leftrightarrow x_{i+1}$ and $x_{j}\leftrightarrow x_{j+1}$ for $|i-j|\geq 2$. 
Eq.~\eqref{eq:YBE_SN} is simply the requirement that $\{R_{j,j+1}\}^{N-1}_{j=1}$ must generate a representation of the permutation group $S_N$--indeed, one way to define $S_N$ is to present it via the generators $\{R_{j,j+1}\}^{N-1}_{j=1}$ and the relations in Eq.~\eqref{eq:YBE_SN}~\cite{kassel2008braid}.
If this representation of $S_N$ %
is not one-dimensional, we say Eq.~\eqref{eq:wavefuntion_exchange_para} defines a type of parastatistical particles, or paraparticles for short. 

We now make a comparison to non-Abelian anyons~\cite{wittenQuantumFieldTheory1989,MR1991,kitaev2003fault,kitaev2006anyons,Nayak2008NAAnyons} in 2D topological phases~\cite{Wen2017Zoo}, which can be viewed as a combination of Eq.~\eqref{eq:wavefuntion_exchange_anyon} and Eq.~\eqref{eq:wavefuntion_exchange_para}. 
In this case we associate $R_{j,j+1}$ to $\hat{E}^{\text{ccw}}_{j,j+1}$ %
and since $\hat{E}^{\text{ccw}}_{j,j+1}=(\hat{E}^{\text{cw}}_{j,j+1})^{-1}$, we should associate $R^{-1}_{j,j+1}$ to $\hat{E}^{\text{cw}}_{j,j+1}$. 
Repeating the arguments above~[one has to carefully keep track of the space-time trajectories in Eqs.~\eqref{eq:YBE1} and \eqref{eq:YBE2}, as shown in Fig.~\ref{fig:YBEgraphical}],
we find that for non-Abelian anyons, the matrices $\{R_{j}\}^{n-1}_{j=1}$ need to satisfy all the relations in Eq.~\eqref{eq:YBE_SN} except the first line, which is equivalent to saying that $\{R_{j,j+1}\}^{N-1}_{j=1}$ generate a representation of the braid group $B_N$~\cite{fredenhagenSuperselectionSectorsBraid1989,kassel2008braid}. 
Therefore in the special case of 2D, parastatistics can be considered as a special case of non-Abelian statistics satisfying additionally $R_{j,j+1}^2=\mathds{1}$, however, anyons cannot be consistently defined in higher dimensions while parastatistics is consistently defined in any dimension. 
Another key difference between anyons and paraparticles is that there does not exist an exactly solvable free particle theory for genuine anyons~(meaning that $R_{j,j+1}^2\neq\mathds{1}$), while there exists a solvable theory of free $R$-paraparticles~\cite{wang2023para}, as we will present in Sec.~\ref{sec:solution}. By a ``solvable free particle theory of anyons'', we mean a many-body Hamiltonian describing a system of non-interacting anyons freely moving in space, similar to Eq.~\eqref{eq:u1_sym_H}, such that one can obtain the exact many body spectrum by solving the one particle spectrum. 

\begin{figure}
	\center{\includegraphics[width=0.8\linewidth]{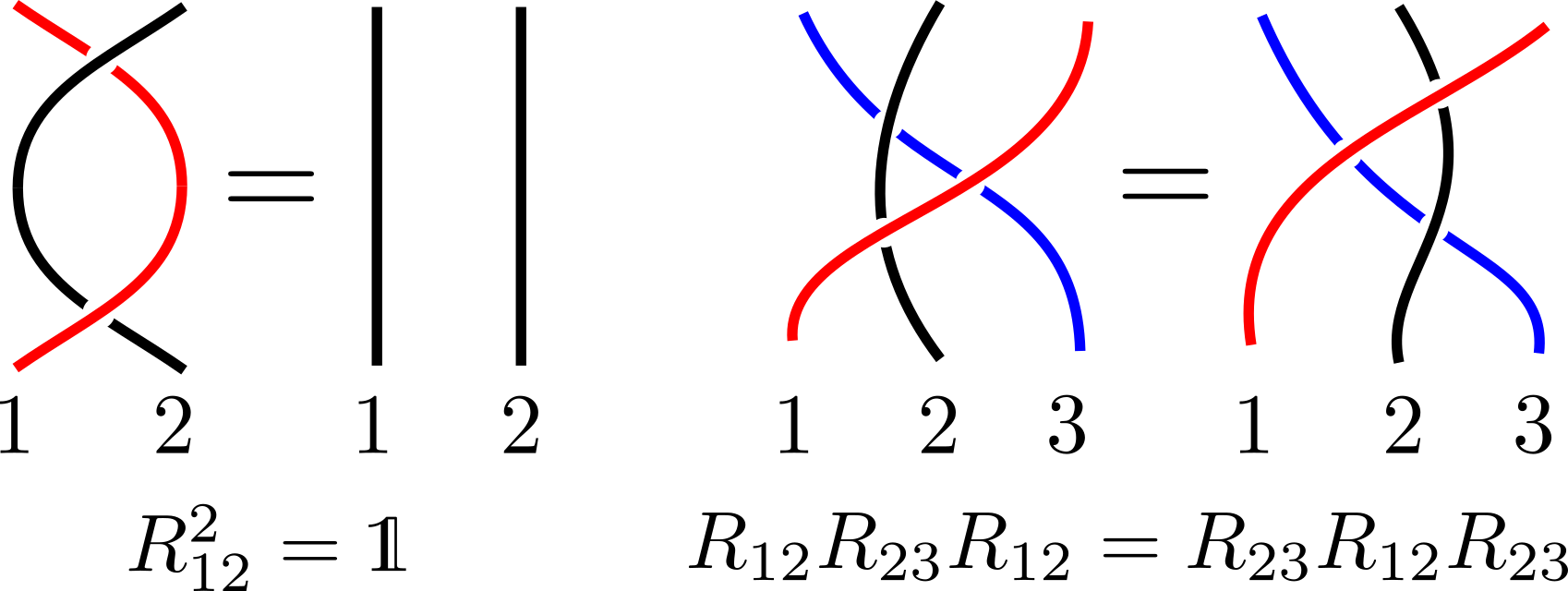}}
	\caption{\label{fig:YBEgraphical} Graphical illustrations of the derivations in  Eqs.~(\ref{eq:involutive_condition}-\ref{eq:YBE2}). Parastatistics satisfy both relations, while anyons in 2+1D only satisfy the second relation. } 
\end{figure}

\subsubsection{The indistinguishability criterion}\label{sec:indistcrit}
We now give a precise definition of indistinguishable particles, and discuss different notions of indistinguishability. 

Consider any quantum state $\ket{\Psi_{x_1,x_2,\ldots,x_n}}$ with quantum particles at positions $x_1,x_2,\ldots,x_n$. In what sense can we say that all these particles are identical? A natural idea is to consider a thought experiment in which we (adiabatically) exchange the positions of any two particles, say the ones at $x_i$ and $x_j$, and ask if it is physically~(experimentally) possible to distinguish the two quantum states before and after the exchange, i.e., $\ket{\Psi_{x_1,x_2,\ldots,x_n}}$ and $\hat{E}_{ij}\ket{\Psi_{x_1,x_2,\ldots,x_n}}$, where $\hat{E}_{ij}$ is the unitary exchange operator. 
Thus we first need to define quantum state indistinguishability. Depending on which class of physical observables we are able to measure in experiment, we can have different notions of state indistinguishability:
\begin{definition}\label{def:abs_indist_criterion}
Let $\ket{\Psi}$ and $\ket{\Phi}$ be two quantum states of a physical system. We say that $\ket{\Psi}$ and $\ket{\Phi}$ are \textit{absolutely indistinguishable} if they cannot be distinguished by any measurement, i.e., if
\begin{equation}\label{eq:abs_indist_criterion}
  \braket{\Psi|\hat{O}|\Psi}=\braket{\Phi|\hat{O}|\Phi} 
\end{equation}
for any observable (i.e., any Hermitian operator) $\hat{O}$. 
\end{definition}
Absolute indistinguishability of $\ket{\Psi}$ and $\ket{\Phi}$ forces them to be equal up to a phase $\ket{\Psi}=e^{i\phi}\ket{\Phi}$. Nevertheless, for most parts of this paper, we will be discussing identical particles in condensed matter physics and quantum field theories, where locality and symmetry fundamentally constrains which class of observable we are able to measure. It is therefore natural to incorporate locality and symmetry into the definition, leading to the following weaker notion of state indistinguishability
\begin{definition}\label{def:weak_indist_criterion}
We say that two quantum states $\ket{\Psi}$ and $\ket{\Phi}$ of a quantum many body system are \textit{locally indistinguishable}  if Eq.~\eqref{eq:abs_indist_criterion} holds for all local observables $\hat{O}$. Furthermore, let $G$ be a symmetry of the system, and let $\{\hat{U}(g)|g\in G\}$ be its representation on the state space. We say that $\ket{\Psi}$ and $\ket{\Phi}$ are $G$-protected indistinguishable if Eq.~\eqref{eq:abs_indist_criterion} holds for all $G$-symmetric observables $\hat{O}$, i.e., all observables $\hat{O}$ satisfying
\begin{equation}
[\hat{U}(g),\hat{O}]=0,\quad \forall g\in G. 
\end{equation}
Similarly, $G$-protected local indistinguishability is defined by combining the above two notions. 
\end{definition}
We are now ready to define indistinguishable particles in quantum many body systems:
\begin{definition}\label{def:weak_indistinguishable_particles}
	A quantum state $\ket{\Psi_{x_1,x_2,\ldots,x_n}}$ describes a system of [$G$-protected] (locally) indistinguishable
    particles at positions $x_1,x_2,\ldots,x_n$ if the states $\{\hat{E}_{ij}\ket{\Psi_{x_1,x_2,\ldots,x_n}}~|~i,j\in\{1,2,\ldots,n\}\}$ are [$G$-protected] (locally) indistinguishable.
    That is, if for all $i,j\in\{1,2,\ldots,n\}$, we have
	\begin{equation}\label{eq:local_indist}
\braket{\Psi|\hat{E}^\dagger_{ij}\hat{O}\hat{E}_{ij}|\Psi} =\braket{\Psi|\hat{O}|\Psi}, 
	\end{equation}
	 for all [$G$-symmetric] (local) observables $\hat{O}$.
\end{definition}
One can show that absolute indistinguishability requires that $\hat{E}_{ij}\ket{\Psi}\propto\ket{\Psi}$, which boils down to fermions or bosons~(and Abelian anyons in 2D). By contrast, local indistinguishability allows  
\begin{equation}\label{eq:EijPsitransform}
	\hat{E}_{ij}\ket{\Psi^I}=\sum_J (R_{ij})^I_J \ket{\Psi^J},
\end{equation}
where the internal index $I,J$ cannot be measured locally. The $G$-protected local indistinguishability criterion (iii) also allows Eq.~\eqref{eq:EijPsitransform}, but only requires that $I,J$ cannot be measured locally without breaking the $G$-symmetry.

We mention that being locally indistinguishable is the way how $R$-parastatistics gets around a more recent no-go theorem~\cite{mekonnen2025invariance} that rules out absolutely indistinguishable paraparticles~\footnote{In some sense, Ref.~\cite{mekonnen2025invariance} rules out $G$-protected indistinguishable paraparticles, where $G=S_n$~(the symmetric group) is the  particle exchange symmetry itself, or more generally, the so called ``quantum permutation symmetry''~\cite{mekonnen2025invariance}. In  
condensed matter and high energy physics, we normally do not treat particle exchange symmetry as a fundamental symmetry of a physical system, as it is quite different from a conventional symmetry. A conventional symmetry of a quantum many body system is a representation  $\{\hat{U}(g)|g\in G\}$ of a group $G$~(or a more general algebraic structure such as a Hopf algebra) on the Hilbert space of the quantum system that commutes with the Hamiltonian. Particle exchange symmetry
is a bit unconventional in that it is represented by an infinite sequence $\{\hat{U}_n(g)|g\in S_n\}_{n=0}^\infty$, where $\hat{U}_n(g)$ is the representation of $S_n$ on the $n$-particle subspace, i.e., it is not representation of one algebraic structure on the entire Hilbert space, but representations of an infinite sequence of groups, each acting on a subspace. 
}. 
All $R$-paraparticles described  by the second quantization formulation that we introduce later in this paper are locally indistinguishable rather than absolutely indistinguishable. In Ref.~\cite{wang2023para} and also in the second part of this series, we construct solvable quantum spin system~(in all dimensions) that host emergent $R$-paraparticles. In 2D and 3D, the emergent $R$-paraparticles in these models are locally indistinguishable, while in 1D, they are $G$-protected locally indistinguishable, where $G$ is generally a non-invertible symmetry represented by matrix product operators.

\subsection{Goal of this paper}\label{sec:intro-goal}
The goal of this paper is to supply the missing theoretical foundation of $R$-parastatistics. Our main focus is to investigate its observable content: what qualifies as a local observable, what aspects of an $R$-paraparticle can be physically measured, and how these local observables encode the relevant physical properties of $R$-paraparticles. 
To this end, we significantly extend the observable content of this theory, by giving a general definition of local observables in an $R$-paraparticle theory, and explicitly constructing families of local observables that go beyond  the basic family introduced in Ref.~\cite{wang2023para}. 
Specifically, we introduce local observables that distinguish particle types, and observables that probe internal indices at special point defects, both of which are fundamental to understanding local indistinguishability, and the former also motivates the concept of mutual parastatistics not discussed in Ref.~\cite{wang2023para}, while the latter is crucial for the experimental probes of exchange statistics and the proposed quantum-information applications~\cite{wang2024parastatistics,wang2025secret}. 
Another important family of local observables introduced in this paper consists of those that create or annihilate particle-antiparticle pairs,
which lead to an $R$-paraparticle analog of Bogoliubov-de Gennes mean-field theory of superconductors, and are also crucial for the construction of relativistic quantum field theories of $R$-paraparticles. We also prove a classification theorem for unitary $R$-matrices admitting local pair creation, in which physical concepts of topological twist factor, Frobenius-Schur indicator, and particle-hole symmetry emerge, connecting a subclass of $R$-paraparticles to the topological-order literature. Finally, we study hidden symmetries that transform the internal indices of $R$-paraparticles while preserving the local observable algebra. Such symmetries are naturally described by Hopf algebras and go beyond conventional group symmetries. In addition to their intrinsic interest, these generalized hidden symmetries help us prove several structural theorems on local indistinguishability of $R$-paraparticles. 

Our paper is organized as follows. 
In Sec.~\ref{sec:second_quantization} we review the second quantized theory of $R$-parastatistics %
proposed in Ref.~\cite{wang2023para}, and introduce a general definition of local observables, the concept of mutual parastatistics, and the notion of weak-equivalence between $R$-matrices.  In Sec.~\ref{sec:nontrivialstatistics} we discuss the  exchange statistics of $R$-paraparticles, %
emphasizing its observable physical consequences and the proposed applications in secret communication~\cite{wang2024parastatistics,wang2025secret}. 
In Sec.~\ref{sec:pair-creation-AP} we introduce a family of local observables responsible for local pair-creation, and classify the subclass of $R$-matrices that allow such observables, and discuss related concepts such as topological twist factor, Frobenius-Schur indicators, and particle-hole symmetry, and propose a $R$-parastatistical analog of Majorana fermion operators.  
In Sec.~\ref{sec:solution} we show how to exactly solve free $R$-paraparticle Hamiltonians, including an analog of Bogoliubov-type bilinear Hamiltonians with U(1)-breaking pair-creation/annihilation terms. In Sec.~\ref{sec:generalized_symmetry} we introduce generalized hidden symmetries of $R$-paraparticles, and use them as basis to discuss local indistinguishability and motivate a categorical description of $R$-paraparticles. We summarize the paper in Sec.~\ref{sec:summary_part1}. 

\section{The second quantized theory of $R$-paraparticles}\label{sec:second_quantization} %

In this section we review the second quantization formulation of $R$-parastatistics introduced in Ref.~\cite{wang2023para}, which is an important tool for the study of $R$-paraparticles in both condensed matter physics and high energy quantum field theories. 
Although the second quantization formulation only realizes a subfamily of the parastatistics defined by the first quantization formulation presented in Sec.~\ref{sec:what_is_para}~(see Sec.~\ref{sec:relation_first_quantization}  for the relation between the two formulations%
), it has the important advantage that it is naturally equipped with a locality structure crucial for condensed matter physics and quantum field theory,  
which is not naturally exhibited the first quantization formulation.
A key ingredient for this locality structure is the special family of representations of $S_N$ realized by the so-called $R$-matrices, which we discuss in Sec.~\ref{sec:YBE}. 
After this we present the second quantized theory of $R$-paraparticles with more technical details compared to Ref.~\cite{wang2023para}, and then in Sec.~\ref{sec:mutual_para} we introduce an extension of this formalism that describes mutual parastatistics. 

\subsection{Representations of $S_N$ and constant Yang-Baxter equations}\label{sec:YBE}
In the second quantization formulation, each type of parastatistics is associated to
an $R$-matrix, a four-index tensor $R^{ab}_{cd}$ with $1\leq a,b,c,d\leq m$~(here $m\in \mathbb{Z}$ is called the quantum dimension of the paraparticle) that realizes a solution to Eq.~\eqref{eq:YBE_SN} in the following way. First, we introduce an abstract tensor indexing notation that we will use throughout the paper: 
\begin{definition}{(Abstract tensor indexing)}\label{def:abstractTNindexing}
Any four index tensor $R^{ab}_{cd}$ defines a  linear map $R$ in the product vector space $\mathfrak{I}\otimes \mathfrak{I}$, where $\mathfrak{I}$ is a $m$-dimensional vector space with basis $\{v_1,v_2,\ldots, v_m\}$, and $R$ acts on the basis vectors as $R(v_c\otimes v_d)=\sum_{ab}R^{ab}_{cd}v_a\otimes v_b$. Now define $R_{j,j+1}$ as a  linear map on $\mathfrak{I}^{\otimes N}$ by
\begin{equation}\label{def:Rjjp1}
	R_{j,j+1}=\underset{{(1)}}{\mathds{1}}\otimes\ldots\otimes \underset{(j-1)}{\mathds{1}}\otimes \underset{(j,j+1)}{R}\otimes\underset{(j+2)}{\mathds{1}}\otimes\ldots\otimes \underset{(N)}{\mathds{1}},
\end{equation} 
i.e. $R_{j,j+1}$ applies $R$ on the $j$ and $j+1$-th factor spaces. 
\end{definition}
If we take the representation space of $S_N$ to be $\mathfrak{I}^{\otimes N}$ and define $R_{j,j+1}$ as in Eq.~\eqref{def:Rjjp1},
 then all the relations in Eq.~\eqref{eq:YBE_SN} reduce to the following two equations
\begin{eqnarray}\label{eq:YBE}
	R^2&=&\mathds{1}_{m^2}~~~~~~~~~~~~(\text{in } \mathfrak{I}\otimes \mathfrak{I}),\nonumber\\
	R_{12} R_{23} R_{12}&=& R_{23} R_{12} R_{23}~~~(\text{in } \mathfrak{I}\otimes \mathfrak{I}\otimes \mathfrak{I}).
\end{eqnarray}
These are the relations the tensor $R^{ij}_{kl}$ is required to satisfy. We also have the following tensor graphical representation of Eq.~\eqref{eq:YBE}:
\begin{equation}\label{eq:YBEgraphical}
	\begin{tikzpicture}[baseline={([yshift=-.8ex]current bounding box.center)}, scale=0.5]
		\Rmatrix{0}{\AL}{R}
		\Rmatrix{0}{-\AL}{R}
		\node  at (-\AL,2.5*\AL) {\footnotesize $a$};
		\node  at (\AL,2.5*\AL) {\footnotesize $b$};
		\node  at (-\AL,-2.5*\AL) {\footnotesize $c$};
		\node  at (\AL,-2.5*\AL) {\footnotesize $d$};
	\end{tikzpicture}=
	\begin{tikzpicture}[baseline={([yshift=-.8ex]current bounding box.center)}, scale=0.5]
		\draw[thick] (-\AL,-2*\AL) -- (-\AL,2*\AL);
		\draw[thick] (\AL,-2*\AL) -- (\AL,2*\AL);
		\node  at (-\AL,2.5*\AL) {\footnotesize $a$};
		\node  at (\AL,2.5*\AL) {\footnotesize $b$};
		\node  at (-\AL,-2.5*\AL) {\footnotesize $c$};
		\node  at (\AL,-2.5*\AL) {\footnotesize $d$};
		\node  at (-1.5*\AL,0*\AL) {\footnotesize $\delta$};
		\node  at (1.5*\AL,0*\AL) {\footnotesize $\delta$};
	\end{tikzpicture}~,~~~
	\begin{tikzpicture}[baseline={([yshift=-.8ex]current bounding box.center)}, scale=0.5]
		\Rmatrix{-\AL}{2*\AL}{R}
		\Rmatrix{\AL}{0}{R}
		\Rmatrix{-\AL}{-2*\AL}{R}
		\draw[thick] (-2*\AL,-\AL) -- (-2*\AL,\AL);
		\draw[thick] (2*\AL,\AL) -- (2*\AL,3*\AL);
		\draw[thick] (2*\AL,-\AL) -- (2*\AL,-3*\AL);
		\node  at (-2*\AL,3.5*\AL) {\footnotesize $a$};
		\node  at (0*\AL,3.5*\AL) {\footnotesize $b$};
		\node  at (2*\AL,3.5*\AL) {\footnotesize $c$};
		\node  at (-2*\AL,-3.7*\AL) {\footnotesize $d$};
		\node  at (0*\AL,-3.7*\AL) {\footnotesize $e$};
		\node  at (2*\AL,-3.7*\AL) {\footnotesize $f$};
	\end{tikzpicture}
	=
	\begin{tikzpicture}[baseline={([yshift=-.8ex]current bounding box.center)}, scale=0.5]
		\Rmatrix{\AL}{2*\AL}{R}
		\Rmatrix{-\AL}{0}{R}
		\Rmatrix{\AL}{-2*\AL}{R}
		\draw[thick] (2*\AL,-\AL) -- (2*\AL,\AL);
		\draw[thick] (-2*\AL,\AL) -- (-2*\AL,3*\AL);
		\draw[thick] (-2*\AL,-\AL) -- (-2*\AL,-3*\AL);
		\node  at (-2*\AL,3.5*\AL) {\footnotesize $a$};
		\node  at (0*\AL,3.5*\AL) {\footnotesize $b$};
		\node  at (2*\AL,3.5*\AL) {\footnotesize $c$};
		\node  at (-2*\AL,-3.7*\AL) {\footnotesize $d$};
		\node  at (0*\AL,-3.7*\AL) {\footnotesize $e$};
		\node  at (2*\AL,-3.7*\AL) {\footnotesize $f$};
	\end{tikzpicture},
\end{equation}
where $R^{ab}_{cd}=\!\!
\begin{tikzpicture}[baseline={([yshift=-.6ex]current bounding box.center)}, scale=0.45]
	\Rmatrix{0}{0}{R}
	\node  at (-1.5*\AL,\AL) {\footnotesize $a$};
	\node  at (1.5*\AL,\AL) {\footnotesize $b$};
	\node  at (-1.5*\AL,-\AL) {\footnotesize $c$};
	\node  at (1.5*\AL,-\AL) {\footnotesize $d$};
\end{tikzpicture}$,
and throughout this paper we use tensor graphical notation where open indices are identified on both sides of the equation and contracted indices are summed over, and a line segment represents a Kronecker $\delta$ function.

With $R_{j,j+1}$ defined in Eq.~\eqref{def:Rjjp1}, the $N$-particle state $\mpwf{\Psi}{x_1,x_2,\ldots, x_N}{I}$ in Eq.~\eqref{eq:wavefuntion_exchange_para} has $N$ auxiliary indices $I=(a_1,a_2,\ldots,a_N)$. Under two particle exchange, the state transforms in the following way~(we show the case of $N=3$ where $x_1,x_2$ are exchanged, but generalization to arbitrary $N$ is straightforward)
\begin{equation}\label{eq:wavefuntion_exchange_Rmat}
	\mpwf{\Psi}{x_2,x_1,x_3}{a_1, a_2, a_3}=\sum_{b_1,b_2} R_{a_1a_2}^{b_1 b_2 } \mpwf{\Psi}{x_1,x_2,x_3}{b_1, b_2, a_3},
\end{equation} 
which is obtained by inserting Eq.~\eqref{def:Rjjp1} into Eq.~\eqref{eq:wavefuntion_exchange_para}. %

Before giving examples, we mention a natural equivalence relation between $R$-matrices. We say that two unitary $R$-matrices $R$ and $R'$ are equivalent, denoted by $R\cong R'$, if there exists a unitary matrix $V$ such that
\begin{equation}\label{eq:basistransformRmat}
 R'=%
 (V\otimes V)^{-1}
 R(V\otimes V).
\end{equation}
(For equivalence between non-unitary $R$-matrices, we drop the unitarity condition on $V$.) 
Indeed, if $R$ and $R'$ are related by Eq.~\eqref{eq:basistransformRmat}, then the associated paraparticle theories are related by a basis transformation~(represented by $V$) in the internal space~[note, in particular, that the YBE~\eqref{eq:YBE} is invariant under this basis rotation], and consequently their physical predictions are the same, since all physical properties of a certain type of $R$-paraparticle are independent of the choice of basis in the internal space. In this case, we call the two theories physically equivalent. 
In Sec.~\ref{sec:weak-equivalence}, 
we will introduce a weaker equivalence relation between $R$-matrices which still preserves a subset of physical properties but does not lead to full physical equivalence. 

In summary, every tensor $R^{ab}_{cd}$ satisfying Eq.~\eqref{eq:YBE} uniquely characterizes a type of parastatistical particles. Let us take a closer look at this important equation. The second line of Eq.~\eqref{eq:YBE} is known in the mathematical physics literature as the constant YBE~(whose solutions are called $R$-matrices), which appears in diverse areas including representation theory of the braid group $B_N$~\cite{kohno1987monodromy,Wenzl1990Representation,ZHANG1991625}, theory of knot invariants~\cite{Turaev1988,Kauffman_knot_book}, quasitriangular Hopf algebras~\cite{drinfeld1986quantum,Majid1990,klimyk1997book,kasselQuantumGroups1995}, and quantum computing~\cite{Kauffman_2002,dye2003unitary}. A closely related version of the YBE~(with a spectral parameter) plays a fundamental role in Bethe ansatz solvable models~\cite{Jimbo1989,faddeev1996algebraic,korepinQuantumInverseScattering1993,baxter2016exactly}. $R$-matrices that further satisfy the first line of Eq.~\eqref{eq:YBE} are called ``involutive'', and have also been studied~\cite{etingof1999set,HIETARINTA1992245,RUMP2007153,cedo2010involutive,cedo2014braces,guarnieri2017skew,SMOKTUNOWICZ201886,LECHNER2019106769}.

In the following we present some basic examples of $R$-matrices, including ones that lead to non-trivial parastatistics. Notice that the integer $m\geq 1$ is a free parameter for all the following examples~(except Example~\ref{ex:setth}), and it is always assumed that $ 1\leq a,b,c,d\leq m$. That they are solutions to Eq.~\eqref{eq:YBE} can be checked by straightforward computation. \\
\textbf{Example \customlabel{ex:decoupled}{1}}~(trivial).
$R^{ab}_{cd}=\pm\delta_{ad}\delta_{bc}$. Eq.~\eqref{eq:wavefuntion_exchange_Rmat} becomes
$$	\mpwf{\Psi}{x_2,x_1,x_3}{a_2, a_1, a_3}=\pm \mpwf{\Psi}{x_1,x_2,x_3}{a_1, a_2, a_3}.$$
But this %
type of exchange behavior is the same as fermions~($-$) or bosons~($+$) with an internal degree of freedom~(such as spin or flavor), and importantly, when two particles are exchanged, these internal indices are carried along with the particles without changing. We therefore call this type of $R$-matrices and the associated $R$-parastatistics  trivial. [Throughout this paper, whenever we say that a certain aspect of an $R$-matrix~(such as exclusion or exchange statistics) is trivial, we mean that aspect of the $R$-matrix in question is the same as in this example.]   %
\\ 
\textbf{Example \customlabel{ex:Green}{2}}~(Green's parastatistics~\cite{Green1952}).
$R^{ab}_{cd}=\pm\delta_{ad}\delta_{bc}(-1)^{\delta_{ab}}$. Eq.~\eqref{eq:wavefuntion_exchange_Rmat} becomes
$$	
\mpwf{\Psi}{x_2,x_1,x_3}{a_2, a_1, a_3}=\pm (-1)^{\delta_{a_1a_2}} \mpwf{\Psi}{x_1,x_2,x_3}{a_1, a_2, a_3}.
$$
This  essentially leads to Green's parastatistics~\cite{Green1952}~(see App.~\ref{app:Green-theory} for the precise connection). %
Intuitively, for the ($+$) sign, we get $m$ flavors of fermions that are mutually symmetric~(order-$m$ para-fermions in Green's formulation), while for the ($-$) sign, we get $m$ flavors of bosons that are mutually antisymmetric~(order-$m$ parabosons).
However, %
interpreting the extra indices $a_1,a_2,a_3$ as a flavor index glosses over the important issue of local indistinguishability discussed in Sec.~\ref{sec:indistcrit}, as it is often understood that particle flavor can be locally measured while paraparticles are assumed to be locally indistinguishable. We will come back to the discussion of local indistinguishability later in Sec.~\ref{sec:local-indist-revisit}. 
\\
\textbf{Example \customlabel{ex:1m}{3}}
$R^{ab}_{cd}=\pm\delta_{ac}\delta_{bd}$
, or $R=\pm\mathds{1}_{m^2}$.  Eq.~\eqref{eq:wavefuntion_exchange_Rmat} becomes
$$	%
\mpwf{\Psi}{x_2,x_1,x_3}{a_1, a_2, a_3}=\pm \mpwf{\Psi}{x_1,x_2,x_3}{a_1, a_2, a_3}.
$$
This will be our primary example of unitary $R$-matrices with non-trivial exclusion statistics, as we show later in Sec.~\ref{sec:exclusion_statistics_calc}. \\
\textbf{Example \customlabel{ex:1m1}{4}}
For $m\geq 2$, let
\begin{eqnarray}\label{eq:example_1n1Rmatrix}
	R^{ab}_{cd}=\lambda_{ab}\xi_{cd}-\delta_{ac}\delta_{bd},
\end{eqnarray}
where $\lambda, \xi$ are $m\times m$ constant matrices satisfying 
\begin{equation}\label{eq:lambdac}
	\lambda \xi \lambda^T \xi^T=\mathds{1}_{m}, ~~~ \mathrm{Tr}(\lambda \xi^T)=2.
\end{equation}
Eq.~\eqref{eq:lambdac} guarantees that $R$ satisfies Eq.~\eqref{eq:YBE}. 
For example, we can take $\lambda=e^{-M},\xi=-e^M$, where $M$ is an $m\times m$ antisymmetric matrix $M^T=-M$ with complex entries satisfying $\mathrm{Tr}[e^{-2M}]=-2$~(one can get an explicit solution using a block-diagonal ansatz for $M$ with maximum block size $2$). 
For this example, we are mostly interested in the case $m\geq 3$, where it gives examples of $R$-matrices with nontrivial exclusion statistics that allow pair creation terms introduced later in Sec.~\ref{sec:U1breaking}. A caveat, though, is that all solutions to Eq.~\eqref{eq:lambdac} lead to non-unitary $R$-matrices for $m\geq 3$~\footnote{For $m=2$, Eq.~\eqref{eq:lambdac} is satisfied if $\xi=\lambda^*$ is unitary. This does leads to a unitary $R$-matrix, but its exclusion statistics $z_R(x)=1+mx+x^2=(1+x)^2$ is trivial.}.   
\\
\textbf{Example \customlabel{ex:setth}{5}}%
Define an $R$-matrix with $m=4$ as follows: let $S=\{1,2,3,4\}$ and $r:S\times S\to S\times S$ be an bijective map defined as 
\begin{equation}\label{eq:seth-R}
	r(a,b)=
	\left(
	\begin{array}{cccc}
		43 & 12 & 24 & 31 \\
		21 & 34 & 42 & 13 \\
		14 & 41 & 33 & 22 \\
		32 & 23 & 11 & 44 \\
	\end{array}
	\right)_{ab},
\end{equation}
where we use $ab$ as a short hand for $(a,b)$. For example, $r(1,1)=(4,3)$ and $r(3,2)=(4,1)$. The map $r$ in Eq.~\eqref{eq:seth-R} satisfies the set-theoretical YBE~\cite{etingof1999set}
\begin{equation}\label{eq:sethYBE}
	r^2=\mathrm{id}_{S\times S},\quad	r_{12}r_{23}r_{12}=r_{23}r_{12}r_{23},
\end{equation}
where in the second equation both sides are bijective maps from the set $S\times S\times S$ to itself,  $r_{12}=r\times \mathrm{id}_S$, and $r_{23}=\mathrm{id}_S\times r$. 
Now we define the $R$-matrix as
\begin{equation}\label{eq:set-thR}
	R\ket{a,b}=-\ket{b',a'},~~~\forall a,b\in S, 
\end{equation}
where $(b',a')=r(a,b)$. 
It then follows from 
Eq.~\eqref{eq:sethYBE} that $R$ satisfies the YBE~\eqref{eq:YBE}. We will later use this as a primary example to illustrate the quantum information aspects of $R$-paraparticles in Sec.~\ref{sec:QI_statistics}.\\
\textbf{Example \customlabel{ex:setth-ext}{6}}
Define an $R$-matrix with arbitrary $m$ by:
\begin{equation}\label{eq:RA4Z3}
	R\ket{a,b}=\varphi(a,b)\ket{b+1,a-1},
\end{equation}
where $a,b\in\{1,2,\ldots,m\}$ are understood modulo m, and $\varphi(a,b)=(-1)^{\delta_{a,b+1}}$. This $R$-matrix can be understood as a set-theoretical solution similar to Ex.~\ref{ex:setth}, but twisted by a phase factor $\varphi(a,b)$. %

The above six families of $R$-matrices are shown in Tab.~\ref{tab:Hilbert_series}, along with a summary of the key physical properties of the associated paraparticle theory. 
Many more %
examples of $R$-matrices are known in the literature~\cite{etingof1999set,HIETARINTA1992245,RUMP2007153,cedo2010involutive,cedo2014braces,guarnieri2017skew,SMOKTUNOWICZ201886,LECHNER2019106769}. However,  up to now there is not yet a complete classification of all possible solutions to the YBE~\eqref{eq:YBE}, although there is a classification of unitary $R$-matrices up to an equivalence relation~\cite{LECHNER2019106769}~(see Thm.~\ref{thm:classification_Herm_R}), and in Sec.~\ref{sec:pair-creation-AP} we give a  classification of $R$-matrices describing self-dual paraparticles~(paraparticles that are their own antiparticles).   
In Sec.~\ref{sec:cat-description} we describe a simple method to construct new $R$-matrices out of existing ones via fusion~\cite{Turaev1988,Wenzl1990Representation}. 

\begin{table*}
	\centering
	{\renewcommand{\arraystretch}{1.5}
		\begin{tabular}{|c|c|c|c|c|c|c|}
			\hline
			Ex. & \ref{ex:decoupled}& 	\ref{ex:Green} & \ref{ex:1m}  & \ref{ex:1m1} & \ref{ex:setth} & \ref{ex:setth-ext} \\
			\hline
			$R^{ab}_{cd}$ & $-\delta_{ad}\delta_{bc}$ & $\delta_{ad}\delta_{bc}(-1)^{\delta_{ab}}$&$-\delta_{ac}\delta_{bd}$&$\lambda_{ab}\xi_{cd}-\delta_{ac}\delta_{bd}$& $-\braket{a,b|r(c,d)}$ & $\varphi(c,d)\delta_{a,d+1}\delta_{b,c-1}$\\
			\hline
			$z_R(x)$  & $(1+x)^m$	 & $(1+x)^m$ &$1+m x$ &$1+mx+x^2$& $(1+x)^4$& $(1+x)^m$\\
			\hline
			Nontrivial Exclusion    & \No & \No & \Yes & \Yes & \No & \No \\
			\hline
			Nontrivial Exchange    & \No & \Yes & \Yes & \Yes & \Yes & \Yes\\
			\hline
			$d>1$ QSM    & \Yes & \Yes & \No & \No & \Yes & \Yes\\
			\hline
            $D_R$    & $1$ & $2^{\floor{m/2}}$ & $\infty$ & $\infty$ & $8$ & finite\\
			\hline
            $R$ is simple    & \No & \No & \Yes & \Yes & \Yes & \Yes \\
			\hline
Example gauge group & $Z_2,Q_{4n}$ & $D_8,A_4, Z_m\ltimes Z^{\times m}_2$ &\No & \No &  $D_8\ltimes Z_2^{\times 3}$  & $m=3:~A_4\times Z_3$\\
            \hline
			Pair creation & \Yes & \Yes & \No & \Yes & \Yes & \Yes\\
            \hline 
            Twist factor $\theta$ & $-1$ & $-1$ & - & - & $-1$ & $-1$\\
			\hline 
			$(\alpha,\eta_\alpha,\nu)$ & $(s,-1,1)$ & $(d,-1,1)$ & - & $(\xi,+1,\text{-})$ & $(\pi_{12},-1,1)$ & Not self-dual \\
			                               & $(a,+1,-1)$ &  &  &  &  & \\
			\hline
			Relativistic QFT & \Yes & \Yes& \No & \No & \Yes & \Yes\\
			\hline
			Comment    & fermions & Essentially Green~\cite{Green1952} &  & Non-unitary &   &  \\
			\hline
		\end{tabular}
	}
	\caption{\label{tab:Hilbert_series} Examples of involutive $R$-matrices and a summary of the key properties of the different types of paraparticle they define.   $z_R(x)$ is the Hilbert series of $R$ defined in Eq.~\eqref{eq:hilbert_series}, which physically corresponds to the single mode partition function via Eq.~\eqref{eq:single_mode_Z}. 
    The row $d>1$ QSM indicates whether the $R$-paraparticle is realizable in solvable quantum spin models~(QSMs) in dimension $d>1$. $D_R$ is the minimal twisting dimension of $R$ defined in Sec.~\ref{sec:weak-equivalence}, and $D_R=\infty$ means that the $R$-paraparticle is not weakly equivalent to ordinary fermions and bosons, hence beyond the description of symmetric fusion categories.  
    Example gauge group gives an example of the discrete gauge group $G$ such that the $R$-paraparticle can be realized in 3D deconfined $G$-gauge theory~(here ``\No'' means Exs.~\ref{ex:1m} and \ref{ex:1m1} cannot be realized in deconfined gauge theories).  
    The ``pair creation'' row indicates whether the theory allows pair creation terms introduced in Sec.~\ref{sec:U1breaking}. 
    The topological twist factor $\theta$ and the Frobenius-Schur indicator $\nu$ are defined in Thm.~\ref{thm:unitary_simple_self-dualR}, while
    $\alpha$ and $\eta_\alpha$ are defined in Eqs.~(\ref{eq:Ralpha_def-graphical}) and (\ref{eq:eta_R}), respectively. 
    In this row, $s$, $a$, and $d$ represent any $m\times m$ symmetric, antisymmetric, and diagonal matrix, respectively, and $\pi_{12}$ is the $4\times 4$ permutation matrix that permutes 1 and 2 while leaving 3 and 4 unchanged. %
	}
\end{table*}

\subsection{The second quantization algebra with $R$-matrix commutation relations}
Given any involutive $R$-matrix, we define the second quantization formalism of the corresponding $R$-paraparticles by introducing the following commutation relations~(CRs) between the paraparticle creation and annihilation operators $\hat{\psi}^\pm_{i,a}$
\begin{eqnarray}\label{eq:fundamental_Rcommu}
	\hat{\psi}^-_{i,a} \hat{\psi}^+_{j,b}&=&\sum_{cd}R^{ac}_{bd} \hat{\psi}_{j, c}^+ \hat{\psi}_{i,d}^-+\delta_{ab}\delta_{ij},\nonumber\\%
	\hat{\psi}^+_{i,a} \hat{\psi}^+_{j,b}&=&\sum_{cd}R^{cd}_{ab} \hat{\psi}_{j,c}^+ \hat{\psi}_{i,d}^+,\nonumber\\%
	\hat{\psi}^-_{i,a} \hat{\psi}^-_{j,b}&=&\sum_{cd}R^{ba}_{dc} \hat{\psi}_{j,c}^- \hat{\psi}_{i,d}^-.%
\end{eqnarray}
Here $i,j$ are called mode indices, which can label particle position or momentum. Other locally observable properties of the paraparticles, such as spin and flavor, should also be absorbed into $i,j$. By contrast, $a,b,c,d$ are internal indices which are assumed to be locally unobservable~(we will clarify this later in Sec.~\ref{sec:observe_parastatistics} and Sec.~\ref{sec:local-indist-revisit}).
Later in Sec.~\ref{sec:state_space} we explicitly define the action of $\hat{\psi}^\pm_{i,a}$ on the Fock space. 
Notice that the special case $R^{ab}_{cd}=\pm\delta_{ad}\delta_{bc}$ %
gives back fermions~($-$) and bosons~($+$) with an internal degree of freedom. 
While the second quantization algebra in Eq.~\eqref{eq:fundamental_Rcommu} is consistently defined for any $R$-matrix satisfying Eq.~\eqref{eq:YBE},
in this paper we mainly focus on unitary $R$-matrices, i.e.,  $\sum_{a,b} R^{ab}_{cd} (R^{ab}_{ef})^*=\delta_{ce}\delta_{df}$, which is satisfied by all examples in Tab.~\ref{tab:Hilbert_series} except Ex.~\ref{ex:1m1}. The restriction to unitary $R$-matrices will become physically natural later in Sec.~\ref{sec:physical_interpretation} when we connect the $R$-matrix to the exchange statistics of paraparticles.
With a unitary $R$, we will see  $\hat{\psi}^+_{i,a}=(\hat{\psi}^-_{i,a})^\dagger$ after we define the action of $\hat{\psi}^\pm_{i,a}$ on the state space in Sec.~\ref{sec:action_parafield_basis}. 

It is sometimes useful to adopt a tensor graphical representation of the CRs in Eq.~\eqref{eq:fundamental_Rcommu} to better visualize its algebraic structure and simplify calculations. This graphical representation is shown in Fig.~\ref{fig:RMQA}. We also find it useful to suppress the auxiliary indices and write Eq.~\eqref{eq:fundamental_Rcommu} in the following way
\begin{eqnarray}\label{eq:fundamental_Rcommu_supress}
	\hat{\psi}^-_{i} \hat{\psi}^+_{j}&=&\hat{\psi}_{j}^+ \hat{\psi}_{i}^-\cdot\Rmp+\delta_{ij}\delta,\nonumber\\%
	\hat{\psi}^+_{i} \hat{\psi}^+_{j}&=&\hat{\psi}_{j}^+ \hat{\psi}_{i}^+ \cdot R,\nonumber\\
	\hat{\psi}^-_{i} \hat{\psi}^-_{j}&=&\hat{\psi}_{j}^- \hat{\psi}_{i}^-\cdot \Rmm. %
\end{eqnarray}
\begin{figure}
	\center{\includegraphics[width=0.8\linewidth]{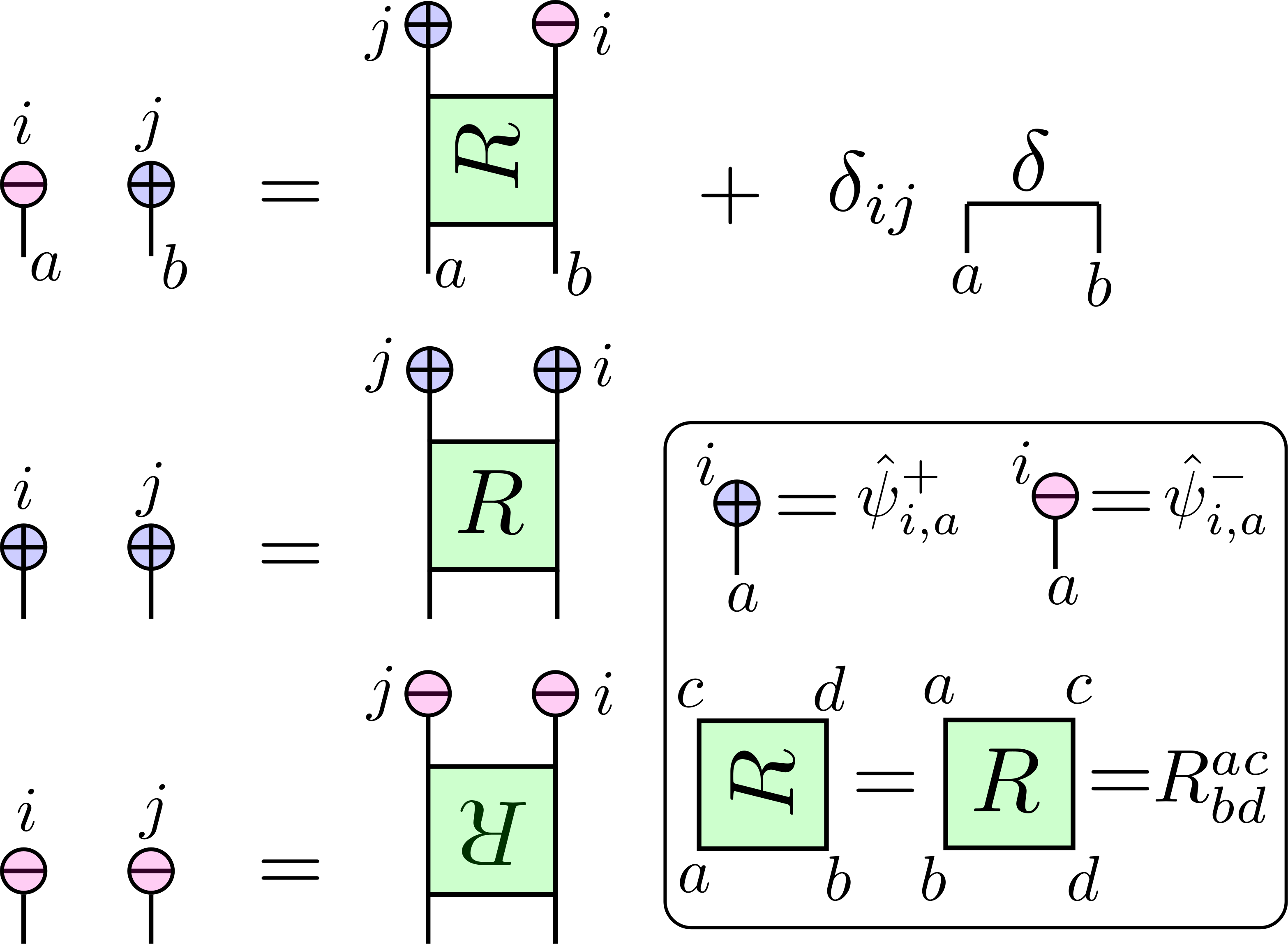}}
	\caption{\label{fig:RMQA} Tensor graphical representation of the quadratic CRs in Eq.~\eqref{eq:fundamental_Rcommu}. Open indices are identified~(in order from left to right) on both sides of the equation, while contracted indices are summed over.   Operator multiplications of $\hat{\psi}^\pm_{o,a}$ and $\hat{\psi}^\pm_{j,b}$ are ordered from left to right. }%
\end{figure}

We finally mention that the second quantization algebra in Eq.~\eqref{eq:fundamental_Rcommu} is a special case of quantum Weyl algebras~\cite{GIAQUINTO1995QWA}, and also a special case of Zamolodchikov-Faddeev algebras~\cite{zamolodchikovFactorizedSmatricesTwo1979,sklyaninQuantummechanicalApproachCompletely2011,lechnerFockRepresentationsZF2020,lechnerAlgebraicConstructiveQuantum2015,ZamolodchikovFaddeevAlgebra}, and closely related to the more general class of Koszul algebras~\cite{priddyKoszulResolutions1970,polishchuk2005quadratic}. However, to our knowledge, only this special class defined by Eq.~\eqref{eq:fundamental_Rcommu} with a multimode structure and an involutive $R$-matrix can physically describe parastatistics and lead to a free-particle theory~(see, however, Ref.~\cite{sanchez2025reconstruction} for a recent approach to extend Eq.~\eqref{eq:fundamental_Rcommu} to a larger class of Koszul algebras while preserving a simple Fock space structure and free particle solvability).

\subsection{Physical observables and Hamiltonians}\label{sec:locality}
Having defined the algebra of $R$-paraparticle operators, the next important step is to define physical observables--an important ingredient of a quantum theory that determines its physical content~(predicting power). In defining physical observables, one has to respect locality~(causality) principle, which is a fundamental principle in condensed matter physics and quantum field theory that requires spatially~(space-like) separated local operators to commute. More specifically, in %
a relativistic quantum field theory,
for any two local operators $\hat{O}_{S_1}$ and $\hat{O}_{S_2}$ supported on bounded regions $S_1$ and $S_2$, respectively, we have
\begin{equation}\label{eq:OS1S2commute}
[\hat{O}_{S_1},\hat{O}_{S_2}]=0,
\end{equation}
if $S_1$ and $S_2$ are space-like separated. %
This is enough to guarantee causality, i.e., %
no signal can travel faster than light. 
In non-relativistic condensed matter systems, Eq.~\eqref{eq:OS1S2commute} is required to hold for any spatially separated regions $S_1\cap S_2=\emptyset$. In this case, 
Ref.~\cite{wang2020tightening} proved that as long as a Hamiltonian is algebraically local and all local terms have uniformly bounded norm, then a Lieb-Robinson bound~\cite{liebFiniteGroupVelocity1972} holds~(as well as many other results based on it, e.g. gapped ground states have exponentially decaying correlations~\cite{hastings2006,nachtergaele2006}), which gives an effective lightcone of causality. 

To guarantee the locality principle, in the $R$-paraparticle theory, we define local observables in the following way:
\begin{definition}\label{def:general_LO}
Let $S$ be a bounded region of space, and let $\bar{S}$ be its complement~(in the relativistic case, $S$ is a bounded region of space-time, and $\bar{S}$ is its causal complement~\cite{haag2012book}). Denote by $\FA_S$ the algebra generated by $\{\hat{\psi}^\pm_{i,a}|i\in S,1\leq a\leq m\}$.
A local observable $\hat{O}$ in $S$ is an element of $\FA_S$ that commutes with all $\hat{\psi}^\pm_{j,a}$ for all $j\in \bar{S}$, i.e., $\hat{O}\in\FA_S$ satisfies~\footnote{We note that an even more expansive definition of local observable is possible, since in principle one requires only $[\hat{O},\hat{O}']=0$ when ${\hat O}$ and ${\hat O}'$ have non-overlapping support. 
In the rest of this paper we will not consider these more general classes of local observables for simplicity. }
\begin{equation}\label{eq:generalLOcondition}
[\hat{O},\hat{\psi}^\pm_{j,a}]=0,~\forall j\in \bar{S},
\end{equation}
and in this case we call $\hat{O}$ a local observable if it is Hermitian, i.e., $\hat{O}^\dagger=\hat{O}$. 
\end{definition}
It is clear that local observables defined this way satisfies the locality principle in Eq.~\eqref{eq:OS1S2commute}, since one can always decompose $\hat{O}_{S_2}$ as a sum of products of $\hat{\psi}^\pm_{j,a}$ for $j\in S_2$ and apply Eq.~\eqref{eq:generalLOcondition} with $\hat{O}=\hat{O}_{S_1}$. 
As a simple example, 
in the usual case of ordinary fermions $R=-1$, Eq.~\eqref{eq:generalLOcondition} simply requires that $\hat{O}$ is composed of even products of fermionic operators, which recovers the usual definition of local observables in fermionic quantum systems. 

The explicit construction of local observables in a general $R$-paraparticle theory is more involved. In this section, we focus on a relatively simple family, called the basic family of local observables, generated by 
contracted bilinear operators %
\begin{equation}\label{eq:def_e_ab}
	\hat{e}_{ij}\equiv \sum^m_{a=1} \hat{\psi}^+_{i,a}\hat{\psi}^-_{j,a}.
\end{equation}
\begin{definition}\label{def:phys_obs_basic}
For each bounded region of space $S$, the basic family of local observables on $S$ are Hermitian elements in the algebra generated by $\{\hat{e}_{ij}|i,j\in S\}$.   %
\end{definition}
As an example, $\hat{O}_{S}=\hat{e}_{ij}\hat{e}_{ji}$ with $i,j\in S$ is a local observable in $S$, since $\hat{e}^\dagger_{ij}=\hat{e}_{ji}$. Soon in Sec.~\ref{sec:bondLA} we will show that the basic family of local observables satisfy the aforementioned locality condition in  Eq.~\eqref{eq:generalLOcondition}, and hence Eq.~\eqref{eq:OS1S2commute}. %
Later in Sec.~\ref{sec:pair-creation-AP}, we will consider another important family of local observables involving particle-antiparticle pair creation, and then in 
Sec.~\ref{sec:local-indist-revisit} 
we will study the algebra of all allowed local observables as a whole, and discuss the local indistinguishability of $R$-paraparticles in the second quantization formalism. 

Once we have defined local observables, we can move on to define local Hamiltonians. A local Hamiltonian $\hat{H}$ is defined to be a sum of local observables $\hat{H}=\sum_{S} h_{S} \hat{O}_S$, where $h_{S}\in \mathbb{R}$ and the summation is over local regions $S$ whose diameters are smaller than some constant cutoff independent of the system size $N$. This definition of local observables and Hamiltonians guarantees unitarity, i.e., time evolution $\hat{U}=e^{-i\hat{H}t}$ generated by a Hamiltonian operator $\hat{H}$ is unitary~\footnote{This argument assumes that the $R$-matrix is unitary. If not, then in general we do not have $\hat{\psi}^+_{i,a}=(\hat{\psi}^-_{i,a})^\dagger$. Nevertheless, even for non-unitary $R$-matrices we can still define a suitable Hermitian norm on the state space such that $\hat{e}^\dagger_{ij}=\hat{e}_{ji}$, for $1\leq i,j\leq N$~(see Sec.~S2.D.5 of the supplementary information of Ref.~\cite{wang2023para}). Then we can use the same construction of basic local observables here, which still guarantees unitarity of the quantum theory.}.

\subsubsection{The Lie algebra of contracted bilinear operators}\label{sec:bondLA}
We now show that the basic family of local observables generated by $\{\hat{e}_{ij}\}_{1\leq i,j\leq N}$~(Definition~\ref{def:phys_obs_basic}) satisfies the locality condition in Eq~\eqref{eq:generalLOcondition}. 
To this end we derive the following CR between $\hat{e}_{ij}$ and $\hat{\psi}^\pm_{k,b}$:
\begin{eqnarray}\label{eq:commu_Eab_psi_p}
	[\hat{e}_{ij}, \hat{\psi}^+_{k,b}]&=&\delta_{jk}\hat{\psi}^+_{i,b},\nonumber\\
	{}[\hat{e}_{ij}, \hat{\psi}^-_{k,b}]&=&-\delta_{ik}\hat{\psi}^-_{j,b}.
\end{eqnarray}
The first line of Eq.~\eqref{eq:commu_Eab_psi_p} is derived as follows
\begin{eqnarray}\label{eq:commu_Eab_psi_p-supp}
	[\hat{e}_{ij}, \hat{\psi}^+_{k,b}]&=&
    \adjincludegraphics[width=9ex,valign=c]{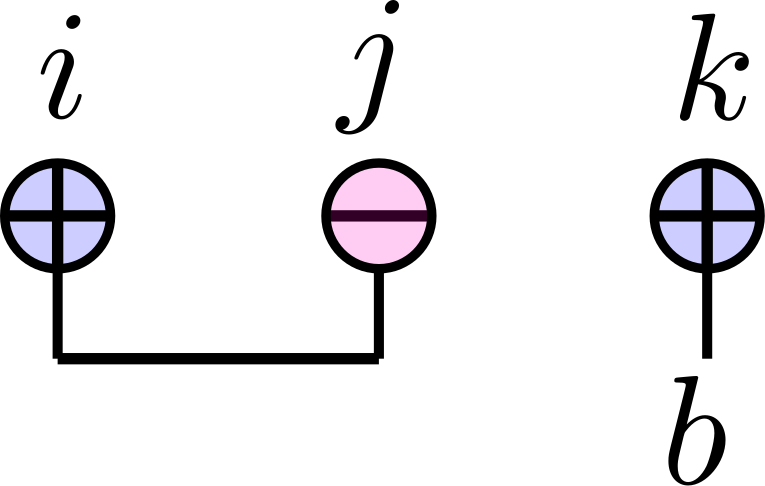}~~-~~
    \adjincludegraphics[width=9ex,valign=c]{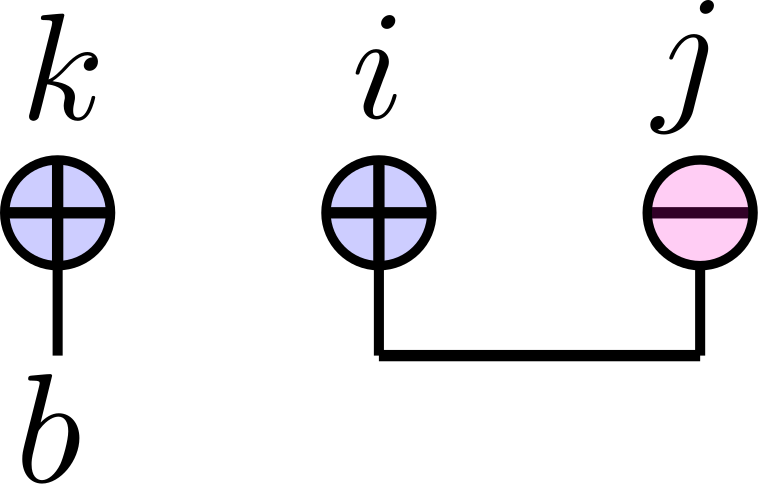}\nonumber\\
	&=&
    \adjincludegraphics[width=9ex,valign=c]{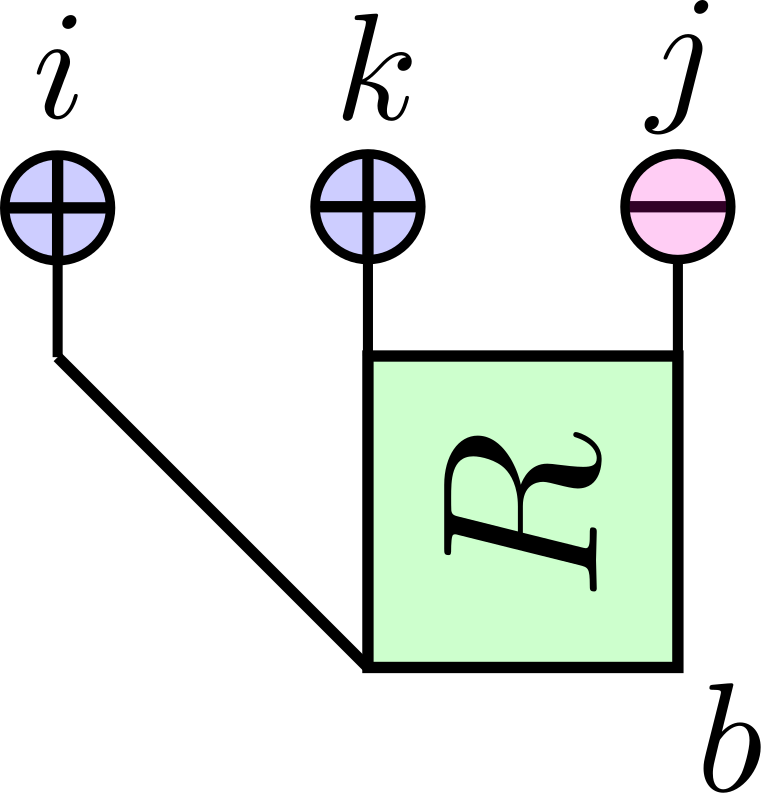}~~+~~
    \adjincludegraphics[width=9ex,valign=c]{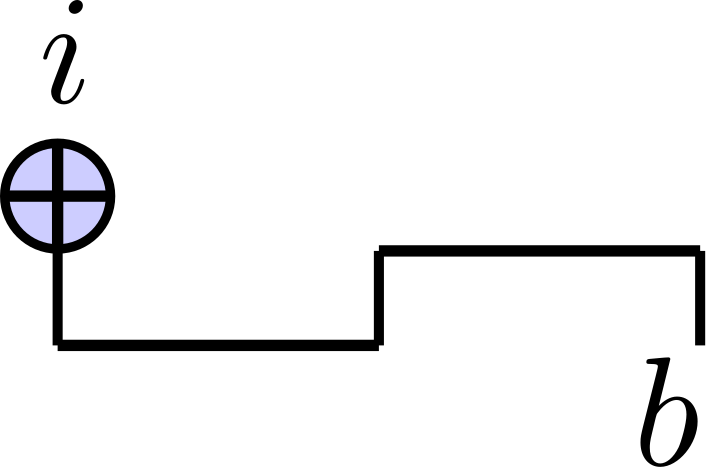}~\delta_{jk}
	-~~
    \adjincludegraphics[width=9ex,valign=c]{Figures/LA-1/epsi_deriv-1-2.png}\nonumber\\
	&=&
    \adjincludegraphics[width=9ex,valign=c]{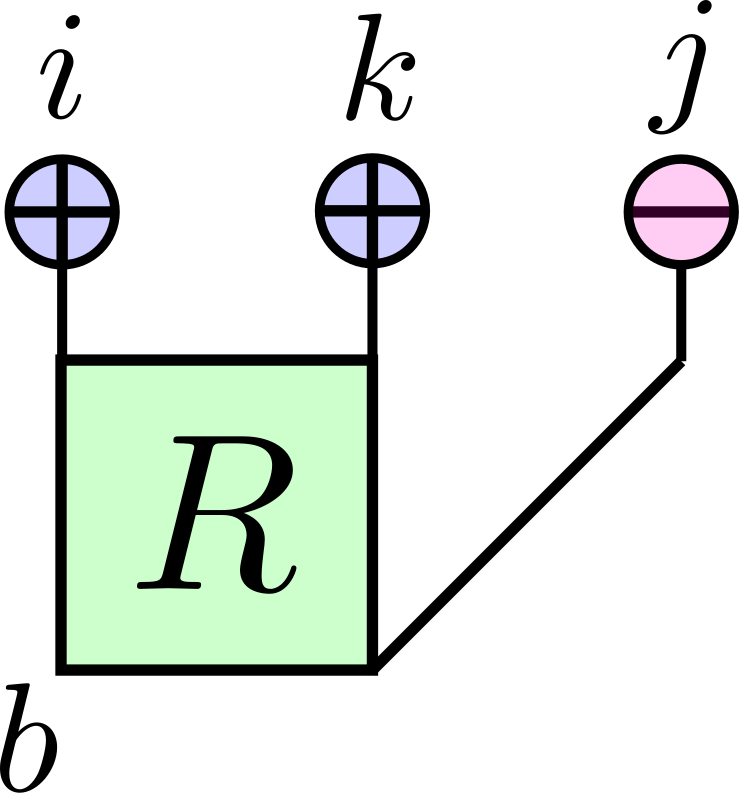}~~-~~
    \adjincludegraphics[width=9ex,valign=c]{Figures/LA-1/epsi_deriv-1-2.png}~~+\delta_{jk}\hat{\psi}^+_{i,b}\nonumber\\
	&=&\delta_{jk}\hat{\psi}^+_{i,b},
\end{eqnarray}
where in the second~(last) line we used the first~(second) line of Eq.~\eqref{eq:fundamental_Rcommu}.

The second line of Eq.~\eqref{eq:commu_Eab_psi_p} is derived in a similar way. Since~(by Definition~\ref{def:phys_obs_basic}) any basic local observable $\hat{O}$ in a bounded region $S$ can be written as a sum of products of $\hat{e}_{ij}$, where $i,j\in S$,
Eq.~\eqref{eq:commu_Eab_psi_p} implies that $\hat{O}$ satisfies the locality condition in Eq.~\eqref{eq:generalLOcondition}. 

We remark that the commutator $[\hat{e}_{ij}, \hat{e}_{kl}]$ also has a very simple form
\begin{eqnarray}\label{eq:commu_Eab_Ecd}
	[\hat{e}_{ij}, \hat{e}_{kl}]&=&\sum_b[\hat{e}_{ij}, \hat{\psi}^+_{k,b}\hat{\psi}^-_{l,b}]\nonumber\\
	&=&\sum_b[\hat{e}_{ij}, \hat{\psi}^+_{k,b}]\hat{\psi}^-_{l,b}+\sum_b\hat{\psi}^+_{k,b}[\hat{e}_{ij}, \hat{\psi}^-_{l,b}]\nonumber\\
	&=&\sum_b\delta_{jk}\hat{\psi}^+_{i,b}\hat{\psi}^-_{l,b}-\sum_b\hat{\psi}^+_{k,b}\delta_{il}\hat{\psi}^-_{j,b}\nonumber\\
	&=&\delta_{jk}\hat{e}_{il}-\delta_{il}\hat{e}_{kj},
\end{eqnarray}
where in the second line we used Eq.~\eqref{eq:commu_Eab_psi_p}.
Eq.~\eqref{eq:commu_Eab_Ecd} is the CR between the basis elements $\{\hat{e}_{ij}\}_{1\leq i,j\leq N}$ of the $\mathfrak{gl}_N$ Lie algebra, where $\hat{e}_{ij}$ represents the matrix that has $1$ in the $i$-th row and $j$-th column and zero everywhere else. 

Besides locality, this $\mathfrak{gl}_N$ Lie algebra structure plays another crucial role in the theory of $R$-paraparticles:
it guarantees that all bilinear Hamiltonians are exactly solvable~(Sec.~\ref{sec:solution}), thereby extending the notion of free particles to $R$-paraparticles. Indeed, 
Eq.~(\ref{eq:commu_Eab_psi_p}) means that, under the adjoint action of this $\mathfrak{gl}_N$ Lie algebra,  
the creation operators $\{\hat{\psi}^+_{j,a}\}^N_{j=1}$ for each $a$ transform in the fundamental representation of  $\mathfrak{gl}_N$, and the annihilation operators $\{\hat{\psi}^-_{j,a}\}^N_{j=1}$  transform in its fundamental conjugate representation. 
This allows us to construct eigenmodes $k$ of the Hamiltonian $\hat{H}$ such that the ${\tilde\psi}_{k,a}^\pm$ formed from  the associated linear combinations of ${\hat \psi}_{j,a}^\pm$ satisfy $[\hat{H},\tilde{\psi}^\pm_{k,a}]=\pm \epsilon_{k}\tilde{\psi}^\pm_{k,a}$, thereby obtaining the whole spectrum of $\hat{H}$. 
A particularly important family of physical observables are the particle number operators 
\begin{equation}\label{eq:def_n_a}
	\hat{n}_i\equiv \hat{e}_{ii}= \sum^m_{i=1} \hat{\psi}^+_{i,a}\hat{\psi}^-_{i,a}.
\end{equation}
It follows from Eq.~\eqref{eq:commu_Eab_Ecd} that they mutually commute $[\hat{n}_i,\hat{n}_j]=0$, so they have a complete set of common eigenstates. %
Meanwhile, Eq.~\eqref{eq:commu_Eab_psi_p}  gives $[\hat{n}_{i}, \hat{\psi}^\pm_{j,b}]=\pm\delta_{ij}\hat{\psi}^\pm_{j,b}$, meaning that $\hat{\psi}^+_{j,b}$~($\hat{\psi}^-_{j,b}$) increases~(decreases) the eigenvalue of $\hat{n}_j$ by $1$, and does not change the eigenvalue of $\hat{n}_i$ for $j\neq i$. This justifies the terminology creation and annihilation operators, since $\hat{\psi}^+_{j,b}$~($\hat{\psi}^-_{j,b}$) creates~(annihilates) a particle in the mode $j$. We also define the total particle number operator $\hat{n}=\sum_{i=1}^N \hat{n}_i$, so we have $[\hat{n}, \hat{\psi}^\pm_{j,b}]=\pm\hat{\psi}^\pm_{j,b}$. These CRs involving the number operators are the same as for fermions and bosons. However, we will see later that due to the generalized CRs between $\{\hat{\psi}^\pm_{i,b}\}$ in Eqs.~\eqref{eq:fundamental_Rcommu}, the spectrum of $\{\hat{n}_i \}$ is different for paraparticles.

We close this subsection by giving the most general form of quadratic local operators that commute with the total particle number $\hat{n}$:
\begin{equation}\label{eq:def_e_ijkappa}
	\hat{e}^{\kappa}_{ij}\equiv \sum^m_{a=1} \kappa_{ab}\hat{\psi}^+_{i,a}\hat{\psi}^-_{j,b},
\end{equation}
where $\kappa_{ab}$ is a coefficient matrix. 
The condition in Eq.~\eqref{eq:generalLOcondition} for $\hat{e}^{\kappa}_{ij}$ to be a local operator supported on $S=\{i,j\}$ becomes 
$[\hat{e}^{\kappa}_{ij},\hat{\psi}^\pm_{k,a}]=0~\forall k\notin S$,  
which requires $\kappa$ to satisfy  
\begin{equation}\label{eq:kappaDC}
\adjincludegraphics[height=7ex,valign=c]{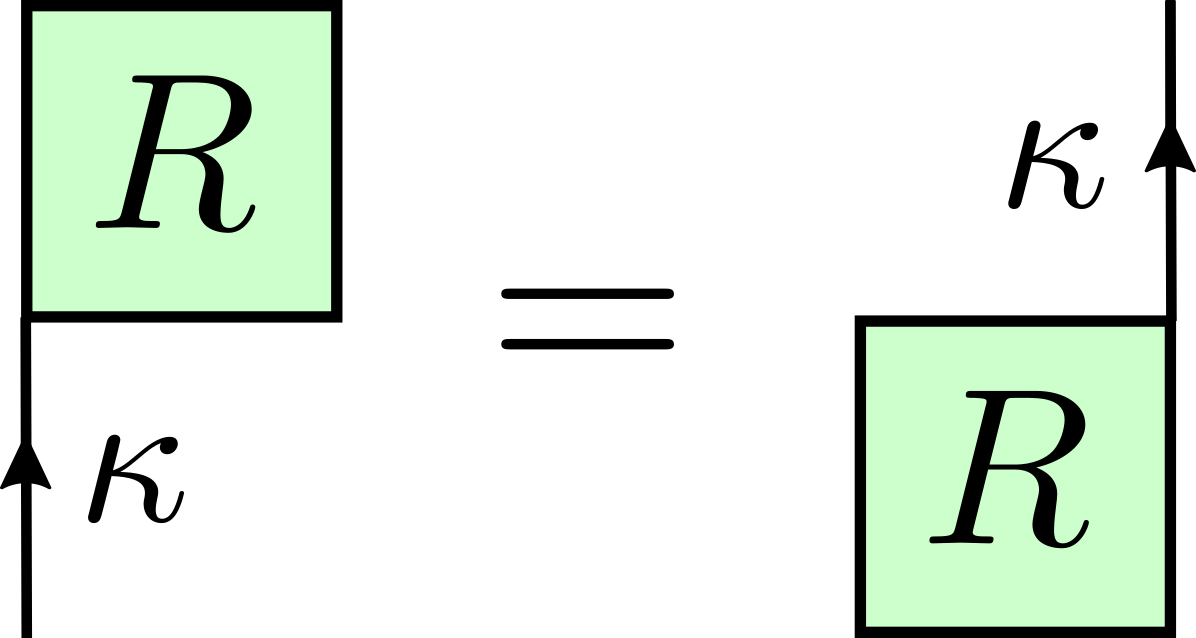},
~~
\adjincludegraphics[height=7ex,valign=c]{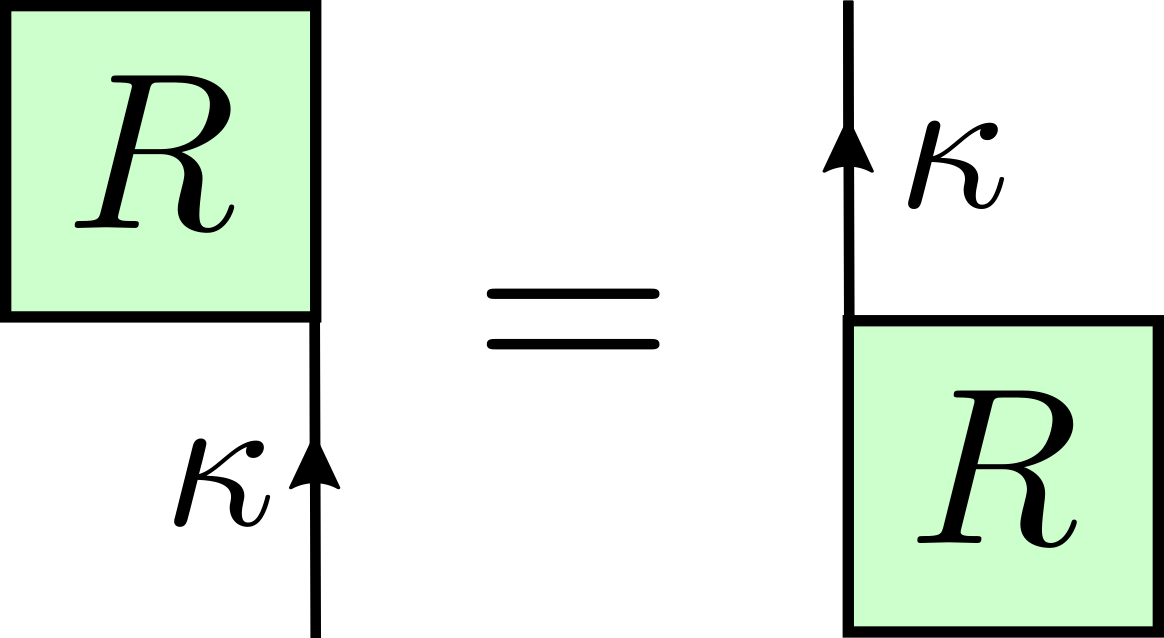}~.
\end{equation}
Note that these two equations are equivalent due to $R^2=\mathds{1}$. Eq.~\eqref{eq:kappaDC} and the general quadratic operators $\hat{e}^{\kappa}_{ij}$ will become important when we discuss the local indistinguishability of $R$-paraparticles in Sec.~\ref{sec:local-indist-revisit}. 
Note that in principle we can also generalize Eqs.~\eqref{eq:commu_Eab_psi_p-supp} and \eqref{eq:commu_Eab_Ecd} to $\hat{e}^{\kappa}_{ij}$ as well, however, we do not need this result in this paper. %

\subsection{The state space}\label{sec:state_space}
We have defined the fundamental CRs~\eqref{eq:fundamental_Rcommu} of paraparticle operators and showed that they lead to locality of the theory. To fully define a quantum theory, we need to specify a many-particle state space~(i.e. the Fock space), and then define the action of the creation and annihilation operators $\hat{\psi}^\pm_{i,a}$ on the state space so that Eq.~\eqref{eq:fundamental_Rcommu} is satisfied. More formally, we need to find a representation of the algebra defined by the CRs in Eq.~\eqref{eq:fundamental_Rcommu}, which is the primary goal of this section. We begin by constructing a basis for the state space in Sec.~\ref{sec:basis_state_space} and then define the action of $\hat{\psi}^\pm_{i,a}$ in Sec.~\ref{sec:action_parafield_basis}.

\subsubsection{A basis for the state space}\label{sec:basis_state_space}
Analogous to the Fock space of fermions and bosons, there is a unique vacuum state $|0\rangle$ satisfying $\hat{\psi}^-_{i,a}\ket{0}=0$, and we always assume it is normalized as $\braket{0|0}=1$. Physically, $|0\rangle$ represents the state with no particles $\hat{n}|0\rangle=0$.  The many-particle state space is spanned by all states of the form  
\begin{equation}\label{eq:psi_statespace_general}
\ket{\psi}=\hat{\psi}^+_{i_1,a_1}\hat{\psi}^+_{i_2,a_2}\ldots \hat{\psi}^+_{i_n,a_n}|0\rangle,
\end{equation}
where $n\in\mathbb{Z}_{\geq 0}$ is the particle number. Our job is to construct a basis for all such states. Without loss of generality, we can always assume that the creation operators in Eq.~\eqref{eq:psi_statespace_general} are ordered such that $i_1\leq i_2\leq \ldots\leq i_n$, since %
an unordered state can always be written as a linear combination of ordered states by repeatedly using the second line of  Eq.~\eqref{eq:fundamental_Rcommu}.

We first focus on the subspace of states in which all particles are in one mode $i_1=i_2=\ldots=i_n=i$. An arbitrary state in this subspace can be written as
\begin{equation}\label{eq:Psi_state_single_mode}
	|\Psi\rangle=\frac{1}{\sqrt{n!}}\sum_{\substack{a_1,\ldots,a_n}}\Psi_{a_1\ldots a_n}\hat{\psi}^+_{i,a_1}\ldots\hat{\psi}^+_{i,a_n}|0\rangle.
\end{equation}
Due to the $R$-CRs between the creation operators $\{\hat{\psi}^+_{i,a}\}$ in Eq.~\eqref{eq:fundamental_Rcommu},  we can always assume without loss of generality that the coefficients $\{\Psi_{a_1\ldots a_n}\}$ satisfy
\begin{equation}\label{eq:V_n_basis}
	\sum_{a'_{j},a'_{j+1}}R^{a_ja_{j+1}}_{a'_ja'_{j+1}}\Psi_{a_1\ldots a'_ja'_{j+1}\ldots a_n}=\Psi_{a_1\ldots a_ja_{j+1}\ldots a_n}
\end{equation} 
for $j=1,2,\ldots,n-1$. Intuitively, Eq.~\eqref{eq:V_n_basis} requires that $\Psi_{a_1\ldots a_n}$ is $R$-symmetric, which in the case of fermions or bosons~($R=\pm 1$) reduces to being totally symmetric or antisymmetric. 
Indeed, an unsymmetrized wavefunction $\tilde{\Psi}$ can always be symmetrized  in the following way~\cite{wang2023para}
\begin{equation}
	\Psi^{A}=\frac{1}{n!}\sum_{A',\sigma\in S_n}\rho_{AA'}(\sigma)\tilde{\Psi}^{A'},
\end{equation}
where $\rho$ is the representation of $S_n$ realized by the $R$-matrix, as defined in Sec.~\ref{sec:YBE}, and we use $A=(a_1,a_2,\ldots,a_n)$ to collectively denote the internal states of the paraparticles. 
This symmetrization does not change the overall quantum state $|\Psi\rangle$ due to the second line in Eq.~\eqref{eq:fundamental_Rcommu}, and it is straightforward to verify that the symmetrized wavefunction $\Psi$ satisfies Eq.~\eqref{eq:V_n_basis}.
	
Now consider the general case $i_1\leq i_2\leq \ldots\leq i_n$. 
Let $\{\Psi^\alpha_{a_1a_2\ldots a_n}\}_{\alpha=1}^{d_n}$ be a complete set of linearly independent solutions to the system of linear equations~\eqref{eq:V_n_basis}.
A basis for the entire state space is constructed as the set of states of the form
\begin{equation}\label{eq:full_basis}
	|{}^{\alpha_1}_{n_1},{}_{n_2}^{\alpha_2},\ldots, {}_{n_N}^{\alpha_N}\rangle=\hat{\Psi}^{(1)+}_{n_1,\alpha_1}\hat{\Psi}^{(2)+}_{n_2,\alpha_2}\ldots\hat{\Psi}^{(N)+}_{n_N,\alpha_N}|0\rangle,%
\end{equation}
where the numbers $\{(n_i,\alpha_i)\}^N_{i=1}$ can be chosen independently for different modes~(with the only constraint being $1\leq \alpha_i\leq d_{n_i}$ for each $i$), and the  operator
\begin{equation}\label{eq:single_mode_creation}
	\hat{\Psi}^{(i)+}_{n,\alpha}\equiv \frac{1}{\sqrt{n!}}\sum_{a_1 a_2\ldots a_n}\Psi^\alpha_{a_1 a_2\ldots a_n}\hat{\psi}^+_{i,a_1}\hat{\psi}^+_{i,a_2}\ldots\hat{\psi}^+_{i,a_n}, 
\end{equation}
creates a multiparticle state in the mode $i$ with occupation number $n$~(eigenvalue of  $\hat{n}_i\equiv \hat{e}_{ii}$). It is straightforward to see that any state in the Fock space can be written as a linear combination of the states in Eq.~\eqref{eq:full_basis}, based on our arguments presented above. %
To prove that Eq.~\eqref{eq:full_basis} defines a basis for the state space, we still need to show that they are linearly independent, which we do below. We start with the easier case of unitary $R$-matrices~[where we actually prove a stronger result that states in Eq.~\eqref{eq:full_basis} are orthonormal], and later generalize to non-unitary $R$-matrices. 
\paragraph{Unitary $R$-matrices}
When the $R$-matrix is unitary, we can consistently set $\hat{\psi}^+_{i,a}=(\hat{\psi}^-_{i,a})^\dagger$~\footnote{First, notice that $\hat{\psi}^+_{i,a}=(\hat{\psi}^-_{i,a})^\dagger$ is fully consistent with the fundamental CRs in Eq.~\eqref{eq:fundamental_Rcommu}, as taking Hermitian conjugate on both sides leaves Eq.~\eqref{eq:fundamental_Rcommu} invariant~(maps the first line to itself and swaps the second and the third lines). Then it can be checked straightforwardly that the explicit matrix representation of $\hat{\psi}^\pm_{i,a}$ defined in Sec.~\ref{sec:action_parafield_basis}  indeed satisfies $\hat{\psi}^+_{i,a}=(\hat{\psi}^-_{i,a})^\dagger$. By contrast, if $R$ is not unitary, the relation $\hat{\psi}^+_{i,a}=(\hat{\psi}^-_{i,a})^\dagger$ is not consistent with  Eq.~\eqref{eq:fundamental_Rcommu} as taking Hermitian conjugate on both sides leads to extra relations that result in an algebraic inconsistency, and the representation of $\hat{\psi}^\pm_{i,a}$ constructed in Sec.~\ref{sec:action_parafield_basis}  does not satisfy $\hat{\psi}^+_{i,a}=(\hat{\psi}^-_{i,a})^\dagger$. 
However, as we proved in Sec.S2.D.5 of the supplementary information of Ref.~\cite{wang2023para}, even for a non-unitary $R$, we can still define a Hermitian inner product on the Fock space such that $\hat{e}^\dagger_{ij}=\hat{e}_{ji}$ is still satisfied, and consequently all physical observables defined in Sec.~\ref{sec:locality} are still Hermitian with respect to this inner product.} and in this case, we require that the coefficients
$\{\Psi^\alpha_{a_1a_2\ldots a_n}\}_{\alpha=1}^{d_n}$ are normalized as 
\begin{equation}\label{eq:singlemodewf_normalization}
	\sum_{a_1,a_2,\ldots, a_n}\Psi^{\beta*}_{a_1a_2\ldots a_n}\Psi^\alpha_{a_1a_2\ldots a_n}=\delta_{\alpha\beta}.
\end{equation}
Then the basis in Eq.~\eqref{eq:full_basis} is orthonormal, i.e.  
\begin{equation}\label{eq:full_basis_orthonormal}
	\Braket{{}^{\beta_1}_{n^{\prime}_1},{}_{n^{\prime}_2}^{\beta_2},\ldots, {}_{n^{\prime}_N}^{\beta_N}|{}^{\alpha_1}_{n^{\vphantom{\prime}}_1},{}_{n^{\vphantom{\prime}}_2}^{\alpha_2},\ldots, {}_{n^{\vphantom{\prime}}_N}^{\alpha_N}}=\prod^N_{j=1}\delta_{n^{\vphantom{\prime}}_jn^{\prime}_j}\delta_{\alpha_j\beta_j}.%
\end{equation}
To prove Eq.~\eqref{eq:full_basis_orthonormal}, we use the fact that for any state $\ket{\Phi}$ satisfying $\hat{\psi}^-_{i,a}\ket{\Phi}=0,\forall a$, and any  $\Psi_{a_1\ldots a_n}$ satisfying Eq.~\eqref{eq:V_n_basis}, 
we have 
\begin{equation}\label{eq:psim_action_basis}
	\hat{\psi}^-_{i,a}\hat{\Psi}^{(i)+}_{n}\ket{\Phi}= \frac{n}{\sqrt{n!}}\sum_{ a_2\ldots a_n}\Psi_{a a_2\ldots a_n}\hat{\psi}^+_{i,a_2}\ldots\hat{\psi}^+_{i,a_n}\ket{\Phi},
\end{equation}
and 
\begin{equation}\label{eq:PsiPsi}
	\hat{\Psi}^{(i)-}_{n,\beta}\hat{\Psi}^{(i)+}_{n,\alpha}\ket{\Phi}=\delta_{\alpha\beta}\ket{\Phi},
\end{equation}
where $\hat{\Psi}^{(i)-}_{n,\beta}=(\hat{\Psi}^{(i)+}_{n,\beta})^\dagger$. Eq.~\eqref{eq:psim_action_basis} is proved by using the first relation in Eq.~\eqref{eq:fundamental_Rcommu} to move $\hat{\psi}^-_{i,a}$ all the way to the right until it hits $|0\rangle$, and the fact that $\Psi_{a_1a_2\ldots a_n}$ satisfies Eq.~\eqref{eq:V_n_basis}. Eq.~\eqref{eq:PsiPsi} is proved by expanding the definition of $\hat{\Psi}^{(i)-}_{n,\beta}$%
~[the Hermitian conjugate of Eq.~\eqref{eq:single_mode_creation}] and then apply Eq.~\eqref{eq:psim_action_basis} $n$ times. 

Eq.~\eqref{eq:full_basis_orthonormal} is proved as follows. Notice that we only need to prove the case $n_j=n'_j,~1\leq j\leq N$, since it is clear that the common eigenstates of 
$\{\hat{n}_j\}^N_{j=1}$ with different eigenvalues are orthogonal. 
\begin{eqnarray}\label{eq:full_basis_orthonormal_proof}
	&&\Braket{{}^{\beta_1}_{n_1},{}_{n_2}^{\beta_2},\ldots, {}_{n_N}^{\beta_N}|{}^{\alpha_1}_{n_1},{}_{n_2}^{\alpha_2},\ldots, {}_{n_N}^{\alpha_N}}\nonumber\\
	&=&\Braket{{}^0_0,{}_{n_2}^{\beta_2},\ldots, {}_{n_N}^{\beta_N}| 	\hat{\Psi}^{(i)-}_{n_1,\beta_1}\hat{\Psi}^{(i)+}_{n_1,\alpha_1}|{}^0_0,{}_{n_2}^{\alpha_2},\ldots, {}_{n_N}^{\alpha_N}}\nonumber\\
	&=&\delta_{\alpha_1\beta_1}\Braket{{}^0_0,{}_{n_2}^{\beta_2},\ldots, {}_{n_N}^{\beta_N}|{}^0_0,{}_{n_2}^{\alpha_2},\ldots, {}_{n_N}^{\alpha_N}}\nonumber\\
	&=&\ldots\nonumber\\
	&=&\prod^N_{j=1}\delta_{\alpha_j\beta_j},
\end{eqnarray}
where in the third line we used the fact that $\hat{\psi}^-_{1,a}\ket{{}^0_0,{}_{n_2}^{\alpha_2},\ldots, {}_{n_N}^{\alpha_N}}=0,~\forall a$ and Eq.~\eqref{eq:PsiPsi}. 

\paragraph{Non-unitary $R$-matrices}
We now show that for non-unitary $R$-matrices, the basis states in Eq.~\eqref{eq:full_basis} are still linearly independent, even though they are no longer orthonormal. In this case we can still formally define a dual basis
\begin{equation}\label{eq:dual_basis_non-unitary}
	\bra{{}^{\beta_1}_{n_1},{}_{n_2}^{\beta_2},\ldots, {}_{n_N}^{\beta_N}}=\bra{0}\hat{\Psi}^{(1)-}_{n_1,\beta_1}\hat{\Psi}^{(2)-}_{n_2,\beta_2}\ldots\hat{\Psi}^{(N)-}_{n_N,\beta_N},
\end{equation}
where 
\begin{equation}\label{eq:single_mode_anni_non-unitary}
	\hat{\Psi}^{(i)-}_{n,\beta}\equiv \frac{1}{\sqrt{n!}}\sum_{a_1 a_2\ldots a_n}\Psi^{\beta*}_{a_1 a_2\ldots a_n}\hat{\psi}^-_{i,a_n}\ldots\hat{\psi}^-_{i,a_2}\hat{\psi}^-_{i,a_1}, 
\end{equation}
Here we use the same normalization for $\{\Psi^\alpha_{a_1a_2\ldots a_n}\}_{\alpha=1}^{d_n}$ given in Eq.~\eqref{eq:singlemodewf_normalization}. Then it is straightforward to see that Eqs.~(\ref{eq:psim_action_basis},\ref{eq:PsiPsi},\ref{eq:full_basis_orthonormal_proof}) are still valid, since we never use the 
unitarity of the $R$-matrix~[or relations like $\hat{\psi}^+_{i,a}=(\hat{\psi}^-_{i,a})^\dagger$] in their derivations. Then Eq.~\eqref{eq:full_basis_orthonormal} implies that the basis states in Eq.~\eqref{eq:full_basis} are linearly independent.  A caveat, however, is that for non-unitary $R$-matrices, 	the dual basis state $\bra{{}^{\beta_1}_{n_1},{}_{n_2}^{\beta_2},\ldots, {}_{n_N}^{\beta_N}}$ defined in Eq.~\eqref{eq:dual_basis_non-unitary} is not the Hermitian conjugate of $\ket{{}^{\beta_1}_{n_1},{}_{n_2}^{\beta_2},\ldots, {}_{n_N}^{\beta_N}}$  in general, due to $\hat{\psi}^+_{i,a}\neq (\hat{\psi}^-_{i,a})^\dagger$. So Eq.~\eqref{eq:full_basis_orthonormal} does not imply that the basis states are orthonormal. 

\subsubsection{Action of $\hat{\psi}^\pm_{i,a}$ on the basis states}\label{sec:action_parafield_basis}
The action of the creation and annihilation operators $\hat{\psi}^\pm_{i,a}$ on the basis states in Eq.~\eqref{eq:full_basis} are completely determined by the fundamental CRs in Eq.~\eqref{eq:fundamental_Rcommu}. 
We can work out their explicit matrix representation 
using the following~(more convenient) relations that are derived from Eq.~\eqref{eq:fundamental_Rcommu} and the definition of $\hat{\Psi}^{(i)+}_{n,\alpha}$ in Eq.~\eqref{eq:single_mode_creation},
\begin{eqnarray}\label{eq:CRpsiPsi}
	\hat{\psi}^\pm_{j,a}\hat{\Psi}^{(i)+}_{n,\alpha}&=&\sum_{b,\beta}T^\pm_{ab,\beta\alpha} \hat{\Psi}^{(i)+}_{n,\beta}\hat{\psi}^\pm_{j,b},~j\neq i,\nonumber\\
	\hat{\psi}^+_{i,a}\hat{\Psi}^{(i)+}_{n,\alpha}&=&\sum_{b,\beta}Y^+_{a,\beta\alpha} \hat{\Psi}^{(i)+}_{n+1,\beta},\nonumber\\
	\hat{\psi}^-_{i,a}\hat{\Psi}^{(i)+}_{n,\alpha}\ket{\Phi}&=&\sum_{b,\beta}Y^-_{a,\beta\alpha} \hat{\Psi}^{(i)+}_{n-1,\beta}\ket{\Phi},
\end{eqnarray} 
where $\ket{\Phi}$ satisfies $\hat{\psi}^-_{i,a}\ket{\Phi}=0,\forall a$, and the tensors $Y^\pm_{a,\beta\alpha}$ are  defined as 
	\begin{eqnarray}\label{eq:YExplicit}
	\sum_{\beta}Y^{+}_{a,\beta\alpha}\Psi^\beta_{a_0 a_1 \ldots a_n}&=&\!\frac{1}{\sqrt{n+1}}\bar{Y}^{a_0 a_1\ldots a_n}_{~a_{\hphantom{0}} b_1\ldots b_n} \Psi^\alpha_{b_1\ldots b_n},\nonumber\\
	\sum_{\beta}Y^-_{a,\beta\alpha}\Psi^\beta_{a_2\ldots a_n}&=&\sqrt{n}\Psi^\alpha_{a a_2\ldots a_n},
\end{eqnarray}
	where 
	\begin{equation}\label{def:YRtensor}
	\bar{Y}_{01\ldots n}=1+R_{01}+R_{12}R_{01}+\ldots+R_{n-1,n}\cdots R_{12}R_{01}.%
\end{equation}
Here we use the abstract tensor indexing introduced in Eq.~(\ref{def:Rjjp1}). The tensors $T^\pm_{ab,\beta\alpha}$ are defined as
\begin{eqnarray}\label{eq:Tpm_matrix_element}
	\sum_\beta T^+_{ab,\beta\alpha}\Psi^\beta_{a_1 a_2 \ldots a_n}&=&
	\sum_{a'_1,\ldots,a'_n}\Psi^\alpha_{a'_1,a'_2\ldots,a'_n}\nonumber\\
	&&\times 
	\begin{tikzpicture}[baseline={([yshift=-.4ex]current bounding box.center)}, scale=.8]
		\TpRmatrix{0}{0}{a_1}{a}{a'_1}{}
		\TpRmatrix{1}{0}{a_2}{}{a'_2}{}
		\draw[dotted, thick] (3*\AL, 0) -- (5*\AL,0);
		\TpRmatrix{3}{0}{a_n}{}{a'_n}{b}
	\end{tikzpicture},\nonumber\\
	\sum_\beta T^-_{ab,\beta\alpha}\Psi^\beta_{a_1 a_2 \ldots a_n}&=&
	\sum_{a'_1,\ldots,a'_n}\Psi^\alpha_{a'_1 a'_2\ldots a'_n}\\
	&&\times
	\begin{tikzpicture}[baseline={([yshift=-.4ex]current bounding box.center)}, scale=.8]
		\TmRmatrix{0}{0}{a_1}{a}{a'_1}{}
		\TmRmatrix{1}{0}{a_2}{}{a'_2}{}
		\draw[dotted, thick] (3*\AL, 0) -- (5*\AL,0);
		\TmRmatrix{3}{0}{a_n}{}{a'_n}{b}
	\end{tikzpicture}.\nonumber
\end{eqnarray} 
Using Eq.~\eqref{eq:CRpsiPsi}, we can obtain an explicit matrix representation of $\{\hat{\psi}^\pm_{j,a}\}$ acting on the basis states in Eq.~\eqref{eq:full_basis}
\begin{eqnarray}\label{eq:MPOactionpsi}
	\hat{\psi}^+_{j,a}\ket{{}^{\alpha_1}_{n_1},{}_{n_2}^{\alpha_2},\ldots, {}_{n_N}^{\alpha_N}}&=&\sum_{\beta_1,\ldots,\beta_j}\ket{{}^{\beta_1}_{n_1},\ldots,{}_{n_j+1}^{\beta_j}, {}_{n_{j+1}}^{\alpha_{j+1}},\ldots, {}_{n_N}^{\alpha_N}}\nonumber\\
	&&\!\!\!\!\times\begin{tikzpicture}[baseline={([yshift=.4ex]current bounding box.center)}, scale=0.8]
		\node  at (-1.5*\AL,0*\AL) {\footnotesize $a$};
		\umatr{0}{0}{T^+}
		\quantumindices{0}{0}{\beta_1}{\alpha_1}
		\umatr{2*\AL}{0}{T^+}
		\quantumindices{2*\AL}{0}{\beta_2}{\alpha_2}
		\draw[dotted, thick] (3*\AL, 0) -- (5*\AL,0);
		\umatr{6*\AL}{0}{T^+}
		\quantumindices{6*\AL}{0}{\beta_{j-1}}{\alpha_{j-1}}
		\Ytrian{8*\AL}{0}{+}{}
		\quantumindices{8*\AL}{0}{\beta_j}{\alpha_j}
	\end{tikzpicture}\!\!,\nonumber\\
	\hat{\psi}^-_{j,a}\ket{{}^{\alpha_1}_{n_1},{}_{n_2}^{\alpha_2},\ldots, {}_{n_N}^{\alpha_N}}&=&\sum_{\beta_1,\ldots,\beta_j}\ket{{}^{\beta_1}_{n_1},\ldots,{}_{n_j-1}^{\beta_j}, {}_{n_{j+1}}^{\alpha_{j+1}},\ldots, {}_{n_N}^{\alpha_N}}\nonumber\\
	&&\!\!\!\!\times\begin{tikzpicture}[baseline={([yshift=.4ex]current bounding box.center)}, scale=0.8]
		\node  at (-1.5*\AL,0*\AL) {\footnotesize $a$};
		\umatr{0}{0}{T^-}
		\quantumindices{0}{0}{\beta_1}{\alpha_1}
		\umatr{2*\AL}{0}{T^-}
		\quantumindices{2*\AL}{0}{\beta_2}{\alpha_2}
		\draw[dotted, thick] (3*\AL, 0) -- (5*\AL,0);
		\umatr{6*\AL}{0}{T^-}
		\quantumindices{6*\AL}{0}{\beta_{j-1}}{\alpha_{j-1}}
		\Ytrian{8*\AL}{0}{-}
		\quantumindices{8*\AL}{0}{\beta_j}{\alpha_j}
	\end{tikzpicture}\!\!,%
\end{eqnarray}
where $
\begin{tikzpicture}[baseline={([yshift=-.6ex]current bounding box.center)}, scale=.8]
	\umatr{0}{0}{T^\pm}
	\quantumindices{0}{0}{\beta}{\alpha}
	\paraindices{0}{0}{a}{b}
\end{tikzpicture}=T^\pm_{ab,\beta\alpha}$, and $
\begin{tikzpicture}[baseline={([yshift=-.6ex]current bounding box.center)}, scale=.8]
	\Ytrian{0}{0}{\pm}
	\quantumindices{0}{0}{\beta}{\alpha}
	\node  at (-1.5*\AL,0*\AL) {\footnotesize $a$};
\end{tikzpicture}=Y^\pm_{a,\beta\alpha}$. 
Eq.~\eqref{eq:MPOactionpsi} applies to paraparticle theories defined by any $R$-matrix, including non-unitary ones.

It is straightforward to verify that the explicit matrix representation of $\{\hat{\psi}^\pm_{j,a}\}$ in Eq.~\eqref{eq:MPOactionpsi} satisfies all the CRs in Eq.~\eqref{eq:fundamental_Rcommu}. This can be done using elementary tensor network manipulations. We will prove this explicitly later in Part II when we study the properties of the  MPO Jordan-Wigner transformation~(JWT)~\cite{wang2023para} in the 1D spin models with emergent paraparticles, as 
the MPO representation of $\{\hat{\psi}^\pm_{j,a}\}$ in Eq.~\eqref{eq:MPOactionpsi} is essentially the same as the MPO JWT in the spin models. Indeed, we will show in Part II that the second quantization algebra in Eq.~\eqref{eq:fundamental_Rcommu} along with its representation theory studied in this section provides an alternative way to construct the 1D spin model. 

The MPO string in Eq.~\eqref{eq:MPOactionpsi} is reminiscent of the familiar fact in the second quantization of fermions that the action of fermionic operators on the particle number basis involves a fermion minus sign, and implementing such a fermion minus sign computationally is equivalent to doing a JWT. For paraparticles, the fermion minus sign becomes the paraparticle $R$-matrix, and the JW string becomes an MPO string of $R$~[or more precisely the tensor $T^\pm$, which is constructed out of $R$ and $\{\Psi^\alpha_{a_1a_2\ldots a_n}\}_{\alpha=1}^{d_n}$ in Eq.~\eqref{eq:Tpm_matrix_element}]. 

\begin{remark}\label{rmk:Fock_irrep}
We finally make some technical remarks on the formal aspects of second quantization algebra and its Fock space representation. 
Denote by $\mathcal{X}_{R,N}$ the unital associative algebra over $\mathbb{C}$ generated by  $\{\hat{\psi}^\pm_{i,b}|1\leq i \leq N, 1\leq b\leq m\}$ modulo all the relations in Eq.~\eqref{eq:fundamental_Rcommu}. 
One can show that the Fock space representation of $\mathcal{X}_{R,N}$ constructed above is always irreducible~\cite{GIAQUINTO1995QWA}, 
and it is the unique~(up to isomorphism) irrep in which the spectrum of $\hat{n}$ is bounded from below. Furthermore, if %
the single mode Hilbert space is finite-dimensional with dimension $d$, %
then the algebra $\mathcal{X}_{R,N}$ is finite dimensional and simple~\cite{GIAQUINTO1995QWA}, with  the unique irrep being the Fock space representation above. In this case, $\mathcal{X}_{R,N}$ is isomorphic to $M_D[\C]$, the algebra of $D\times D$ matrices, with $D=d^N$ being the Fock space dimension. 
\end{remark}

\subsection{Generalized exclusion statistics}\label{sec:exclusion_statistics_calc}
\begin{figure}
	\center{
    \includegraphics[width=.95\linewidth]{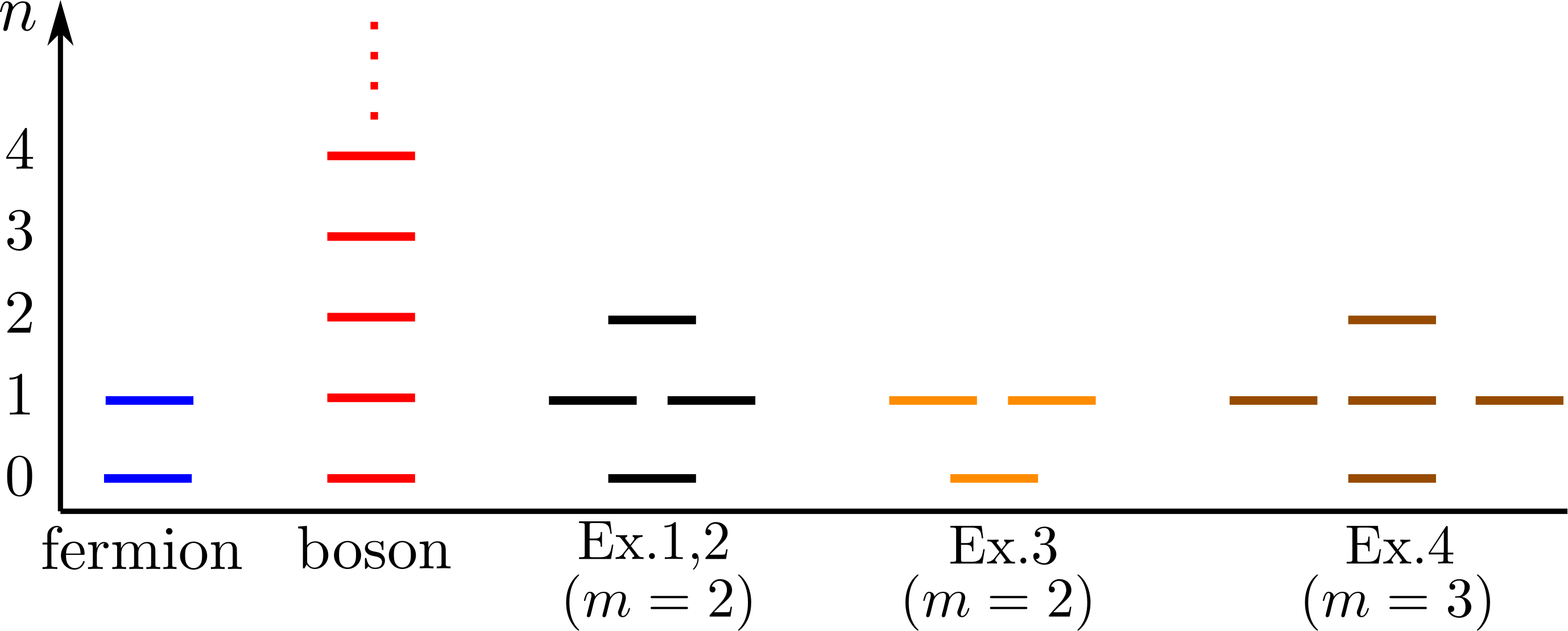}
    }
	\caption{\label{fig:level_statistics} The generalized exclusion statistics of paraparticles defined by the $R$-matrices in Exs.~\ref{ex:decoupled}-\ref{ex:1m1} of Tab.~\ref{tab:Hilbert_series}, and a comparison to ordinary fermions and bosons. Exs.~\ref{ex:setth} and \ref{ex:setth-ext} both have trivial exclusion statistics isomorphic to ordinary fermions with some internal degrees of freedom, similar to Exs.~\ref{ex:decoupled} and \ref{ex:Green}, so we omit them here.  %
    }
\end{figure}
The structure of the state space reveals a distinctive feature of paraparticles--the generalized exclusion statistics, which is encoded in    
the numbers $\{d_n\}_{n\geq 0}$, where $d_n$ is the number of linearly independent solutions to %
Eq.~\eqref{eq:V_n_basis}. As illustrated in Fig.~\ref{fig:level_statistics}, the number $d_n$ corresponds to the  number of orthonormal $n$-particle states in a single mode. 
For example, for the $R$-matrix in Ex.~\ref{ex:1m} with $m=2$, each mode $i$ has three possible quantum states: a vacuum state $\ket{0}$, and two singly occupied states $\hat{\psi}^+_{i,1}\ket{0}$ and $\hat{\psi}^+_{i,2}\ket{0}$. This kind of exclusion statistics is different from any type of free fermions or bosons, representing a distinct type of free particle exclusion statistics. Ex.~\ref{ex:1m1} has a more exotic type of exclusion statistics: each mode $i$ has $m+2$ possible quantum states in total, a vacuum state $\ket{0}$, $m$ linearly independent singly occupied states $\{\hat{\psi}^+_{i,a}\ket{0}\}_{a=1}^m$, and a doubly occupied state $\ket{2}$, which is distinct from ordinary fermions for any $m\geq 3$. 

It turns out to be  helpful to encode the numbers $\{d_n\}_{n\geq 0}$ into a ``characteristic polynomial'' of the $R$-paraparticle, known in the mathematics literature as the \textit{Hilbert series} of the $R$-matrix~\cite{polishchuk2005quadratic}, defined as
\begin{equation}\label{eq:hilbert_series}
	z_R(x)\equiv \sum_{n=0}^\infty d_n x^n.
\end{equation}
The Hilbert series $z_R(x)$ provides a quick way to show the non-triviality~(i.e. distinct from fermions and bosons) of the exclusion statistics. The reason is that any $R$-matrix describing a system of ordinary fermions and bosons must have a Hilbert series of the trivial form $z_R(x)=(1+x)^p (1-x)^{-q}$, where $p$~(and $q$) is the number of flavors of fermions~(and bosons) in this system.  
Therefore the $R$-matrix given in Ex.~\ref{ex:1m} must define a non-trivial type of exclusion statistics for $m\geq 2$, since $z_R(x)=1+mx$ is not equal to  
$(1+x)^p (1-x)^{-q}$ for any integers $p,q$. 
Ex.~\ref{ex:1m1} is similarly non-trivial for $m\geq 3$. A caveat, however, is that $z_R(x)$ only gives a sufficient condition for non-triviality, and having a trivial $z_R(x)$ does not imply that the corresponding paraparticle theory is completely equivalent to fermions or bosons. 
For example, %
for the $R$-matrix in Ex.~\ref{ex:setth}, %
its Hilbert series $z_R(x)=(1+x)^4$ is the same as free fermions with an SU$(4)$ symmetry, but the exchange statistics of these emergent paraparticles are still physically~(observably) distinct from fermions and bosons, as we demonstrate in Sec.~\ref{sec:nontrivialstatistics}. 

The physical meaning of $z_R(x)$ can be understood by computing the grand canonical partition function for a single mode of a non-interacting system at temperature $T$. 
Suppose that each particle in this mode carries energy $\epsilon$~(i.e., the   Hamiltonian is $\hat{H}=\epsilon\hat{n}$). Then 
\begin{equation}\label{eq:single_mode_Z}
	\mathrm{Tr} [e^{-\beta \epsilon \hat{n}}]=z_R(e^{-\beta\epsilon}),
\end{equation}
where $\beta=1/(k_B T)$, $k_B$ is Boltzmann's constant, and we have absorbed the chemical potential $\mu$ into $\epsilon$.
When generalizing to multi-mode Hamiltonians $\hat{H}=\sum_i \epsilon_i\hat{n}_i$, a key simplifying feature is that
the partition function of the whole system factorizes into products of single-mode partition functions exactly as for fermions and bosons~\footnote{This is one of the advantages of the $R$-matrix parastatistics compared to Green's parastatistics, for which the many-mode partition functions do not simply factorize this way~\cite{Stoilova2020}, and is very hard to compute in general, apart from a few simple cases.}, i.e., 
\begin{equation}
	Z=\Tr[e^{-\beta \hat{H}}]=\prod_i z_R(e^{-\beta\epsilon_i}),
\end{equation} 
which is essential for the exact solution of free paraparticle systems we present in Sec.~\ref{sec:solution}.

We now mention a useful mathematical result that relates the numbers $\{d_n\}_{n\geq 0}$  to the $S_n$ representation $\rho_{R,n}$ realized by the $R$-matrix~(see Sec.~\ref{sec:YBE}):
\begin{proposition}\label{prop:dn-char-rel}
For $\tau\in S_n$, let $\chi_{R,n}(\tau)=\Tr[\rho_{R,n}(\tau)]$ be the character of $\rho_{R,n}$. Then we have
  \begin{equation}\label{eq:dn-char-rel}
		d_n=\frac{1}{n!}\sum_{\tau\in S_n}\chi_{R,n}(\tau).
	\end{equation}  
\end{proposition}
\begin{proof}
By definition, $d_n$ is the number of linearly independent solutions to 
Eq.~\eqref{eq:V_n_basis}, which exactly counts how many times the trivial  representation of $S_n$ occurs in the direct sum decomposition of $\rho_{R,n}$. The operator
\begin{equation}\label{eq:PnProjTrivialSubrep-main}
    P_n=\frac{1}{n!}\sum_{\tau\in S_n}\rho_{R,n}(\tau)
\end{equation}
gives the projection to the trivial subrepresentation of $\rho_{R,n}$, and therefore we have $d_n=\Tr[P_n]$, leading to Eq.~\eqref{eq:dn-char-rel}. 
\end{proof}

For explicit computations of the Hilbert series, we refer to the accompanying Mathematica code of Ref.~\cite{wang2023para}, which computes $z_R(x)$~(up to a given degree $n$) for any input involutive $R$-matrix. This is achieved by directly computing the numbers  $\{d_n\}_{n\geq 0}$ from  their definition,  i.e., $d_n$ is the number of linearly independent solutions to Eq.~\eqref{eq:V_n_basis}. Here we summarize some general facts. First, 
note that Eq.~\eqref{eq:V_n_basis} does not put any restriction on $\Psi$ for $n=0$ and $n=1$, so we always have $d_0=1, d_1=m$ for any $R$-matrix. The physical meaning of this is clear: we always have one vacuum state $|0\rangle$ and $m$ degenerate singly occupied states. 
Second, there is a  useful identity relating the Hilbert series of the $R$-matrices $R$ and $-R$~(note that $-R$ also satisfies the YBE in Eq.~\eqref{eq:YBE} if $R$ does)~\cite{polishchuk2005quadratic}:  %
\begin{equation}\label{eq:hRh_mRrelation}
	z_{R}(-x)z_{-R}(x)=1. 
\end{equation}
which allows us to compute the exclusions statistics $\{d_n\}_{n\geq 0}$ of $-R$ if the exclusion statistics of $R$ is known. For example, for the $R$-matrix in Ex.~\ref{ex:1m1}, we have $z_{-R}(x)=1/(1-mx+x^2)$, from which we obtain $d_0=1, d_1=m, d_2=m^2-1$, and $d_{n+1}=m d_n-d_{n-1}$ for $n\geq 1$. 

Finally, Lemma~\ref{lemma:OmegaLemma-Ralpha} shows that if $\Tr_1[R]=\theta\mathds{1}$ for some constant $\theta$, then $\theta$ and $m$ completely determine its Hilbert series. For example, for the $R$-matrices  in Exs.~\ref{ex:decoupled}, \ref{ex:Green}, \ref{ex:setth}, and \ref{ex:setth-ext} of Tab.~\ref{tab:Hilbert_series}, we have $\theta=-1$, therefore their Hilbert series are all the same as that of Ex.~\ref{ex:decoupled}: $z_R(x)=(1+x)^m$.

\subsection{Mutual parastatistics}\label{sec:mutual_para}
Up to now we have assumed that there is only one type of  paraparticle in our physical system, described by a single $R$-matrix. One may wonder what happens if we have different types of paraparticles~(described by different $R$-matrices) in one system. In particular, what should be the commutation relation between different types of particles? A simple consistent choice is to demand that the creation/annihilation operators for different types of particles commute. Physically, this means that different types of paraparticles are ``relatively bosonic'', i.e., have trivial mutual statistics between each other.  However, it is also possible that some different types of particles have nontrivial mutual statistics. In this section we give a general framework to treat physical systems with several different types of paraparticles and derive the algebraic constraints on their mutual particle statistics. 
The observable physical consequences of mutual parastatistics and their application in secret communication will be discussed in Sec.~\ref{sec:observe_parastatistics}.

Let $\mathcal{T}$ be the collection of particle types in a physical system. Each particle type $\psi\in \mathcal{T}$ has its own quantum dimension $m_\psi$ and $R$-matrix $R^{(\psi\psi)}$ describing its exchange statistics. For any two types of particles $\psi,\varphi\in \mathcal{T}$, let us use $R^{(\psi\varphi)}$ to describe their mutual particle statistics. The commutation relations between $\psi,\varphi$ are given as follows.
\begin{eqnarray}\label{eq:fundamental_Rcommu-rela}
	\hat{\psi}^-_{i,a} \hat{\varphi}^+_{j,b}&=&\sum_{cd}[R^{(\varphi\psi)}]^{ac}_{bd}~\hat{\varphi}_{j,c}^+ \hat{\psi}_{i,d}^-+\delta_{\varphi\psi}\delta_{ij}\delta_{ab},\nonumber\\%
	\hat{\psi}^+_{i,a} \hat{\varphi}^+_{j,b}&=&\sum_{cd}[R^{(\psi\varphi)}]^{cd}_{ab}~\hat{\varphi}_{j,c}^+ \hat{\psi}_{i,d}^+,\nonumber\\%
	\hat{\psi}^-_{i,a} \hat{\varphi}^-_{j,b}&=&\sum_{cd}[R^{(\psi\varphi)}]^{ba}_{dc}~\hat{\varphi}_{j,c}^- \hat{\psi}_{i,d}^-.%
\end{eqnarray}
The tensor graphical representation for Eq.~\eqref{eq:fundamental_Rcommu-rela} are given in Fig.~\ref{fig:RMQA-rela}.
Notice that the special case $\psi=\varphi$ gives us back Eq.~\eqref{eq:fundamental_Rcommu}. Furthermore, an argument similar to that shown in Fig.~\ref{fig:YBEgraphical} shows that the mutual $R$-matrices have to satisfy the following generalized YBE
\begin{equation}\label{eq:YBEgraphical-rela}
\adjincludegraphics[height=10ex,valign=c]{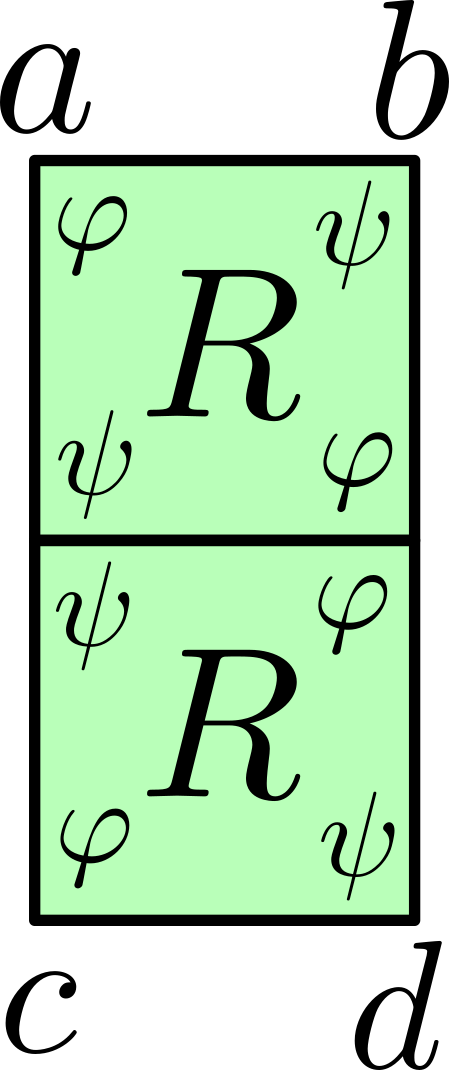}=
\adjincludegraphics[height=10ex,valign=c]{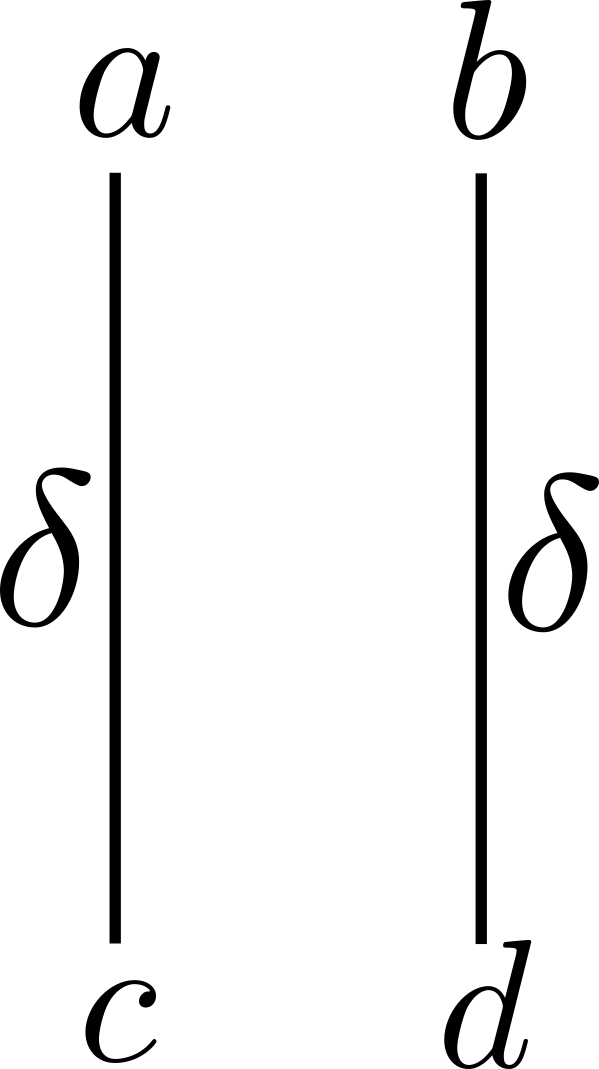},\quad \adjincludegraphics[height=15ex,valign=c]{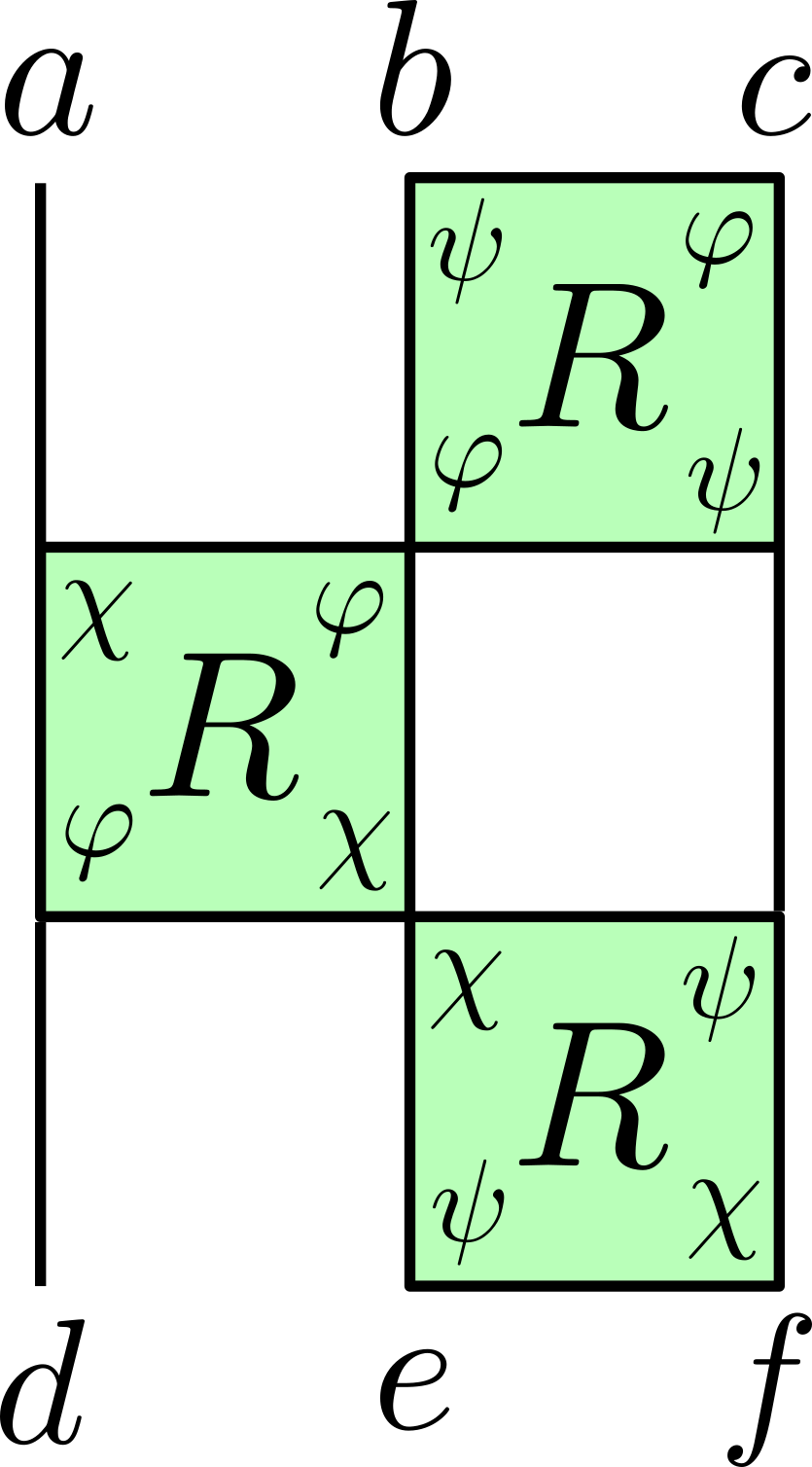}=
\adjincludegraphics[height=15ex,valign=c]{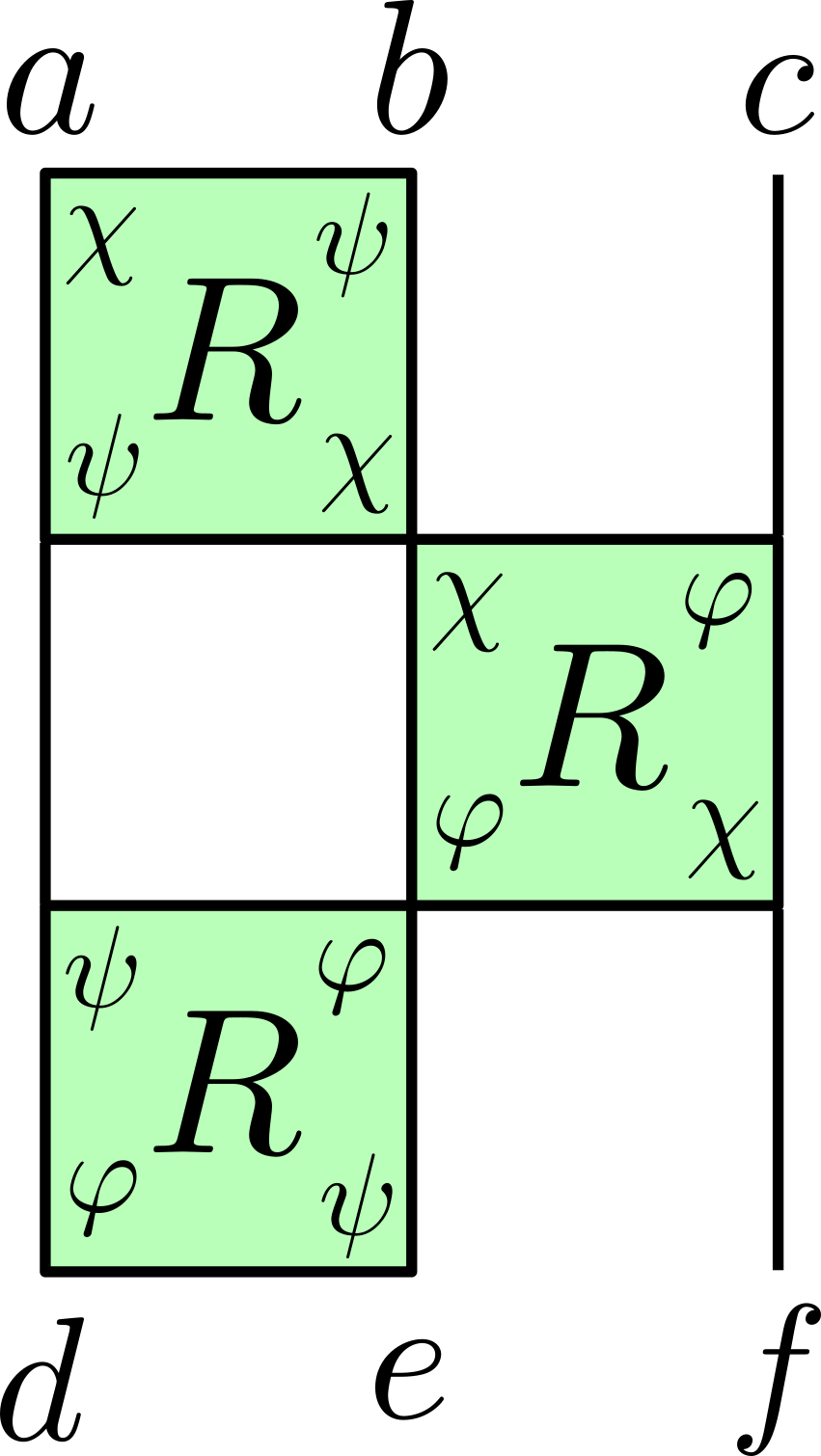},
\end{equation}
for any $\psi,\varphi,\chi\in\mathcal{T}$~(not necessarily mutually distinct),
where we use the tensor graphical notation for the relative $R$-matrix $[R^{(\psi\varphi)}]^{ab}_{cd}=
\adjincludegraphics[height=6ex,valign=c]{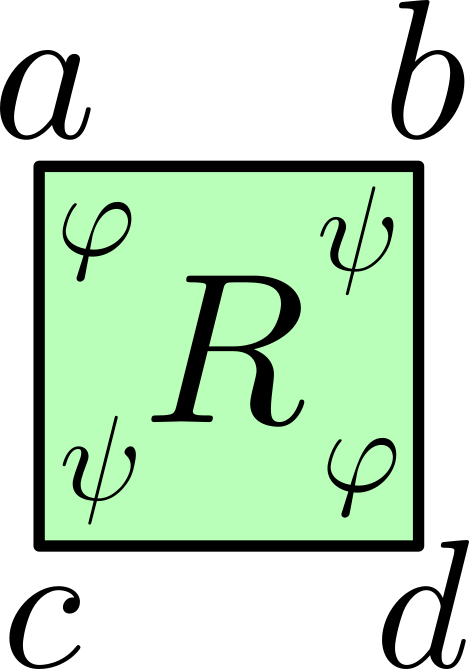}$ and similarly for others.

A basis for the Fock space of the whole system can be constructed in a similar way as in Sec.~\ref{sec:basis_state_space}. For simplicity, we show the case $|\T|=2$ with $\T=\{\psi,\varphi\}$, and we focus on the single mode case, the generalization to the multimode case with any number of particle types is straightforward. In this case, it is clear that any state in the Fock space can be represented as a superposition of states of the form~[analogous to Eq.~\eqref{eq:psi_statespace_general}]
\begin{equation}\label{eq:statespace_mutual}
\ket{\psi}=\hat{\psi}^+_{a_1}\ldots \hat{\psi}^+_{a_n}\hat{\varphi}^+_{b_1}\ldots \hat{\varphi}^+_{b_n}|0\rangle,
\end{equation}
by using the second line of Eq.~\eqref{eq:fundamental_Rcommu-rela} to move all $\hat{\psi}^+$ to the left of $\hat{\varphi}^+$. 
A basis for the single mode Fock space can be constructed as
\begin{equation}\label{eq:full_basis-mutual}
	|{}^{\alpha_1}_{n_1},{}_{n_2}^{\alpha_2}\rangle=\hat{\Psi}^{+}_{n_1,\alpha_1}\hat{\Phi}^{+}_{n_2,\alpha_2}|0\rangle,%
\end{equation}
where as in Eq.~\eqref{eq:full_basis}, the numbers $\{(n_i,\alpha_i)\}^2_{i=1}$ can be chosen independently for different particle types %
and
$\hat{\Psi}^{(1)+}_{n,\alpha}$ is defined in Eq.~\eqref{eq:full_basis} with $\Psi^\alpha_{a_1 a_2\ldots a_n}$ satisfying Eq.~\eqref{eq:V_n_basis} with $R=R^{(\psi\psi)}$ and $\hat{\Phi}^{(2)+}_{n_2,\alpha_2}$ is given by the analogous expression for $\hat{\varphi}^+$. 
This basis is orthonormal if all the (self and mutual) $R$-matrices are unitary and the normalization in Eq.~\eqref{eq:singlemodewf_normalization} is used. The action of paraparticle operators $\hat{\psi}^\pm_{a}, \hat{\varphi}^\pm_{a}$ on the state space can be worked out in a similar way as in Sec.~\ref{sec:action_parafield_basis}. 

We now mention a useful trick that allows us to treat systems with multiple types of paraparticles using our previous formalism that involves only one type of paraparticle. Indeed, we can always combine all particle types in $\mathcal{T}$ into a single ``composite type'' as follows. Let $A=(\psi,a)$ be a collective index that labels the internal~(parastatistical) states of the composite particle, where $\psi\in\mathcal{T}$ and $a=1,2,\ldots,m_\psi$, and let  $\hat{\Psi}^\pm_{i,A}=\hat{\psi}^\pm_{i,a}$.  We define the $R$-matrix for the composite type paraparticle as follows: for $A=(\psi,a), B=(\varphi,b), C=(\chi,c)$ and $D=(\gamma,d)$, let 
\begin{equation}\label{eq:compositeR-def}
\mathbf{R}^{AB}_{CD}=\delta_{\psi\gamma}\delta_{\varphi\chi} [R^{(\varphi\psi)}]^{ab}_{cd}.
\end{equation} 
Then it is straightforward to verify  %
that Eq.~\eqref{eq:fundamental_Rcommu-rela} for all $\psi,\varphi\in\mathcal{T}$ is equivalent to 
\begin{eqnarray}\label{eq:fundamental_Rcommu-rela-full}
	\hat{\Psi}^-_{i,A} \hat{\Psi}^+_{j,B}&=&\sum_{CD}\mathbf{R}^{AC}_{BD} \hat{\Psi}_{j,C}^+ \hat{\Psi}_{i,D}^-+\delta_{AB}\delta_{ij},\nonumber\\%
	\hat{\Psi}^+_{i,A} \hat{\Psi}^+_{j,B}&=&\sum_{CD}\mathbf{R}^{CD}_{AB} \hat{\Psi}_{j,C}^+ \hat{\Psi}_{i,D}^+,\nonumber\\%
	\hat{\Psi}^-_{i,A} \hat{\Psi}^-_{j,B}&=&\sum_{CD}\mathbf{R}^{BA}_{DC} \hat{\Psi}_{j,C}^- \hat{\Psi}_{i,D}^-,%
\end{eqnarray}
which is exactly the original Eq.~\eqref{eq:fundamental_Rcommu} for $\hat{\Psi}^\pm_{i,A}$, and furthermore, the generalized YBE in Eq.~\eqref{eq:YBEgraphical-rela} for all $\psi,\varphi,\chi\in\mathcal{T}$ is equivalent to the original YBE for $\mathbf{R}$
\begin{eqnarray}\label{eq:YBE-composite}
	\mathbf{R}^2&=&\mathds{1}_{m^2},\nonumber\\
	\mathbf{R}_{12} \mathbf{R}_{23} \mathbf{R}_{12}&=& \mathbf{R}_{23} \mathbf{R}_{12} \mathbf{R}_{23}.
\end{eqnarray}
This trick allows us to treat a collection of different types of paraparticles as a single composite type, where most of our previous results apply.

As a simple example, consider a physical system with two types of paraparticles $\mathcal{T}=\{\psi,\varphi\}$, and let 
$[R^{(\psi\varphi)}]^{ab}_{cd}=\adjincludegraphics[height=6ex,valign=c]{Figures/MutualStatistics-1/Rpsiphi.png}$. 
 Then, in addition to the YBE for $R^{(\psi\psi)}$ and $R^{(\varphi\varphi)}$, %
 we also have the following two equations for their mutual $R$-matrix, which are obtained from the second relation in Eq.~\eqref{eq:YBEgraphical-rela},
\begin{equation}\label{eq:Rpsiphiconstraint}
\adjincludegraphics[width=9ex,valign=c]{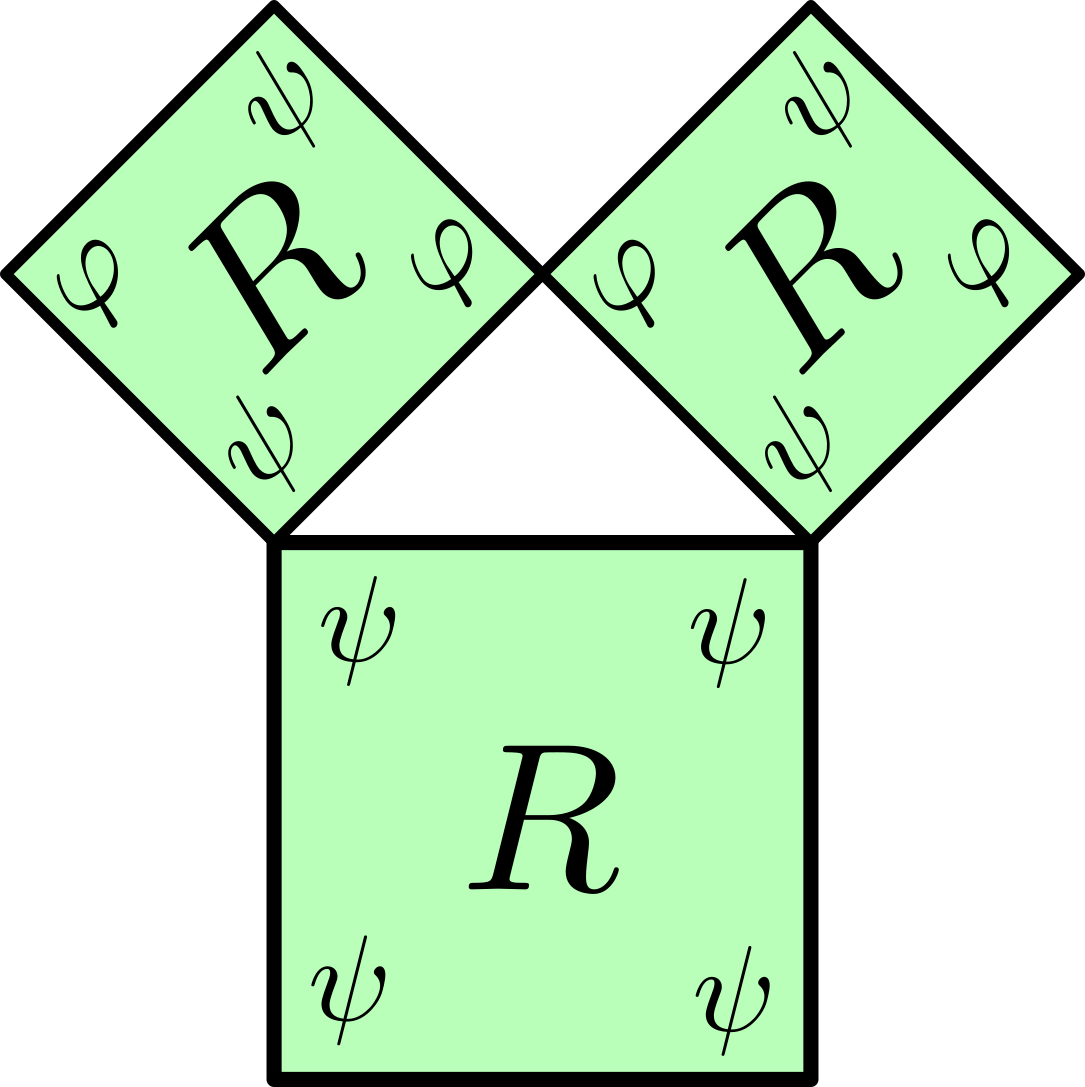}
=\adjincludegraphics[width=9ex,valign=c]{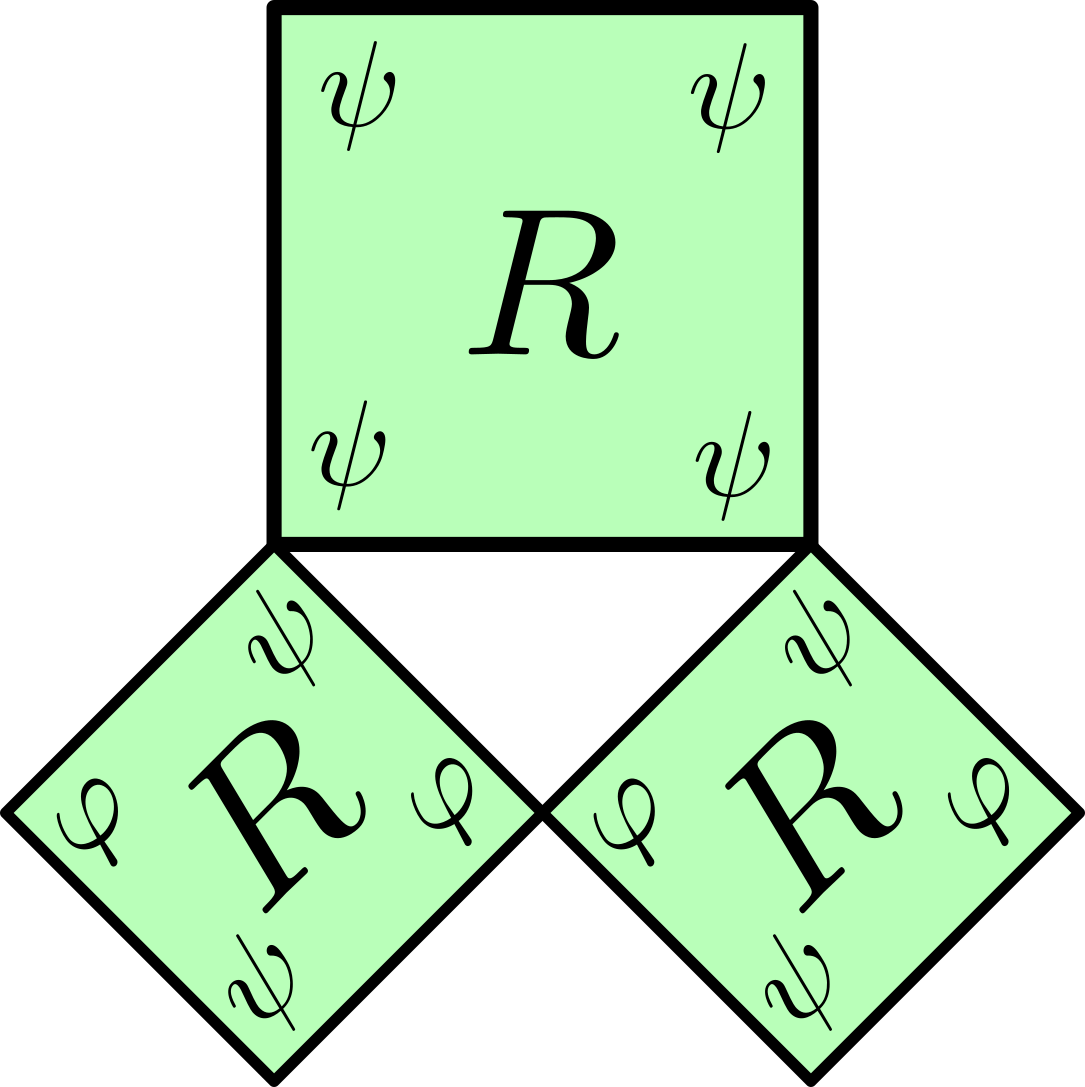},\quad 
\adjincludegraphics[height=9ex,valign=c]{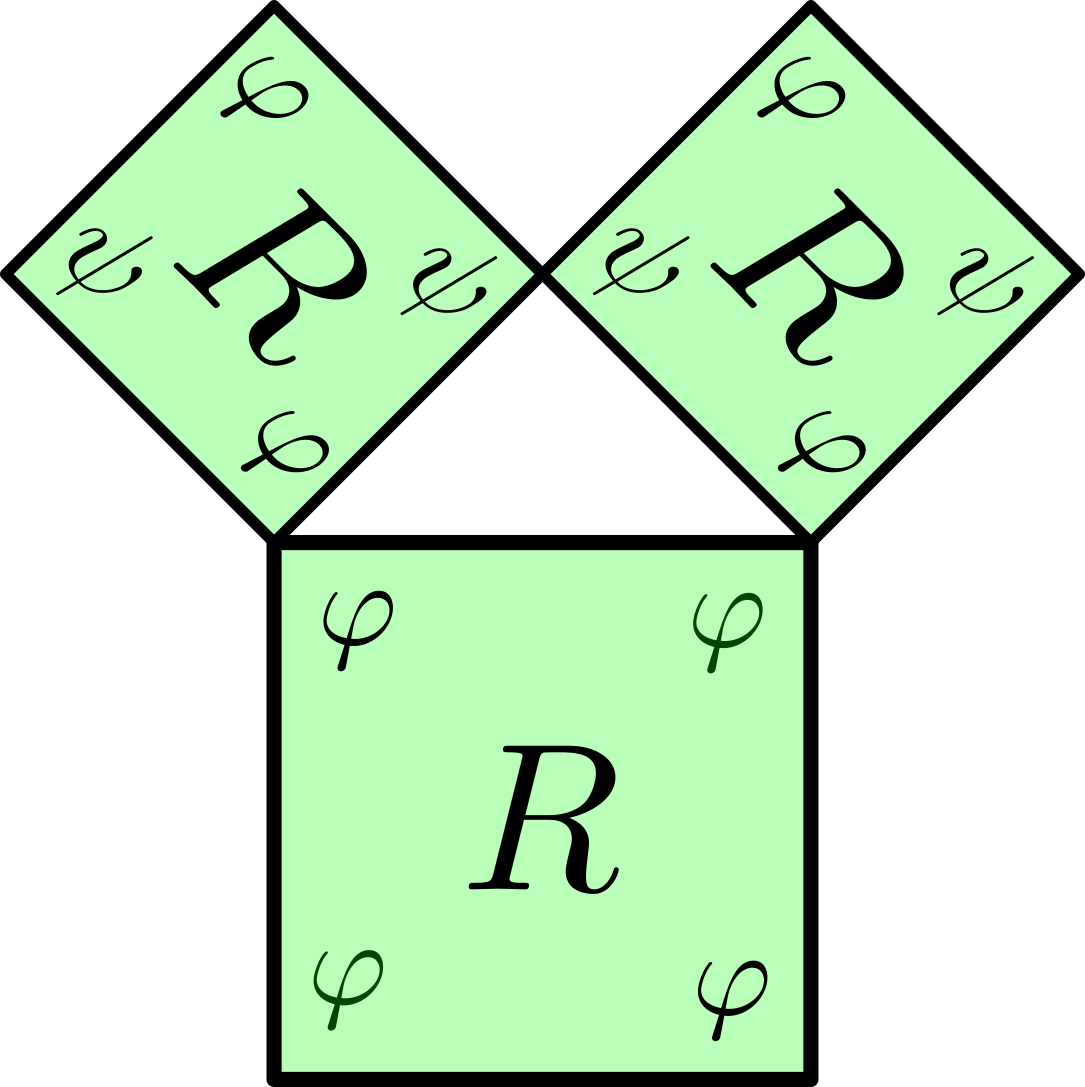}=
\adjincludegraphics[height=9ex,valign=c]{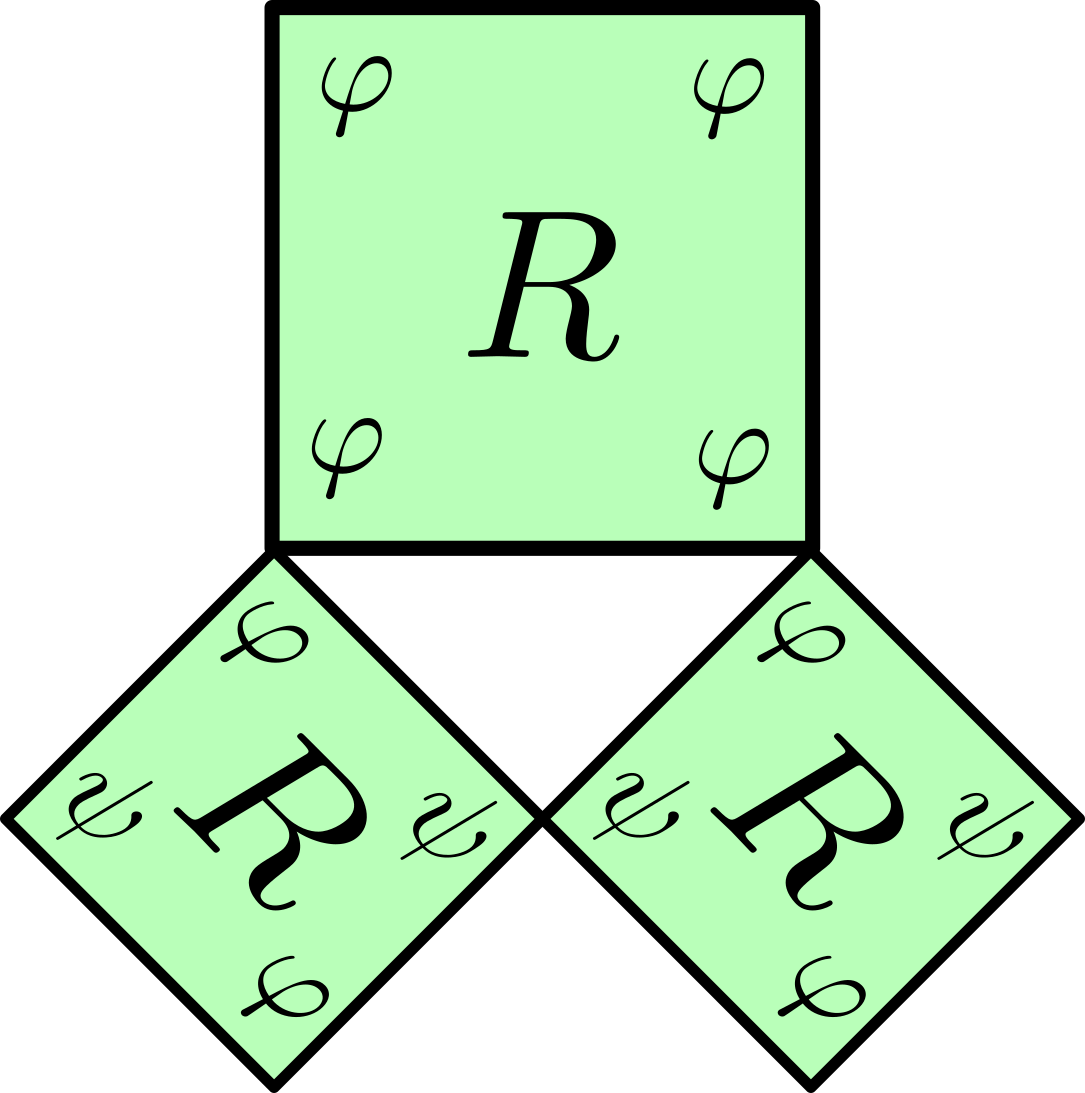}.
\end{equation}
Note that in principle we also have two similar equations for $[R^{(\varphi\psi)}]^{ab}_{cd}=
\adjincludegraphics[height=6ex,valign=c]{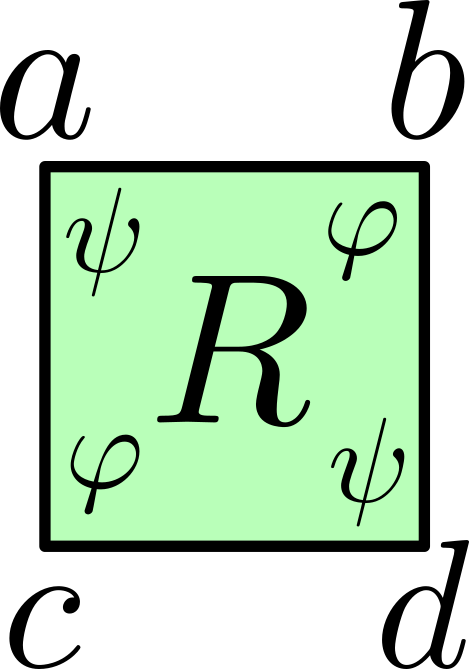}$; however, the first relation in Eq.~\eqref{eq:YBEgraphical-rela} fixes $R^{(\varphi\psi)}=[R^{(\psi\varphi)}]^{-1}$, and one can check that once $R^{(\varphi\psi)}$ satisfies Eq.~\eqref{eq:Rpsiphiconstraint}, the corresponding constraints for $R^{(\varphi\psi)}=[R^{(\psi\varphi)}]^{-1}$ are satisfied automatically. 

As an aside, we note that Eq.~\eqref{eq:Rpsiphiconstraint} along with Eq.~\eqref{eq:compositeR-def} gives a method of constructing new $R$-matrices out of existing ones. Namely, given any two $R$-matrices $R$ and $R'$ describing paraparticles $\psi$ and $\varphi$, respectively, we can solve Eq.~\eqref{eq:Rpsiphiconstraint}~(with $R^{(\psi\psi)}=R$ and $R^{(\varphi\varphi)}=R'$) to find a mutual $R$-matrix $S=R^{(\psi\varphi)}$ compatible with $R$ and $R'$, and then use Eq.~\eqref{eq:compositeR-def} to get a new $R$-matrix $\mathbf{R}$. We denote $\mathbf{R}=R\oplus_{S} R'$, and we call  $\mathbf{R}$ the direct sum of $R$ and $R'$ twisted by the mutual $R$-matrix $S$. We also call $\mathbf{R}$ a direct sum of $R$ and $R'$, due to the fact that Eq.~\eqref{eq:Rpsiphiconstraint} often have many different solutions, each solution leads to such a twisted direct sum. In particular, note that if we take $S^{ab}_{cd}=\delta_{bc}\delta_{ad}$ to be the swap matrix, then   Eq.~\eqref{eq:Rpsiphiconstraint} is always satisfied, and in this case we call $\mathbf{R}$ simply the (untwisted) direct sum of $R$ and $R'$, denoted by $\mathbf{R}=R\oplus R'$. As an example, the $R$-matrix of Green's paraparticle~(Ex.~\ref{ex:Green}) can be constructed by recursively using this direct sum construction, starting from the trivial $R$-matrix $R=\pm 1$.  
The following theorem is useful for computing the Hilbert series of the resulting composite $R$-matrix $\mathbf{R}$:
\begin{theorem} Given any system $\T$ of $R$-paraparticles with potentially nontrivial mutual statistics, the Hilbert series of the composite $R$-matrix $\mathbf{R}$ in Eq.~\eqref{eq:compositeR-def} is equal to the product of the Hilbert series of its individual components, i.e.,
\begin{equation}
z_{\mathbf{R}}(x)=\prod_{\psi\in\mathcal{T}}z_{R^{(\psi)}}(x). 
\end{equation}
\end{theorem}
\begin{proof}
The single mode particle number operator of the composite paraparticle $\Psi$ is decomposed as follows
\begin{equation}\label{eq:def_n_Psi}
	\hat{n}^{(\Psi)}= \sum_{\psi\in \T} \hat{n}^{(\psi)},
\end{equation}
which directly follows from definition. 
Since the spectrum of $\hat{n}^{(\psi)}$ are mutually independent for different particle types, as we mentioned below Eq.~\eqref{eq:full_basis-mutual},
we have 
\begin{eqnarray}\label{eq:single_mode_Z-mutual}
	z_{\mathbf{R}}(x)&=&\mathrm{Tr} \left[e^{-\beta \epsilon \hat{n}^{(\Psi)}}\right]\nonumber\\
    &=&\mathrm{Tr} \left[\prod_{\psi\in \T} e^{-\beta \epsilon \hat{n}^{(\psi)}}\right]\nonumber\\
    &=&\prod_{\psi\in\mathcal{T}}z_{R^{(\psi)}}(x),
\end{eqnarray}
where $e^{-\beta\epsilon}$ and we used the definition of $z_{\mathbf{R}}(x)$ in Eq.~\eqref{eq:single_mode_Z}.
\end{proof}

$R$-matrices constructed from the above direct sum construction will be called decomposable,  and we give the formal definition below
\begin{definition}\label{def:indecomposable}
An $R$-matrix is called decomposable if it is  equivalent~[with respect to the equivalence relation defined in Eq.~\eqref{eq:basistransformRmat}] to a direct sum of the form  $R_1\oplus_S R_2$, where both $R_1$ and $R_2$ have quantum dimension at least $1$, and it is called indecomposable otherwise. 
\end{definition}
The indecomposability of $R$-matrices will become important in later sections when discussing indistinguishability of $R$-paraparticles. 

Below we give some more concrete examples of systems of $R$-paraparticles with nontrivial mutual parastatistics. Let $M\geq 3$ be an integer. Consider a system with $M$ different types of paraparticles, with particle types labeled by elements of the cyclic group $\T=\Z/M\Z$. All of them %
have the same quantum dimension $m=M$, and we label the  basis states of their internal space also by elements of $\Z/M$. We define their mutual $R$-matrix~(including self-$R$-matrices as special cases) by  
\begin{eqnarray}\label{eq:RA4Z3-mutual}
R^{(\psi\varphi)}\ket{a,b}&=&(-1)^{\delta_{a,b+\psi}}\ket{b+\psi,a-\varphi},
\end{eqnarray}
with $\psi,\varphi,a,b\in \Z/M$, and all the arithmetic are understood modulo $M$. The composite $R$-matrix in Eq.~\eqref{eq:compositeR-def} for this system is given by
\begin{equation}\label{eq:RA4Z3-composite}
\mathbf{R}\ket{(\psi,a),(\varphi,b)}=(-1)^{\delta_{a,b+\psi}}\ket{(\varphi,b+\psi),(\psi,a-\varphi)}.
\end{equation}
It is straightforward to verify that $\mathbf{R}$ satisfies Eq.~\eqref{eq:YBE-composite}, for example, both sides of the second line maps $\ket{(\psi,a),(\varphi,b),(\chi,c)}$ to the state
$\ket{(\chi,c+\psi+\varphi),(\varphi,b+\psi-\chi),(\psi,a-\varphi-\chi)}$ times the phase factor $(-1)^{\delta_{a,b+\psi}+\delta_{b,c+\varphi}+\delta_{a,c+\psi+\varphi}}$. Notice that $R^{(0,0)}$ is the $R$-matrix for Green's para-fermion~[Ex.~\ref{ex:Green} with $(+)$ sign], and $R^{(1,1)}$ is the $R$-matrix in Ex.~\ref{ex:setth-ext}, and $R^{(0,1)}$ and $R^{(1,0)}$ define their mutual parastatistics. 

Interestingly, even ordinary fermions can have nontrivial mutual statistics. 
In the simplest case where $\psi$ and $\varphi$ are fermions, i.e., both $R$ and $R'$ are taken from Ex.~\ref{ex:decoupled} in Tab.~\ref{tab:Hilbert_series} with quantum dimensions $m$ and $m'$~(assume $m<m'$ without loss of generality), respectively,  the general solution to  Eq.~\eqref{eq:Rpsiphiconstraint} has the following form:
\begin{equation}
\adjincludegraphics[height=6ex,valign=c]{Figures/MutualStatistics-1/Rpsiphi.png}=\delta_{ad}\delta_{bc}\phi_{a,b}%
\end{equation}
where $\{\phi_{a,b}\}$ are arbitrary phases. 
\begin{figure}
	\center{\includegraphics[width=0.8\linewidth]{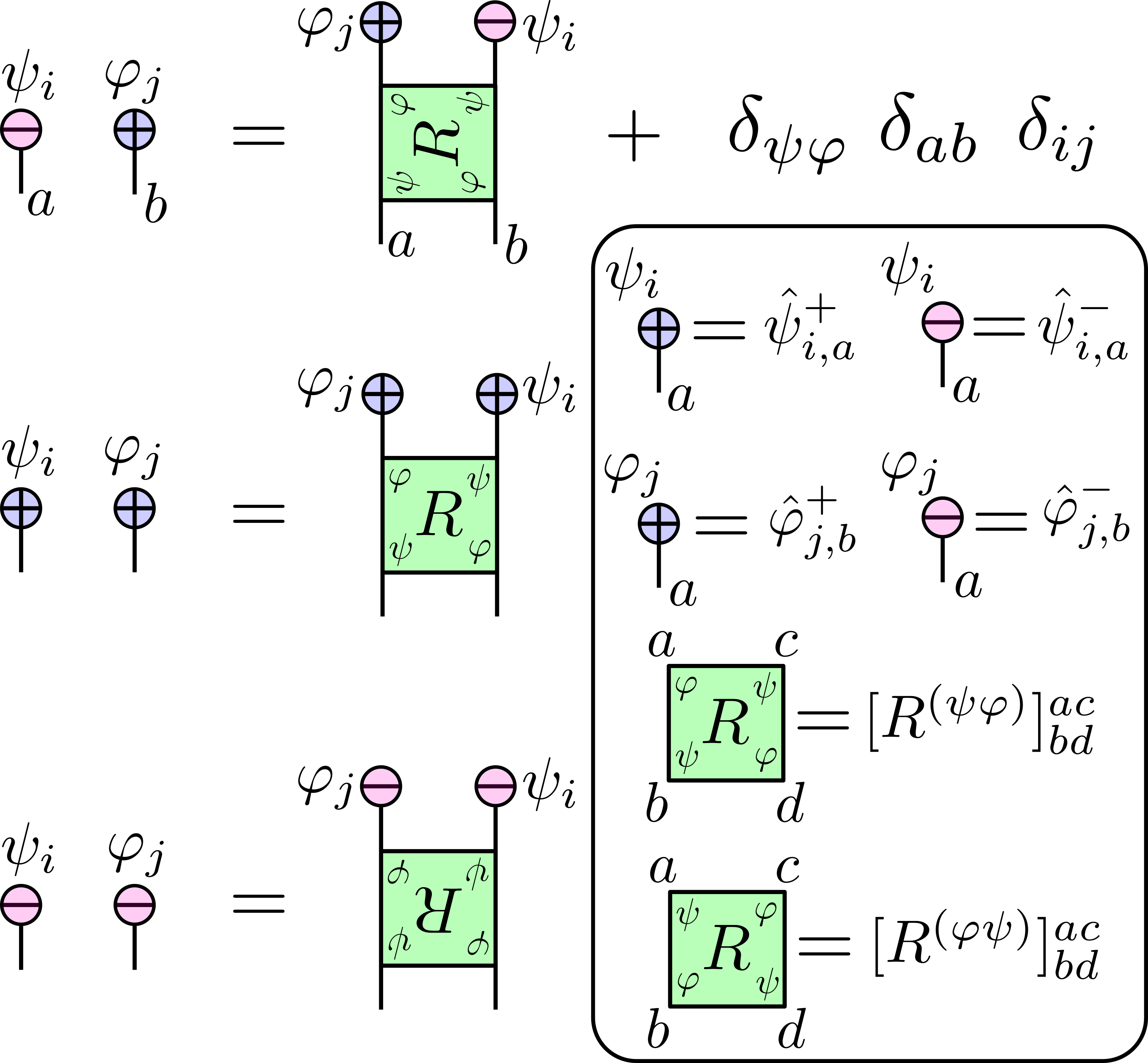}}
	\caption{\label{fig:RMQA-rela} Tensor graphical representation of the quadratic CRs in Eq.~\eqref{eq:fundamental_Rcommu-rela}, using the same convention as in Fig.~\ref{fig:RMQA}.  }
\end{figure}

We finally mention the definition of local observables 
for a system with different particle types involving nontrivial mutual parastatistics. In this more general case, one should replace $\hat{\psi}^\pm_{i,a}$ by $\hat{\Psi}^\pm_{i,A}$~[defined above in Eq.~\eqref{eq:compositeR-def}] in Definition~\ref{def:general_LO}; equivalently, this requires Eq.~\eqref{eq:generalLOcondition} to hold for all particle types $\psi\in\mathcal{T}$.
For example, the local operators $\hat{e}^{}_{ij}$  defined in Eq.~\eqref{eq:def_e_ab} can be generalized to
\begin{equation}\label{eq:def_e_ab_psi}
	\hat{e}^{(\psi)}_{ij}\equiv \sum^{m_\psi}_{a=1} \hat{\psi}^+_{i,a}\hat{\psi}^-_{j,a}
\end{equation} 
for all $\psi\in\mathcal{T}$.  %
Indeed, $\hat{e}^{(\psi)}_{ij}$ is a special case of the more general quadratic local operator $\hat{e}^{\kappa}_{ij}$ in Eq.~\eqref{eq:def_e_ijkappa}, if we use the composite  $\mathbf{R}$~[defined in Eq.~\eqref{eq:compositeR-def}] to treat a system of $R$-paraparticles, the connection can be seen as follows.
For any particle type $\psi\in\mathcal{T}$, we
define a matrix $p_\psi$, called the internal space projector to $\psi$, as $(p_\psi)_{BC}= \delta_{\psi\varphi}\delta_{\psi\chi}\delta_{bc}$, where we use the notation $B=(\varphi,b),~C=(\chi,c)$ introduced above Eq.~\eqref{eq:compositeR-def}. It is then straightforward to show that $\kappa=p_\psi$ satisfies Eq.~\eqref{eq:kappaDC}, and in this case $
\hat{e}^{\kappa}_{ij}$ in Eq.~\eqref{eq:def_e_ijkappa} is the same as $\hat{e}^{(\psi)}_{ij}$ in Eq.~\eqref{eq:def_e_ab_psi}. %

\subsection{Weak equivalence between $R$-paraparticles}\label{sec:weak-equivalence}
In Sec.~\ref{sec:YBE}, we defined two $R$-matrices to be equivalent if they are related by a basis transformation in the internal space, see Eq.\eqref{eq:basistransformRmat}. In this section we introduce a weaker equivalence relation between $R$-matrices, which 
helps us distinguish between two fundamentally different classes of $R$-paraparticles, and understand the boundary between $R$-paraparticles and ordinary fermions and bosons, their realization in higher dimensional topological phases, and the relation to no-go theorems.
\begin{definition}
	Let $R$ and $R'$ be %
    $R$-matrices with the same quantum dimension $m$. %
    An MPO intertwiner~(or a morphism) from $R$ to $R'$ is four-index tensor 
	$W^{pq}_{ab}=
	\scalebox{1}{
		\begin{tikzpicture}[baseline={([yshift=-.8ex]current bounding box.center)}, scale=.8]
			\umatr{0}{0}{W}{}
			\paraindices{0}{0}{a}{b}
			\quantumindices{0}{0}{p}{q}
		\end{tikzpicture}
	}
	$ with $1\leq a,b\leq m$ and $1\leq p,q\leq d_W$~(where $d_W\in\mathbb{Z}_{\geq 1}$), satisfying   
	\begin{equation}\label{eq:MPOintertwinerR}
		\begin{tikzpicture}[baseline={([yshift=.4ex]current bounding box.center)}, scale=.8]
			\umatrix{0}{0}{W}{}
			\umatrix{1}{0}{W}{}
			\Rmatrix{0.5}{-1}{R}
		\end{tikzpicture}= 
		\begin{tikzpicture}[baseline={([yshift=.4ex]current bounding box.center)}, scale=.8]
			\Rmatrix{0.5}{1}{R'}
			\umatrix{0}{0}{W}{}
			\umatrix{1}{0}{W}{}
		\end{tikzpicture}.	
	\end{equation}
	We denote Eq.~\eqref{eq:MPOintertwinerR} by $R\xrightarrow{W} R'$. 
	We say that two unitary $R$-matrices $R$ and $R'$ are weakly equivalent if there exists a dual unitary tensor $W$ such that $R\xrightarrow{W} R'$, and we denote this weak equivalence relation by $R\sim_W R'$, or simply $R\sim R'$.
\end{definition}
Note that if there exists an invertible MPO intertwiner with $d_W=1$, then the two $R$-matrices are equivalent $R\cong R'$, as previously defined in Eq.\eqref{eq:basistransformRmat}.

It is straightforward to verify the following:
\begin{eqnarray}
	R\xrightarrow{W} R'&\Leftrightarrow&  R'\xrightarrow{W^{-1}} R\nonumber\\
	R_1\xrightarrow{W} R_2,~R_2\xrightarrow{V} R_3&\Rightarrow&  R_1\xrightarrow{W\boxtimes V} R_3,
\end{eqnarray}
where $(W\boxtimes V)_{ab}=\sum_c W_{ac}\otimes V_{cb}$. Therefore, $\sim$ indeed consistently defines an equivalence relation in the mathematical sense.

In particular, if a unitary $R$-matrix satisfies $R\sim\pm X$, we say that the $R$-paraparticle is weakly equivalent to an ordinary boson~$(+)$ or fermion~$(-)$, 
and we define $D_R=\min\{d_W| R\sim_{W} \pm X\}$. %
In Ref.~\cite{wang2025secret} it is shown that any $R$-matrix describing a simple particle type $\psi$ in a symmetric fusion category $\calC$ is 
weakly-equivalent to the trivial $R$-matrix $R=\theta X$, where $\theta=\pm 1$ is the topological twist factor of $\psi$. 
Intuitively, $D_R$ measures the ``distance'' of a given type of $R$-paraparticle from ordinary fermions and bosons, and we will see later in Sec.~\ref{sec:upperboundImax} that $D_R$ sets an upper bound on the information transfer capability of the $R$-paraparticle. The value of $D_R$ for our primary examples of $R$-matrices are given in Tab.~\ref{tab:Hilbert_series}. 

It is straightforward to show~(see Proposition~\ref{prop:weakeqvHseries} below) that weak equivalence preserves the exclusion statistics of $R$-matrices, i.e., if $R\sim R'$, then $z_R(x)=z_{R'}(x)$, where $z_R(x)$ is the Hilbert series of $R$. Therefore, $R$-matrices that have nontrivial exclusion statistics, such as $R=-\mathds{1}$ with $m\geq 2$, are not weakly equivalent to trivial $R$-matrices of the form $\pm X$, and are therefore beyond the description of SFCs~[they are generally described by (possibly non-rigid) $\C$-linear symmetric monoidal categories with infinitely many isomorphic classes of simple objects]. 

In part II of this series, we will discuss how weak equivalence of $R$-paraparticles is physically implemented in condensed matter realizations of $R$-paraparticles, i.e., how two different models realizing weakly equivalent $R$-paraparticles are related at the microscopic level. In short, 
in 1D solvable spin model realizations, weak equivalence can be implemented by an MPO intertwiner of the spin chain Hamiltonian, or alternatively by adding auxiliary sites accompanied with a finite depth quantum circuit transformation. For realizations of $R$-paraparticles in higher dimensional topological phases, weak equivalence is naturally implemented by an invertible domain wall in the system. 

\begin{proposition}\label{prop:weakeqvHseries}
	Let $\rho_{R,n}$ denote the representation of the symmetric group $S_n$ generated by the $R$-matrix, as defined in Sec.~\ref{sec:YBE}. If two unitary $R$-matrices $R_1$ and $R_2$ are weakly equivalent, then we have\\
	(1) $\rho_{R_1,n}$ and $\rho_{R_2,n}$ are equivalent representations of $S_n$ for any positive integer $n$. In particular, their characters are equal, i.e., $\chi_{R_1,n}(\tau)=\chi_{R_2,n}(\tau),~\forall \tau\in S_n$;\\
	(2) they have the same Hilbert series: $z_{R_1}(x)=z_{R_2}(x)$.
\end{proposition}
\begin{proof}
	(1) is proved by noticing that a dual-unitary MPO intertwiner $W$ defines an isomorphism between the representations $\rho_{R_1,n}\otimes \id_W$ and $\id_W\otimes \rho_{R_2,n} $ on the tensor product space $V_m^{\otimes n}\otimes V_W$, where $V_W$ is the vector space corresponding to the horizontal indices of $W$ with $\dim(V_W)=d_W$. (2) follows from (1) since the Hilbert series $z_R(x)$ is completely determined by the character $\chi_{R,n}$ via Eq.~\eqref{eq:dn-char-rel}. 
\end{proof}

\section{The  exchange statistics of paraparticles}\label{sec:nontrivialstatistics}
In Sec.~\ref{sec:second_quantization} we established the second quantization formulation of parastatistics based on the $R$-matrix CRs in Eq.~\eqref{eq:fundamental_Rcommu}, and briefly explained how this formalism connects to the first quantization description of exchange statistics in Sec.~\ref{sec:intro}. In this section we present a more illuminating derivation of exchange statistics in the second quantization formalism. Specifically, we provide two different viewpoints to interpret parastatistics in our theory--a wavefunction interpretation~[Sec.~\ref{sec:relation_first_quantization}] and a physical~(operational) interpretation~[Sec.~\ref{sec:physical_interpretation}]. Then in Sec.~\ref{sec:observe_parastatistics} we show how to experimentally observe parastatistics based on the result of Sec.~\ref{sec:physical_interpretation}. Finally, in Sec.~\ref{sec:QI_statistics} we present some physical consequences and potential applications of parastatistics from the viewpoint of quantum information theory. 

\subsection{The wavefunction interpretation}\label{sec:relation_first_quantization}
In the following we provide an alternative viewpoint to understand parastatistics as a symmetry property of the many particle wavefunction. 
To begin, we note that an arbitrary $n$-particle state in the Fock space can always be represented as
\begin{equation}\label{eq:Psi_state_full}
	|\Psi\rangle=\frac{1}{\sqrt{n!}}\sum_{\substack{A,x_1,\ldots,x_n}}\Psi^{A}(x_1,\ldots,x_n)\hat{\psi}^+_{x_1,a_1}\ldots\hat{\psi}^+_{x_n,a_n}|0\rangle,
\end{equation}
where $\Psi^{A}(x_1,\ldots,x_n)$ is called the many particle wavefunction and we use $A=(a_1,a_2,\ldots,a_n)$ to collectively denote the internal states of the paraparticles. Due to the $R$-CRs between the creation operators $\{\hat{\psi}^+_{x,a}\}$ in Eq.~\eqref{eq:fundamental_Rcommu},  we can always assume without loss of generality that the many particle wavefunction satisfies the following symmetrization property
\begin{equation}
	\Psi^{A}( x_{\bar{\sigma}(1)},\ldots,x_{\bar{\sigma} (n)})=\sum_{A'}\rho_{AA'}(\sigma)\Psi^{A'}(x_1,\ldots,x_n),
\end{equation}
for any $\sigma\in S_n$, where $\bar{\sigma}=\sigma^{-1}$, and $\rho$ is the representation of $S_n$ realized by the $R$-matrix, as defined in Sec.~\ref{sec:YBE}. 
For example, in the case $n=3$ and $\sigma=(12)$, we have
\begin{equation}\label{eq:wavefuntion_exchange_Rmat-wf}
	\Psi^{a_1 a_2 a_3}(x_2,x_1,x_3)=\sum_{b_1,b_2} R^{a_1a_2}_{b_1 b_2 } \Psi^{b_1 b_2 a_3}(x_1,x_2,x_3).
\end{equation} 
Indeed, an unsymmetrized wavefunction $\tilde{\Psi}$ can always be symmetrized  in the following way~\cite{wang2023para}
\begin{equation}
\Psi^{A}(x_1,\ldots,x_n)=\frac{1}{n!}\sum_{A',\sigma\in S_n}\rho_{AA'}(\sigma)\tilde{\Psi}^{A'}( x_{\sigma(1)},\ldots,x_{\sigma (n)})
\end{equation}
This symmetrization does not change the overall quantum state $|\Psi\rangle$ due to the second line in Eq.~\eqref{eq:fundamental_Rcommu}, and it is straightforward to verify that the symmetrized wavefunction $\Psi$ satisfies Eq.~\eqref{eq:wavefuntion_exchange_Rmat-wf}. Eq.~\eqref{eq:wavefuntion_exchange_Rmat-wf} gives the wavefunction interpretation of parastatistics: when one swaps the arguments of the many particle wavefunction, its internal label $A$ transforms in the representation $\rho$ of the symmetric group $S_n$ generated by the $R$-matrix. For bosons, $\rho$ is  
the trivial representation, for fermions, it is the sign representation, while parastatistics correspond to more general higher dimensional representations.

\subsection{The physical interpretation}\label{sec:physical_interpretation}
In the following we derive the physical exchange statistics of paraparticles starting from the second quantization algebra in Eq.~\eqref{eq:fundamental_Rcommu}.
To begin, we introduce a convenient notation to label a general quantum state with $n$ paraparticles
\begin{equation}\label{eq:2ndQ_basis_n_particle_space}
	\ket{0;i_1^{a_1} i_2^{a_2} \ldots i_n^{a_n}}=\hat{\psi}^+_{i_1,a_1}\hat{\psi}^+_{i_2,a_2}\ldots\hat{\psi}^+_{i_n,a_n}\ket{0},
\end{equation}
where  $i_1,i_2,\ldots, i_n$ are the positions of the paraparticles, and in this section we assume that they are mutually different for simplicity. 
The numbers $a_1,a_2,\ldots,a_n\in \{1,2,\ldots,m\}$ label the internal states of the paraparticles, as before. When the positions $i_1,\ldots, i_n$ are fixed, there are $m^n$ linearly independent internal states, spanning a $m^n$-dimensional subspace which we denote as $\Hil_{i_1i_2 \ldots i_n}$.

\begin{figure}
\centering
\includegraphics[width=.48\linewidth]{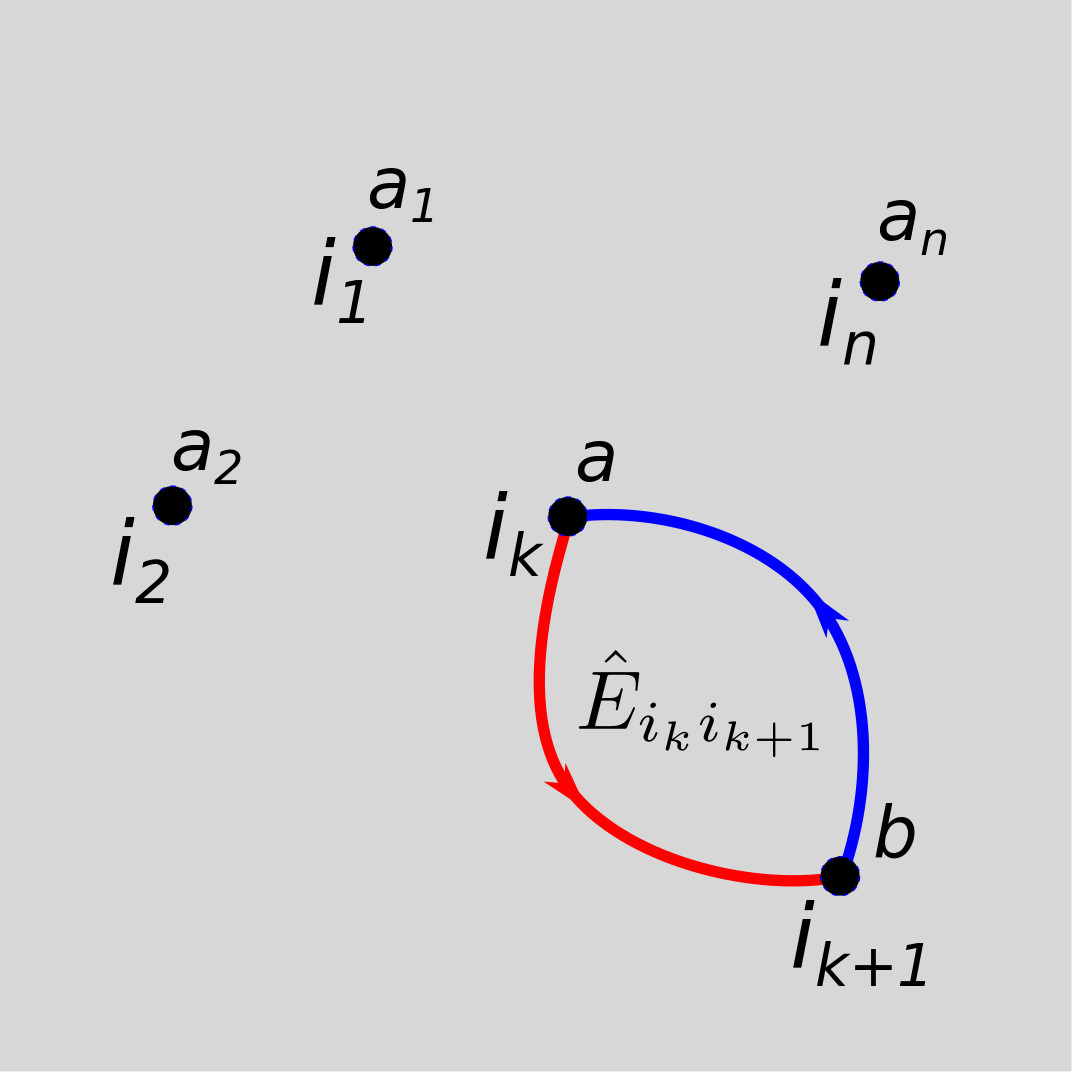}
\includegraphics[width=.48\linewidth]{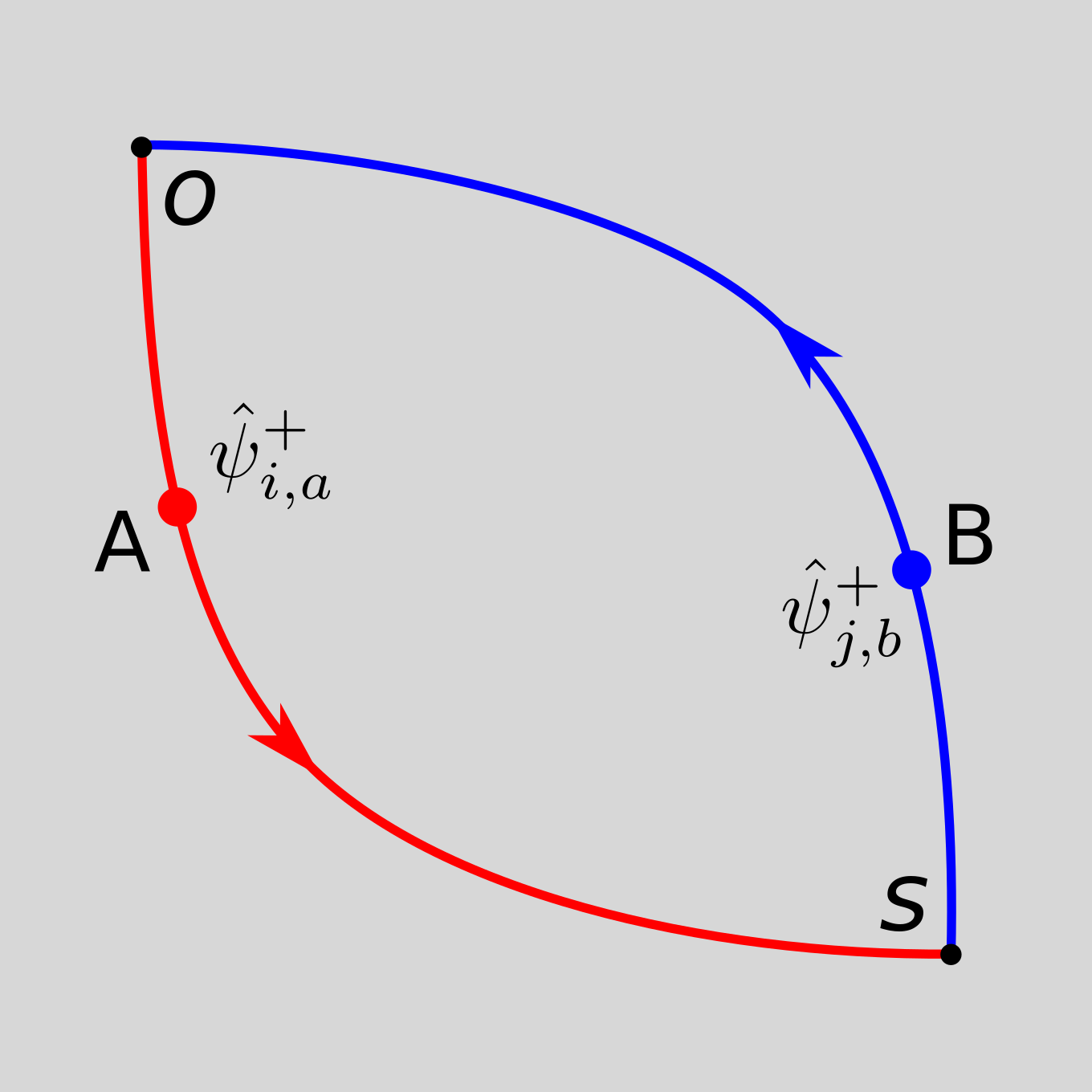}
\caption{\label{fig:physical_exchange} (left) Deriving the exchange statistics of paraparticles. $\hat{E}_{i_k i_{k+1}}$ is a unitary operator that exchanges the position of the two paraparticles at positions $i_k $ and $i_{k+1}$ along the designated paths; (right) Illustrating the secure communication protocol between Alice~(A) and Bob~(B), as discussed in Sec.~\ref{sec:secret_communication}.}
\end{figure}
Let $\hat{E}_{ij}$ be a unitary operator that induces the exchange of the positions of the paraparticles at $i$ and $j$:
\begin{equation}\label{eq:Uij_action}
	\hat{E}_{ij}\hat{\psi}^+_{i,a}\hat{E}_{ij}^\dagger=\hat{\psi}^+_{j,a},\quad	\hat{E}_{ij}\hat{\psi}^+_{j,a}\hat{E}_{ij}^\dagger=\hat{\psi}^+_{i,a},\quad \forall a.
\end{equation}
Note that such an operator can always be decomposed as a product of elementary local exchange operators of the form $\hat{E}_{kl}=e^{i\frac{\pi}{2} (\hat{e}_{kl}+\hat{e}_{lk})}$ along two paths connecting $i$ and $j$, as shown in Fig.~\ref{fig:physical_exchange}, and we assume that the two chosen paths do not go through any other particles.  
$\hat{E}_{ij}$ acts on the $n$ particle state Eq.~\eqref{eq:2ndQ_basis_n_particle_space} as 
\begin{eqnarray}\label{eq:braiding_derivation}
	\hat{E}_{ij}\ket{0;\ldots i^{a}j^{b}\ldots }&\equiv& \hat{E}_{ij}(\ldots)\hat{\psi}^+_{i,a}\hat{\psi}^+_{j,b}(\ldots)\ket{0}\nonumber\\
	&=&(\ldots)(\hat{E}_{ij}\hat{\psi}^+_{i,a}\hat{E}_{ij}^\dagger)(\hat{E}_{ij}\hat{\psi}^+_{j,b}\hat{E}_{ij}^\dagger)(\ldots)\ket{0}\nonumber\\
	&=&(\ldots)\hat{\psi}^+_{j,a}\hat{\psi}^+_{i,b}(\ldots)\ket{0}\nonumber\\
		&=&\sum_{a',b'}R^{b'a'}_{ab}(\ldots)\hat{\psi}^+_{i,b'}\hat{\psi}^+_{j,a'}(\ldots)\ket{0}\nonumber\\
	&=&\sum_{a',b'}R^{b'a'}_{ab}|0;\ldots i^{b'}j^{a'}\ldots \rangle,
\end{eqnarray}
where $\ldots$ collectively denotes other labels that are unaffected by the exchange, in the second line we applied Eq.~\eqref{eq:Uij_action} and the invariance of $\ket{0}$ under $\hat{E}_{ij}$, and in the third line we used the fundamental CR~\eqref{eq:fundamental_Rcommu}  between $\hat{\psi}^+_{j,a}$ and $\hat{\psi}^+_{i,b}$. 
Eq.~\eqref{eq:braiding_derivation} defines the physical meaning of the $R$-matrix as the unitary rotation of the $n$-particle state space $\Hil_{i_1i_2 \ldots i_n}$ that results from physically exchanging paraparticles. 
We emphasize, however, that experimentally preparing and measuring the state in Eq.~\eqref{eq:2ndQ_basis_n_particle_space} is challenging due to the nature of paraparticles being topological excitations, making it challenging to physically observe the effect of parastatistics.
Later in Sec.~\ref{sec:observe_parastatistics} we show a way to do this using a sequence of local operations and measurements, which shows a striking difference from ordinary fermions and bosons, leading to %
a potential application in quantum information theory~(Sec.~\ref{sec:QI_statistics}). %

We also remark that while the above derivation works in any spatial dimension, in the special case of 1D, exchange statistics is not uniquely defined~[there are different ways to exchange particles other than described in Eq.~\eqref{eq:Uij_action}], since exchanging two particles in 1D necessarily involve colliding them in the process, and the result generally depends on the microscopic details of the system and the exchange process, and is not robust against perturbations. By contrast, for the realization of paraparticles in  condensed matter systems in 2D or 3D~\cite{wang2023para,wang2024parastatistics,wang2025secret}, exchange statistics is uniquely defined and robust against local perturbations. 

\subsection{Observing the exchange statistics}\label{sec:observe_parastatistics}
In the physical interpretation of $R$-parastatistics, 
the $R$-matrix defines the unitary rotation of the internal states of the paraparticles when one swaps the position of two paraparticles. How can we experimentally observe this unitary rotation? 
This requires preparing the state used in the procedure~(which requires access to otherwise ``hidden'' internal index) and then reading out the otherwise-hidden internal index. At first glance, this seems to be forbidden by the local-indistinguishability criterion we established back in Sec.~\ref{sec:indistcrit}, where we stated explicitly below Eq.~\eqref{eq:EijPsitransform} that the hidden internal index of an $R$-parapaticle cannot be accessed by any local measurement. Yet there is still a way out. 

It turns out that, in condensed matter realizations of $R$-paraparticles, if an $R$-paraparticle gets close to a special type of point-like defect in the system, then it becomes possible to access its internal state using local operations. This does not contradict the local indistinguishability criterion in Sec.~\ref{sec:indistcrit}, as the paraparticles remain  locally indistinguishable in the bulk of the matter.  Such defects are called ``black defects'' in Ref.~\cite{wang2025secret} and are described and studied using the general theory of topological defects~\cite{Kitaev2012gappedboundary,kong2014braided,Barkeshli2019SETclassification} in 2D and 3D topological phases of matter. 
In the second part of this series, we will provide a systematic way to study black defects within the second quantization framework of $R$-parastatistics and present physical realizations of these black defects in the exactly solvable models with emergent $R$-paraparticle; for now, we only describe the relevant properties of black defects at a high level by introducing an additional assumption: %
\begin{assumption}\label{assump:observe_internal}
There exist %
two special points  $\oA,\oB$ in the physical system where one can measure the internal state of a paraparticle. More precisely, 
there exist observables $\hat{O}_{\oA}$, $\hat{O}'_{\oB}$ localized around $\oA,\oB$, respectively, satisfying
\begin{eqnarray}\label{eq:localmeasurementatcorner_0}
	\hat{O}_{\oA}\ket{0;i_1^{a_1} \ldots i_n^{a_n}}&=&a_1 \ket{0; i_1^{a_1} \ldots i_n^{a_n}},\text{ if } i_1=\oA,\nonumber\\
	\hat{O}'_{\oB}\ket{0;i_1^{a_1} \ldots i_n^{a_n}}&=&a_n \ket{0;i_1^{a_1} \ldots i_n^{a_n}},\text{ if } i_n=\oB.
\end{eqnarray}	
Note that if we instead have $i_k=\oA$ for some $k>1$~(or $i_k=\oB$ for some $k<n$), then we need to use the second line in Eq.~\eqref{eq:fundamental_Rcommu} to swap $i_k$ all the way to the front~(back) before applying Eq.~\eqref{eq:localmeasurementatcorner_0} to compute the action of  $\hat{O}_{\oA}$~(or $\hat{O}'_{\oB}$). %
\end{assumption}
The observable $\hat{O}_{\oA}$ in Eq.~\eqref{eq:localmeasurementatcorner_0} can be simply constructed as
\begin{equation}\label{eq:OoA}
\hat{O}_{\oA}=\sum_a a\hat{\psi}^+_{\oA,a}\hat{\psi}^-_{\oA,a}.
\end{equation} 
One can verify that the first line in Eq.~\eqref{eq:localmeasurementatcorner_0} follows from Eq.~\eqref{eq:fundamental_Rcommu}. 
Importantly, for any local observable $\hat{O}_S$ defined in Definition~\ref{def:general_LO} away from $\oA$~(i.e., $\oA\notin S$), we have $[\hat{O}_{\oA},\hat{O}_{S}]=0$~[which follows from Eq.~\eqref{eq:generalLOcondition}], therefore, introducing the ``extra'' local observable $\hat{O}_{\oA}$ only at $\oA$ does not break the locality of the quantum theory. 
The construction of $\hat{O}'_{\oB}$ in the second quantization formulation is trickier and will be presented in part II. %
The conclusion is that it can also be introduced in a way compatible with locality.  
Notice that the existence of the special points  $\oA,\oB$ breaks the translation invariance of the system, however, this is very natural in condensed matter systems with boundary or defects~(impurity). Indeed, in our 2D solvable quantum spin models constructed in Ref.~\cite{wang2023para}, %
$\oA,\oB$ are special points on the boundary, where both $\hat{O}_{\oA}$ and $\hat{O}'_{\oB}$ are strictly local observables in the spin model. 

Assumption~\ref{assump:observe_internal} extends the observable content of the theory beyond that discussed in Sec.~\ref{sec:locality}, and with this 
assumption, it is straightforward to experimentally observe parastatistics. For simplicity, let us consider the exchange of two paraparticles, i.e., the case $n=2$ in Eq.~\eqref{eq:braiding_derivation}, and choose $i=\oA$ and $j=\oB$. 
At $t=0$, right before the exchange, the two particles are located at $\oA,\oB$, and we measure $\hat{O}_{\oA}$ and $\hat{O}'_{\oB}$ to obtain their internal states $a$ and $b$, respectively, and the quantum state collapse into $|0;i^a j^b\rangle$.
After the  exchange process the quantum state of the system evolves to the last line of Eq.~\eqref{eq:braiding_derivation}, and we can measure $\hat{O}_{\oA}$ and $\hat{O}'_{\oB}$ again to see their internal state evolution. For example,
with the set-theoretical $R$-matrix in Eqs.~(\ref{eq:seth-R},\ref{eq:set-thR}), the final state is $-\ket{0;i^{b'}j^{a'}}$, where $(b',a')=r(a,b)$, and the measurement results $a'$ and $b'$ are definite. For example, if we start with $a=b=1$, we end up measuring $b'=4,a'=3$. That is, the internal states of the paraparticles undergo a non-trivial unitary rotation even though the two particles keep distance with each other throughout the whole process. 
Such a nontrivial rotation is present for all the unitary $R$-matrices in Tab.~\ref{tab:Hilbert_series} except Ex.~\ref{ex:decoupled}~(note that to observe this nontrivial rotation in Ex.~\ref{ex:Green}, one needs to measure in a different basis for internal states). 
This is in stark contrast with fermions and bosons, in which case we would measure $a'=a$ and $b'=b$, i.e. the indices are simply carried with the particles without any change.

\subsection{Physical consequence of the exchange statistics--a quantum information viewpoint}\label{sec:QI_statistics}
In this section we discuss some notable physical consequences and potential applications of the exchange statistics of paraparticles from  a quantum information viewpoint, in the form of a novel protocol for secret communication~\cite{wang2024parastatistics,wang2025secret}~[Sec.~\ref{sec:secret_communication}] and entanglement generation~[Sec.~\ref{sec:entanglement_generation}] over long distance that are robust against local noise and eavesdropping when realized in 2D or 3D topologically ordered systems. 

\subsubsection{Secret communication}\label{sec:secret_communication}
It has been shown~\cite{wang2024parastatistics,wang2025secret} that the exchange statistics of paraparticles as described in Secs.~\ref{sec:physical_interpretation} and \ref{sec:observe_parastatistics} leads to a novel secure communication protocol that is physically impossible to be intercepted by local measurements. More precisely, by exploiting  the exchange statistics in Eq.~\eqref{eq:braiding_derivation}, two parties holding paraparticles can communicate a message to each other merely by exchanging their positions and manipulating the internal states of the paraparticles, without ever coming close to each other, and without leaving any trace behind that can be detected by a third party.  To implement this protocol, we need the additional assumption~\ref{assump:observe_internal}~(the presence of black defects~\cite{wang2025secret} at special points $\oA,\oB$) %
that one can perform an arbitrary unitary rotation of the internal state of a paraparticle when it is near $\oA$ or $\oB$. Specifically, we assume that for any $m\times m$ unitary matrix $U\in\mathrm{SU}(m)$, there exist unitary operators $\hat{U}_{\oA}$, $\hat{U}'_{\oB}$ localized around $\oA,\oB$, respectively, satisfying
\begin{eqnarray}\label{eq:localoperationatcorner_0}
	\hat{U}_{\oA}\ket{0;i_1^{a_1} \ldots i_n^{a_n}}&=&\sum_{b_1}U_{ b_1 a_1} \ket{0; i_1^{b_1} \ldots i_n^{a_n}},\text{ if } i_1=\oA,\nonumber\\
	\hat{U}'_{\oB}\ket{0;i_1^{a_1} \ldots i_n^{a_n}}&=&\sum_{b_n}U_{ b_n a_n} \ket{0;i_1^{a_1} \ldots i_n^{b_n}},\text{ if } i_n=\oB.\nonumber\\
\end{eqnarray}	
Notice that analogous to Eq.~\eqref{eq:OoA}, $\hat{U}_{\oA}$ can be implemented by exponentiating local observables of the form
\begin{equation}\label{eq:colorLOatoA}
\hat{F}^{a,b}_{\oA}=\mu\hat{\psi}^+_{\oA,a}\hat{\psi}^-_{\oA,b}+\mathrm{h.c.},\quad \mu\in\C. 
\end{equation} 
As before, we have $[\hat{F}^{a,b}_{\oA},\hat{O}_{S}]=0$ for any local observable $\hat{O}_{S}$ supported on a bounded region $S$ away from $\oA$~(i.e., $\oA\notin S$), therefore, introducing $\hat{F}^{a,b}_{\oA}$ as an extra local observable only at $\oA$ does not break locality. 
The construction of a similar observable $\hat{F}^{\prime a,b}_{\oB}$ at $\oB$ is trickier and will be presented in part II. %

With this additional assumption, we now briefly explain the secure communication protocol. The process is illustrated in the right panel of Fig.~\ref{fig:physical_exchange}. In this protocol, two parties, who we call Alice and Bob, exchange their positions by going along the respective paths shown in Fig.~\ref{fig:physical_exchange}, keeping a large distance from each other throughout the process. If each of them carries a paraparticle and is allowed to perform arbitrary local quantum operations at their respective positions, they can secretly send a message to each other using the following strategy. At the beginning of the exchange process, Alice encodes her message into the internal state $a$ of her paraparticle, and similarly Bob encodes his message into the internal state $b$ of his paraparticle. This can be achieved by a combination of the local measurement in Eq.~\eqref{eq:localmeasurementatcorner_0} and local unitary operation in Eq.~\eqref{eq:localoperationatcorner_0} at $\oA$ and $\oB$, respectively. After this, the state of the whole system is 
$\ket{\Psi(0)}=\ket{0;\oA^a \oB^b}$.  When the exchange is complete, the state evolves to
\begin{eqnarray}\label{eq:Psi3}
	\ket{\Psi(T)}=\ket{0;\oB^a \oA^b}=\sum_{a',b'}R^{b'a'}_{ab}\ket{G;\oA^{b'} \oB^{a'}},
\end{eqnarray}
where we use Eq.~\eqref{eq:braiding_derivation}. Then Alice and Bob measure $\hat{O}'_{\oB}$ and $\hat{O}_{\oA}$ in Eq.~\eqref{eq:localmeasurementatcorner_0} 
to obtain $a'$ and $b'$, respectively. Importantly, as long as $R^{b'a'}_{ab}$ is not of the trivial product form $R^{b'a'}_{ab}\neq p_{a'a}q_{b'b}$, knowing $(a,a')$ allows Alice to extract some information about $b$, and similarly for Bob, allowing them to communicate.

As a concrete example, consider the set-theoretical $R$-matrix in Eq.~\eqref{eq:seth-R}.
With this $R$-matrix, the measurement results $a',b'$ are definite, and in this case,
knowing $(b,b')$ completely determines $a$~(we assume that both players know the $R$-matrix beforehand), and similarly, knowing $(a,a')$ completely determines $b$. For example, if Alice has $(a,a')=(2,3)$, then she searches in the second row of the matrix for the column that has $a'=3$. This turns out to be the fourth column, so she determines that $b=4$ and $b'=1$.  
This allows the players to secretly send two bits of information to each other.

This protocol demonstrates a  clear physical distinction between paraparticles and ordinary fermions or bosons with an internal symmetry, such as color or flavor. The statistics of the latter is described by $R^{b'a'}_{ab}=\pm\delta_{a'a}\delta_{b'b}$. 
With such an $R$-matrix, Alice will simply obtain $a'=a$ and Bob will obtain $b'=b$, which contains no information about each other's number. 

In the condensed matter realization of emergent paraparticles~\cite{wang2023para} that we introduce in Part II, 
this protocol is robust against any kind of eavesdropping via local measurements: during the exchange process, when both Alice and Bob are far away from $\oA$ and $\oB$, an eavesdropper cannot obtain any information about the players' numbers $a$ and $b$ using any local measurements anywhere in the system, even including measurements at $\oA,\oB$ or near the position of Alice and Bob. The reason is that when the paraparticles are far away from $\oA,\oB$, the information $a,b$ is stored non-locally in the system and is topologically protected against local measurements. This is discussed in detail in Refs.~\cite{wang2024parastatistics,wang2025secret}.
We also point out that mutual parastatistics can also be used to carry out this secret communication protocol~[for example, using the mutual $R$-matrix in Eq.~\eqref{eq:RA4Z3-mutual}], using essentially the same strategy mentioned above~(the only difference is that Alice and Bob carry different types of $R$-paraparticles, and the specific encoding and decoding algorithms generally depend on the $R$-matrix).

\subsubsection{Entanglement generation}\label{sec:entanglement_generation}
The protocol described in Sec.~\ref{sec:secret_communication} also provides a robust way to create quantum entanglement over long distance. For example, consider again the set-theoretical $R$-matrix in Eq.~\eqref{eq:seth-R}. If Alice and Bob initialize the internal state of their paraparticles to be $\ket{1}$ and $(\ket{1}+\ket{2}+\ket{3}+\ket{4})/2$, then using Eq.~\eqref{eq:Psi3}, after the exchange, the internal state of the two paraparticles will become the maximally entangled state $(\ket{34}+\ket{21}+\ket{42}+\ket{31})/2$. Importantly,  using a similar argument given in Refs.~\cite{wang2024parastatistics,wang2025secret},  the amount of quantum entanglement that can be generated this way does not decay in the distance between $\oA$ and $\oB$ even when noise is present. 
This again demonstrates 
a clear physical distinction between paraparticles and ordinary fermions or bosons with an internal symmetry~(or more generally $R$-matrices of the trivial product form $R^{b'a'}_{ab}= p_{a'a}q_{b'b}$), where no entanglement can be generated this way.

\subsubsection{Upper bounds on information transfer and entanglement generation}\label{sec:upperboundImax}
We have seen that the secret communication protocol proposed in Sec.~\ref{sec:secret_communication} allows a small amount of information~(not exceeding $2\log m$) to be transferred between the players in one round of exchange, if each party holds one paraparticle. Naively, one might expect that by holding more and more paraparticles, 
the players may be able to transfer an arbitrary amount of information in one round of exchange. Yet this is not the case for any type of $R$-paraparticle currently known to be realizable in 2D or 3D topological phases, all of which are described by symmetric fusion categories and satisfy the following upper bound~(proof will be given in part II)
\begin{equation}\label{eq:parapartilce-infobound}
I_{\max}\leq E_{\max}\leq 2\log(D_R),
\end{equation}
where $D_R$ is the minimal twisting dimension of the $R$-matrix  defined in Sec.~\ref{sec:weak-equivalence}, %
and $I_{\max}$~($E_{\max}$) is the maximal amount of information that can be transferred~(the maximal amount of entanglement that can be generated) in one round of exchange, if each player is allowed to hold an arbitrary number of paraparticles along with a quantum ancilla of any dimension. %

More generally, in any known 2D or 3D topological phases, we have the following conjectured upper bound
\begin{equation}\label{eq:universal_infobound}
I_{\max}(\ket{G})\leq E_{\max}(\ket{G})\leq 2\log(D)\leq S_{\mathrm{top}}(\ket{G}),
\end{equation}
where $I_{\max}(\ket{G})$ is the maximal amount of information that can be transferred in one round of exchange~(subject to the constraint stated at the beginning of Sec.~\ref{sec:secret_communication}), if each player is allowed to hold any number of quasiparticles of any type~(even including non-Abelian anyons in 2D) in the system with gapped ground state $\ket{G}$. 
It is straightforward to show~\cite{wang2025secret} that both $I_{\max}(\ket{G})$ and $E_{\max}(\ket{G})$ are invariants under local unitary transformations~\cite{Chen2010LUT}, i.e., $I_{\max}(\ket{G})=I_{\max}(\hat{U}\ket{G})$ and  $E_{\max}(\ket{G})=E_{\max}(\hat{U}\ket{G})$ for any gapped ground state $\ket{G}$ and local unitary transformation $\hat{U}$, and therefore both quantities only depend on the universal long-distance properties of the underlying phase of matter.  
The number $D$ in Eq.~\eqref{eq:universal_infobound} is the total quantum dimension of all point-like topological quasiparticles in the system, and $S_{\mathrm{top}}$ is the topological entanglement  of $\ket{G}$. Note that %
in any topological phase hosting an $R$-paraparticle, we always have $D_R\leq D$, as $D_R$ is also the total quantum dimension of the symmetric fusion subcategory generated by the $R$-paraparticle. 

These upper bounds have an important implication that the exotic type of $R$-paraparticle in Ex.~\ref{ex:1m} is unlikely to be realizable in 2D or 3D gapped phases of matter, as we argue below. 

\subsubsection{Violation of the maximal information bound and the realizability of beyond-SFC paraparticles}
We now study the quantum information properties  of $R$-paraparticles in Ex.~\ref{ex:1m}, with 
$R=-\mathds{1}_{m\times m}$. 
We will see that this type of paraparticles exhibits a fundamentally different behavior from all those $R$-paraparticles known to be realizable in topological phases, in that it allows the two parties to transfer an unbounded amount of information in a single round of the exchange, thereby violating the upper bounds in Eqs.~\eqref{eq:parapartilce-infobound} and \eqref{eq:universal_infobound}. 

In the following we assume this type of paraparticle to be ``realizable'' in the sense that there exists a gapped phase of matter in 2D or 3D hosting this type of paraparticle, with two far-separated black defects satisfying Assumption~\ref{assump:observe_internal}~(or more precisely, all the axioms of emergent parastatistics formulated in Ref.~\cite{wang2025secret} with $R=-\mathds{1}_{m\times m}$). 

The protocol we give here is similar to the one given in Sec.~\ref{sec:secret_communication}, but this time each player holds $n$ paraparticles. Specifically, at the beginning of the exchange process, Alice encodes $n$ dits of information $a_1a_2\ldots a_n$ into the internal states of her paraparticles, and similarly for Bob, using a combination of the local measurement in Eq.~\eqref{eq:localmeasurementatcorner_0} and local unitary operation in Eq.~\eqref{eq:localoperationatcorner_0} at $\oA$ and $\oB$, respectively. 
The state of the whole system becomes
\begin{equation}
\ket{\Psi(0)}=\ket{0;i_1^{a_1}i_2^{a_2}\ldots i_n^{a_n}j_1^{b_1}j_2^{b_2}\ldots j_n^{b_n}}.
\end{equation}
Here we assume that Alice's paraparticles are located at $n$ different points around Alice's position $i_1,i_2, \ldots, i_n$, to avoid subtleties related to exclusion statistics, and similarly for Bob.  
After the exchange is complete, let us assume that Alice's paraparticles are  at $\tilde{i}_1,\tilde{i}_2,\ldots,\tilde{i}_n$, close to $\oB$, and similarly for Bob. The state evolves to
\begin{eqnarray}\label{eq:Psi3-Ex3}
\ket{\Psi(T)}&=&\ket{0;\tilde{i}_1^{a_1}\tilde{i}_2^{a_2}\ldots \tilde{i}_n^{a_n}\tilde{j}_1^{b_1}\tilde{j}_2^{b_2}\ldots \tilde{j}_n^{b_n}}\nonumber\\
	&=&(-1)^{n^2}\ket{0;\tilde{j}_1^{a_1}\tilde{j}_2^{a_2}\ldots \tilde{j}_n^{a_n}\tilde{i}_1^{b_1}\tilde{i}_2^{b_2}\ldots \tilde{i}_n^{b_n}}
\end{eqnarray}
where we have used the second line of Eq.~\eqref{eq:fundamental_Rcommu} $n^2$ times in total. Since $\tilde{j}_1,\tilde{j}_2,\ldots,\tilde{j}_n$ are now  close to $\oA$, Bob can now measure local observables of the form in Eq.~\eqref{eq:colorLOatoA} to obtain $a_1,\ldots,a_n$, and similarly Alice can measure $b_1,\ldots, b_n$ at $\oB$. This allows the players to transfer an unbounded amount of information in a single round of the exchange process. By contrast, we have seen in Eq.~\eqref{eq:universal_infobound} that for any quasiparticle in known 2D/3D topological phases described by braided/symmetric fusion categories, the total amount of information that can be transferred in a single round of exchange is bounded from above even if the players can carry arbitrary number of quasiparticles.  
Moreover, the violation of the last inequality in Eq.~\eqref{eq:universal_infobound} suggests that this type of paraparticles are unlikely realizable in 2D or 3D gapped phases as it implies an unbounded topological entanglement entropy. Still, it is hard to give a fully rigorous proof to rule them out, since to our knowledge, there is currently no proof that the $S_{\mathrm{top}}$ must be upper bounded in any 2D and 3D gapped ground state. Notice that an unbounded   $S_{\mathrm{top}}$ does not necessarily contradict the area-law conjecture: even if one can prove an area-law inequality $c_1|\partial A|+c_1'\leq S(A)\leq c_2|\partial A|+c_2'$, we still cannot obtain an upper bound on $S_{\mathrm{top}}$ unless we have an ultra-strong inequality with $c_1=c_2$, due to having negative terms in $S_{\mathrm{top}}$. Furthermore, the realization of this type of paraparticles in much-less-understood 2D or 3D gapless phases is also possible.

\section{Pair creation and antiparticles}\label{sec:pair-creation-AP}
In Sec.~\ref{sec:locality} we provided a general definition of local observables~(Definition~\ref{def:general_LO}) in an $R$-paraparticle theory and explicitly constructed a basic family of local observables generated by the contracted bilinear operators $\hat{e}_{ij}$ in Eq.~\eqref{eq:def_e_ab}.
Hamiltonians and local observables constructed from  $\{\hat{e}_{ij}\}$ all %
have a global $U(1)$ symmetry $\hat{\psi}^\pm_{i,a}\to e^{\pm i\phi}\hat{\psi}^\pm_{i,a} $, corresponding to conservation of total particle number, $[\hat{n},\hat{e}_{ij}]=0$. 
In this section, we will show that, for a  certain class of $R$-matrices, we can construct 
a family of local operators that breaks this $U(1)$ particle number symmetry and create pairs of paraparticles from the vacuum. Such pair creation operators naturally define the notion of antiparticles in an $R$-paraparticle theory, which is crucial for several physical reasons. 
First, particle-antiparticle pair creation is crucial for constructing relativistic quantum field theories~(QFTs) of $R$-paraparticles, since all free particle QFTs we know have pair creation, and
there are theorems saying that for relativistic QFTs satisfying mild conditions, every particle has an antiparticle~\cite{Fredenhagen1981Antiparticles,WightmanPCTbook}. 
Second, the pair creation and annihilation operators we introduce are $R$-parastatistical analogs of the $\hat{\psi}^\dagger_i\hat{\psi}^\dagger_j$ terms found in the Bogoliubov mean-field theory of fermionic superconductors~\cite{Bogoljubov1958,Valatin1958}, and they preserve the solvability of bilinear Hamiltonians, as we show explicitly in Sec.~\ref{sec:u1_breaking_solu}. This naturally leads to a Bogoliubov-type mean-field theory of $R$-paraparticle superconductors, which can be realized in a new family of exactly solvable quantum spin systems generalizing the ones constructed in Ref.~\cite{wang2023para}~(to be presented in part II). %
This also gives an $R$-parastatistical generalization of Majorana fermions, as we present in Sec.~\ref{sec:gen_Majorana}.  
Third, the classification of unitary $R$-matrices with pair creation naturally lead to several important physical  concepts of $R$-paraparticles, such as topological twist factor and the Frobenius-Schur indicator, which are gauge-invariant, intrinsic topological properties, and directly connect to corresponding concepts in TQFT and topological order literature~\cite{wittenQuantumFieldTheory1989,Reshetikhin-Turaev1991,kitaev2006anyons,kong2014braided,simon2020topological_protobook}. 

In this following, we begin by defining pair creation for self-dual $R$-paraparticles in Sec.~\ref{sec:U1breaking}, and then present the more general case in Sec.~\ref{sec:pair-creation-nonSD}.

\subsection{Pair creation for self-dual paraparticles}\label{sec:U1breaking}
We begin in Sec.~\ref{sec:condLOpaircreation} by constructing pair creation and annihilation operators and define the relevant class of $R$-matrices compatible with local pair creation. 
Then in Sec.~\ref{sec:LApaircreationbilin} we study the Lie algebra structure generated by all the bilinear operators including pair creation and annihilation, which is important for solving bilinear Hamiltonians in Sec.~\ref{sec:u1_breaking_solu}. In Sec.~\ref{sec:TPSpin-FSI-PHS} we present a classification theorem of unitary $R$-matrices compatible with pair creation, and discuss relevant physical concepts such as topological twist factor, Frobenius-Schur indicator, and particle-hole symmetry. 
Finally in Sec.~\ref{sec:gen_Majorana} we present an $R$-parastatistical generalization of Majorana fermions, which is sometimes more convenient for describing self-dual $R$-paraparticles.

\subsubsection{Constructing local pair creation operators}\label{sec:condLOpaircreation}
In general, we call an $R$-paraparticle $\psi$ self-dual if there exists a local operator $\hat{O}$~(in the sense of Definition~\ref{def:general_LO}) that creates a pair of $\psi$ from the vacuum $\ket{0}$. 
More concretely, we introduce the following ansatz for pair annihilation and creation operators~\footnote{While there can also exist non-bilinear local pair creation operators, it is straightforward to show, using the normal ordering of paraparticle operators, that if an $R$-paraparticle theory admits any local pair creation operator, then it must also admit bilinear pair creation operators. Since our main purpose here is to understand which subclass of $R$-matrices admit local pair creation, there is no loss of generality by restricting to the simpler class of bilinear pair creation operators. }:
\begin{eqnarray}\label{def:e_pm_ab}
\hat{e}^{-}_{ij}&=&\sum_{a,b}\alpha_{ba}\hat{\psi}_{i,a}^-\hat{\psi}_{j,b}^-,\nonumber\\	
	\hat{e}^{+}_{ij}&=&-\sum_{a,b}\alpha'_{ab}\hat{\psi}_{i,a}^+\hat{\psi}_{j,b}^+,
\end{eqnarray}
where the tensors $\alpha$ and $\alpha'$ are $m\times m$ constant matrices 
that are introduced to satisfy the locality condition in Eq.~\eqref{eq:generalLOcondition}.  %
[The minus sign convention in the second line of Eq.~\eqref{def:e_pm_ab} will be explained later in Sec.~\ref{para:herm_paircreation}.]  
More precisely, we require that all the commutators $[\hat{e}^{\pm}_{ij},\hat{\psi}^+_{k,c}]$ and $[\hat{e}^{\pm}_{ij},\hat{\psi}^-_{k,c}]$ 
vanish when $k\notin\{i,j\}$~\footnote{Actually, this condition is sufficient  for locality and solvability but not necessary. Indeed, for locality and solvability of bilinear Hamiltonians, we only need $[\hat{e}^\mu_{ij},\hat{e}^\nu_{kl}]=0$ whenever $\{i,j\}\cap\{k,l\}=\emptyset$, where $\mu,\nu\in\{+,-\}$. The problem about this generalization is that if $[\hat{e}^{\pm}_{ij},\hat{\psi}^+_{k,c}]\neq 0$ for $k\notin\{i,j\}$, then 
$\hat{\psi}^+_{k,c}$ loses its physical meaning of creating one point-like particle at position $k$, as it may excite local observables $\hat{e}^{\pm}_{ij}$ far away, and the  state $\hat{\psi}^+_{k,c}\ket{0}$ may not even have finite excitation energy in the thermodynamic limit. 
We therefore do not discuss this more general possibility in this paper.
}. 
This sets a constraint on the $R$-matrix and the tensors $\alpha,\alpha'$, which we derive in the following. 
We begin by computing $[\hat{e}^-_{ij}, \hat{\psi}^+_{k,b}]$: 
\begin{eqnarray}\label{eq:commu_Eabp_psi}
	[\hat{e}^-_{ij}, \hat{\psi}^+_{k,b}]&=&
    \adjincludegraphics[width=9ex,valign=c]{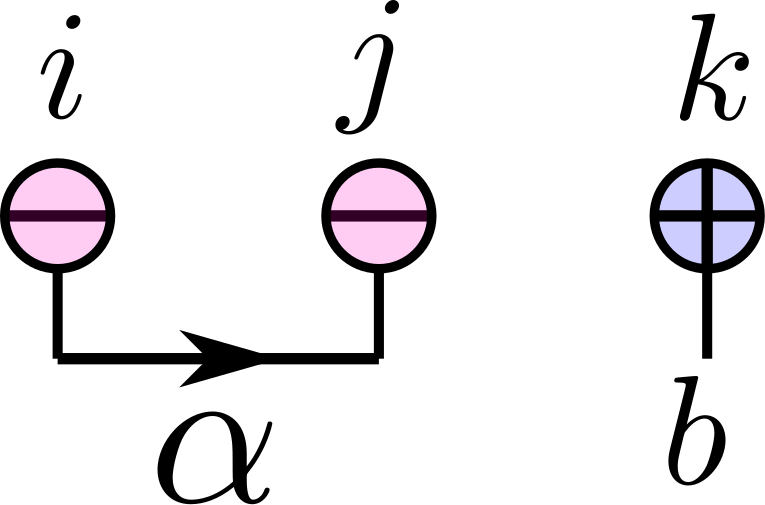}~~-~~
    \adjincludegraphics[width=9ex,valign=c]{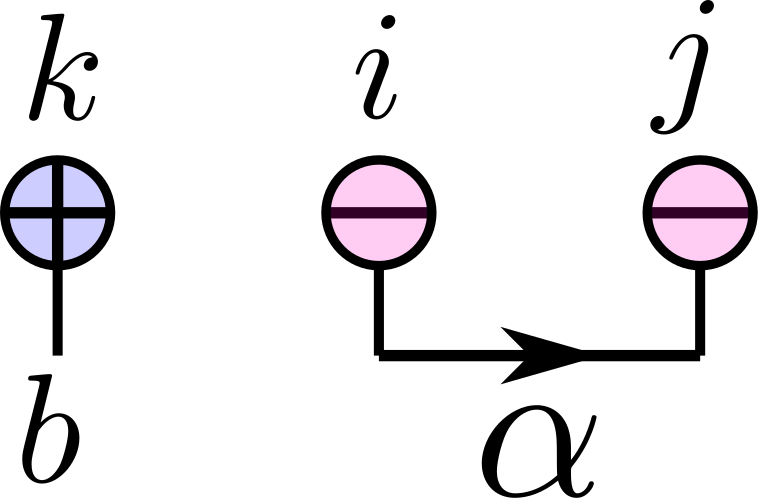}\nonumber\\
	&=&
    \adjincludegraphics[width=9ex,valign=c]{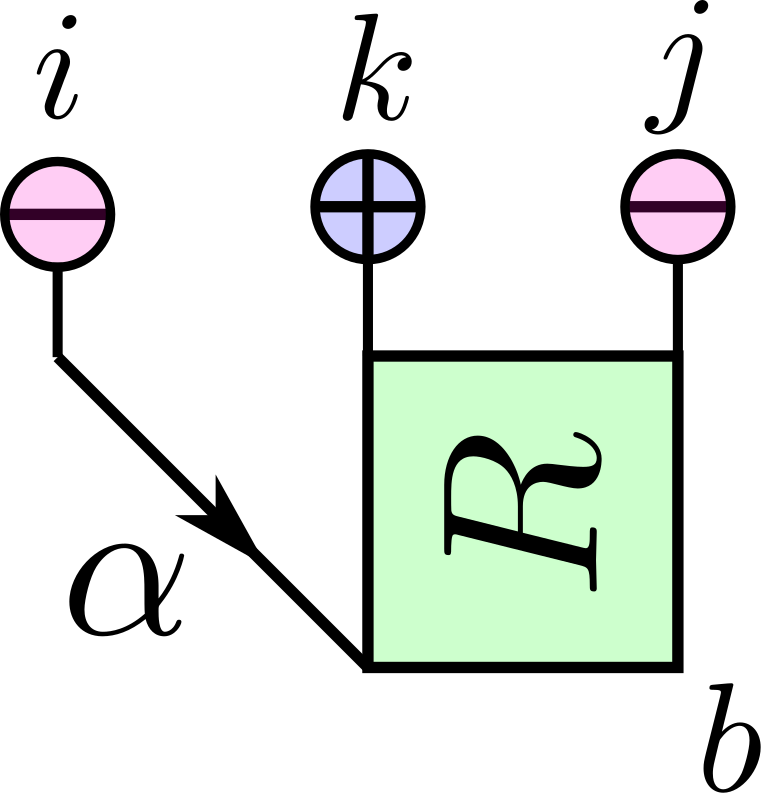}~~+~~
    \adjincludegraphics[width=9ex,valign=c]{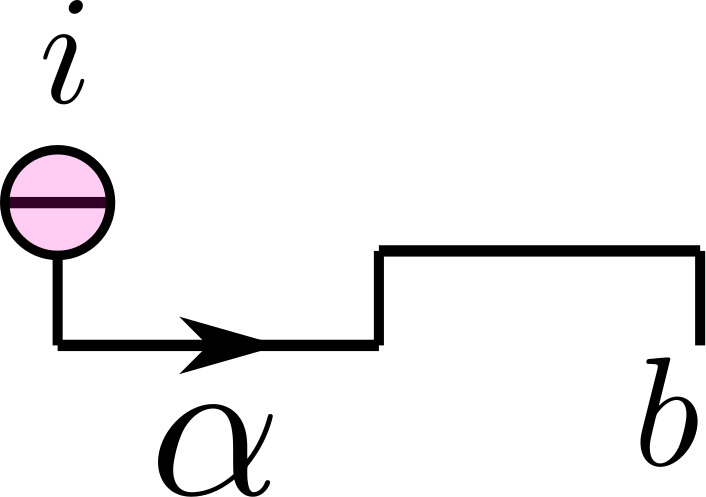}~\delta_{jk}
	-~~
    \adjincludegraphics[width=9ex,valign=c]{Figures/LA-1/empsip_deriv-1-2.png}\nonumber\\
	&=&\delta_{ik}~~
    \adjincludegraphics[width=9ex,valign=c]{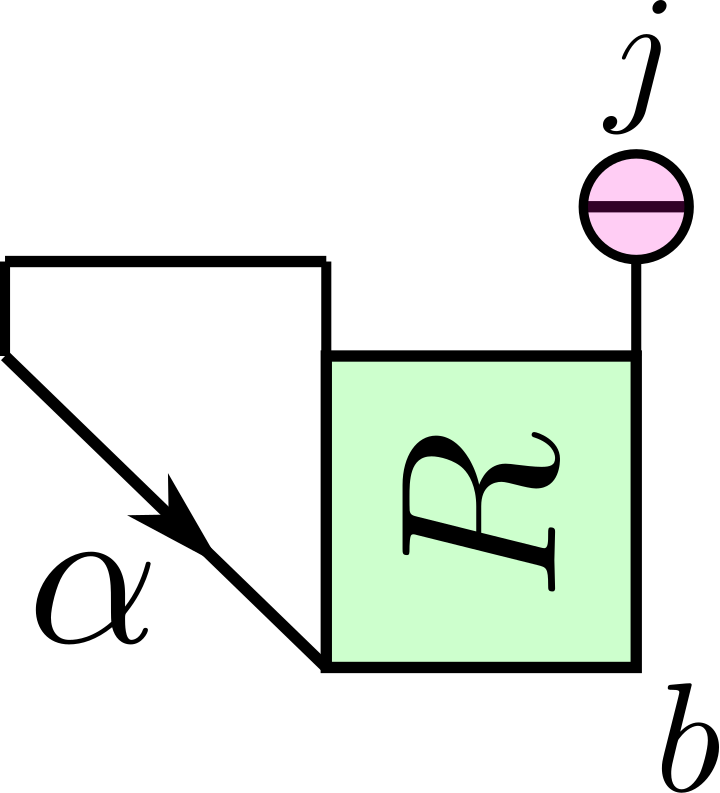}~~+\delta_{jk}
    \adjincludegraphics[width=9ex,valign=c]{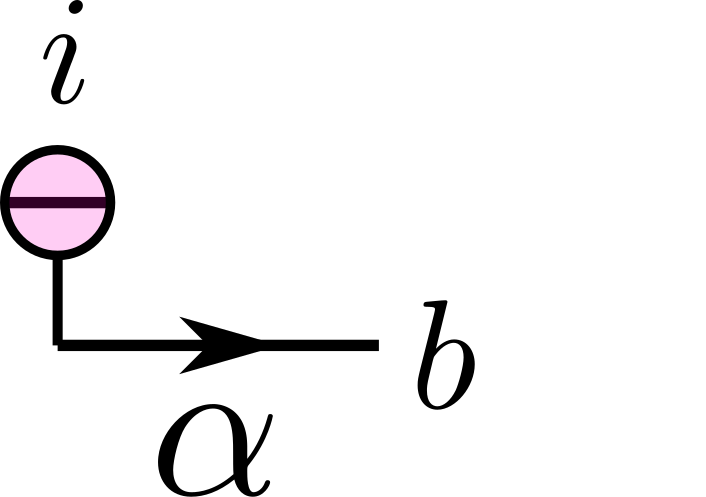}\nonumber\\
	&&+~~
    \adjincludegraphics[width=9ex,valign=c]{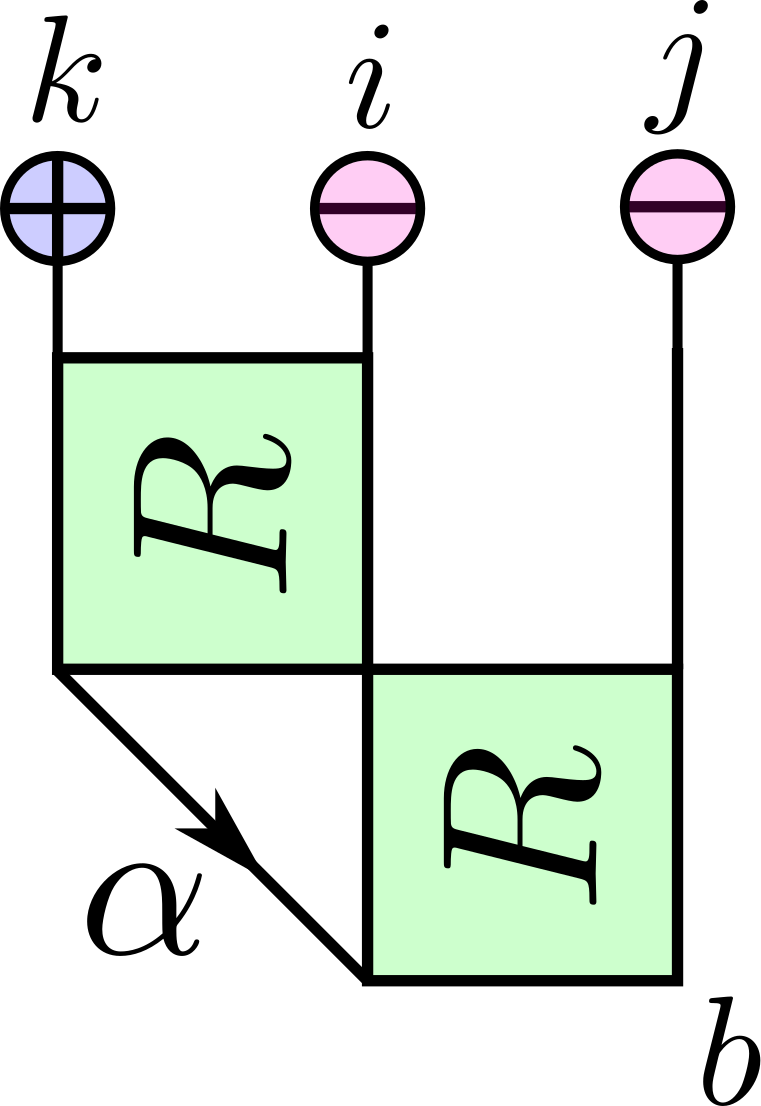}~~-~~
    \adjincludegraphics[width=9ex,valign=c]{Figures/LA-1/empsip_deriv-1-2.png},
\end{eqnarray}
where in the second and third line we used the first line of Eq.~\eqref{eq:fundamental_Rcommu}. 
The first two terms %
vanish when $k\notin\{i,j\}$, so we only need to require that the two normally ordered cubic terms cancel, which happens when 
\begin{equation}\label{eq:RRalpha}
	\adjincludegraphics[width=24ex,valign=c]{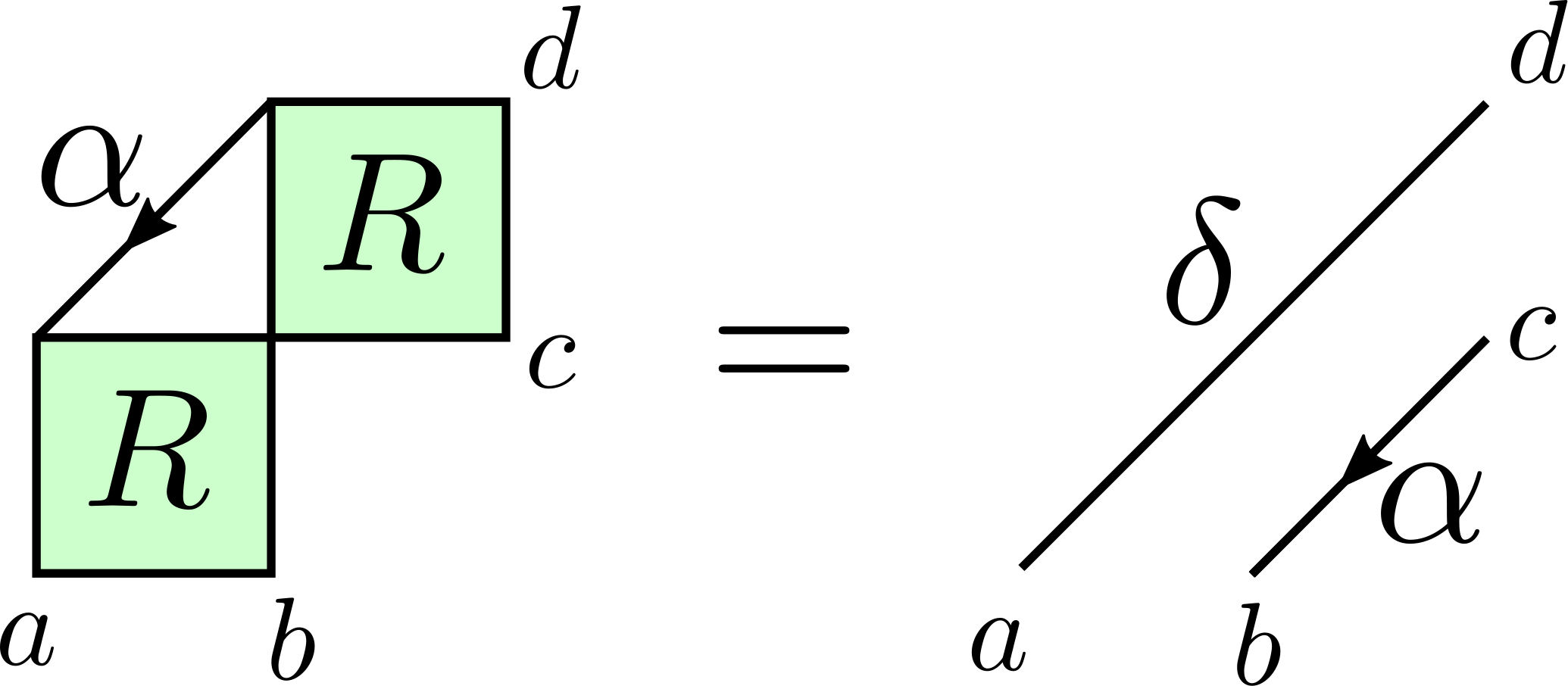}~.
\end{equation}
Using $R^2=\mathds{1}$, Eq.~\eqref{eq:RRalpha} is equivalent to
\begin{equation}\label{eq:Ralpha_def-graphical}
\adjincludegraphics[width=22ex,valign=c]{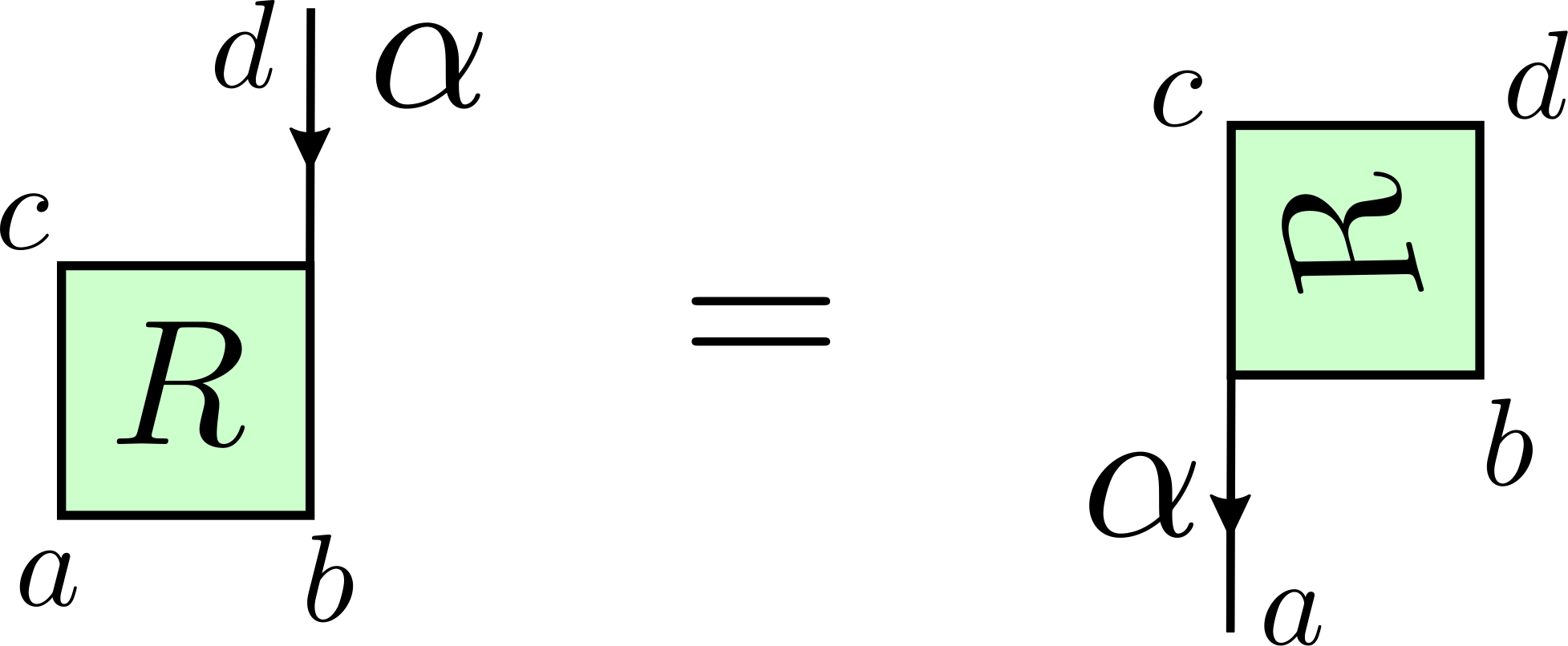}~.
\end{equation}
Next, we compute $[\hat{e}^+_{ij}, \hat{\psi}^+_{k,b}]$:
\begin{eqnarray}\label{eq:commu_Eabp_psip}
	[\hat{e}^+_{ij}, \hat{\psi}^+_{k,b}]&=&-\adjincludegraphics[width=9ex,valign=c]{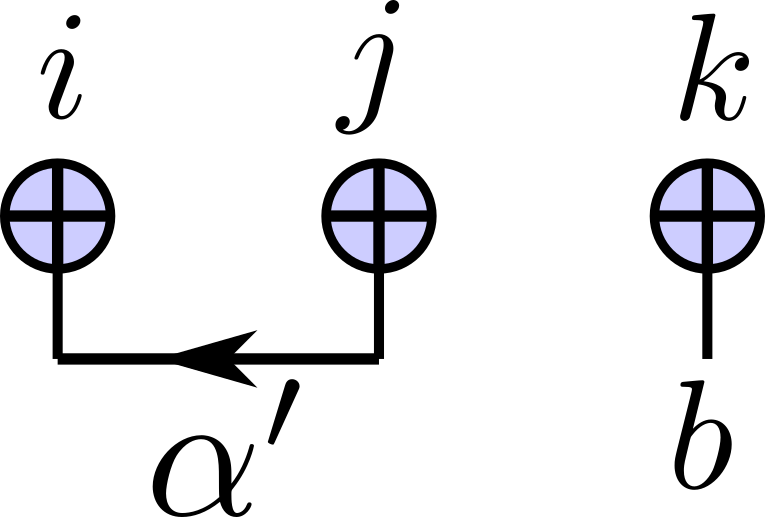}~~+~~
    \adjincludegraphics[width=9ex,valign=c]{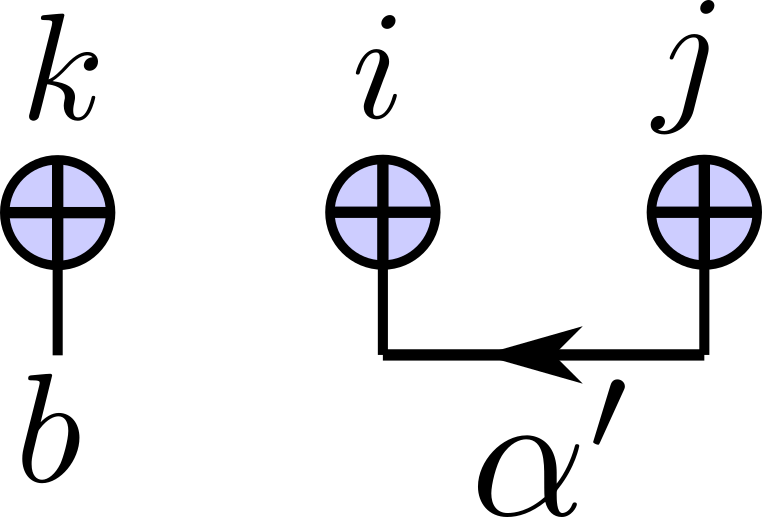}~,\nonumber\\
	&=&-
    \adjincludegraphics[width=9ex,valign=c]{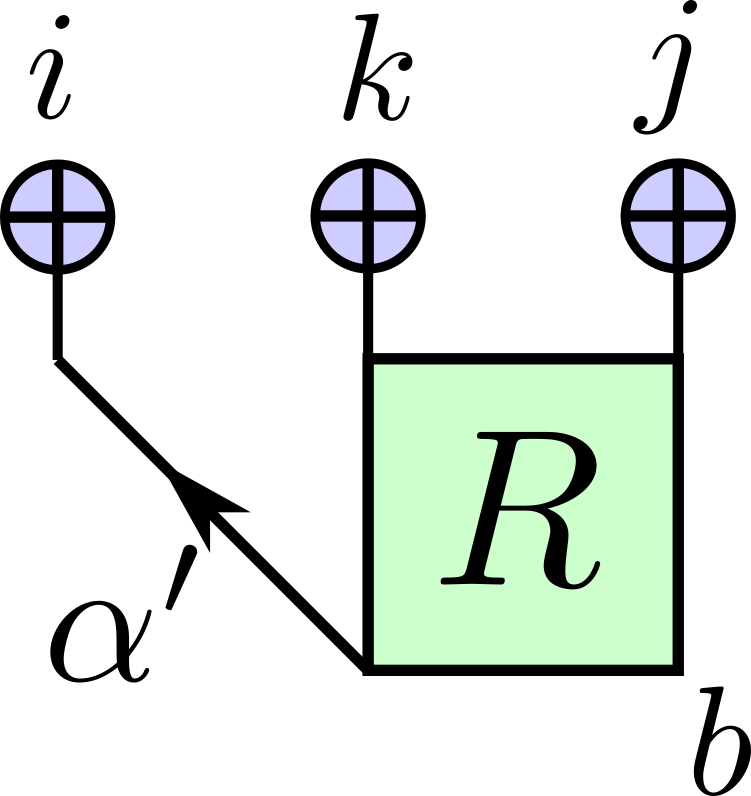}~~+~~
    \adjincludegraphics[width=9ex,valign=c]{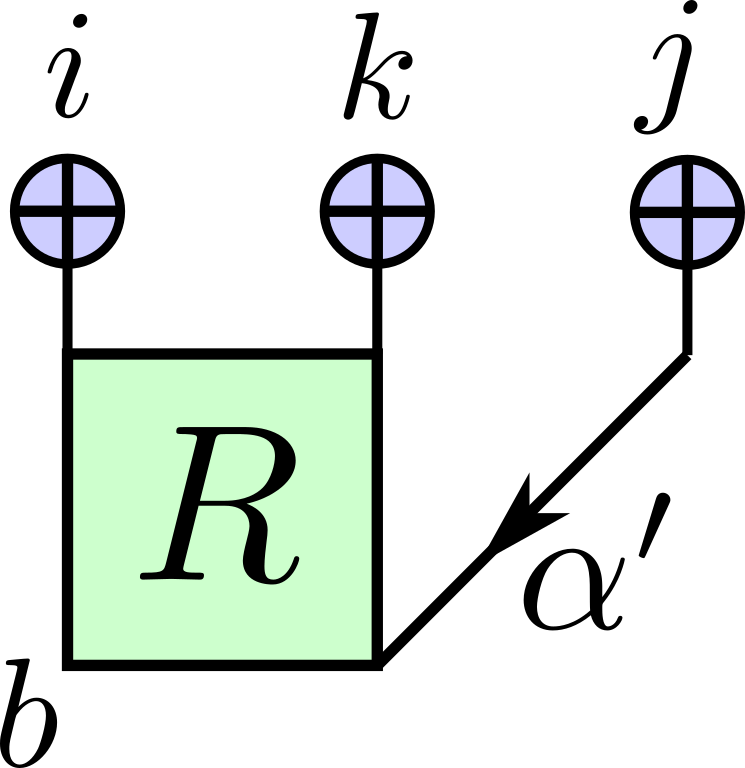}~,
\end{eqnarray}
where we use the second line of Eq.~\eqref{eq:fundamental_Rcommu}. The two cubic terms in the second line of Eq.~\eqref{eq:commu_Eabp_psip} cancel if
\begin{equation}\label{eq:Ralpha_prime}
	\adjincludegraphics[width=22ex,valign=c]{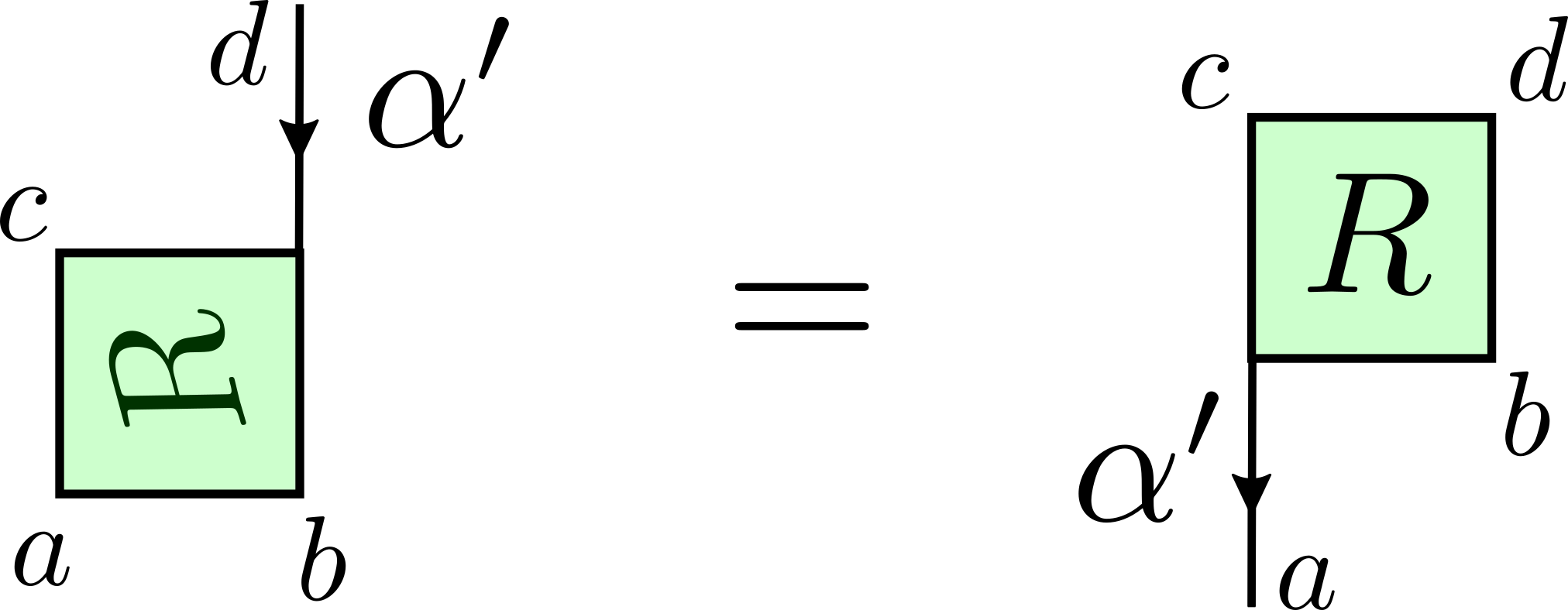}~.
\end{equation}
Note that Eq.~\eqref{eq:Ralpha_prime} follows from  Eq.~\eqref{eq:Ralpha_def-graphical} if we define  $\alpha'=\alpha^{-1}$. %
Using similar calculations, one can show that as long as $\alpha$ and $\alpha'$ satisfy Eqs.~\eqref{eq:Ralpha_def-graphical} and \eqref{eq:Ralpha_prime}, the  commutators $[\hat{e}^{\pm}_{ij},\hat{\psi}^-_{k,c}]$ also %
vanish when $k\notin\{i,j\}$.

\begin{figure}
\center{\includegraphics[width=.85\linewidth]{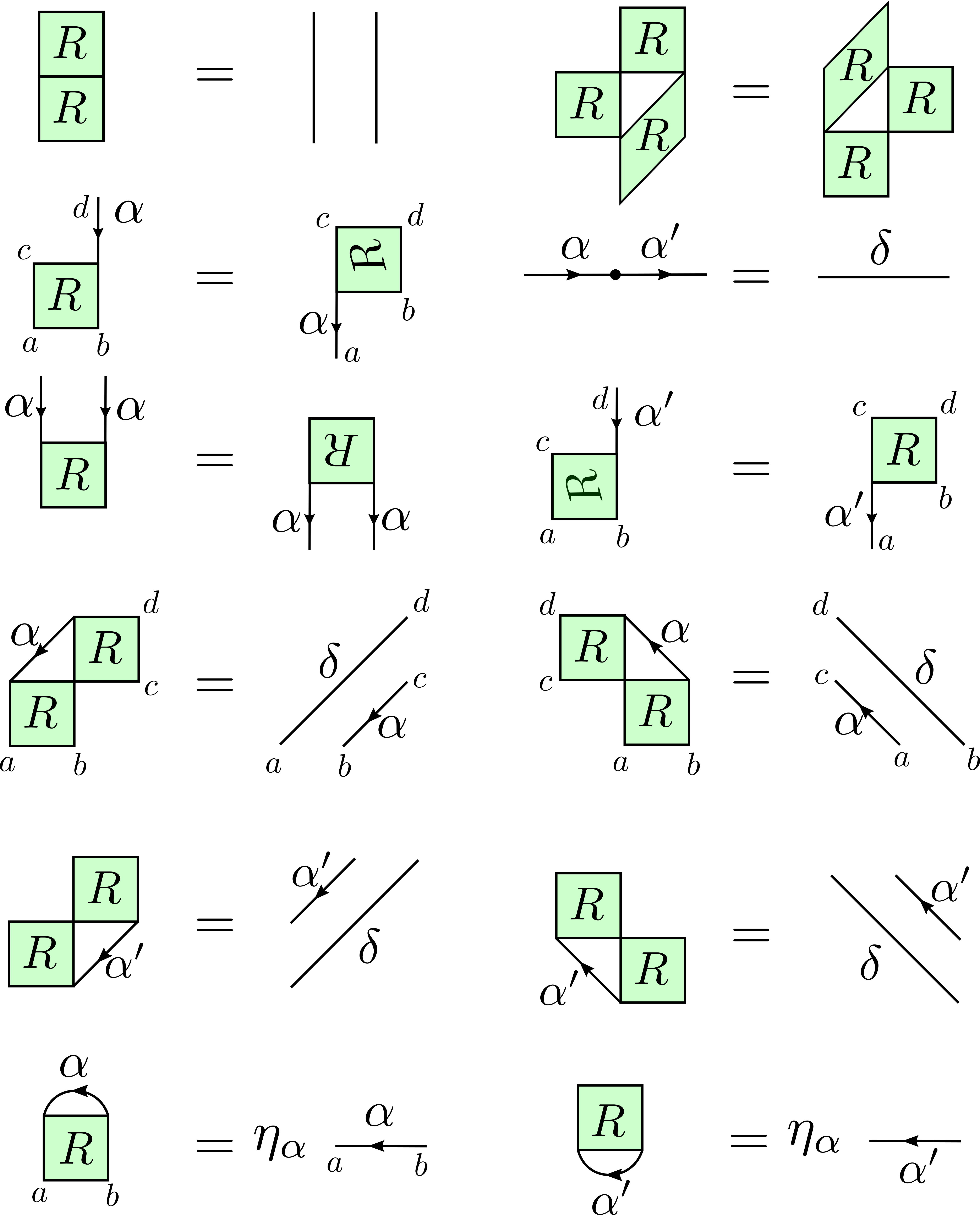}}
	\caption{\label{fig:Ralphaeqn} A list of tensor equations satisfied by $\alpha,\alpha'$, and $R$, useful for computational purpose. We use the graphical notation $\alpha_{ab}=\adjincludegraphics[width=7ex,valign=c]{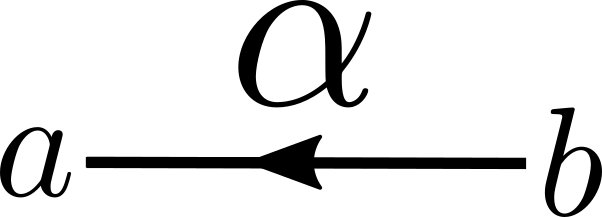}$ 
    and $\alpha'_{ab}=\adjincludegraphics[width=7ex,valign=c]{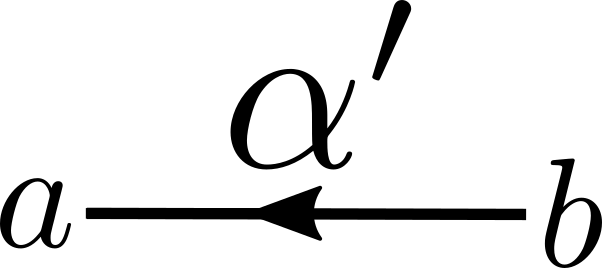}$. The first line is the YBE in Eq.~\eqref{eq:YBE}, the second line is the definition of $\alpha$ and $\alpha'$. The third to the fifth lines are derived from the first two lines. The fourth and fifth lines are used in calculating $[\hat{e}^\pm_{ij},\hat{\psi}_{k,c}^\pm]$--they guarantee that the~(normal-ordered) cubic terms cancel. The last line gives the definition of $\eta_\alpha$. }
\end{figure}

In summary, local pair creation of $\psi$ is possible %
if there exists an invertible $m\times m$ constant matrix $\alpha$ satisfying Eq.~\eqref{eq:Ralpha_def-graphical}. This defines the subclass of $R$-matrices compatible with local pair creation. %
Not every $R$-matrix satisfies this condition; for example, for the $R$-matrix $R^{ab}_{cd}=\pm\delta_{ac}\delta_{bd}$ with $m\geq 2$~(Ex.~\ref{ex:1m} of Sec.~\ref{sec:YBE}), there does not exist an invertible $\alpha$ satisfying this equation. 
But the $R$-matrix defined in Eqs.~(\ref{eq:example_1n1Rmatrix},\ref{eq:lambdac}) satisfies this condition provided that $\lambda \xi=-\mathds{1}_m$~[this is satisfied by the construction below Eq.~\eqref{eq:lambdac}] in which case we can take $\alpha=\xi$ and $\alpha'=-\lambda$, then  Eqs.~\eqref{eq:Ralpha_def-graphical} and \eqref{eq:Ralpha_prime} follow by straightforward calculation. 
The $R$-matrix in Ex.~\ref{ex:setth} also satisfies this condition, where we have 
\begin{equation}
\alpha=\alpha'=\pi_{12}=\begin{pmatrix}
	0 & 1 & 0 & 0\\
	1 & 0 & 0 & 0\\
	0 & 0 & 1 & 0\\
	0 & 0 & 0 & 1\\
\end{pmatrix}.
\end{equation}
Note that the above treatment assumes that $\alpha$ is invertible, which allows us to set $\alpha'=\alpha^{-1}$. In general, there can be situations where a nonzero $\alpha$ satisfying Eq.~\eqref{eq:Ralpha_def-graphical} exists, but is not invertible, see Remark~\ref{rmk:example-noninvalpha} for an example. However, if the $R$-matrix is simple~(in the sense of Definition~\ref{def:simpleR-matrix}), and a nonzero $\alpha$ satisfying Eq.~\eqref{eq:Ralpha_def-graphical} exists, then $\alpha$ must be unique~(up to a multiplicative constant) and invertible,
as stated in Lemma~\ref{lemma:alpha_unique_inv}. 

Let us now come back to Eq.~\eqref{eq:commu_Eabp_psi}. The second term in the last line of Eq.~\eqref{eq:commu_Eabp_psi} can also be simplified. One can show that if $\alpha_{cd}$ satisfies Eq.~\eqref{eq:Ralpha_def-graphical}, so does $\sum_{a,b}\alpha_{ab}R^{ab}_{cd}$. Since all the $\alpha$ satisfying Eq.~\eqref{eq:Ralpha_def-graphical} form a linear space~(indeed, in all examples we know this space is one-dimensional), we can take $\alpha$ to be a left eigenvector of $R$ with eigenvalue $\eta_{\alpha}=\pm 1$:
\begin{equation}\label{eq:eta_R}
	\sum_{ab}\alpha_{ab}R^{ab}_{cd}=\eta_\alpha {\alpha}_{cd},
\end{equation}
which is shown in the last line of Fig.~\ref{fig:Ralphaeqn}. For example, for the $R$-matrix in Eq.~\eqref{eq:example_1n1Rmatrix}, we have $\alpha=\xi$ and therefore $\eta_\alpha=+1$. For the $R$-matrix in Ex.~\ref{ex:setth}, we have $\eta_\alpha=-1$.
The tensor graphical representations of the various equations satisfied by $\alpha$ and $\alpha'$ are summarized in Fig.~\ref{fig:Ralphaeqn}. 

\subsubsection{Lie algebra of the bilinear operators $\{\hat{e}^{}_{ij},\hat{e}^{\pm}_{ij}\}$}\label{sec:LApaircreationbilin}
We now study the Lie algebra structure generated by all the bilinear operators $\{\hat{e}^{}_{ij},\hat{e}^{\pm}_{ij}\}$, which is important for solving bilinear Hamiltonians in Sec.~\ref{sec:u1_breaking_solu}. 
The CRs between $\hat{e}^{\pm}_{ij}$ and $\hat{\psi}^\pm_{k,c}$ are summarized as follows:
\begin{eqnarray}\label{eq:commu_Eab_pm_psi}
	[\hat{e}^+_{ij}, \hat{\psi}^+_{k,c}]&=&[\hat{e}^-_{ij}, \hat{\psi}^-_{k,c}]=0,\nonumber\\ 
	{}[\hat{e}^-_{ij}, \hat{\psi}^+_{k,c}]&=&\delta_{jk}\bar{\psi}^-_{i,c}+\delta_{ik}\bar{\psi}^-_{j,c}\eta_\alpha,\nonumber\\
	{}[\hat{e}^+_{ij}, \bar{\psi}^-_{k,c}]&=&\eta_\alpha \delta_{jk}\hat{\psi}^+_{i,c}+\delta_{ik}\hat{\psi}^+_{j,c}.
\end{eqnarray}
where $\bar{\psi}^-_{i,c}\equiv \sum_a \hat{\psi}^-_{i,a}\alpha_{ca}$. 
Along with Eq.~\eqref{eq:commu_Eab_psi_p}, one can calculate the commutators between $\hat{e}^+_{ij}, \hat{e}^-_{ij}$ and $\hat{e}_{ij}$, similar to the calculation done in Eq.~\eqref{eq:commu_Eab_Ecd}:
\begin{eqnarray}\label{eq:commu_Eab_Ecd_sosp}
	[\hat{e}_{ij}, \hat{e}_{kl}]&=&\delta_{jk}\hat{e}_{il}-\delta_{il}\hat{e}_{kj},\nonumber\\
	{}[\hat{e}_{ij}, \hat{e}_{kl}^+]&=&\delta_{jk}\hat{e}^+_{il}+\delta_{jl}\hat{e}^+_{ki},\nonumber\\
	{}[\hat{e}_{kl}^-,\hat{e}_{ij}]&=&\delta_{ki}\hat{e}^-_{jl}+\delta_{li}\hat{e}^-_{kj},\nonumber\\
	{}[\hat{e}_{kl}^+, \hat{e}_{ij}^+]&=&[\hat{e}_{kl}^-, \hat{e}_{ij}^-]=0,\nonumber\\
	{}[\hat{e}_{kl}^+, \hat{e}_{ij}^-]&=&\delta_{ik}\hat{e}_{jl}+\delta_{jl}\hat{e}_{ki}+\eta_\alpha(\delta_{jk}\hat{e}_{li}+\delta_{il}\hat{e}_{kj})\nonumber\\
	&&+\mathrm{Tr}(\alpha'\alpha^T)(\delta_{jk}\delta_{li}+\eta_\alpha\delta_{ik}\delta_{jl}).
\end{eqnarray}
The constant term that appears in the last equation of Eq.~\eqref{eq:commu_Eab_Ecd_sosp} can be absorbed into the  Lie algebra generators as  follows:
\begin{eqnarray}
\hat{e}^\prime_{ij}&=&\hat{e}_{ij}+\frac{\eta_\alpha}{2}\mathrm{Tr}(\alpha'\alpha^T)\delta_{ij},\nonumber\\
\hat{e}^{\prime\pm}_{ij}&=&\hat{e}_{ij}^\pm.
\end{eqnarray}
These operators satisfy Eq.~\eqref{eq:commu_Eab_Ecd_sosp} but without the constant term in the last line, which are the defining relations of either the  $\mathfrak{sp}_{2N}$ or the $\mathfrak{so}_{2N}$ Lie algebra~\cite{humphreys_LA}, corresponding to $\eta_\alpha=+1$ or $\eta_\alpha=-1$, respectively. Then Eq.~\eqref{eq:commu_Eab_psi_p} and \eqref{eq:commu_Eab_pm_psi} imply that $\hat{\psi}^+_{i,c}$ transform in the fundamental representations of these Lie algebras, while $\bar{\psi}^-_{i,c}$ transform in the fundamental conjugate representation. 
This allows us to construct eigenmodes of a general bilinear Hamiltonian $\hat{H}$ using an $R$-paraparticle generalization of Bogoliubov transformation, thereby obtaining the whole spectrum of $\hat{H}$, as will be presented in Sec.~\ref{sec:u1_breaking_solu}. 
Thus, similar to the number conserving case, the  $\mathfrak{sp}_{2N}$ or $\mathfrak{so}_{2N}$ Lie algebra structure still guarantees 
solvability of bilinear Hamiltonians.

\subsubsection{Topological twist factor, Frobenius-Schur indicator, and particle-hole symmetry}\label{sec:TPSpin-FSI-PHS}
We now present a classification theorem of unitary $R$-matrices that admit pair creation, in which several important physical concepts such as topological twist factor and Frobenius-Schur indicator naturally arises, and the discussion of hermiticity and particle-hole symmetry also relies on this theorem. 
\paragraph{A general classification result on unitary $R$-matrices admitting pair-creation}
\begin{restatable}{thm}{unitarysimpleselfdualR} \label{thm:unitary_simple_self-dualR}
	Let $R$ be a unitary involutive simple $R$-matrix such that 
    there exists a nonzero $\alpha$ satisfying Eq.~\eqref{eq:Ralpha_def-graphical}. Then
	we have\\
	(1) $\alpha$ is unique~(up to a multiplicative constant) and invertible;\\
	(2) we can always normalize $\alpha$ such that $\alpha^*\alpha=\nu \mathds{1}$, where $\nu=\pm 1$ is called the Frobenius-Schur indicator of the paraparticle;\\
	(3) $\alpha^T=\nu\alpha$, therefore,  $\alpha$ is unitary  with the above normalization;\\
	(4) $R$ is a dual-unitary tensor, or equivalently, $R$ satisfies 
    \begin{equation}\label{eq:DUcondition}
\adjincludegraphics[width=22ex,valign=c]{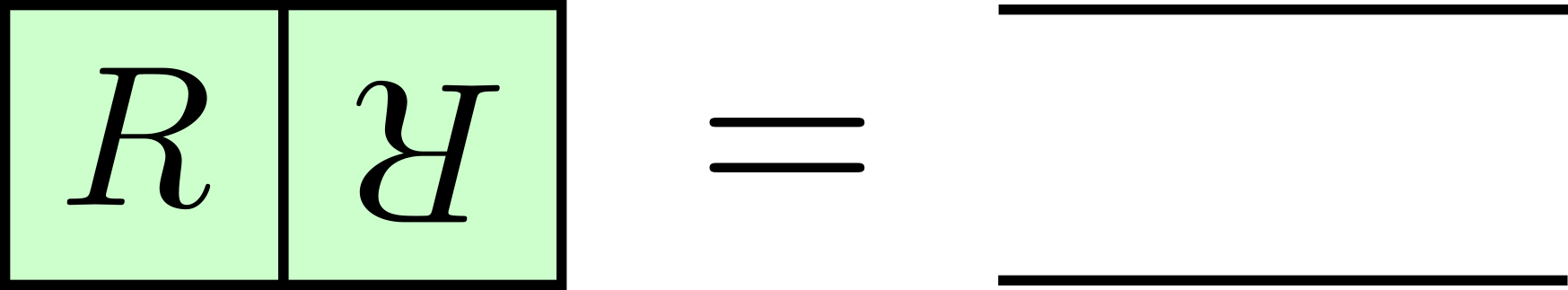}     .   
    \end{equation}
	(5) $\Tr_1[R]=\Tr_2[R]=\thetaR \mathds{1}_m$, where $\thetaR=\nu\eta_\alpha$ is called the topological twist factor, and $z_R(x)=(1-\thetaR x)^{-\thetaR m}$.
\end{restatable}
The proof of this theorem is given in App.~\ref{sec:unitary-R-with-alpha}. 
As a non-technical summary, 
Thm.~\ref{thm:unitary_simple_self-dualR} implies that if the $R$-matrix is unitary, involutive,  simple, and allows pair-creation, then $\alpha$ is invertible and can be normalized to a canonical form, 
the Frobenius-Schur indicator $\nu$ and the topological twist factor $\thetaR $ are well-defined and can be directly extracted from $R$, and
the exclusion statistics of the associated $R$-paraparticle is either fermionic or bosonic,  depending on the topological twist factor $\thetaR$. Note however, that this does not imply that the $R$-paraparticle is equivalent or weakly-equivalent to fermions or bosons, as the equality of the Hilbert series $z_{R_1}(x)=z_{R_2}(x)$ is a necessary but not sufficient condition for weak equivalence of $R_1$ and $R_2$. %
We also note that both Ex.~\ref{ex:1m} and Ex.~\ref{ex:1m1} in Tab.~\ref{tab:Hilbert_series} evades the conclusion of Thm.~\ref{thm:unitary_simple_self-dualR} for having nontrivial exclusion statistics, the former evades by not allowing pair creation, while the latter evades by not being unitary~(consequently, $\nu$ and $\thetaR$ are not well-defined for both Ex.~\ref{ex:1m} and Ex.~\ref{ex:1m1}). 

The physical meaning of the Frobenius-Schur indicator $\nu$ can be seen from the following calculation: %
\begin{eqnarray}\label{eq:FSI-2ndQ}
	\hat{e}^-_{ij}\hat{\psi}^+_{i,a}(\hat{e}^-_{jk})^\dagger\ket{0}&=&[\hat{e}^-_{ij},\hat{\psi}^+_{i,a}](-\thetaR\hat{e}^+_{jk})\ket{0}\nonumber\\
	&=&-\thetaR \eta_\alpha \bar{\psi}^-_{j,a}\hat{e}^+_{jk}\ket{0}\nonumber\\
	&=&-\nu [\bar{\psi}^-_{j,a},\hat{e}^+_{jk}]\ket{0}\nonumber\\
	&=&\nu \hat{\psi}^+_{k,a}\ket{0}.
\end{eqnarray}
Eq.~\eqref{eq:FSI-2ndQ} shows the relation between two different ways of moving a paraparticle from position $i$ to position $k$: one can either directly move it using $\hat{e}_{ik}$, or one  can first create a pair of paraparticles at positions $j$ and $k$, and then annihilate a pair at $i$ and $j$. The phase difference between these two processes is the Frobenius-Schur indicator $\nu$. A diagrammatic way to present Eq.~\eqref{eq:FSI-2ndQ} is
\begin{equation}\label{eq:FSI-graphical}
	\adjincludegraphics[height=9ex,valign=c]{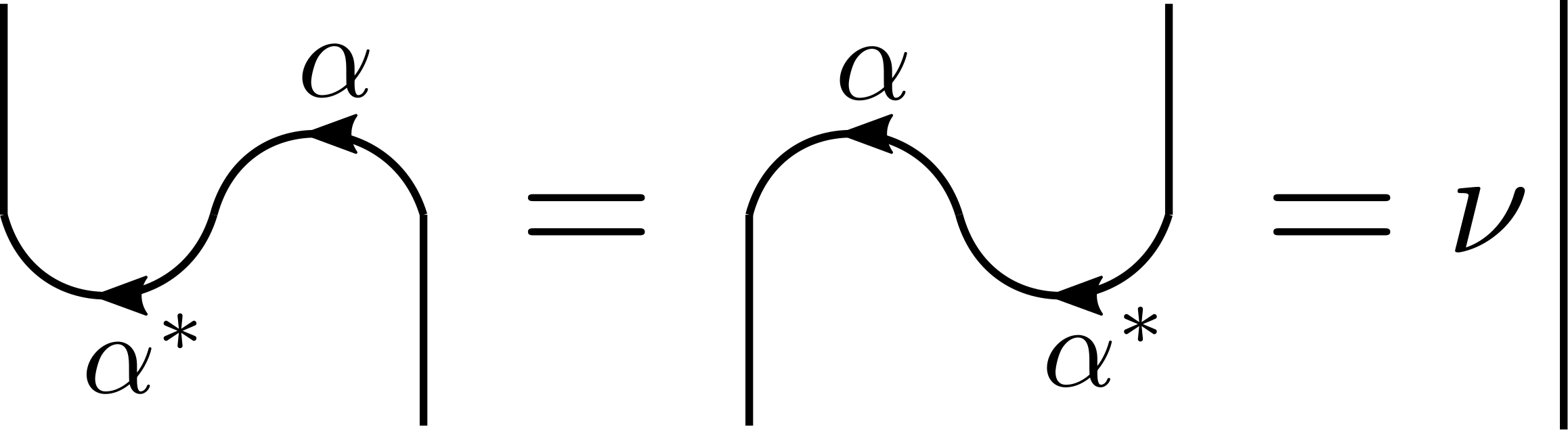}~~,
\end{equation}
which connects to the definition of the Frobenius-Schur indicator in more general unitary fusion categories~\cite{kitaev2006anyons,simon2020topological_protobook}. 

The physical meaning of the topological twist factor $\thetaR$ is more transparent if we draw the $R$-matrix as a crossing, as in Secs.~\ref{sec:mutual_para} and  \ref{sec:wick_thm}:
\begin{equation}\label{eq:TPS-graphical}
	\adjincludegraphics[height=12ex,valign=c]{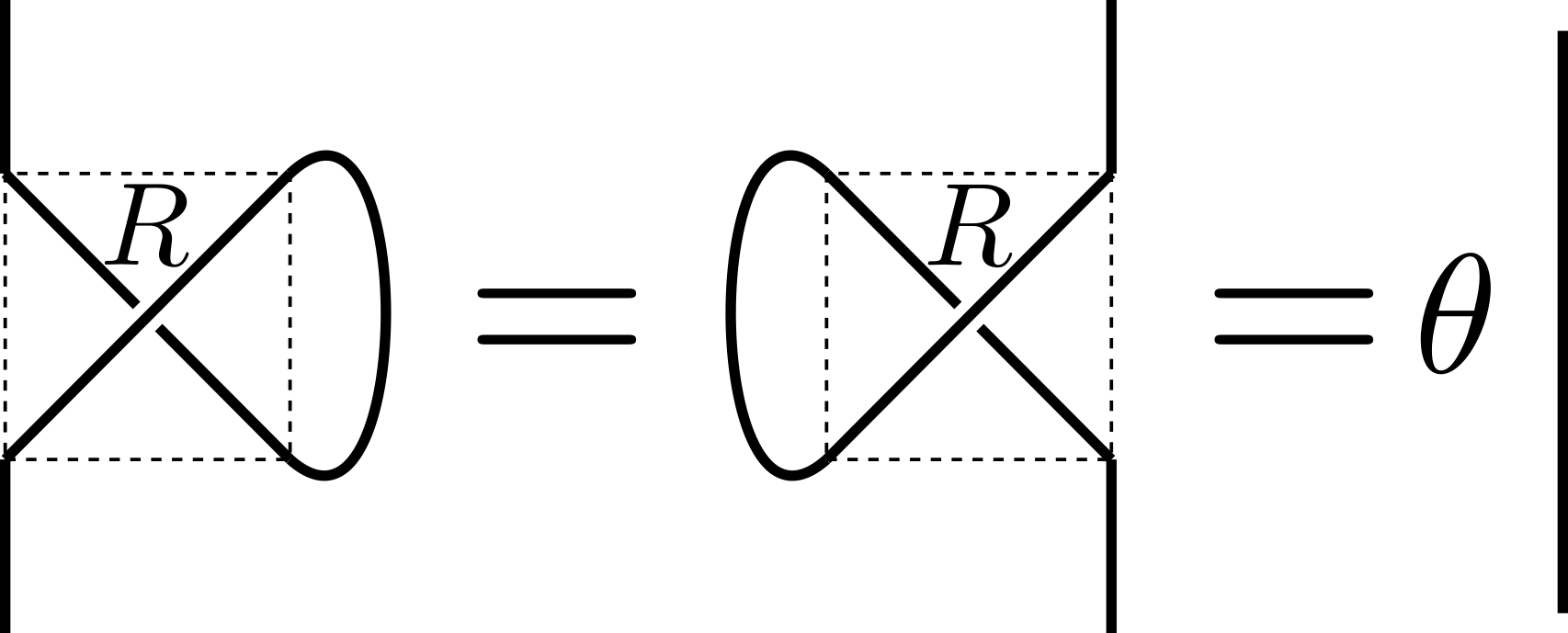}~~,
\end{equation}
where the partial trace on $R$ is drawn as a contraction of indices. 
This again connects to the definition of the topological twist factor in general unitary braided fusion categories~\cite{kitaev2006anyons,kong2014braided,simon2020topological_protobook}. We will discuss the topological twist factor in greater detail in part II after we present the realization of $R$-paraparticles in topological phases. 

\paragraph{Hermitian conjugate of $\hat{e}^+_{ij}$}\label{para:herm_paircreation}
We now discuss the important issue of hermiticity, which becomes more subtle when we include the pair creation and annihilation terms. For simplicity, let us first focus on the case where the $R$-matrix is unitary~\footnote{If the $R$-matrix is not unitary, but still has a finite Hilbert series~(e.g. Ex.~\ref{ex:1m1} in Tab.~\ref{tab:Hilbert_series}), then 
Thm.~S2.5 in the supplementary information of Ref.~\cite{wang2023para}
applied to $\mathfrak{sp}_{2N}$ or $\mathfrak{so}_{2N}$ still  guarantees that there exists a Hermitian inner product on the state space such that $(\hat{e}_{ij})^\dagger=\hat{e}_{ji}$ and $(\hat{e}^\pm_{ij})^\dagger= \hat{e}^\mp_{ij}$. 
For non-unitary $R$-matrices with infinite Hilbert series, $\{\hat{e}_{ij},\hat{e}^\pm_{ij}\}$ form an infinite dimensional representation of the Lie algebra $\mathfrak{sp}_{2N}$ or $\mathfrak{so}_{2N}$ that is generally not a direct sum of finite dimensional representations, therefore %
the aforementioned theorem
does not apply.
(By contrast, without the U(1)-breaking terms $\hat{e}^\pm_{ij}$, $\{\hat{e}_{ij}\}$ form an infinite dimensional representation of $\mathfrak{sl}_N$ that decomposes as a direct sum of finite dimensional representations, hence the aforementioned theorem applies to that case.) We will not treat this complicated case in this paper.} and simple, i.e., $R^\dagger R=\mathds{1}_{m^2}$ and $\DC[R]\cong\C$. Under these conditions, Thm.~\ref{thm:unitary_simple_self-dualR} applies, which %
leads to $\alpha'=\nu\alpha^*$. We compute $(\hat{e}^-_{ij})^\dagger$ as follows:
\begin{equation}
	(\hat{e}^-_{ij})^\dagger=\left(\adjincludegraphics[width=5ex,valign=c]{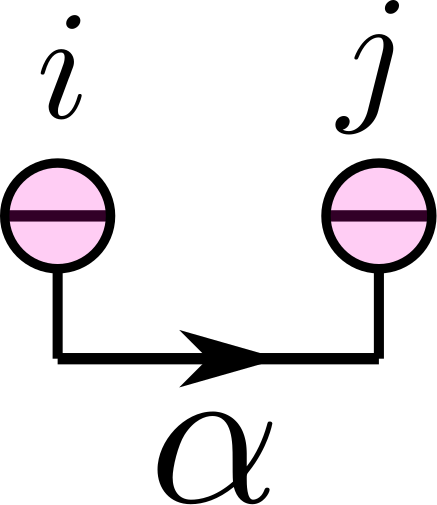}\right)^\dagger
	=
    \adjincludegraphics[width=5ex,valign=c]{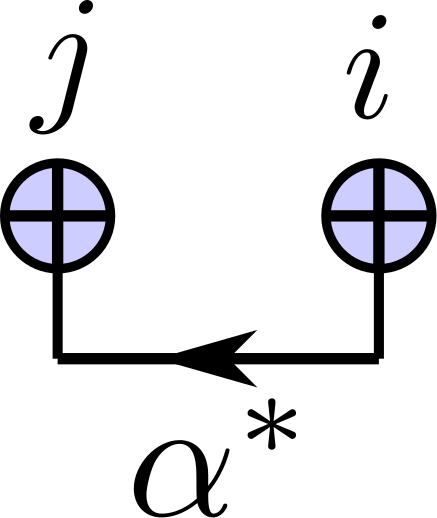}
	=\nu~
    \adjincludegraphics[width=5ex,valign=c]{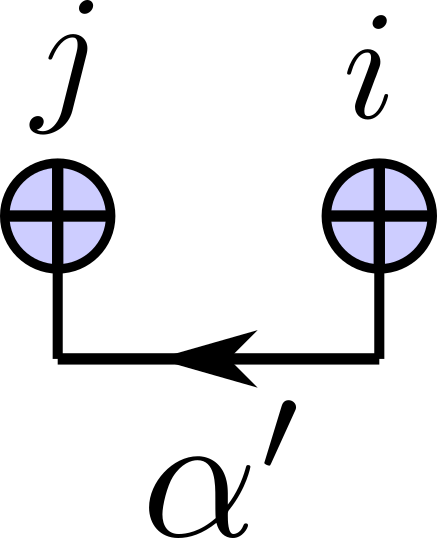}
	=-\nu\eta_\alpha\hat{e}^+_{ij},
\end{equation}
where in the third step we use $\alpha'=\nu\alpha^*$ and in the last step we used $\hat{e}^\pm_{ij}=\eta_\alpha \hat{e}^\pm_{ji}$~[which directly follows from Eqs.~\eqref{eq:fundamental_Rcommu_supress}, \eqref{def:e_pm_ab} and \eqref{eq:eta_R}].

Therefore we have $(\hat{e}^-_{ij})^\dagger=-\thetaR\hat{e}^+_{ij}$ where $\thetaR=\nu\eta_\alpha$. Meanwhile,
Thm.~\ref{thm:unitary_simple_self-dualR} says that $\thetaR$ also controls the Hilbert series: $z_R(x)=(1-\thetaR x)^{-\thetaR m}$. Therefore, we can summarize the above results as follows: %
for any $R$-matrix satisfying the conditions in Thm.~\ref{thm:unitary_simple_self-dualR}, 
we have 
\begin{eqnarray}\label{eq:hermiticity_anomoly}
(\hat{e}^-_{ij})^\dagger&=&\hat{e}^+_{ij}\Leftrightarrow 	\thetaR=-1 \Leftrightarrow z_R(x)=(1+x)^{m},\\
(\hat{e}^-_{ij})^\dagger&=&-\hat{e}^+_{ij}\Leftrightarrow 	\thetaR=1 \Leftrightarrow z_R(x)=(1-x)^{-m}.\nonumber
\end{eqnarray}
The minus sign in the second case may seem weird, but this happens already in the case of ordinary bosons with our definition Eq.~\eqref{def:e_pm_ab}, where we have 
\begin{equation}
\hat{e}^-_{ij}=\hat{b}_i\hat{b}_j,\quad 	\hat{e}^+_{ij}=-\hat{b}^\dagger_i\hat{b}^\dagger_j. 
\end{equation}
Importantly, the sign convention that we use in Eq.~\eqref{def:e_pm_ab} is to guarantee that $\hat{e}^+_{ij}$ is the Hermitian conjugate of $\hat{e}^-_{ij}$ in the fundamental representation of the corresponding Lie algebra. 
We will come back to the hermiticity issue when we solve Bogoliubov-type free-paraparticle Hamiltonians in Sec.~\ref{sec:u1_breaking_solu}, where we will see the consequence of the potential minus sign in  Eq.~\eqref{eq:hermiticity_anomoly}.

\paragraph{Particle-hole symmetry}
To discuss particle-hole symmetry, it is
useful to explicitly write down the CRs between $\hat{\psi}^+_{i,a}$ and $\bar{\psi}^-_{j,b}\equiv \sum_c\hat{\psi}^-_{j,c}\alpha_{bc}$ for this family of $R$ matrices, since they have simpler transformation rules than those given in  Eq.~\eqref{eq:commu_Eab_pm_psi}. This is done by contracting all the $\hat{\psi}^-_{i,a}$  with $\alpha_{ca}$ in the tensor equations in Fig.~\ref{fig:RMQA} and applying  the graphical  rules in Fig.~\ref{fig:Ralphaeqn}. We obtain~(suppressing auxiliary indices)
\begin{eqnarray}\label{eq:fundamental_Rcommu_NNC}
	\bar{\psi}^-_{i} \hat{\psi}^+_{j}&=&\hat{\psi}_{j}^+ \bar{\psi}_{i}^-\cdot R+\delta_{ij}\alpha,\nonumber\\%
	\hat{\psi}^+_{i} \hat{\psi}^+_{j}&=&\hat{\psi}_{j}^+ \hat{\psi}_{i}^+ \cdot R,\nonumber\\
	\bar{\psi}^-_{i} \bar{\psi}^-_{j}&=&\bar{\psi}_{j}^- \bar{\psi}_{i}^-\cdot R.%
\end{eqnarray}
The advantage of Eq.~\eqref{eq:fundamental_Rcommu_NNC} over Eq.~\eqref{eq:fundamental_Rcommu_supress} is that the $R$-matrix appears in all the CRs with the same orientation. 
We use $\mathcal{X}_{R,N}$ to denote the algebra  generated by  $\{\hat{\psi}^+_{i,b},\bar{\psi}^-_{i,b}|1\leq i \leq N, 1\leq b\leq m\}$  modulo all the relations in Eq.~\eqref{eq:fundamental_Rcommu_NNC}, which is isomorphic to the one defined in Remark~\ref{rmk:Fock_irrep} via a change of basis in the generators described above. 

To define particle-hole transformation, we first define an automorphism $\mathfrak{C}$ of the second quantization algebra $\mathcal{X}_{R,N}$ via its action on the generators 
\begin{equation}\label{eq:particle-hole-trans}
    \mathfrak{C}[\hat{\psi}^+_{i,a}]=\bar{\psi}^-_{i,a},\quad \mathfrak{C}[\bar{\psi}^-_{i,a}]=-\eta_\alpha\hat{\psi}^+_{i,a}.
\end{equation} 
It is straightforward to check that this transformation 
preserves all the defining relations of $\mathcal{X}_{R,N}$ in Eq.~\eqref{eq:fundamental_Rcommu_NNC}, therefore, $\mathfrak{C}$ extends uniquely to an automorphism of $\mathcal{X}_{R,N}$. 
    
The above discussion is valid for any involutive $R$-matrix. In order to implement the automorphism $\mathfrak{C}$ as a unitary transformation on the Fock space, we need to specialize to the subclass of $R$-matrices for which $z_R(x)$ is finite. In this case we have $\mathcal{X}_{R,N}\cong M_D[\C]$, as discussed  in Remark~\ref{rmk:Fock_irrep}. It is known that every automorphism of $M_D[\C]$ must be inner~\cite{etingof2011introduction}, meaning that 
there exists $\hat{C}\in \mathcal{X}_{R,N}$ such that 
\begin{equation}\label{eq:PHsymm-inner}
\mathfrak{C}[\hat{x}]=\hat{C}\hat{x}\hat{C}^{-1},\quad\forall\hat{x}\in \mathcal{X}_{R,N}.
\end{equation} 
Physically, $\hat{C}$ implements the particle-hole transformation in the Fock space. When the $R$-matrix is unitary and $z_R(x)$ is finite, %
$\mathcal{X}_{R,N}$ is a $\C^*$-algebra, and  $\mathfrak{C}$ %
is a $*$-automorphism, i.e., $\mathfrak{C}[\hat{x}^\dagger]=\mathfrak{C}[\hat{x}]^\dagger$ for all $\hat{x}\in \mathcal{X}_{R,N}$~(the last claim follows from the two conclusions $\alpha^*\alpha=\nu \mathds{1}$ and $\nu\eta_\alpha=-1$ from  Thm.~\ref{thm:unitary_simple_self-dualR}). 
In this case, we can assume $\hat{C}$ to be unitary, since every $*$-automorphism of the $*$-algebra   $\mathcal{X}_{R,N}\cong M_D[\C]$ is inner. 
Furthermore, we have $\hat{C}^2\hat{\psi}^\pm_{i,a}\hat{C}^{-2}=\mathfrak{C}^2[\hat{\psi}^\pm_{i,a}]=-\eta_\alpha\hat{\psi}^\pm_{i,a}=\nu\hat{\psi}^\pm_{i,a}=\nu^{\hat{n}}\hat{\psi}^\pm_{i,a}\nu^{-\hat{n}}$, hence $\nu^{-\hat{n}}\hat{C}^2$ is a central element of $\mathcal{X}_{R,N}\cong M_D[\C]$. Hence $\nu^{-\hat{n}}\hat{C}^2=\lambda \mathds{1}$, where $\lambda$ is a constant phase factor. %
Since Eq.~\eqref{eq:PHsymm-inner} only defines $\hat{C}$ up to a multiplicative constant, we can always rescale $\hat{C}$ such that $\lambda=1$, leading to
\begin{equation}
\hat{C}^2=\nu^{\hat{n}}.
\end{equation}

Below we apply the particle-hole symmetry to prove the following equality which we will use in %
App.~\ref{sec:unitary-R-with-alpha} 
for proving Thm.~\ref{thm:unitary_simple_self-dualR}:
\begin{equation}\label{eq:Tralphaalpha'}
\Tr(\alpha^T\alpha')=-\eta_\alpha n_\mathrm{max},
\end{equation}
where $n_\mathrm{max}$ is the maximal occupation number of each mode~(i.e., largest eigenvalue of $\hat{n}_i$). 
Here we assume $z_R(x)$ to be finite but we do not assume $R$ to be unitary~(in particular, the proof below does not use Thm.~\ref{thm:unitary_simple_self-dualR}).  
We begin by computing
\begin{equation}\label{eq:particle-hole_n_i}
\mathfrak{C}[\hat{n}_i]=\sum_a\mathfrak{C}[\hat{\psi}^+_{i,a}]\mathfrak{C}[\hat{\psi}^-_{i,a}]=-\eta_\alpha\Tr(\alpha^T\alpha')-\hat{n}_i,
\end{equation}
	where we use Eqs.~\eqref{eq:fundamental_Rcommu_NNC} and \eqref{eq:particle-hole-trans}, the definition $\alpha' = \alpha^{-1}$, and the fact that $\mathfrak{C}$ is an automorphism. 
    Applying Eq.~\eqref{eq:PHsymm-inner}, we see that $-\eta_\alpha\Tr(\alpha^T\alpha')-\hat{n}_i=\hat{C}\hat{n}_i\hat{C}^{-1}$ must have exactly the same spectrum as $\hat{n}_i$, which is $\{0,1,\ldots,n_{\mathrm{max}}\}$, leading to Eq.~\eqref{eq:Tralphaalpha'}.  %

\subsubsection{A parastatistical generalization of Majorana fermions}\label{sec:gen_Majorana}
It will turn out to be useful in many situations to rewrite the CRs in Eq.~\eqref{eq:fundamental_Rcommu_NNC} in a more compact way as follows. 
First, we define $\hat{\phi}_{j,\zeta, b} \equiv
		(\hat{\psi}^+_{j,b},
		\bar{\psi}^-_{j,b})_\zeta$ for $\zeta=1,2$, and introduce 
collective labels $I=(i,\zeta_I),J=(j,\zeta_J)$.
Then we can rewrite all the CRs in Eq.~\eqref{eq:fundamental_Rcommu_NNC} in a single line as
\begin{eqnarray}\label{eq:fundamental_Rcommu_NNC_compact}
	\hat{\phi}_{I,a} \hat{\phi}_{J,b}=\sum_{c,d}R^{cd}_{ab}\hat{\phi}_{J,c} \hat{\phi}_{I,d}-\eta_\alpha \alpha_{ab}s_{IJ},%
\end{eqnarray}
where $\delta_{IJ}=\delta_{ij}\delta_{\zeta_I\zeta_J}$, and 
\begin{equation}
s_{IJ}\equiv \delta_{ij}\begin{pmatrix}
		0 & \mathds{1}\\
		-\eta_\alpha \mathds{1} & 0
	\end{pmatrix} _{\zeta_I\zeta_J} .
\end{equation}
The simple CR in Eq.~\eqref{eq:fundamental_Rcommu_NNC_compact} will be used in the exact solution of Bogoliubov-type Hamiltonians of $R$-paraparticles in Sec.~\ref{sec:u1_breaking_solu}, and also in constructing $R$-paraparticle quantum field theories. 

When $\eta_\alpha=-1$, we can also have an equivalent formulation of the theory in terms of $R$-parastatistical Majorana fermion operators. Define $\hat{\chi}_{j\zeta a}\equiv i^{\delta_{\zeta,-1}} (\hat{\psi}^+_{j,a}+\zeta \bar{\psi}^-_{j,a})$, where $\zeta=\pm1$. One can then check that Eq.~\eqref{eq:fundamental_Rcommu_NNC}~[or Eq.~\eqref{eq:fundamental_Rcommu_NNC_compact}] is equivalent to the following CR between $\hat{\chi}_{j\zeta a}$ 
\begin{equation}\label{eq:paraClifford}
\hat{\chi}_{I,a}\hat{\chi}_{J,b}=\sum_{c,d}\hat{\chi}_{J,c}\hat{\chi}_{I,d} R^{cd}_{ab}+2\delta_{IJ}\alpha_{ab},
\end{equation}
where we use the same collective labels as above. 
Eq.~\eqref{eq:paraClifford} is an $R$-parastatistical generalization of the Clifford algebra of Majorana fermion operators~\cite{Kitaev2001Majorana} where $R=-1$ and $\alpha=1$. 

Within this $R$-Majorana fermion formulation, the basic family of local observables is generated by bilinear operators of the form $\alpha'_{ab}\hat{\chi}_{Ia}\hat{\chi}_{Jb}$, and a general local observable $\hat{O}_S$ in a bounded region $S$ satisfies $[\hat{O}_S,\hat{\chi}_{Ia}]=0$ for any $I\in \bar{S}$.   %
The most general solvable bilinear~(non-interacting) Hamiltonian is
\begin{eqnarray}
	\hat{H}=\sum_{I,J,a,b}t_{IJ}\alpha'_{ab}\hat{\chi}_{Ia}\hat{\chi}_{Jb}.
\end{eqnarray}
For local Hamiltonians, we require that the tunneling matrix $t_{IJ}$ is zero for far separated $I$ and $J$. One can then check that the CRs in Eq.~\eqref{eq:paraClifford}  guarantee the locality of local observables and local Hamiltonians, and that $[\hat{H}, \hat{\chi}_{Ia}]$ is linear in $\{\hat{\chi}_{Ia}\}$.

\subsection{Pair creation for non-self-dual paraparticles}\label{sec:pair-creation-nonSD}
We now generalize the formalism of Sec.~\ref{sec:U1breaking} to non-self-dual $R$-paraparticles. Let  $\T$ be the set of particle types in a system of $R$-paraparticles, as considered in Sec.~\ref{sec:mutual_para}, and let $\psi,\bar{\psi}$ be two distinct simple particle types. 
We call $\bar{\psi}$ the antiparticle of $\psi$ if there exists a local operator $\hat{O}$~(in the sense of Definition~\ref{def:general_LO}) that creates a pair of particles $\psi$ and $\bar{\psi}$ from the vacuum $\ket{0}$. 
The ansatz for pair annihilation and creation operators used in Eq.~\eqref{def:e_pm_ab} is generalized as follows
\begin{eqnarray}\label{eq:e_pm_ab-psipsib}
\hat{e}^{(\psi\bar{\psi})-}_{ij}&=&\sum_{a,b}\alpha^{(\bar{\psi}\psi)}_{ba}\hat{\psi}_{i,a}^-\hat{\bar{\psi}}_{j,b}^-,\nonumber\\	
	\hat{e}^{(\psi\bar{\psi})+}_{ij}&=&-\sum_{a,b}\alpha^{\prime(\psi\bar{\psi})}_{ab}\hat{\psi}_{i,a}^+\hat{\bar{\psi}}_{j,b}^+,
\end{eqnarray}
where as before,  the tensors $\alpha^{(\bar{\psi}\psi)}$ and $\alpha^{\prime(\psi\bar{\psi})}$ are $m\times m$ constant matrices  
such that $\hat{e}^{(\psi\bar{\psi})-}_{ij}$ and $\hat{e}^{(\psi\bar{\psi})+}_{ij}$ commute with 
$\hat{\varphi}^\pm_{k,a}$ for any $k\notin\{i,j\}$ and any $\varphi\in \T$. 
This sets a condition between $\alpha^{(\bar{\psi}\psi)}$ and $\alpha^{\prime(\psi\bar{\psi})}$ and the relevant mutual $R$-matrices: %
\begin{equation}\label{eq:evcoevdef}
\adjincludegraphics[height=11ex,valign=c]{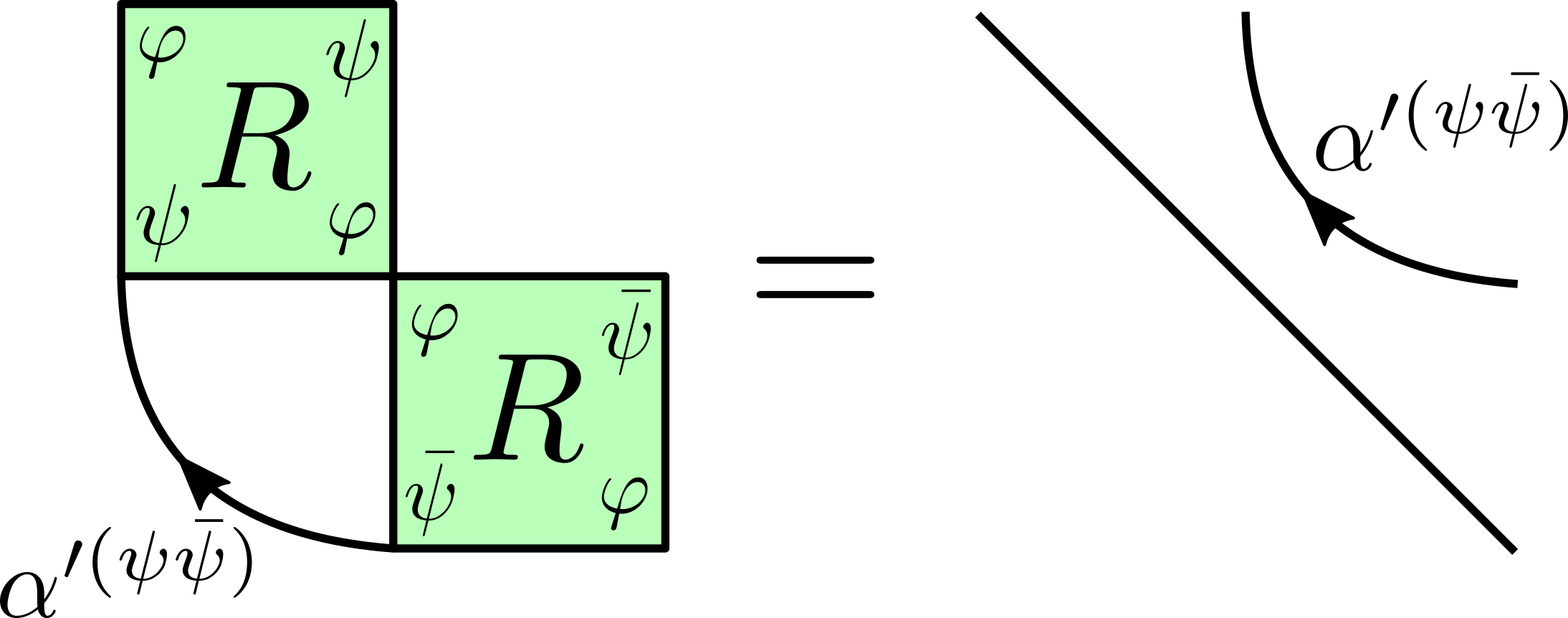}~,
\quad
\adjincludegraphics[height=11ex,valign=c]{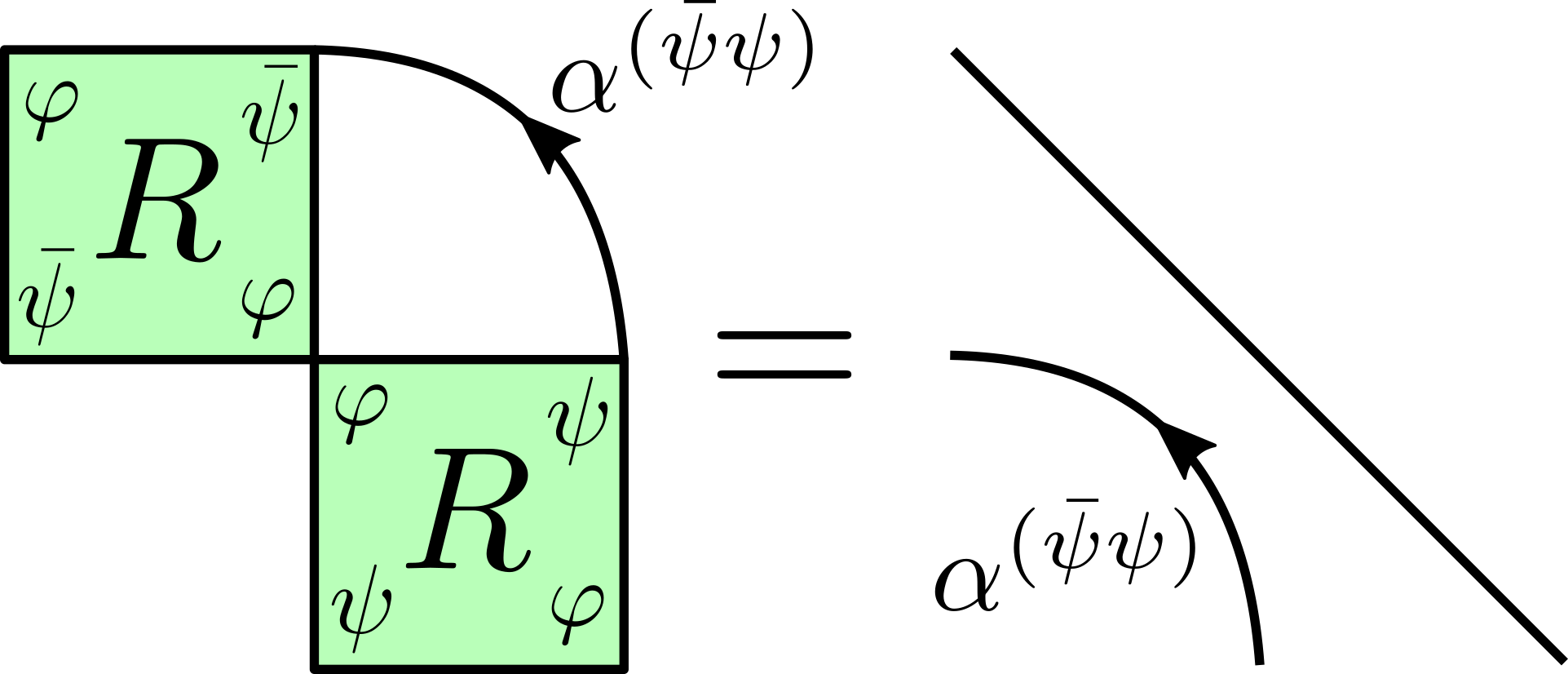}~,	
\end{equation}
which generalizes Eq.~\eqref{eq:RRalpha}. 
The equality $\alpha\alpha'=\mathds 1$ is now generalized to  the following 
\begin{equation}\label{eq:snakelemma}
	\adjincludegraphics[height=9ex,valign=c]{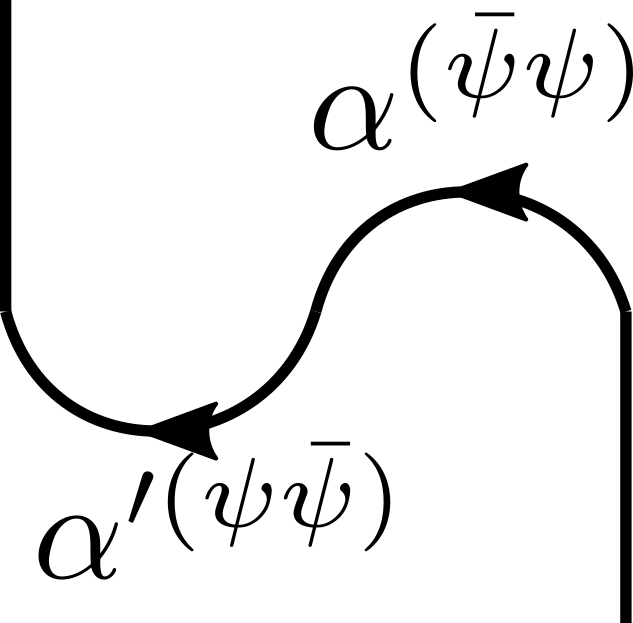}~=~\id_\psi~
    \adjincludegraphics[height=9ex,valign=c]{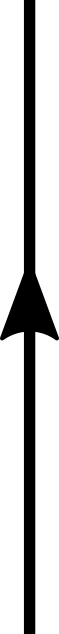}~~,
    \quad
	\adjincludegraphics[height=9ex,valign=c]{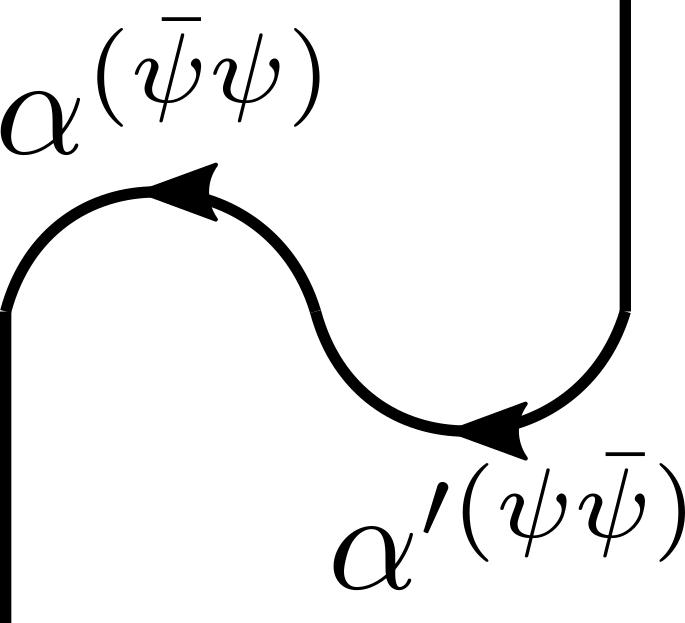}~=~\id_{\bar{\psi}}~
    \adjincludegraphics[height=9ex,valign=c]{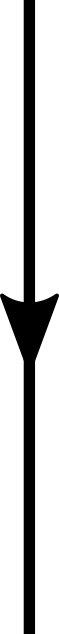}~~,
\end{equation}
which is reminiscent of the snake identity in rigid monoidal categories~\cite{TenCat_EGNO}.

The Lie algebra relations of all bilinear local operators can be computed exactly as in Sec.~\ref{sec:LApaircreationbilin}, and we do not repeat the calculations here, but point out that bilinear Hamiltonians still remain exactly solvable. 

In fact, if we use the composite $R$-matrix $\mathbf{R}$ defined in Eq.~\eqref{eq:compositeR-def}, then we can have a collective description of pair creation and annihilation of all particle types in the system, and it essentially leads us back to the formalism in Sec.~\ref{sec:U1breaking}
with $R$ replaced by $\mathbf{R}$. For example, we can write pair annihilation and creation operators collectively as
\begin{eqnarray}\label{eq:e_pm_ab-Psi}
\hat{e}^{-}_{ij}&=&\sum_{A,B}\boldsymbol{\alpha}_{BA}\hat{\Psi}_{i,A}^-\hat{\Psi}_{j,B}^-,\nonumber\\	
	\hat{e}^{+}_{ij}&=&-\sum_{A,B}\boldsymbol{\alpha}^{\prime}_{AB}\hat{\Psi}_{i,A}^+\hat{\Psi}_{j,B}^+,
\end{eqnarray}
where $\mathbf{R}$,  $\boldsymbol{\alpha}$, and $\boldsymbol{\alpha}'$ satisfy the first five lines  in Fig.~\ref{fig:Ralphaeqn}. The relation between Eq.~\eqref{eq:e_pm_ab-psipsib} and Eq.~\eqref{eq:e_pm_ab-Psi} is seen in the following expansion 
\begin{equation}
\hat{e}^{-}_{ij}=\sum_{\psi\in T} \hat{e}^{(\psi\bar{\psi})-}_{ij},
\end{equation}
where $\alpha^{(\bar{\psi}\psi)}_{ba}=\boldsymbol{\alpha}_{(\bar{\psi},b),(\psi,a)}$, and similarly for $\hat{e}^{+}_{ij}$.

\begin{fact}{(Standard construction for antiparticle)} If the $R$-matrix of $\psi$ satisfies the condition in Eq.~\eqref{eq:DUcondition}~(e.g., when $R$ is dual-unitary), then we can always define the antiparticle $\bar\psi$ by defining the mutual $R$-matrices as follows
\begin{equation}\label{eq:standard-dual}
\adjincludegraphics[height=14ex,valign=c]{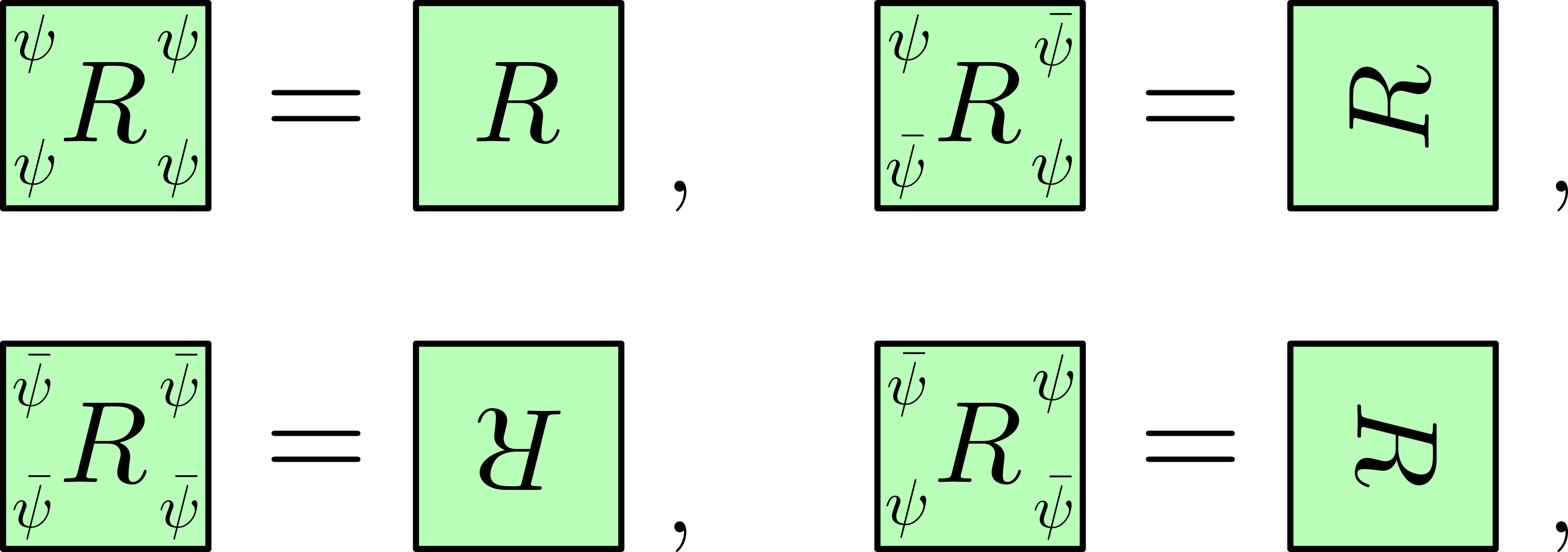}.	
\end{equation}
It is straightforward to verify that this assignment satisfy the generalized YBE in Eq.~\eqref{eq:YBEgraphical-rela}, hence it consistently defines the mutual statistics between $\psi$ and $\bar\psi$. Moreover, Eqs.~(\ref{eq:evcoevdef},\ref{eq:snakelemma}) are satisfied by taking 
$ \alpha^{(\bar{\psi}\psi)}=\alpha^{\prime(\psi\bar{\psi})}=\mathds{1}$, 
hence $\bar\psi$ is indeed the antiparticle of $\psi$ according to our earlier definition. 

The subtle point here is that if $R$ satisfies Eq.~\eqref{eq:Ralpha_def-graphical}, then from Fig.~\ref{fig:Ralphaeqn} we see that $R$ is actually isomorphic to $\Rmm$ in the sense of
Eq.~\eqref{eq:basistransformRmat}, where $\alpha$ gives the basis transformation, and we will see later in Definition~\ref{def:morphisms} that $\psi$ and $\bar\psi$ describe isomorphic particle types, i.e., $\psi\cong \bar\psi$ is a self-dual paraparticle. In this case we should use the formalism in Sec.~\ref{sec:U1breaking}, where $\psi$ and $\bar\psi$ are created by the same set of creation operators $\hat\psi^+_{i,a}$, instead of using an independent field system for $\bar\psi$.   
\end{fact}

We end this section by mentioning two  conjectures about unitary involutive $R$-matrices admitting pair creation: 
\begin{conjecture}\label{conj:nonsimpleR}
Let $R$ be a unitary involutive $R$-matrix. \\
(1). If there exists an invertible $\alpha$ satisfying Eq.~\eqref{eq:Ralpha_def-graphical}, then $R$ must be dual unitary;\\
(2). If $R$ is dual unitary, then we have $\thetaR\equiv \Tr_1[R]=\Tr_2[R]\in \DC[R]$ and $\thetaR^2=\mathds{1}$. 
\end{conjecture}
We expect these conjectures to be correct as they follow from natural expectations from physical and categorical considerations. Note that both conjectures are correct for unitary simple self-dual $R$-matrices, as shown by Thm.~\ref{thm:unitary_simple_self-dualR}. These conjectures are the major technical challenges for generalizing Thm.~\ref{thm:unitary_simple_self-dualR} to non-simple unitary $R$-matrices, or equivalently~(via the composite $\mathbf R$-matrix trick in Sec.~\ref{sec:mutual_para} and the direct sum decomposition Thm.~\ref{thm:directsum-decomposition}), to an $R$-paraparticle system with any number of particle types.

\section{Exact Solution of free paraparticles}\label{sec:solution}
In this section we present an important feature of our second formulation: it guarantees the solvability of general bilinear Hamiltonians
\begin{equation}\label{eq:u1_sym_H}
	\hat{H}=\sum_{1\leq i,j\leq N} h_{ij} \hat{e}_{ij} =\sum_{\substack{1\leq i,j\leq N\\1\leq a\leq m}} h_{ij}  \hat{\psi}^+_{i,a}\hat{\psi}^-_{j,a},
\end{equation} 
where the coefficient matrix $h_{ij}$ is required to be Hermitian $h_{ij}^*=h_{ji}$ so that $\hat{H}^\dagger=\hat{H}$. We will see that the solution boils down to finding the single-particle spectrum, which is obtained by diagonalizing the coefficient matrix $h_{ij}$. The whole many body spectrum is obtained by independently occupying  single-particle levels, each level satisfying the generalized exclusion principle discussed in Sec.~\ref{sec:exclusion_statistics_calc}.  Therefore, our construction extends the concept of free particle systems, from free fermions and bosons to free paraparticles.

The exact solution of free particle systems plays a fundamental role in our understanding of quantum many-body physics. First, free particle systems already display rich physics and are interesting in their own right. They exhibit interesting phases, such as the integer quantum Hall effect~\cite{thoulessQuantizedHallConductance1982}, topological insulators and superconductors~\cite{kitaevPeriodicTableTopological2009,TPSC-RMP}, and have interesting dynamical properties such as Anderson localization~\cite{andersonAbsenceDiffusionCertain1958}. Second, they are often the starting point for understanding interacting many body systems. Two of the most important approximations used in studying interacting systems, namely mean-field theory and perturbation theory, are often based on the solutions of free particles. Third, the technique of solving free particles contributes to the exact solutions of some strongly-interacting spin models~(via spin-fermion mapping), such as the 2D classical Ising model~\cite{Lieb1964Two-dimensional} or 1D quantum Ising model~\cite{pfeutyOnedimensionalIsingModel1970,dziarmaga2005dynamics}, the 1D quantum XY model~\cite{LSM1961Twosoluble}, and Kitaev's honeycomb models~\cite{kitaev2006anyons,Feng2007JWKitaev,Yao2007Exact,Mandal2009Exactly,Yao2011Fermionic,takagi2019concept,wangTopologicalCorrelationsThreedimensional2023,chapmanUnifiedGraphTheoreticFramework2023}. Finally, free particle solutions are also employed in a number of numerical techniques, such as the determinantal quantum Monte Carlo algorithm~\cite{whiteNumericalStudyTwodimensional1989}. 
We mention that the exact solution of free $R$-paraparticle systems have already played an important role in some recent works that apply $R$-paraparticle theory to study exotic quantum matter, such as the $R$-paraparticle analog of SYK model~\cite{liSpectralFormFactor2025,liNote$q2$Rparafermionic2025}, emergent $R$-paraparticles in interacting periodic chain~\cite{schurichtParastatisticsRevealedPeierls2025}, and the $R$-paraparticle analog of Luttinger liquid~\cite{salinelBosonization$R$paraparticleLuttinger2025}.

In Sec.~\ref{sec:u1_sym_solu} we present the exact solution to the family of $U(1)$ symmetric Hamiltonians in Eq.~\eqref{eq:u1_sym_H}, and in Sec.~\ref{sec:u1_breaking_solu} we extend the method to Hamiltonians containing the $U(1)$ breaking terms $\hat{e}^\pm_{ij}$ introduced in Sec.~\ref{sec:U1breaking}. 
\subsection{$U(1)$ symmetric case}\label{sec:u1_sym_solu}
One important property of the $R$-CRs in Eq.~\eqref{eq:fundamental_Rcommu_supress} is that they are invariant under $U(N)$ transformations of $\{\hat{\psi}^\pm_{i,a}\}$:
\begin{eqnarray}\label{eq:fourier_transform}
	\hat{\psi}^-_{i,a}&=& \sum^N_{k=1} U_{ki}^{*} \tilde{\psi}^-_{k,a},\nonumber\\
	\hat{\psi}^+_{i,a}&=& \sum^N_{k=1} U_{ki} \tilde{\psi}^+_{k,a},
\end{eqnarray}
where $U_{ki}$ is a $N\times N$ unitary matrix. For example, in translationally invariant systems it is useful to take $U_{ki}=\frac{1}{\sqrt{N}}e^{-i\vec{k}\cdot\vec{i}}$.  Inserting Eq.~\eqref{eq:fourier_transform} into Eq.~\eqref{eq:fundamental_Rcommu_supress}, we get
\begin{eqnarray}\label{eq:fundamental_Rcommu_kspace}
	\hat{\psi}^-_{k} \hat{\psi}^+_{p}&=&\hat{\psi}_{p}^+ \hat{\psi}_{k}^-\cdot\Rmp+\delta_{kp}\delta,\nonumber\\%
	\hat{\psi}^+_{k} \hat{\psi}^+_{p}&=&\hat{\psi}_{p}^+ \hat{\psi}_{k}^+ \cdot R,\nonumber\\
	\hat{\psi}^-_{k} \hat{\psi}^-_{p}&=&\hat{\psi}_{p}^- \hat{\psi}_{k}^-\cdot \Rmm. %
\end{eqnarray}
Therefore, the operators  $\{\hat{\psi}^\pm_{k,a}\}$ satisfy exactly the same CRs as $\{\hat{\psi}^\pm_{i,a}\}$. 
Notice that most of our discussions in Sec.~\ref{sec:second_quantization}, including Sec.~\ref{sec:bondLA} and Sec.~\ref{sec:state_space}, only assume the CRs in Eq.~\eqref{eq:fundamental_Rcommu_supress}, therefore, those results that hold for $\{\hat{\psi}^\pm_{i,a}\}$ also apply to $\{\hat{\psi}^\pm_{k,a}\}$. In particular,  the operators $\{\hat{\psi}^\pm_{k,a}\}$ lead to an equivalent basis for the same many particle state space, the  single-particle mode $k$ created by $\{\hat{\psi}^+_{k,a}\}^m_{a= 1}$ obeys the same exclusion statistics discussed in Sec.~\ref{sec:exclusion_statistics_calc}, and the bilinear operators $\hat{e}_{kp}=  \sum_a \hat{\psi}^+_{k,a}\hat{\psi}^-_{p,a}$   form an equivalent basis for the Lie algebra $\mathfrak{gl}_N$ discussed in Sec.~\ref{sec:bondLA}, with the same CRs as in  Eqs.~(\ref{eq:commu_Eab_psi_p}-\ref{eq:commu_Eab_Ecd}).

Inserting Eq.~\eqref{eq:fourier_transform} into Eq.~\eqref{eq:u1_sym_H}, the bilinear Hamiltonian transforms into
\begin{equation}\label{eq:u1_sym_H_kspace}
	\hat{H}=\sum_{1\leq k,p\leq N} h'_{kp} \hat{e}_{kp}=\sum_{\substack{1\leq k,p\leq N\\1\leq a\leq m}} h'_{kp} \hat{\psi}^+_{k,a}\hat{\psi}^-_{p,a},
\end{equation} 
where 
\begin{equation}\label{eq:coef_h_kp}
	h'_{kp}  =\sum_{1\leq i,j\leq N}U_{ki} h_{ij}U_{kj}^*=[UhU^\dagger]_{kp}.
\end{equation} 
We can therefore choose the unitary matrix $U$ to diagonalize the Hermitian coefficient matrix $h_{ij}$, with real eigenvalues $\{\epsilon_{k}\}^N_{k=1}$. The Hamiltonian then becomes 
\begin{equation}\label{eq:u1_sym_H_diag}
	\hat{H}=\sum^N_{k=1} \epsilon_{k} \hat{n}_{k}.
\end{equation} 
From the discussions in Sec.~\ref{sec:state_space} we know that the particle number operators $\{\hat{n}_k\}^N_{k=1}$ mutually commute and take independent eigenvalues. Their common eigenstates, as defined in Sec.~\ref{sec:state_space}, are $|{}^{\alpha_1}_{n_1},{}_{n_2}^{\alpha_2},\ldots, {}_{n_N}^{\alpha_N}\rangle$, with energy eigenvalues $E=\sum^N_{k=1} \epsilon_{k} n_{k}$, and $1\leq \alpha_j\leq d_{n_j}$ specifies the single-particle exclusion statistics.

We now calculate some physical observables at finite temperature $T$, and the corresponding zero temperature results can be obtained by taking the limit $T\to 0$. The partition function is a product of single mode partition functions in Eq.~\eqref{eq:single_mode_Z}
\begin{equation}\label{eq:part_func_Z}
	Z(\beta)\equiv \mathrm{Tr} [e^{-\beta \hat{H}}]=\prod_k z_R(e^{-\beta\epsilon_k}),
\end{equation}
where $z_R(x)$ is the Hilbert series defined in Eq.~\eqref{eq:single_mode_Z}. The free energy is
\begin{equation}\label{eq:free_energy}
	F(\beta)\equiv -\frac{1}{\beta}\ln Z(\beta)=-\frac{1}{\beta}\sum_k \ln z_R(e^{-\beta\epsilon_k}).
\end{equation}
The thermal expectation value of occupation number operator for mode $k$ is
\begin{equation}\label{eq:n_k_expectation}
	\langle \hat{n}_k\rangle_\beta\equiv \frac{\mathrm{Tr} [\hat{n}_k e^{-\beta \hat{H}}]}{ \mathrm{Tr} [e^{-\beta \hat{H}}]}=\frac{z'_R(e^{-\beta\epsilon_k})e^{-\beta\epsilon_k}}{z_R(e^{-\beta\epsilon_k})},
\end{equation}
and $\langle \hat{e}_{kp}\rangle_\beta=0$ for $k\neq p$. Fig.~\ref{fig:Tempdep} plots $\langle \hat{n}_k\rangle_\beta$ as a function of $\beta\epsilon_k$ for the $R$-matrices in examples~\ref{ex:1m} and \ref{ex:1m1} of Sec.~\ref{sec:YBE} with $m=5$, which shows a distinct finite temperature behavior of paraparticles compared to ordinary fermions and bosons. 
\begin{figure}
	\center{\includegraphics[width=0.7\linewidth]{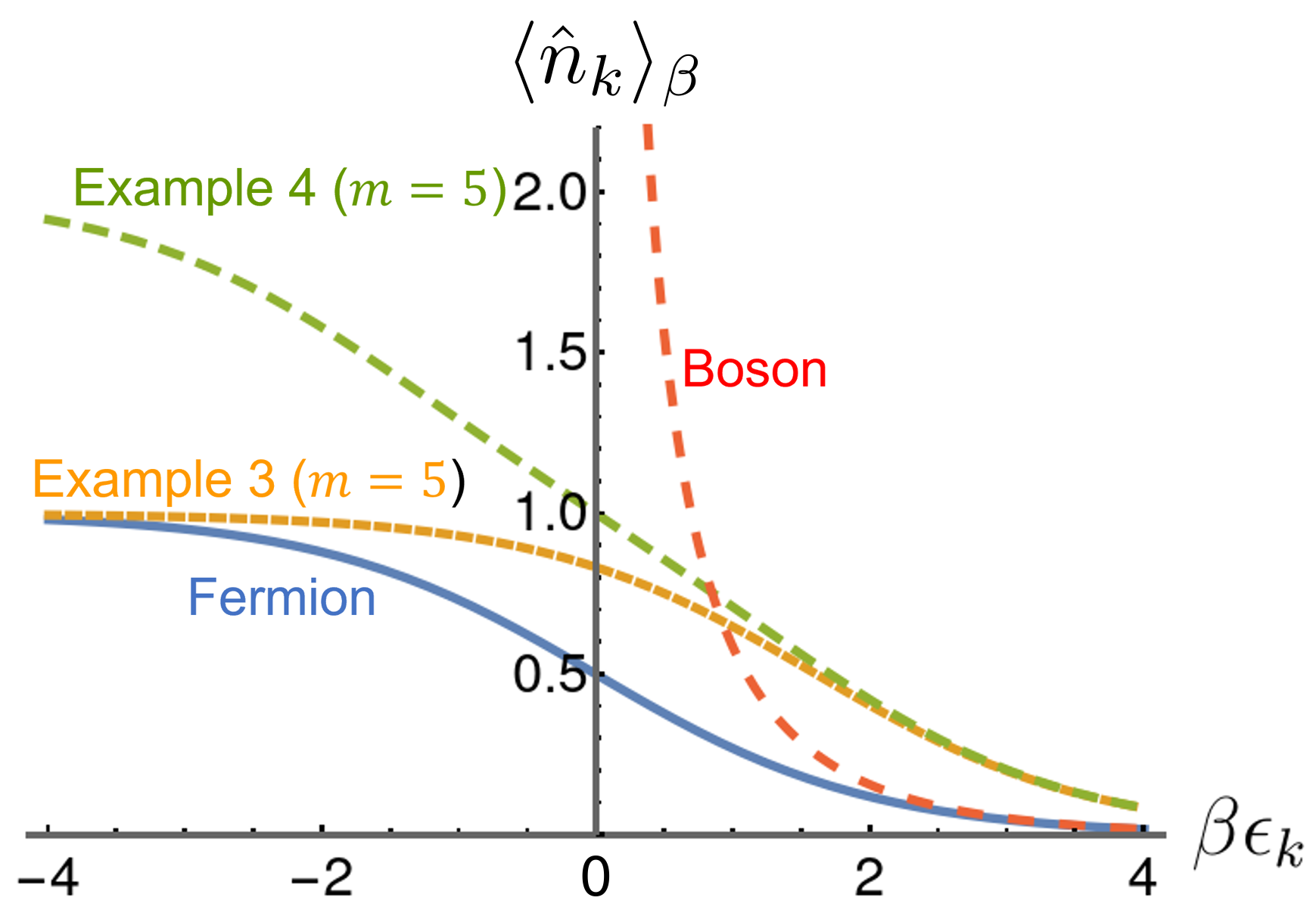}}
	\caption{\label{fig:Tempdep} The thermal expectation value of single mode occupation number $\langle \hat{n}_k\rangle_\beta$: comparison between paraparticles~(labeled by the $R$-matrices in examples~\ref{ex:1m} and \ref{ex:1m1} of Sec.~\ref{sec:YBE} with $m=5$) and ordinary fermions and bosons. }
\end{figure}

The thermal average for real space local operators $ \hat{e}_{ij}$ are obtained from Eq.~\eqref{eq:fourier_transform} and the result for $\langle \hat{e}_{kp}\rangle_\beta$
\begin{equation}
	\langle \hat{e}_{ij}\rangle_\beta=\sum_k U_{ki} U_{kj}^*\langle \hat{n}_k\rangle. 
\end{equation} 
Observables that are products of $\{\hat{e}_{ij}\}$~(or $\{\hat{e}_{kp}\}$) can be calculated using the Lie algebra relations in Eq.~\eqref{eq:commu_Eab_Ecd} and the results for  $\langle \hat{e}_{kp}\rangle_\beta$. For example, the thermal expectation value of quadratic products of  $\{\hat{e}_{ij}\}$~(e.g. $\hat{e}_{ij}\hat{e}_{ji}$ and $\hat{n}_i \hat{n}_j$) can always be written as linear combinations of $\langle \hat{e}_{kp}\hat{e}_{pk}\rangle_\beta$, which can be obtained as follows. First, we have
\begin{eqnarray}\label{eq:ekp_epk_1}
	\langle \hat{e}_{kp}\hat{e}_{pk}\rangle_\beta&=&\frac{1}{Z}\mathrm{Tr}[ \hat{e}_{kp}\hat{e}_{pk}e^{-\beta \hat{H}}]\nonumber\\
	&=&\frac{1}{Z}\mathrm{Tr}[ \hat{e}_{pk}\hat{e}^{-\beta \hat{H}}\hat{e}_{kp}]\nonumber\\
	&=&\langle \hat{e}_{pk}\hat{e}_{kp}\rangle_\beta e^{\beta(\epsilon_p-\epsilon_k)},
\end{eqnarray}
where in the last line we used $[\hat{H},\hat{e}_{kp}]=(\epsilon_k-\epsilon_p)\hat{e}_{kp}$. On the other hand, using Eq.~\eqref{eq:commu_Eab_Ecd}, we have 
\begin{eqnarray}\label{eq:ekp_epk_2}
	\langle \hat{e}_{kp}\hat{e}_{pk}\rangle_\beta-\langle \hat{e}_{pk}\hat{e}_{kp}\rangle_\beta&\equiv&\langle [\hat{e}_{kp},\hat{e}_{pk}]\rangle_\beta\nonumber\\
	&=&\langle \hat{n}_{k}\rangle_\beta-\langle \hat{n}_{p}\rangle_\beta.
\end{eqnarray}
Combining Eqs.~\eqref{eq:ekp_epk_1} and \eqref{eq:ekp_epk_2}, we get
\begin{eqnarray}\label{eq:ekp_epk}
	\langle \hat{e}_{kp}\hat{e}_{pk}\rangle_\beta=\frac{\langle \hat{n}_{k}\rangle_\beta-\langle \hat{n}_{p}\rangle_\beta}{1-e^{\beta(\epsilon_k-\epsilon_p)}}.
\end{eqnarray}
Higher-order products of $\{\hat{e}_{kp}\}$ can be obtained similarly. 

We can also calculate the unequal time correlators between local observables, for example
\begin{equation}\label{eq:unequal_t_correlator}
	\langle [\hat{n}_i(t),\hat{n}_j(0)]\rangle_\beta=\sum_{k,p} U_{ki}U_{pi}^* U_{kj}^* U_{pj} \langle \hat{n}_{k}-\hat{n}_{p}\rangle_\beta e^{it(\epsilon_k-\epsilon_p)}.
\end{equation}
Eq.~\eqref{eq:unequal_t_correlator} is obtained by expanding $\hat{n}_i(t)$ and $\hat{n}_j(0)$ as linear combinations of $\hat{e}_{kp}$ and then using Eq.~\eqref{eq:ekp_epk}, as well as the solution $\hat{e}_{kp}(t) = \hat{e}_{kp} e^{i (\epsilon_p-\epsilon_k) t}$ to the Heisenberg equations of motion.

An alternative method to calculate higher-order products of $\{\hat{e}_{kp}\}$ is to use the generalization of Wick's theorem to non-interacting paraparticles that we develop in Sec.~\ref{sec:wick_thm}. This allows straightforward, systematic computation of arbitrary products. 

\subsection{With $U(1)$ breaking terms}\label{sec:u1_breaking_solu}
In this section we extend the method above to the family of bilinear Hamiltonians involving $U(1)$ breaking terms $\hat{e}^\pm_{ij}$, as introduced in Sec.~\ref{sec:U1breaking}. The most general form of this family of Hamiltonians is 
\begin{eqnarray}\label{eq:u1_breaking_H}
	\hat{H}&=&\sum_{1\leq i,j\leq N} (M_{ij} \hat{e}_{ij}+P_{ij}\hat{e}^-_{ij}/2+\eta_\alpha N_{ij}\hat{e}^+_{ij}/2 )\\
	&=&\sum_{i,j,a}M_{ij}  \hat{\psi}^+_{i,a}\hat{\psi}^-_{j,a}+\sum_{i,j,a,b}\left(P_{ij}\alpha_{ba}\hat{\psi}^-_{i,a}\hat{\psi}^-_{j,b}/2\right.\nonumber\\
	&&\hspace{1.1in}{}-\left.\eta_\alpha N_{ij}\alpha'_{ab}\hat{\psi}^+_{i,a}\hat{\psi}^+_{j,b}/2\right)\nonumber\\
	&=&\!\!\sum_{i,j,a,b}\frac{\alpha'_{ab}}{2}\begin{pmatrix}
		\hat{\psi}^+_{i,a} & \bar{\psi}^-_{i,a}
	\end{pmatrix}\!\!\begin{pmatrix}
		M_{ij} & N_{ij}\\
		P_{ij} & -M_{ji}
	\end{pmatrix}\!\!	
	\begin{pmatrix}
		\bar{\psi}^-_{j,b} \\
		-\eta_\alpha \hat{\psi}^+_{j,b} 
	\end{pmatrix}%
        -C\nonumber
\end{eqnarray}
where the matrices $P_{ij},N_{ij}$ can be taken to satisfy $P^T=\eta_\alpha P, N^T=\eta_\alpha N$~(since the relations $\hat{e}^\pm_{ij}=\eta_\alpha \hat{e}^\pm_{ji}$ allows us to replace $P_{ij}$ by $\bar{P}_{ij}\equiv( P_{ij} + \eta_\alpha P_{ji} )/2$ without changing the Hamiltonian), and the constant  $C=\eta_\alpha \mathrm{Tr}[\alpha^T\alpha']\mathrm{Tr}[M]/2$ comes from commuting $\bar{\psi}^-_{i,a}$ with $\hat{\psi}^+_{j,b}$  using the first CR in Eq.~\eqref{eq:fundamental_Rcommu_NNC}. 
Note that hermiticity imposes additional constraints on the coefficient matrices $M,N,P$, which we discuss at the end of this section, and our treatment below should work for diagonalizable Hamiltonians in general, even including non-Hermitian ones.

We now introduce $2N\times 2N$ matrices $h$ and $s$ and a $2N$ dimensional vector $\phi_b$~($2N$ dimensional for each $b$) 
\begin{equation}\label{eq:def_hsphi}
	h\equiv \begin{pmatrix}
		M & N\\
		P & -M^T
	\end{pmatrix}, 
	s\equiv \begin{pmatrix}
		0 & \mathds{1}\\
		-\eta_\alpha \mathds{1} & 0
	\end{pmatrix}, 	
	\phi_{b} \equiv	\begin{pmatrix}
		\hat{\psi}^+_{b} \\
		\bar{\psi}^-_{b} 
	\end{pmatrix},
\end{equation}
which allows us to suppress the mode labels $i,j$ and write the last line of Eq.~\eqref{eq:u1_breaking_H} in  the more compact form
\begin{equation}\label{eq:H_NNC_compact}
	\hat{H}=\frac{1}{2}\sum_{a,b}\alpha'_{ab}\phi^T_{a}\cdot h\cdot s\cdot \phi_{b} -C,
\end{equation}
and the operators $\phi_b$ satisfy the CRs in Eq.~\eqref{eq:fundamental_Rcommu_NNC_compact}.

We now try to find a canonical transformation of the creation and annihilation operators that preserves their algebra~\eqref{eq:fundamental_Rcommu_NNC_compact} while diagonalizing the Hamiltonian in Eq.~\eqref{eq:H_NNC_compact}.
The matrices $h,s$ defined in Eq.~\eqref{eq:def_hsphi} satisfy 
\begin{equation}\label{def:so_sp}
	sh+h^Ts=0.
\end{equation}
The Lie algebra $\mathfrak{sp}_{2N}$~(for $\eta_\alpha=+1$) or $\mathfrak{so}_{2N}$~(for $\eta_\alpha=-1$) is defined by all such $2N\times 2N$ matrices $h$ satisfying Eq.~\eqref{def:so_sp}~\cite{humphreys_LA}. The corresponding Lie group Sp$(2N)$ or SO$(2N)$ is defined by all $2N\times 2N$ invertible matrices $F$ satisfying
\begin{equation}\label{def:SO_SP}
	F^TsF=s.
\end{equation}
Note that inserting $F=e^{ih}$ into Eq.~\eqref{def:SO_SP} for infinitesimal $h$ reproduces  Eq.~\eqref{def:so_sp}. Also note that due to $s^2=-\eta_\alpha \mathds{1}$, Eq.~\eqref{def:SO_SP} is equivalent to $FsF^T=s$~(to see this, note that $F^TsFs=s^2\propto\mathds{1}$ implies $ Fs F^Ts=s^2$, since every matrix commutes with its inverse). 

It is clear that the CR in Eq.~\eqref{eq:fundamental_Rcommu_NNC_compact} is invariant under the transformation 
$\tilde{\phi}_b= F^T \phi_b$
for any element $F$ of the Lie group~[i.e. for any $F$ satisfying Eq.~\eqref{def:SO_SP}]. Written in terms of $\tilde{\phi}_b$, the Hamiltonian $\hat{H}$ still has the same form in Eq.~\eqref{eq:H_NNC_compact}, with $\tilde{h}=F^{-1}hF$~(the adjoint action of the Lie group on the corresponding Lie algebra). 
We choose $F$ to transform $h$ to the diagonal form 
$\tilde{h}=\mathrm{diag}\{\epsilon_1,\ldots,\epsilon_N,-\epsilon_1,\ldots,-\epsilon_N\}$~(this can always be done if $h$ is a semisimple element~\cite{humphreys_LA} of the corresponding Lie algebra. 
In general, for a semisimple element $l$ of any semisimple Lie algebra $\mathfrak{l}$, there always exists an inner automorphism that transforms $l$ into an element of a given Cartan subalgebra~\cite{humphreys_LA}, the explicit derivation for the current case is provided below).  Then the Hamiltonian becomes 
\begin{equation}\label{eq:H_diag_NNC}
	\hat{H}=\sum_{k}\epsilon_{k} \tilde{n}_k+\frac{\eta_\alpha}{2}\mathrm{Tr}[\alpha^T\alpha']\left(\sum_k\epsilon_{k}-\mathrm{Tr}[M]\right),
\end{equation}
where $\tilde{n}_k=\sum_a\tilde{\psi}^+_{k,a}\tilde{\psi}^-_{k,a}$. Therefore the task of solving the many-body spectrum boils down to finding eigenvalues of the $2N\times 2N$ matrix $h$. Once this is done, Eq.~\eqref{eq:H_diag_NNC} tells us the many-body spectrum, and the energy eigenstates are simultaneous eigenstates of $\{\tilde{n}_k\}^N_{i=1}$, as described in Sec.~\ref{sec:state_space}. This basis of many-body eigenstates satisfy the same generalized exclusion principles described in Sec.~\ref{sec:exclusion_statistics_calc}, since the transformed operators $\tilde{\psi}^+_{k,a},\tilde{\psi}^-_{k,a}$ satisfy the same CRs in Eq.~\eqref{eq:fundamental_Rcommu_supress}. Therefore, the expressions in  Eqs.~(\ref{eq:part_func_Z}-\ref{eq:n_k_expectation}) for the partition function $Z(\beta)$, the free energy $F(\beta)$,  and the single-mode occupation number $\langle\tilde{n}_k\rangle_{\beta}$ are still valid~[up to a constant factor in $Z(\beta)$ and a constant zero point energy shift in $F(\beta)$].

For the purpose of calculating physical observables and correlation functions, it is helpful to use an explicit  parametrization of $F^T$ 
\begin{equation}F^T=\begin{pmatrix}
		U^* & V^*\\
		V & U
	\end{pmatrix},\end{equation} 
where  $U,V$ are $N\times N$ matrices satisfying $UU^\dagger-\eta_\alpha VV^\dagger=\mathds{1}$ and $UV^T=\eta_\alpha VU^T$. Then we can explicitly write down the transformation of the creation and annihilation operators by expanding  $\tilde{\phi}_a= F^T \phi_a$:
\begin{eqnarray}\label{eq:Bogoliubov-transform}
	\tilde{\psi}^+_{k,a}&=&U^*_{ki}\hat{\psi}^+_{i,a}+V^*_{ki}\bar{\psi}^-_{i,a},\nonumber\\
	\tilde{\bar{\psi}}^-_{k,a}&=&U_{ki}\bar{\psi}^-_{i,a}+V_{ki}\hat{\psi}^+_{i,a},
\end{eqnarray}
or the inverse transformation $s\phi_a=F\cdot s\cdot \tilde{\phi}_a$:
\begin{eqnarray}
	\bar{\psi}^-_{i,a}&=&U^*_{ki}\tilde{\bar{\psi}}^-_{k,a}-\eta_\alpha V_{ki}\tilde{\psi}^+_{k,a},\nonumber\\
	\hat{\psi}^+_{i,a}&=&-\eta_\alpha V^*_{ki}\tilde{\bar{\psi}}^-_{k,a}+U_{ki}\tilde{\psi}^+_{k,a},
\end{eqnarray}
which are generalizations of the Bogoliubov transformation. We can then calculate the thermal expectation values of all physical observables, by expanding $\hat{e}_{ij},\hat{e}^\pm_{ij}$ in terms of the eigenmodes $\tilde{e}_{kp},\tilde{e}^\pm_{kp}$. For example, we have
\begin{equation}
	\langle \hat{e}_{ij}\rangle_\beta=\sum_k (U_{ki}U_{kj}^*+V_{ki}^*V_{kj}\eta_\alpha) \langle\tilde{n}_k\rangle_{\beta}+[V^\dagger V]_{ij}\mathrm{Tr}[\alpha^T\alpha'].
\end{equation}
\begin{remark}
Although our treatment above in principle works for non-Hermitian Hamiltonians, for physical applications we often focus on the Hermitian case. Following the discussion in Sec.~\ref{para:herm_paircreation}, if $z_R(x)$ is finite, then we have $(\hat{e}_{ij})^\dagger=\hat{e}_{ji}$ and $(\hat{e}^\pm_{ij})^\dagger= \hat{e}^\mp_{ij}$, and therefore the hermiticity of $\hat{H}$ in Eq.~\eqref{eq:u1_breaking_H} requires 
\begin{equation}
M^\dagger=M,\quad P^\dagger=N. 
\end{equation}
In this case, the  single-particle Hamiltonian $h$ in Eq.~\eqref{eq:def_hsphi} is also Hermitian, implying that the entire spectrum is real, as expected.
 
If $R$ is unitary and $z_R(x)$ is infinite~(i.e. $\thetaR=+1$), then the hermiticity of $\hat{H}$ in Eq.~\eqref{eq:u1_breaking_H} requires instead
\begin{equation}
	M^\dagger=M,\quad P^\dagger=-N. 
\end{equation}
In this case, the  single-particle Hamiltonian $h$ in Eq.~\eqref{eq:def_hsphi} is not Hermitian, and the  single-particle energy may be complex. We emphasize that this is not an idiosyncrasy of paraparticles, as it already happens in the case of ordinary bosons. As a simple example, the single mode Hamiltonian $\hat{H}=\hat{b}^2+(\hat{b}^\dagger)^2$ is superficially Hermitian, but nevertheless has a purely imaginary spectrum. We do not discuss how to treat this subtle issue in the general case in this paper, but we will encounter this issue again in Part II when we generalize Klein-Gordon field theory to paraparticles with $R$-matrix quantization, and there we will explicitly show how to make the theory physical and consistent.
\end{remark}

\subsection{Wick's theorem for free paraparticles}\label{sec:wick_thm}
In this section we generalize Wick's theorem~\cite{wickEvaluationCollisionMatrix1950} to free paraparticles.
Wick's  theorem is a fundamental computational tool in quantum many-body physics~\cite{altland_simons_2010} and relativistic quantum field theories~\cite{peskin2018introduction,weinbergQuantumTheoryFields1995}, as 
it allows us to reduce all expectation values in a free particle theory to a linear combination of products of contractions. %
For simplicity we focus on self-dual $R$-paraparticles, i.e., we assume that there exists an invertible $\alpha$ satisfying Eq.~\eqref{eq:Ralpha_def-graphical}. Below we generalize the treatment in Ref.~\cite{molinari2017notes} to paraparticles. 
We will see that once we correctly generalize the notion of normal ordering and contraction to paraparticles, all the derivations and results for the original Wick's theorem generalize straightforwardly. 

In this section we use $\hat{A}_a$ to denote a generic paraparticle operator of the form 
\begin{equation}\label{eq:def:A_aWick}
	\hat{A}_a=\sum_i\left(f_{i}\hat{\psi}^+_{i,a}+g_{i}\bar{\psi}^-_{i,a}\right),
\end{equation}
where $f_{i}$ and $g_{i}$ are some coefficients, and $\bar{\psi}^-_{i,a}\equiv \sum_b\hat{\psi}^-_{i,b}\alpha_{ab}$. We also introduce the notation
\begin{equation}
	\hat{A}^+_a=\sum_i f_{i}\hat{\psi}^+_{i,a}, \quad \hat{A}^-_a=\sum_i g_{i}\bar{\psi}^-_{i,a},
\end{equation}
i.e., $\hat{A}^+_a$~($\hat{A}^-_a$) is the creation~(annihilation) part of $\hat{A}_a$. We define the normal ordering symbol as follows
\begin{definition}{(Normal ordering)}
	The normal ordering ${\sf N}$ of the operator product $\hat{A}_{1,a_1} \cdots \hat{A}_{n,a_n}$ is defined via its action on each component as
	\begin{eqnarray}\label{normalordering}
		{\sf N}[\hat{A}_{1,a_1}^\pm \cdots \hat{A}_{n,a_n}^\pm ]=\sum_{b_1,\ldots,b_n}&&\hspace{-0.17in}\hat{A}_{\sigma_1,b_1}^+\cdots \hat{A}_{\sigma_k,b_k}^+\hat{A}_{\sigma_{k+1},b_{k+1}}^- \nonumber\\
		&&\hspace{-0.3in}\times\cdots \hat{A}_{\sigma_n,b_n}^- [\rho(\bar{\sigma})]_{a_1,\ldots,a_n}^{b_1,\ldots,b_n}, 
	\end{eqnarray}
	where $\sigma\in S_n$ is a permutation %
    that moves creation operators to the left and annihilation operators to the right, $\bar{\sigma}=\sigma^{-1}$, and  $\rho$ is the representation of the symmetric group $S_n$ generated by the $R$-matrix, as defined in  Sec.~\ref{sec:YBE}.
	Eq.~\eqref{normalordering} is extended to ${\sf N}[\hat{A}_{1,a_1} \cdots \hat{A}_{n,a_n}]$ by linearity~\footnote{Here we remark that, rigorously speaking, ${\sf N}$ should be regarded as a linear map defined on the free algebra generated by $\psi_{i,a}^\pm$, and since ${\sf N}$ does not respect the $R$-CR~\eqref{eq:fundamental_Rcommu}, it does not descend to the quotient algebra, i.e., ${\sf N}$ is not a linear map defined on the second quantization algebra of $R$-paraparticles. }. 
\end{definition}
Note that the normal ordering symbol so defined satisfies the following symmetry property
\begin{align}\label{eq:NOsymmeetry_property}
	{\sf N}[\hat{A}_1\cdots \hat{A}_n] =  {\sf N}[\hat{A}_{\sigma_1}\cdots \hat{A}_{\sigma_n}]\cdot\rho(\bar{\sigma}),
\end{align}
where we suppress the paraparticle indices. 

As an example, we have
\begin{equation}\label{eq:NOexample}
	{\sf N}\left[~\adjincludegraphics[width=15ex,valign=c]{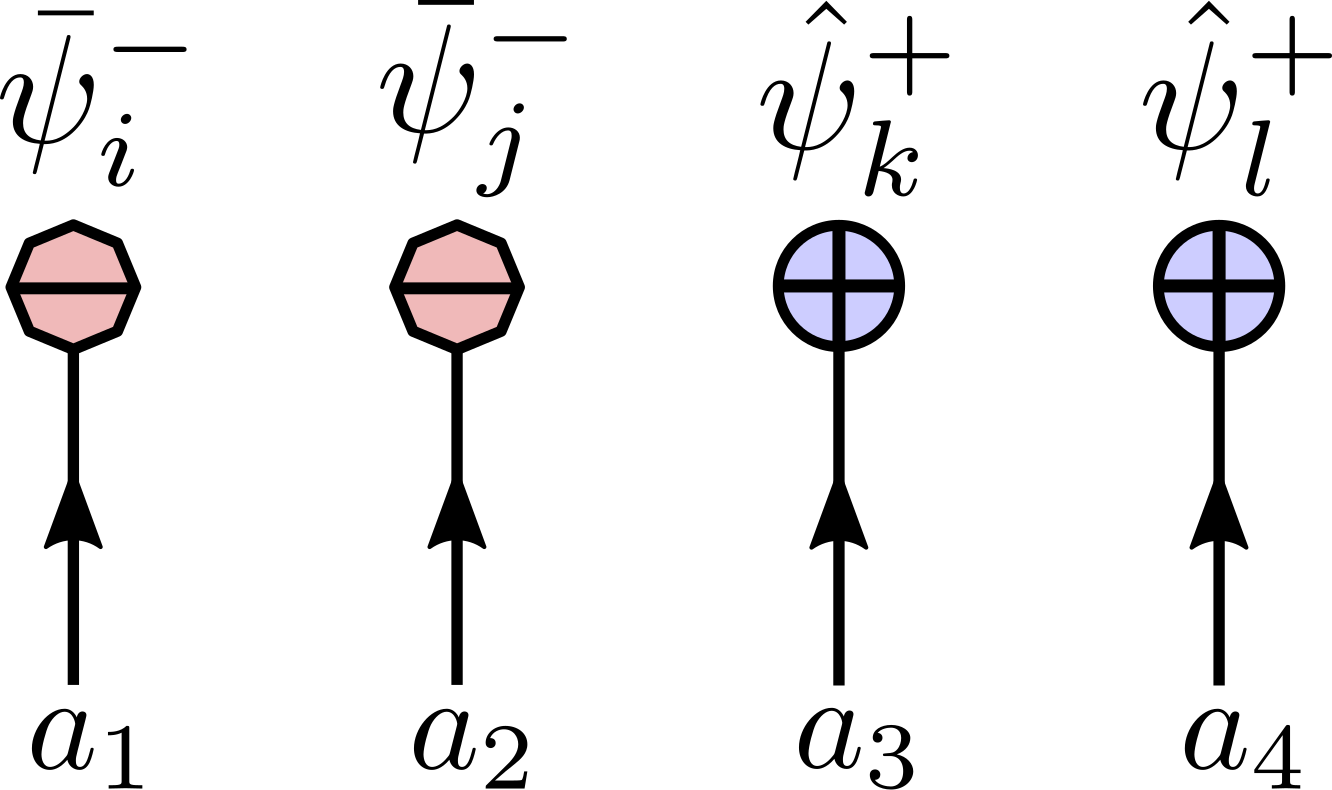}~\right]~=~
    \adjincludegraphics[width=15ex,valign=c]{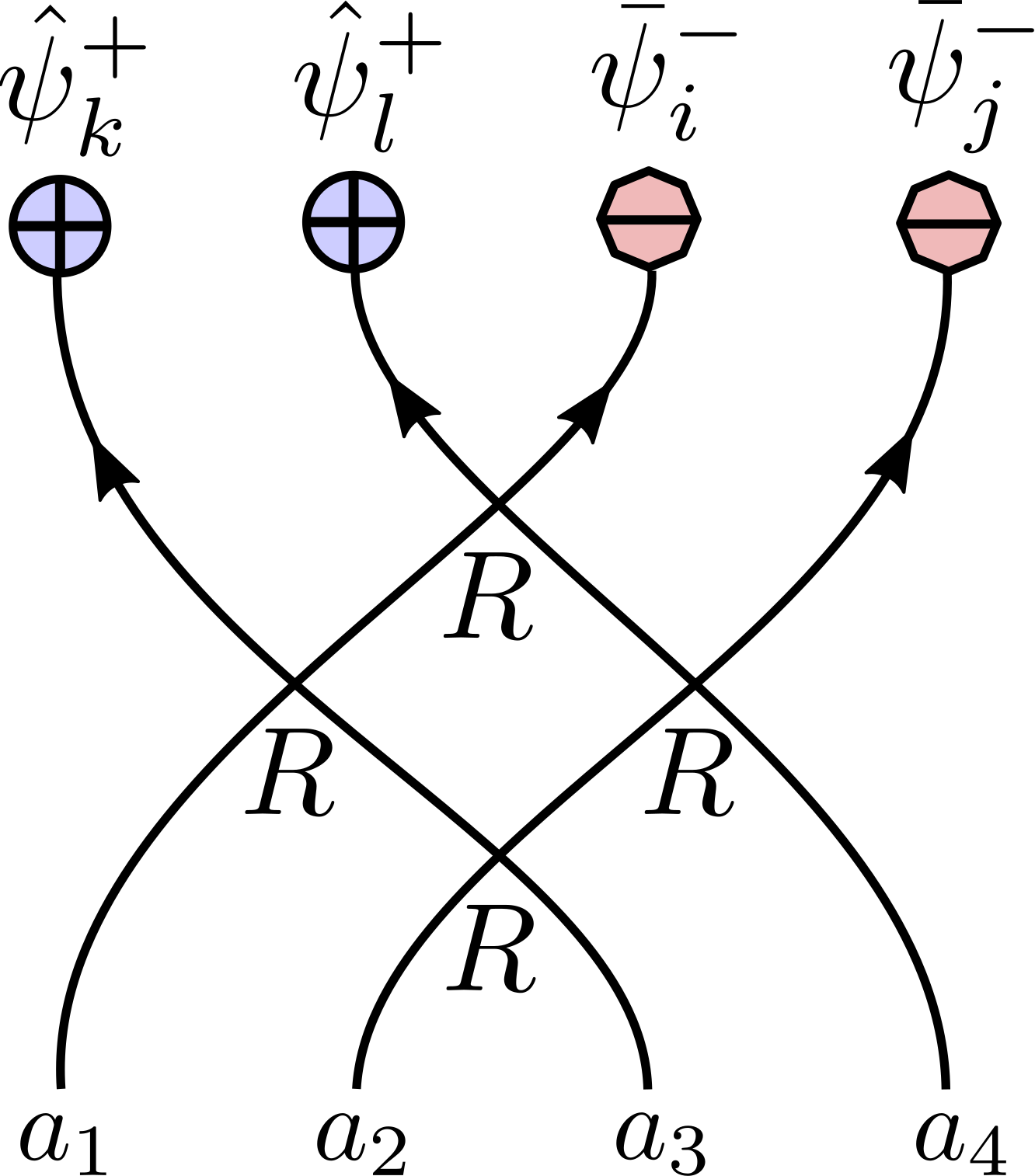}~,
\end{equation}
where in this section we use a labeled crossing to indicate an $R$-matrix 
\begin{equation}\label{eq:R-crossing}
	\adjincludegraphics[height=6ex,valign=c]{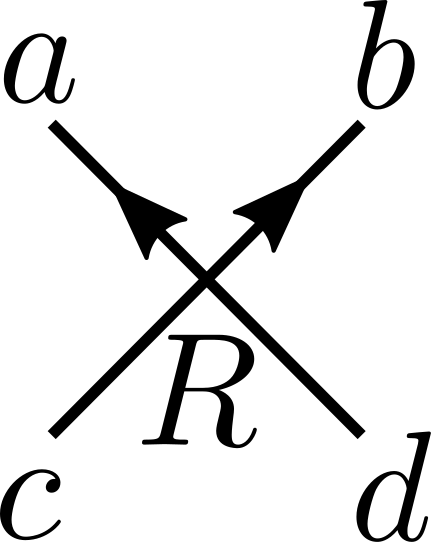}~=~
    \adjincludegraphics[height=6ex,valign=c]{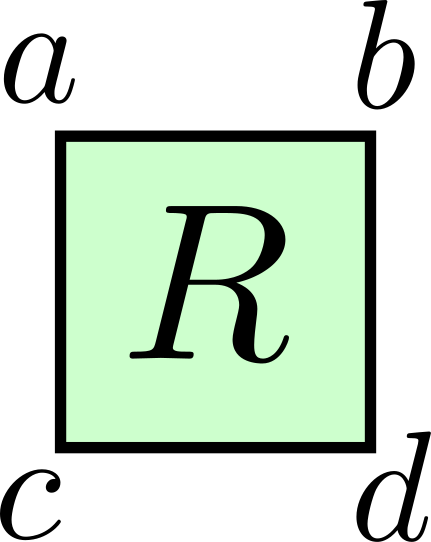}~.
\end{equation}
We now define contraction as follows
\begin{eqnarray}\label{eq:NOcontraction}
	\hat{A}_1\hat{A}_2 &=& (\hat{A}_1^++\hat{A}_1^-)(\hat{A}_2^++\hat{A}_2^-)\nonumber\\
	&=&{\sf N}[\hat{A}_1\hat{A}_2]  + \contraction{}{\hat{A}_1}{}{\hat{A}_2}\hat{A}_1\hat{A}_2,
\end{eqnarray}
where $\contraction{}{\hat{A}_1}{}{\hat{A}_2}\hat{A}_1\hat{A}_2$ (the contraction of $\hat{A}_2\hat{A}_1$) is the constant term obtained from moving $\hat{A}_1^-$ to the right of $\hat{A}_2^+$. %
More explicitly, we have
\begin{definition}{(Contraction)}
The contraction of two paraparticle operators is defined as%
\begin{eqnarray}
    \contraction{}{\hat{A}_1}{}{\hat{A}_2}\hat{A}_1\hat{A}_2&=&\hat{A}_1^-\hat{A}_2^+-\hat{A}_2^+\hat{A}_1^-\cdot R\nonumber\\	
    &=&\langle 0| \hat{A}_1\hat{A}_2| 0\rangle,
\end{eqnarray}
where $\ket{0}$ is the vacuum state annihilated by all $\hat\psi_{i,a}^-$ in Eq.~\eqref{eq:def:A_aWick}. 
	This is extended to non-adjacent contraction as
	\begin{eqnarray}
		\contraction{}{\hat{A}}{(\hat{A}_1\cdots \hat{A}_n)}{\hat{A}'}{\hat{A}}{(\hat{A}_1\cdots \hat{A}_n)}{\hat{A}^\prime}
		= \contraction{}{\hat{A}}{}{\hat{A}^\prime}{\hat{A}}{\hat{A}^\prime} (\hat{A}_{1}\cdots \hat{A}_{n}) \cdot \rho(\sigma),
	\end{eqnarray}
	where $\sigma$ is the corresponding permutation. 
\end{definition}
This definition is extended in an obvious way to more complicated cases involving intersecting contractions, such as
\begin{equation}
	\contraction{}{\hat{A}}{_i\cdots \hat{A}_k \cdots}{\hat{A}_j}
	\contraction[1.5ex]{\hat{A}_i\cdots}{\hat{A}}{_k\cdots \hat{A}_j\cdots }{\hat{A}_l}
	\hat{A}_i\cdots \hat{A}_k \cdots \hat{A}_j \cdots \hat{A}_l.
\end{equation}
For example, we have
\begin{eqnarray}
	\contraction{\hat{A}_1\hat{A}_2}{\hat{A}_3}{\hat{A}_4\hat{A}_5}{\hat{A}_6}
	\contraction[1.5ex]{}{\hat{A}_1}{\hat{A}_2\hat{A}_3\hat{A}_4}{\hat{A}_5}
	\hat{A}_1 \hat{A}_2 \hat{A}_3 \hat{A}_4 \hat{A}_5 \hat{A}_6&=&
    \adjincludegraphics[width=24ex,valign=c]{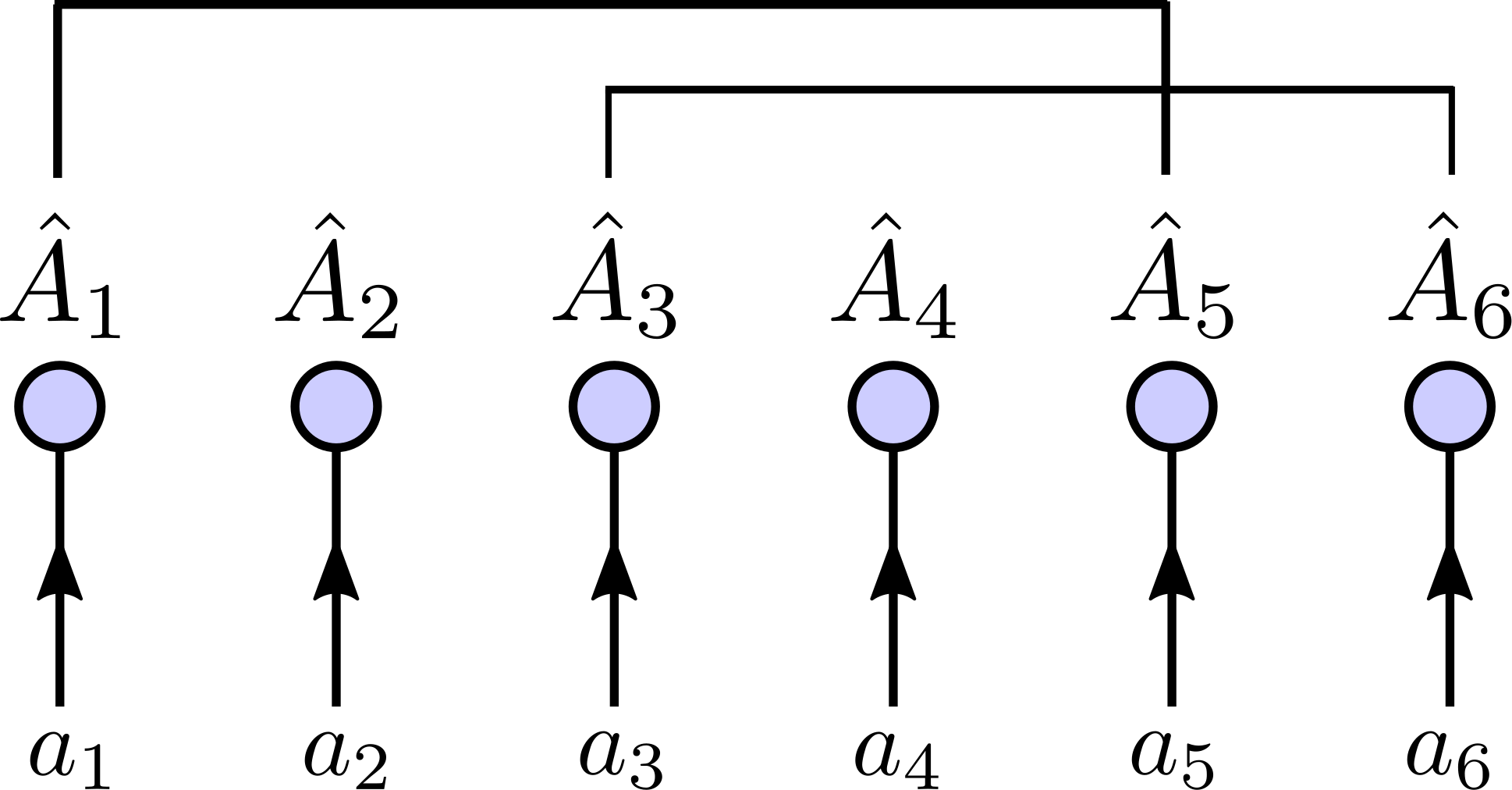}\nonumber\\
	&=&
    \adjincludegraphics[width=24ex,valign=c]{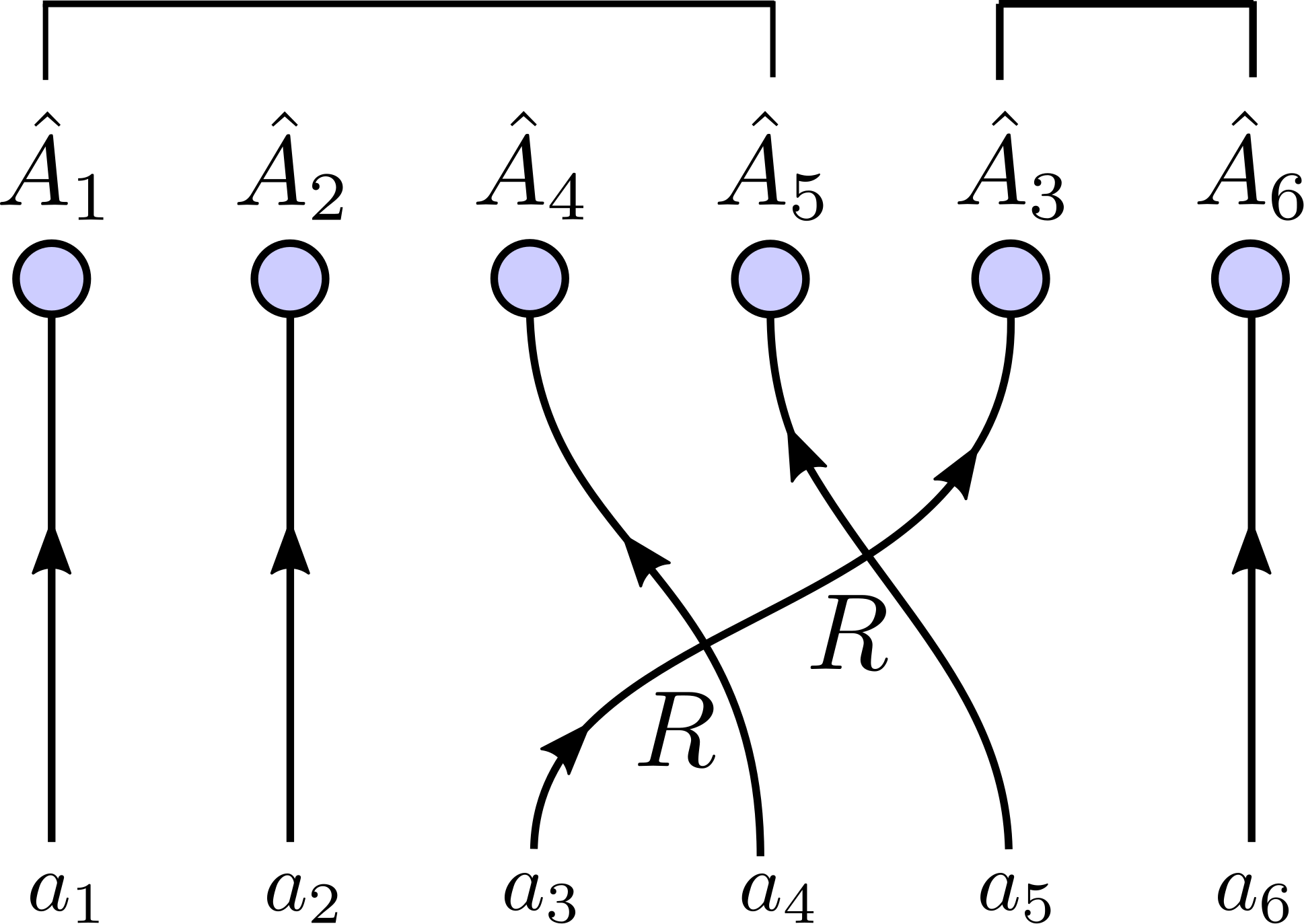}\nonumber\\
	&=&
    \adjincludegraphics[width=24ex,valign=c]{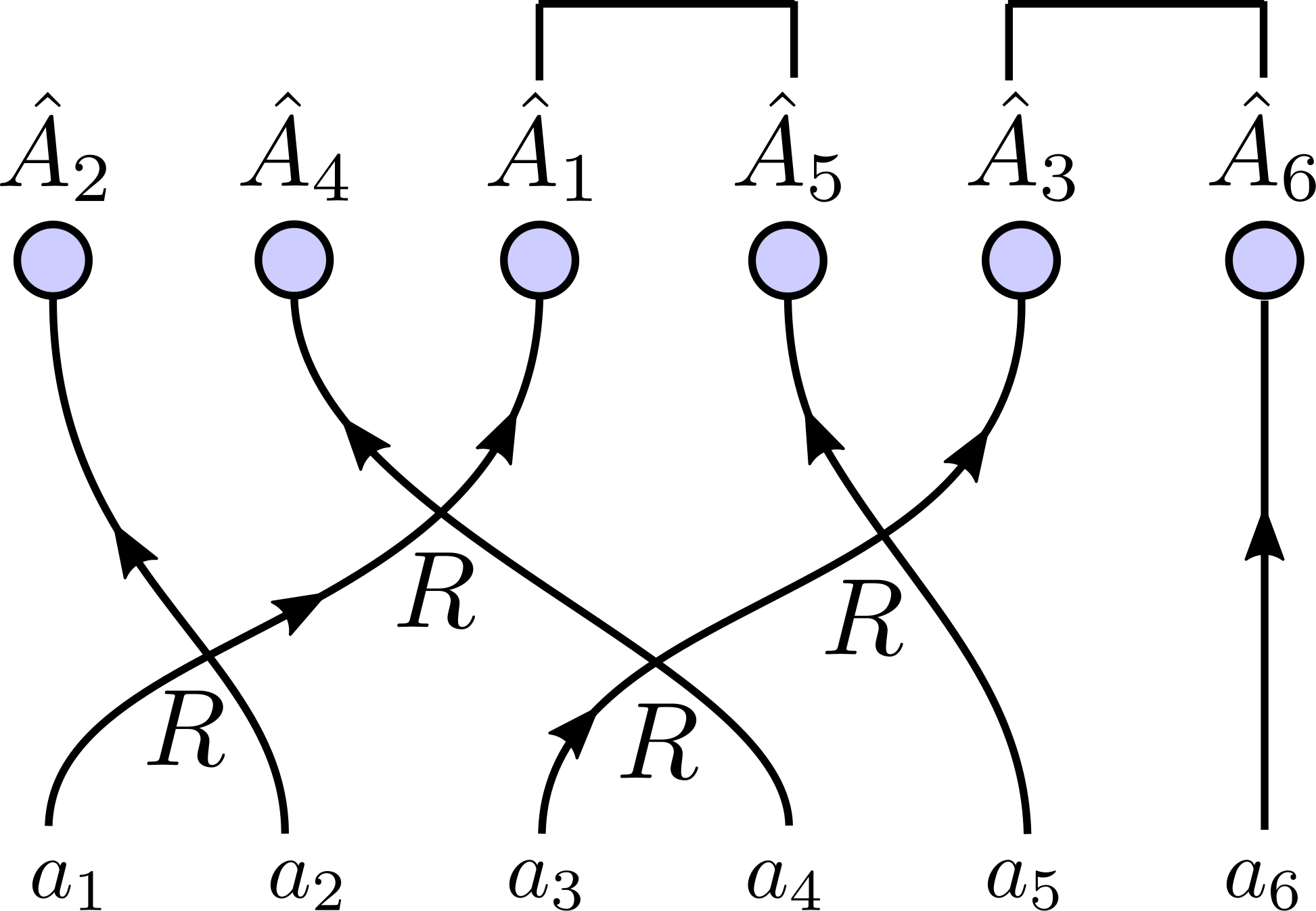}.
\end{eqnarray}

With these preparations, it is straightforward to show that all the derivations and results in Ref.~\cite{molinari2017notes} generalize in a simple way to paraparticles--indeed, most of the results look identical after we suppress the paraparticle indices. For example, we have
\begin{theorem}{(Static Wick's Theorem)}
	For any set of paraparticle operators $\hat{A}_1, \hat{A}_2, \cdots , \hat{A}_n$, where each $\hat{A}_{j,a}$ has the form in Eq.~\eqref{eq:def:A_aWick}, we have
	\begin{align}
		& \hat{A}_1 \hat{A}_2 \cdots \hat{A}_n = {\sf N}[\hat{A}_1 \cdots \hat{A}_n]  \label{WT} \\
		&+\sum_{(ij)} {\sf N}[\hat{A}_1\cdots \contraction{}{\hat{A}_i}{ \cdots }{\hat{A}_j} \hat{A}_i \cdots {A}_j
		\cdots  \hat{A}_n]\nonumber\\
		& +\sum_{(ij)(kl)}  {\sf N}[\hat{A}_1\cdots 
		\contraction{}{\hat{A}}{_i\cdots \hat{A}_k \cdots}{\hat{A}_j}
		\contraction[1.5ex]{\hat{A}_i\cdots}{\hat{A}}{_k\cdots \hat{A}_j\cdots }{\hat{A}_l}
		\hat{A}_i\cdots \hat{A}_k \cdots \hat{A}_j \cdots \hat{A}_l \cdots \hat{A}_n]  \nonumber \\
		& + \ldots             \nonumber
	\end{align}
	where $\ldots$ refers to all higher order contractions.
\end{theorem}

We now move on to  computations involving quantum dynamics. Let $\hat{H}$ be any free paraparticle Hamiltonian of the generic form in Eq.~\eqref{eq:u1_breaking_H}, and let $\hat{A}_a(t)=e^{-i\hat{H}t}\hat{A}_a e^{i\hat{H}t}$ denote the time-evolved operator in the Heisenberg picture.  Since the Heisenberg equations for the $\{A_{ia}\}$ are  a linear set of coupled differential equations, the
$\hat{A}_a(t)$ still have the linear form in Eq.~\eqref{eq:def:A_aWick}, where the coefficients $f_i$ and $g_i$ now depend on time. In time-dependent perturbation theory, we invariably encounter time-ordered  products of field operators, 
so we define the time ordering symbol for paraparticles in the following way:
\begin{definition}{(Time ordering)}
	The time ordering symbol acting on the product of two time-evolved paraparticle operators is defined as 
	\begin{eqnarray}
		{\sf T}[\hat{A}_1(t_1)\hat{A}_2(t_2)]&=&\theta(t_1-t_2)\hat{A}_1(t_1)\hat{A}_2(t_2)\\
		&&{}+\theta(t_2-t_1)\hat{A}_2(t_2)\hat{A}_1(t_1)\cdot R,\nonumber
	\end{eqnarray}
	where $\theta(t_1-t_2)$ is the Heaviside step function,  and the definition is extended to arbitrary product of $n$ operators in similar way as Eq.~\eqref{normalordering}. ${\sf T}$ satisfies a similar symmetry property as that of ${\sf N}$ in Eq.~\eqref{eq:NOsymmeetry_property}.
\end{definition}
A direct computation using Eq.~\eqref{eq:NOcontraction} together with  the symmetry condition analogous to Eq.~\eqref{eq:NOsymmeetry_property} gives
\begin{align}\label{eq:TOcontraction}
	&{\sf T}\hat{A}_1(t_1)\hat{A}_2(t_2)%
	= {\sf N}[\hat{A}_1(t_1)\hat{A}_2(t_2)]+\overbrace{\hat{A}_1(t_1)\hat{A}_2(t_2)},
\end{align}
where the second term is a c-number~(called the time-ordered contraction), and is given by
\begin{eqnarray}
	\overbrace{\hat{A}_1(t_1)\hat{A}_2(t_2)} = \langle 0|{\sf T} \hat{A}_1(t_1)\hat{A}_2(t_2)|0\rangle. 
\end{eqnarray}
We now state Wick's theorem for time-ordered products of paraparticle operators, whose proof is almost identical to the case for ordinary particles after suppressing paraparticle indices:
\begin{theorem}{(Wick's theorem for time-ordered products)}
	For any set of time-evolved paraparticle operators $\hat{A}_1(t_1), \hat{A}_2(t_2), \cdots , \hat{A}_n(t_n)$,
	we have
	\begin{align}
		& {\sf T}[\hat{A}_1(t_1) \cdots \hat{A}_n (t_n)]= {\sf N}[\hat{A}_1(t_1) \cdots \hat{A}_n(t_n)]  \label{WTTO} \\
		&+\sum_{(ij)} {\sf N}[ \hat{A}_1(t_1)\cdots \overbrace{\hat{A}_i(t_i) \cdots \hat{A}_j(t_j)} \cdots  \hat{A}_n(t_n)]\nonumber\\
		& +\sum  {\sf N}[\;\cdots \text{\rm double {\sf T}-contractions} \cdots \;]  + \ldots\nonumber
	\end{align} 
	where $\ldots$ refers to all higher order contractions.
\end{theorem}

\section{Generalized hidden symmetries}\label{sec:generalized_symmetry}
In this section we investigate hidden symmetries of $R$-paraparticles, which are symmetries acting on the internal space of $R$-paraparticles that leave all physical observables invariant. This contrasts with ordinary symmetries, which leave invariant only particular classes of observables and Hamiltonians.  %
We will find that the $R$-matrix second quantization theory of paraparticles enjoys generalized hidden symmetries described by Hopf algebras, which are generalizations of traditional group-like symmetries~\cite{alvarez-gaumeDualityQuantumGroups1990,pasquierCommonStructuresFinite1990,Dijkgraaf1991QuasiHA,mackQuasiHopfQuantum1992,alexanderbaisQuantumSymmetriesDiscrete1992,frohlichQuantumGroupsQuantum1993}~(see also the recent literature in categorical and non-invertible symmetries~\cite{frohlichDualityDefectsRational2007,gaiottoGeneralizedGlobalSymmetries2015,Kong2020Algebraichigher,Ji2020Categorical,cordova2022snowmass,mcgreevyGeneralizedSymmetriesCondensed2023,schafer-namekiICTPLecturesNoninvertible2024}). %
 Such generalized symmetries are present even in systems of free paraparticles, which are not displayed by free fermions  and bosons. %
Generalized hidden symmetries of $R$-paraparticles provide us an important theoretical tool for studying the local indistinguishability of $R$-paraparticles, and naturally lead to a categorical description of $R$-paraparticles, as we show in later parts of this section.

Let us first warm up by describing hidden group symmetries of paraparticles.  Specifically, let $G$ be an arbitrary group.  
By a hidden symmetry, we mean that the Fock space carries a representation of $G$~[denoted by $\hat{\Theta}(g)$ for $g\in G$] that leaves invariant all physical observables and the fundamental CRs in Eq.~\eqref{eq:fundamental_Rcommu}. 
The invariance of all physical observables--that is, $[\hat{\Theta}(g),\hat{e}_{ij}]=0$ for all $g\in G$ and $\hat{e}_{ij}$ defined in Eq.~\eqref{eq:def_e_ab}--requires that the paraparticle operators transform under $G$ as follows~\footnote{Here we only consider $G$-actions that induce linear transformations on the paraparticle operators $\hat{\psi}^\pm_{i,a}$. Under this condition, using the fact that $\hat{e}_{ij}$ transform in the product of the representations of $\hat{\psi}^+_{i,a}$ and $\hat{\psi}^-_{i,a}$, it is not hard to show that the invariance of $\hat{e}_{ij}$  forces $\hat{\psi}^\pm_{i,a}$ to transform in the form of  Eq.~\eqref{eq:G-adjoint-action}.}
\begin{eqnarray}\label{eq:G-adjoint-action}
	\hat{\Theta}(g) \hat{\psi}_{i,a}^{+} \hat{\Theta}(g^{-1})&=&\sum_b \zetapsi_{b a}(g) \hat{\psi}_{i,b}^{+}, \nonumber\\
	\hat{\Theta}(g) \hat{\psi}_{i,a}^{-} \hat{\Theta}(g^{-1})&=&\sum_b  \bar{\zetapsi}_{b a}(g) \hat{\psi}_{i,b}^{-},
\end{eqnarray}
where $\zetapsi$ is a representation of $G$, and $\bar{\zetapsi}$ is its conjugate representation defined by $\bar{\zetapsi}(g)=\zetapsi(g^{-1})^T$ for all $g\in G$. 
Note that we use the same Greek letter $\psi$ to label the $G$-representation $\psi_{ba}(g)$ and the creation and annihilation operators $\hat{\psi}_{i,a}^{\pm}$ associated to the particle type $\psi$, as this will turn out convenient later on when we study hidden symmetry of a system of paraparticles involving different particle types transforming under different $G$-representations, and this abuse of notation will not cause any confusion. 
The LHS of Eq.~\eqref{eq:G-adjoint-action} defines a superoperator
\begin{equation}\label{def:adj-action-G}
	\mathrm{ad}_g(\hat{O})\equiv \hat{\Theta}(g)\hat{O}\hat{\Theta}(g^{-1}),
\end{equation} 
called the adjoint transformation of $g$ on the space of operators. To ensure that Eq.~\eqref{eq:G-adjoint-action} consistently defines a symmetry of the paraparticle system at the quantum level, we require that 
the linear transformation of $\hat{\psi}^\pm_{i,a}$ induced by $\mathrm{ad}_g$
leaves the fundamental CRs~\eqref{eq:fundamental_Rcommu} invariant for all $g\in G$, which puts the following consistency condition on  the $R$-matrix~\footnote{More formally, the second quantization algebra is of the form $\mathfrak{F}=\C[\{\hat{\psi}^\pm_{i,a}\}]/I$, where $\C[\{\hat{\psi}^\pm_{i,a}\}]$ is the free associative algebra with formal generators $\{\hat{\psi}^\pm_{i,a}\}$ and $I$ is the ideal generated by the fundamental CRs~\eqref{eq:fundamental_Rcommu}. 
The linear map $\varphi_g(\hat{\psi}^\pm_{i,a})$ defined by the RHS of Eq.~\eqref{eq:G-adjoint-action} uniquely extends to an algebra automorphism $\tilde{\varphi}_g$ of $\mathfrak{F}$ if and only if it is invertible and preserves the ideal $I$~(a basic fact in abstract algebra), and in the current situation the last condition is equivalent to Eq.~\eqref{eq:G-sym-R-condition}. Once we have an automorphism $\tilde{\varphi}_g$, we can use it to define the action of $\mathrm{ad}_g$ and then define the action of $\hat{\Theta}(g)$ on the Fock space. If by contrast, Eq.~\eqref{eq:G-sym-R-condition} is not satisfied, then there will not exist any $\hat{\Theta}(g)$ satisfying Eq.~\eqref{eq:G-adjoint-action}. 
}:
\begin{equation}\label{eq:G-sym-R-condition}
	R[\zetapsi(g)\otimes \zetapsi(g)]=[\zetapsi(g)\otimes \zetapsi(g)]R,\quad \forall g\in G.
\end{equation}
A graphical derivation of this in the more general setting of Hopf algebra symmetries is given in Eq.~\eqref{eq:Rpp-condition}.
This condition is trivial in the case of ordinary fermions and bosons, and also Ex.~\ref{ex:1m} of Tab.~\ref{tab:Hilbert_series}, but it becomes nontrivial for more general $R$-matrices. %

Finally, we also need to specify how the symmetry operation  $\hat{\Theta}(g)$ transforms the vacuum state $\ket{0}$. For simplicity, in this section we assume the vacuum state $\ket{0}$ to be unique. Then it follows that $\ket{0}$ must be invariant under  $G$~\footnote{More generally, $\ket{0}$ can transform under any 1-dimensional representation of $G$: $\hat{\Theta}(g)\ket{0}=\theta(g)\ket{0},~\forall g\in G$. However, we can always absorb $\theta(g)$ into a redefinition of $\hat{\Theta}(g)$ such that $\ket{0}$ is invariant under the redefined $\hat{\Theta}(g)$, which is also a representation of $G$ satisfying Eq.~\eqref{eq:G-adjoint-action} with $\zeta$ unchanged. Therefore we can always assume without loss of generality that $\theta(g)\equiv 1$. This argument generalizes to the Hopf algebra case as well. }
\begin{equation}\label{eq:G-invariant_vac}
	\hat{\Theta}(g)\ket{0}=\ket{0},\quad\forall g\in G. 
\end{equation}
Eqs.~\eqref{eq:G-adjoint-action} and \eqref{eq:G-invariant_vac} completely  define the action of $\hat{\Theta}(g)$ on the entire Fock space. 

In the following we generalize the above formalism to Hopf algebras, describing a more general class of hidden symmetries for $R$-paraparticles. We first present the formal algebraic description in Sec.~\ref{sec:HAsymmetry}, and then present a tensor network representation in Sec.~\ref{sec:HAsymmetry-TN}, and give some basic examples in Sec.~\ref{sec:HAsymexample}. Then in Sec.~\ref{sec:statistics-as-symmetry} we define a fundamental Hopf algebra hidden symmetry for each $R$-paraparticle system, which provides an alternative viewpoint for mutual parastatistics and the local observable algebra we studied back in Sec.~\ref{sec:second_quantization}. Finally in Secs.~\ref{sec:local-indist-revisit} and \ref{sec:cat-description} we use this hidden symmetry viewpoint to study local indistinguishability and formulate a categorical description of $R$-paraparticles. %

\subsection{Hopf algebra description}\label{sec:HAsymmetry} %
For basic definitions of Hopf algebras, we refer to the standard textbooks~\cite{Majid1995BookFoundationQG,klimyk1997book,kasselQuantumGroups1995}. For a quick introduction to the finite dimensional case and the tensor network representation, we refer to Ref.~\cite{molnar2022matrix} and the appendix of Ref.~\cite{wang2024hopf}. Here for the readers' convenience, we summarize in Tab.~\ref{tab:Group-HA-correspondence} the relevant concepts in Hopf algebras as generalizations of the more familiar concepts in group theory. 
\begin{table*}%
	\centering
	{\renewcommand{\arraystretch}{1.5}
		\begin{tabular}{|c|c|c|}
			\hline
			& Groups	& 	Hopf algebras   \\
			\hline
			Inverse & Inverse axiom & Antipode axiom \\
			& $g^{-1}g=gg^{-1}=1$ &$Sz_{(1)}z_{(2)}=z_{(1)}Sz_{(2)}=\epsilon(z)1$\\
			\hline
			Adjoint action & $\mathrm{ad}_g(\hat{O}):= \hat{\Theta}(g)\hat{O}\hat{\Theta}(g^{-1})$ &  $\mathrm{ad}_z(\hat{O}):= \hat{\Theta}(z_{(1)})\hat{O}\hat{\Theta}(Sz_{(2)})$\\
			& $\mathrm{ad}_g\circ \mathrm{ad}_h=\mathrm{ad}_{gh}$  & $\mathrm{ad}_x\circ \mathrm{ad}_y=\mathrm{ad}_{xy}$\\
			& $\mathrm{ad}_g(\hat{A}\hat{B})=\mathrm{ad}_g(\hat{A})\mathrm{ad}_g(\hat{B})$ & $\mathrm{ad}_z(\hat{A}\hat{B})=\mathrm{ad}_{z_{(1)}}(\hat{A})\mathrm{ad}_{z_{(2)}}(\hat{B})$\\ %
			\hline
			$\hat{O}$ is invariant under $\hat{\Theta}$ & $[\hat{\Theta}(g),\hat{O}]=0\Leftrightarrow \mathrm{ad}_g(\hat{O})=\hat{O}$ & $[\hat{\Theta}(z),\hat{O}]=0\Leftrightarrow\mathrm{ad}_z(\hat{O})=\epsilon(z)\hat{O}$\\
			\hline
			Trivial representation & $\epsilon(g)= 1,~\forall g$ & Counit $\epsilon(z)$\\ %
			\hline
			Dual representation & $\bar{\zeta}(g):=\zeta(g^{-1})^T$ & $\bar{\zeta}(z):=\zeta(Sz)^T$\\
			\hline
			Tensor product representation & $(\zeta\otimes\sigma)(g):=\zeta(g)\otimes\sigma(g)$ & $(\zeta\otimes\sigma)(z):=\zeta(z_{(1)})\otimes\sigma(z_{(2)})$ \\
			\hline
		\end{tabular}
	}
	\caption{\label{tab:Group-HA-correspondence}  Correspondence between various concepts in groups and their generalizations in Hopf algebras. Here all the equalities are assumed to hold for arbitrary  $g,h\in G$ and $x,y,z\in \calA$, and for all quantum operators $\hat{A},\hat{B},\hat{O}$. Note that to avoid clustering of parentheses in writing the action of antipode $S$,  in this paper we write $Sz$ to mean $S(z)$, and $Sxy$ is understood as $S(x)y$, i.e., $S$ is applied on the element appearing right after it.}
\end{table*}

Let $\calA$ be an arbitrary Hopf algebra describing the hidden symmetry of a paraparticle system. Similar to the group case, we assume that
the Fock space carries a representation of $\calA$~[denoted by $\hat{\Theta}(z)$ for $z\in \calA$] that leaves invariant all physical observables and the fundamental CRs in Eq.~\eqref{eq:fundamental_Rcommu}.
We achieve this by directly generalizing Eq.~\eqref{eq:G-adjoint-action} as follows
\begin{eqnarray}\label{eq:HA-adjoint-action}
	\hat{\Theta}[z_{(1)}] \hat{\psi}_{i,a}^{+} \hat{\Theta}[S z_{(2)}]&=&\sum_b \zetapsi_{ba}(z)\hat{\psi}^{+}_{i,b},\nonumber\\
	\hat{\Theta}[z_{(1)}] \hat{\psi}_{i,a}^{-} \hat{\Theta}[S z_{(2)}]&=& \sum_b\bar{\zetapsi}_{ba}(z) \hat{\psi}^{-}_{i,b},
\end{eqnarray}
where $\zetapsi$ is a representation of $\calA$ and $\bar{\zetapsi}$ is its dual representation defined as $\bar{\zetapsi}(z):=\zetapsi(Sz)^T$ for all $z\in \calA$, and we use Sweedler notation $\Delta(z)=z_{(1)}\otimes z_{(2)}$ for the comultiplication of $\calA$. 
The adjoint action in Eq.~\eqref{def:adj-action-G} is generalized as 
\begin{equation}\label{def:adj-action-HA}
	\mathrm{ad}_z(\hat{O}):= \hat{\Theta}(z_{(1)})\hat{O}\hat{\Theta}(Sz_{(2)}),
\end{equation} 
whose basic properties are summarized in Tab.~\ref{tab:Group-HA-correspondence}. 
Using the properties of $\mathrm{ad}_z$ along with Eq.~\eqref{eq:HA-adjoint-action}, it is straightforward to  show that $\hat{\Theta}(z)$ leaves all the physical observables invariant, i.e., 
\begin{equation}\label{eq:HA-eij-commute}
	[\hat{\Theta}(z),\hat{e}_{ij}]=0.
\end{equation}
Similarly, the U$(1)$-breaking terms $\hat{e}^\pm_{ij}$ in Eq.~\eqref{def:e_pm_ab} also commute with $\hat{\Theta}(z)$ if $\alpha$ and $\alpha'$ are invariant under the action of $\calA$ %
\begin{eqnarray}\label{eq:HAsymmetry_alpha-cond}
	[\bar{\zetapsi}(z_{(1)})\otimes\bar{\zetapsi}(z_{(2)})]  \alpha^T&=&\epsilon(z)\alpha^T,\nonumber\\
	{}[\zetapsi(z_{(1)})\otimes\zetapsi(z_{(2)})]  \alpha'&=&\epsilon(z)\alpha'.
\end{eqnarray}
Note that the two lines of Eq.~\eqref{eq:HAsymmetry_alpha-cond} are actually equivalent due to $\alpha'=\alpha^{-1}$. 

To ensure that Eq.~\eqref{eq:HA-adjoint-action} consistently defines a symmetry of the paraparticle system at the quantum level, we require that 
both sides of the fundamental CRs~\eqref{eq:fundamental_Rcommu} transform in the same way under the adjoint transformation $\mathrm{ad}_z$ for all $z\in \calA$~\footnote{More formally, we say that the paraparticle operator algebra defined by the fundamental CRs~\eqref{eq:fundamental_Rcommu} has the structure of an $\calA$-module algebra, where the module action is given by the adjoint transformation $\mathrm{ad}_z$.}, which puts the following condition on  the $R$-matrix:
\begin{equation}\label{eq:HA-sym-R-condition}
	R[\zetapsi(z_{(1)})\otimes \zetapsi(z_{(2)})]=[\zetapsi(z_{(1)})\otimes \zetapsi(z_{(2)})]R,
\end{equation}
which generalizes Eq.~\eqref{eq:G-sym-R-condition}. See Eqs.~(\ref{eq:HA-Rpp},\ref{eq:HA-Rpp-R}) for a graphical derivation of Eq.~\eqref{eq:HA-sym-R-condition}.

Finally, since the counit $\epsilon$ plays the role of the trivial representation of $\calA$, %
Eq.~\eqref{eq:G-invariant_vac} is generalized as
\begin{equation}\label{eq:HAaction-vac}
	\hat{\Theta}(z)\ket{0}=\epsilon(z)\ket{0},
\end{equation}
where $\epsilon$ is the counit of $\calA$. 
Eqs.~\eqref{eq:HA-adjoint-action} and (\ref{eq:HAaction-vac}) completely determine the action of $\calA$ on the Fock space
\begin{widetext}
\begin{align}\label{eq:HAaction-Fock}
    \hspace{-0.3in}\hat{\Theta}(z) \hat{\psi}_{i_1,a_1}^{+} \hat{\psi}_{i_2,a_2}^{+} \cdots \hat{\psi}_{i_n,a_n}^{+} \ket{0}  
	&= \hat{\Theta}(z_{(1)} \epsilon(z_{(2)})) \hat{\psi}_{i_1,a_1}^{+} \hat{\psi}_{i_2,a_2}^{+} \cdots \hat{\psi}_{i_n,a_n}^{+} \ket{0} \hspace{0.45in} \text{co-unit axiom} \\
    &= \hat{\Theta}(z_{(1)}) \hat{\psi}_{i_1,a_1}^{+} \hat{\psi}_{i_2,a_2}^{+} \cdots \hat{\psi}_{i_n,a_n}^{+} \epsilon(Sz_{(2)})\ket{0} \hspace{0.35in} \text{${\hat\Theta}$ is linear, $\epsilon(Sz) = \epsilon(z)$} \nonumber\\
    &= \hat{\Theta}(z_{(1)}) \hat{\psi}_{i_1,a_1}^{+} \hat{\psi}_{i_2,a_2}^{+} \cdots \hat{\psi}_{i_n,a_n}^{+} \hat{\Theta}(Sz_{(2)})\ket{0} \hspace{0.3in} \text{$\hat\Theta$ action on trivial rep.}  \nonumber\\
    &= \mathrm{ad}_z (\hat{\psi}_{i_1,a_1}^{+} \hat{\psi}_{i_2,a_2}^{+} \cdots \hat{\psi}_{i_n,a_n}^{+})\ket{0}\nonumber\\
	&= \mathrm{ad}_{z_{(1)}} (\hat{\psi}_{i_1,a_1}^{+}) \mathrm{ad}_{z_{(2)}}(\hat{\psi}_{i_2,a_2}^{+}) \cdots \mathrm{ad}_{z_{(n)}}(\hat{\psi}_{i_n,a_n}^{+})\ket{0}\hspace{0.3in} \text{Property of }\mathrm{ad}_{z}\nonumber\\
	&= \sum_{b_1,\ldots,b_n}\hat{\psi}_{i_1,b_1}^{+}  \cdots \hat{\psi}_{i_n,b_n}^{+}\ket{0} \zetapsi(z_{(1)})_{b_1 a_1} \cdots  \zetapsi(z_{(n)})_{b_n a_n}, \hspace{0.3in} \text{Eq.~\eqref{eq:HA-adjoint-action} } \nonumber
\end{align}
\end{widetext}
Since the Fock space is spanned by states of the form in Eq.~\eqref{eq:psi_statespace_general}, Eq.~\eqref{eq:HAaction-Fock} uniquely defines the action of $\hat{\Theta}(z)$ on the entire Fock space for all $z\in\calA$.

We finally remark that, while in most part of 
this section,
we focus on Hopf algebra hidden symmetry for simplicity, there are situations where we unavoidably need to work with bialgebra hidden symmetry in which there is no antipode. For example, if the $R$-matrix is not dual unitary~(e.g. Ex.~\ref{ex:1m}), then its fundamental hidden symmetry $\calA_R$~(to be introduced in Sec.~\ref{sec:statistics-as-symmetry}) is only a bialgebra. We will mention later in Remark~\ref{rmk:hiddenBAsymmetry} for how to generalize the above results to the bialgebra case. In general, we recommend avoid working with hidden bialgebra symmetries whenever possible, since the lack of antipode makes computations inconvenient.  

\subsection{Tensor network description}\label{sec:HAsymmetry-TN}
We now introduce an equivalent tensor network description of the Hopf algebra hidden symmetry defined above, which allows us to perform all the derivations in a graphical way. Following Refs.~\cite{molnar2022matrix,wang2024hopf}, let $v_{pq} \in \mathcal A$ be any corepresentation of $\calA$, satisfying
\begin{equation}\label{def:corep}
	\Delta(v_{pq})=\sum^d_{r=1} v_{pr}\otimes v_{rq},\quad \epsilon(v_{pq})=\delta_{pq},
\end{equation}
where $d$ is the dimension of the corepresentation $v_{pq}$. 
We define the following tensors%
\begin{eqnarray}\label{eq:zeta_tensor_def}
	\adjincludegraphics[height=7ex,valign=c]{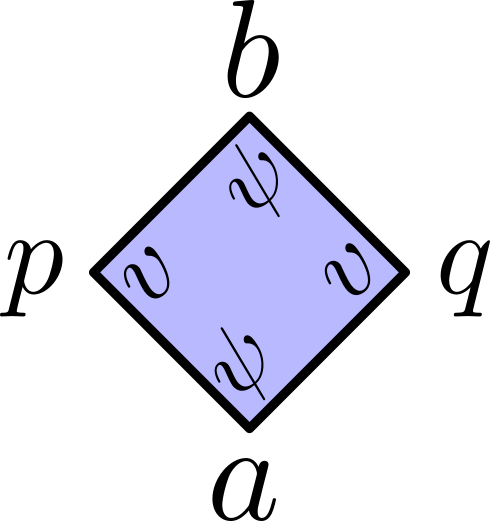}=\zetapsi_{ba}[v_{pq}],
    \quad
	\adjincludegraphics[height=7ex,valign=c]{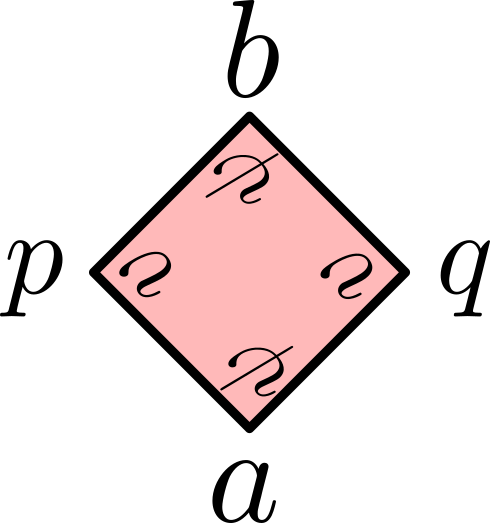}=\zetapsi_{ab}[S v_{pq}].
\end{eqnarray}
We also define the following operator-valued tensors
\begin{eqnarray}\label{def:Thetavpqgraphical}
	\adjincludegraphics[height=7ex,valign=c]{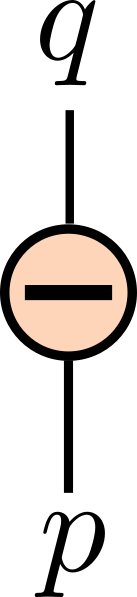}=\hat{\Theta}[v_{pq}],
    \quad
	\adjincludegraphics[height=7ex,valign=c]{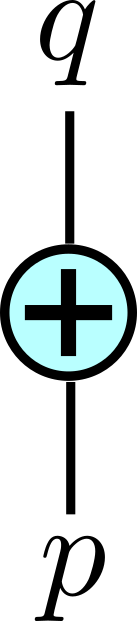}=\hat{\Theta}[S v_{qp}].
\end{eqnarray}
Then by setting $z=v_{pq}$ in the antipode axiom and applying the representations $\zetapsi$ or $\hat{\Theta}$ on both sides, we obtain the following inversion properties of the tensors
\begin{eqnarray}\label{eq:Antipode-graphical}
	&&\adjincludegraphics[height=5ex,valign=c]{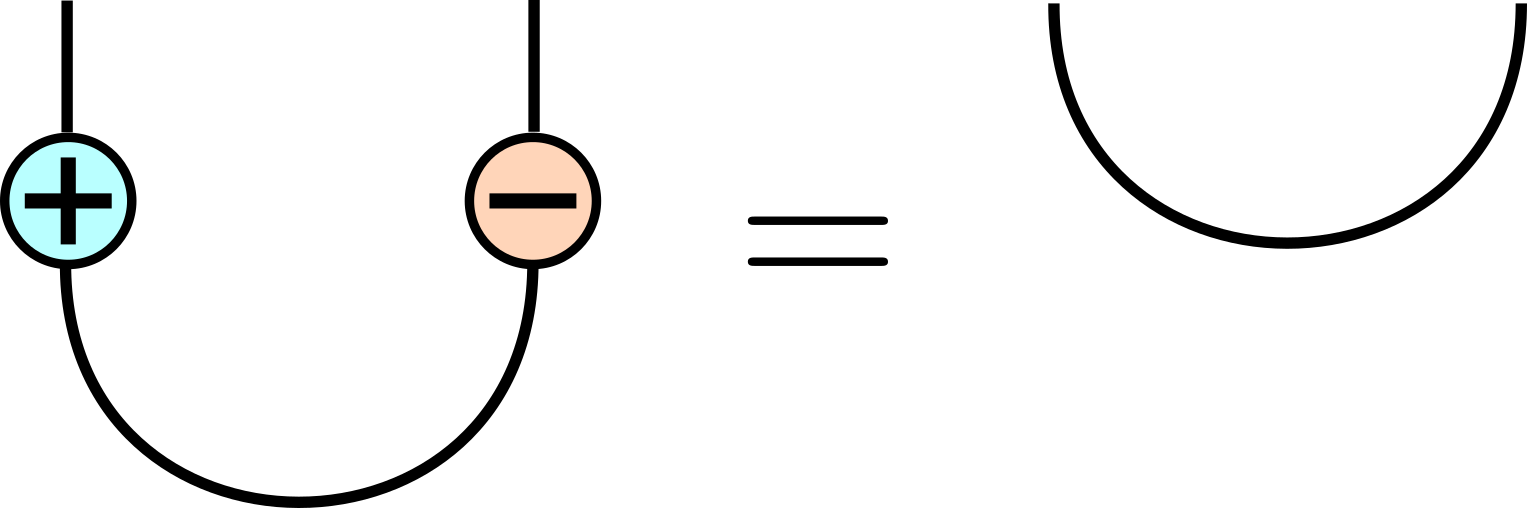},
    \quad
	\adjincludegraphics[height=5ex,valign=c]{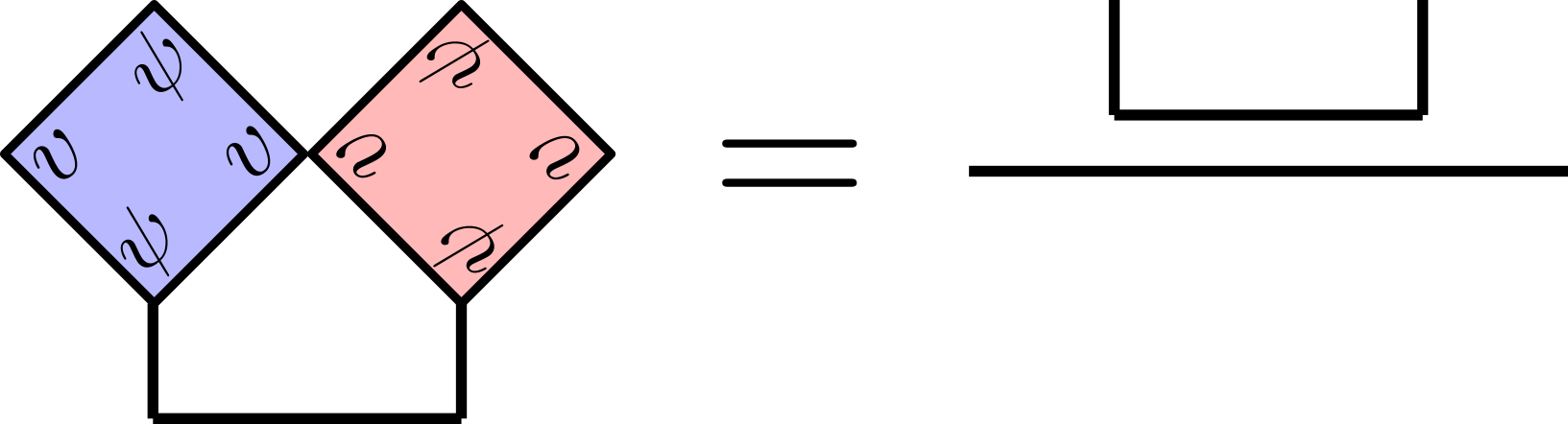},\nonumber\\
	&&
    \adjincludegraphics[height=5ex,valign=c]{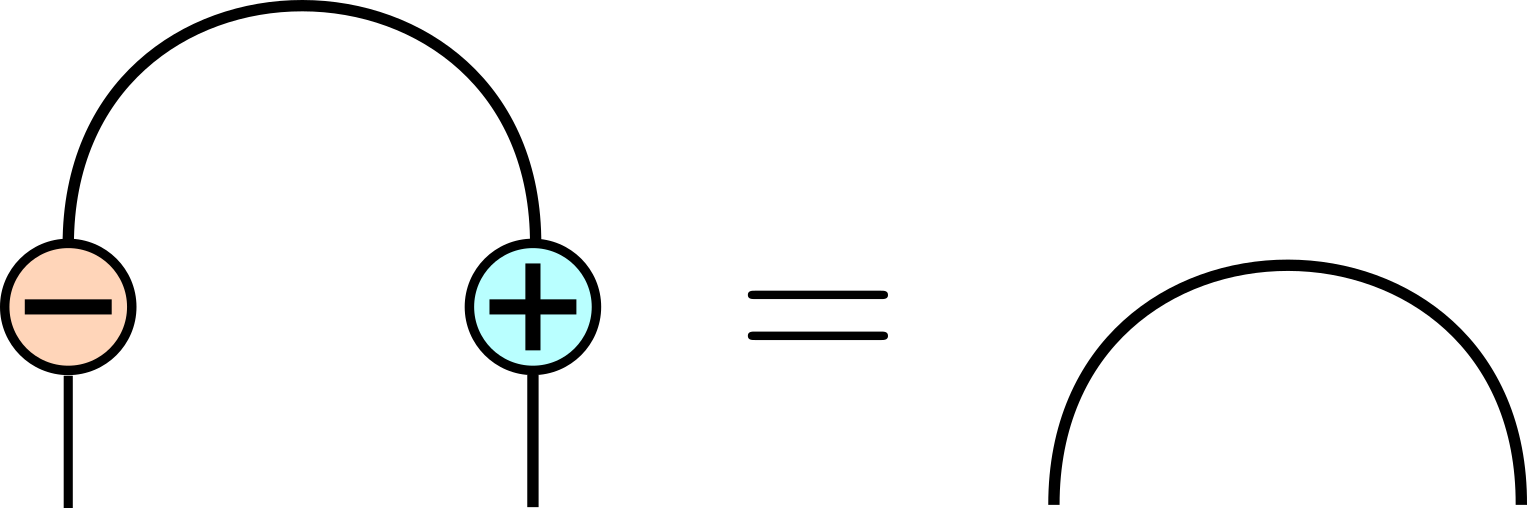},
    \quad
	\adjincludegraphics[height=5ex,valign=c]{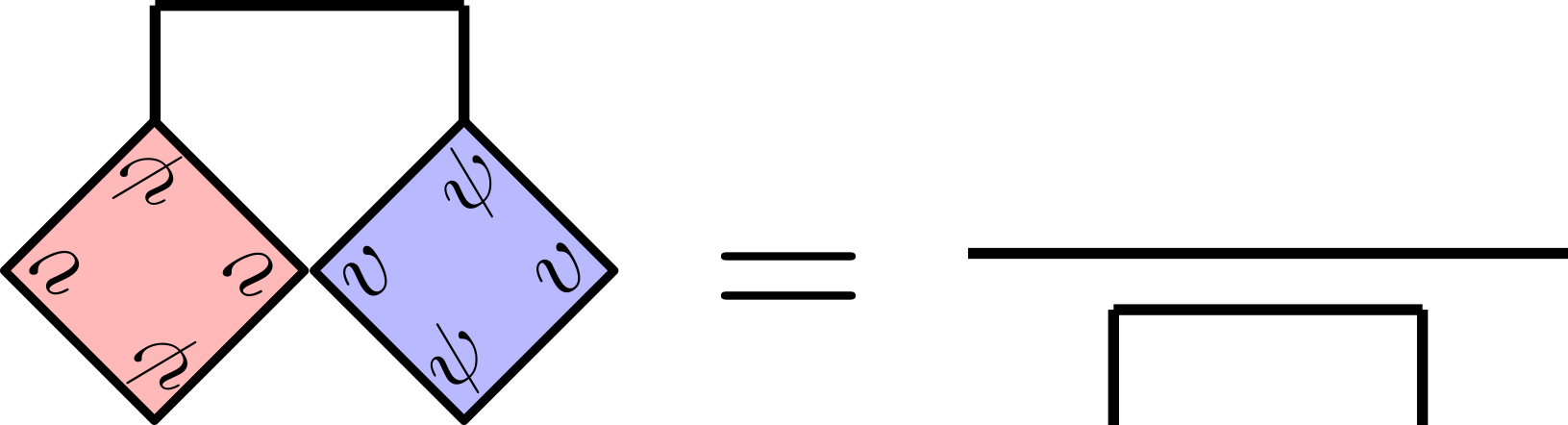}.
\end{eqnarray}
In a similar way, we can translate most of the results in the previous section by setting $z=v_{pq}$ and applying Eq.~\eqref{def:corep} along with other Hopf algebra axioms. For example,   
Eq.~\eqref{eq:HA-adjoint-action} can be rewritten as 
\begin{equation}\label{eq:Adj-action-graphical}
	\adjincludegraphics[height=10ex,valign=c]{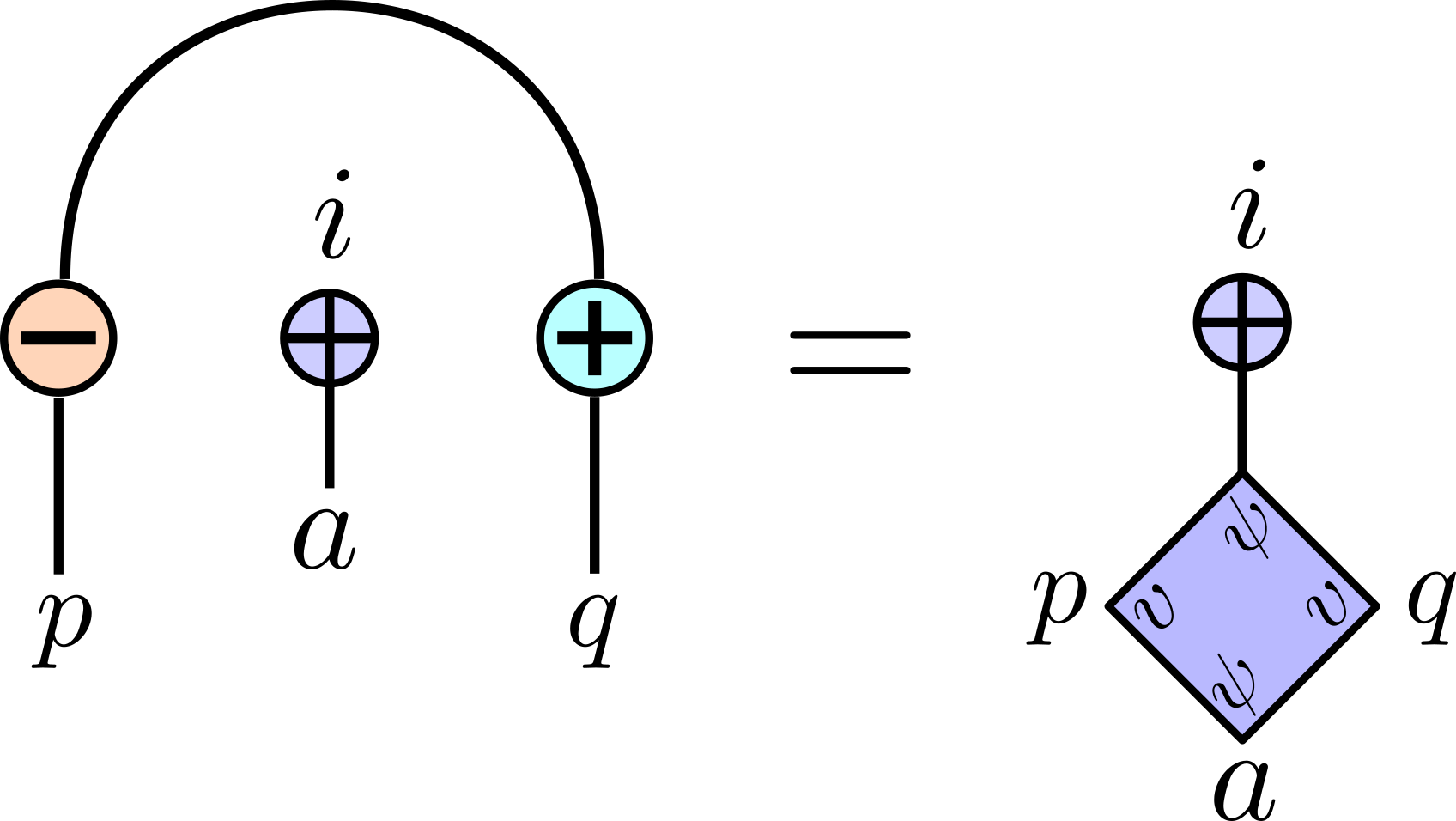},
    \quad  
    \adjincludegraphics[height=10ex,valign=c]{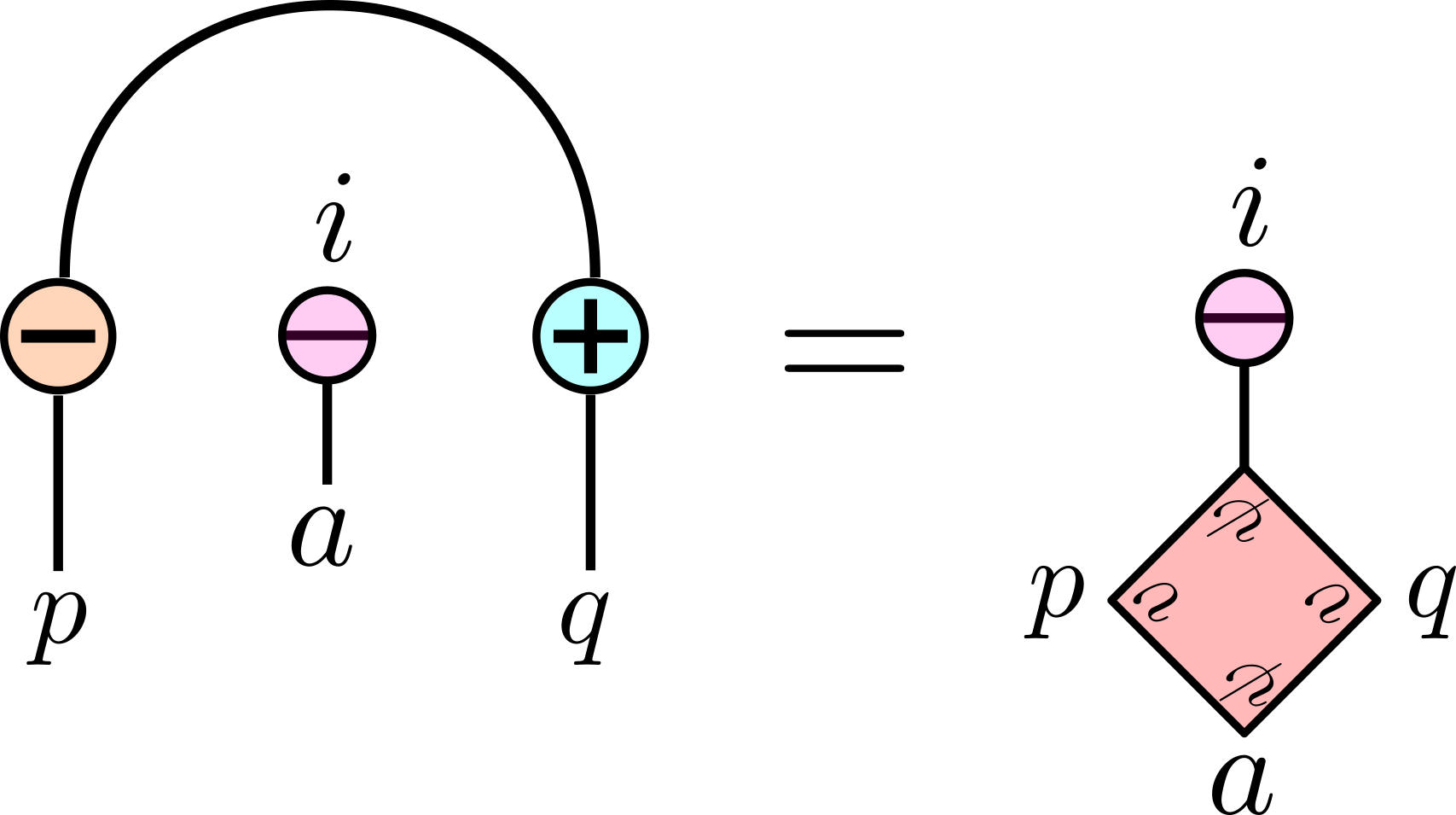}.
\end{equation}
We now show the graphical derivation of Eq.~\eqref{eq:HA-sym-R-condition}, which is the consistency condition obtained by applying $\mathrm{ad}_{v_{pq}}$ on both sides of the fundamental CRs~\eqref{eq:fundamental_Rcommu}. %
Let us begin with the second line of Eq.~\eqref{eq:fundamental_Rcommu}. 
The LHS gives us
\begin{eqnarray}\label{eq:HA-Rpp}
	\adjincludegraphics[height=9ex,valign=c]{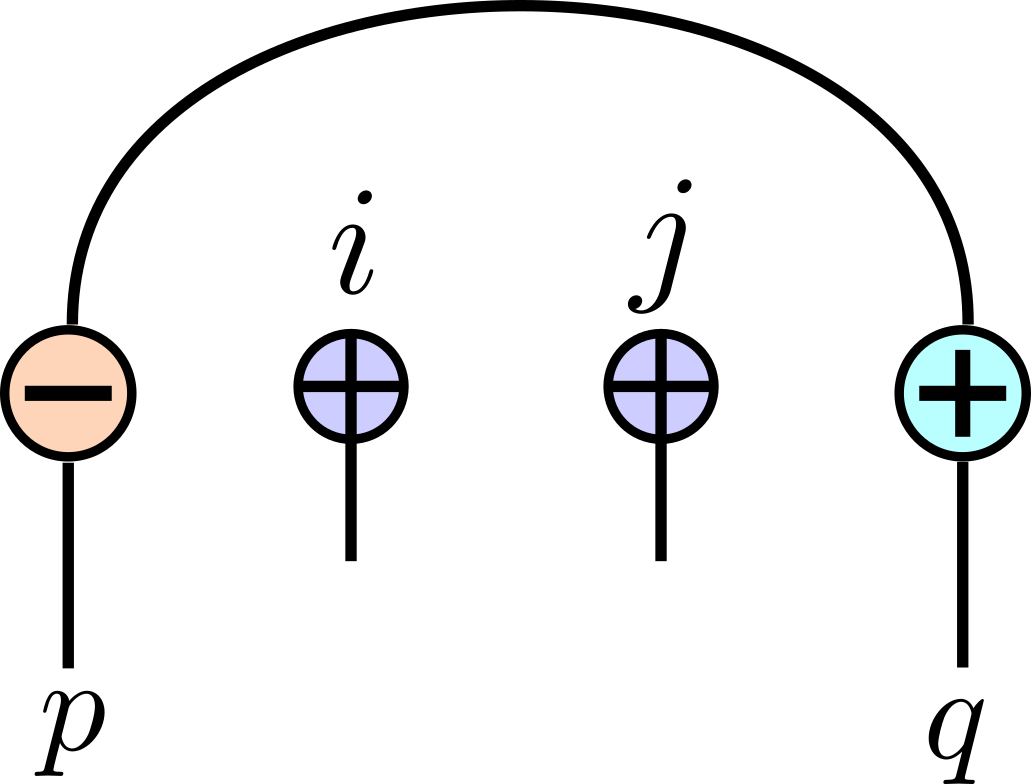}
	&=&
    \adjincludegraphics[height=10ex,valign=c]{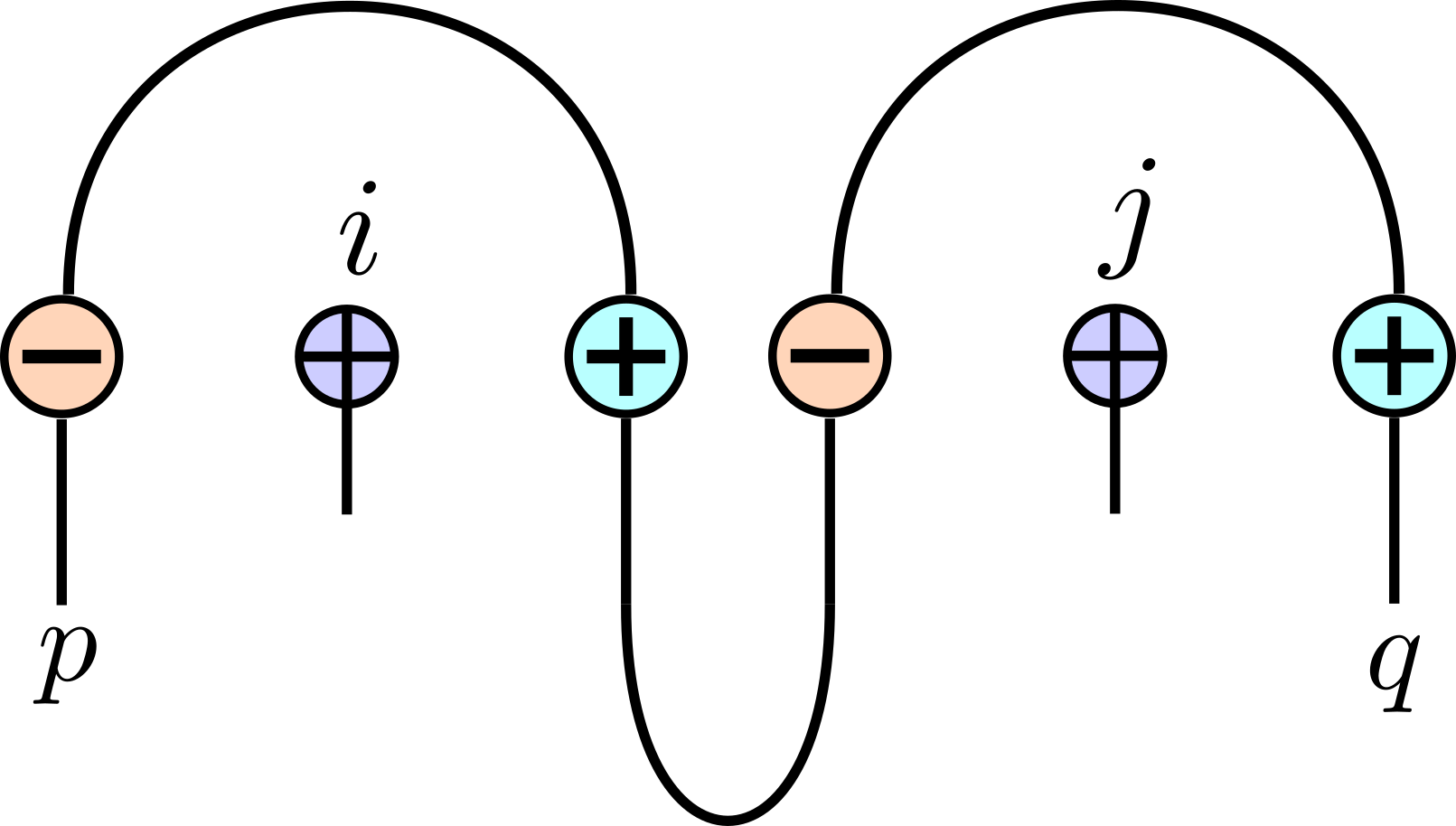}\nonumber\\
	&=&
    \adjincludegraphics[width=9ex,valign=c]{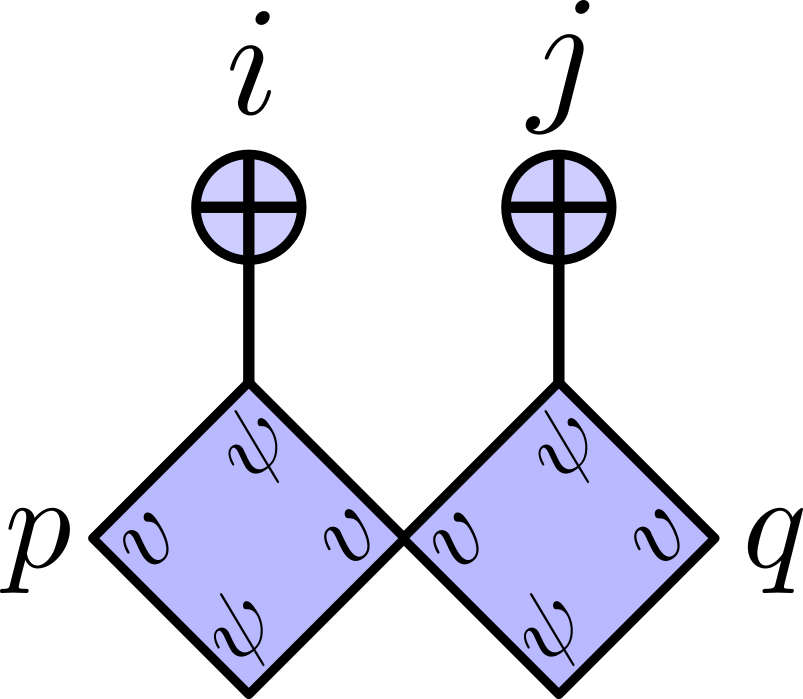}\nonumber\\
	&=&
    \adjincludegraphics[width=9ex,valign=c]{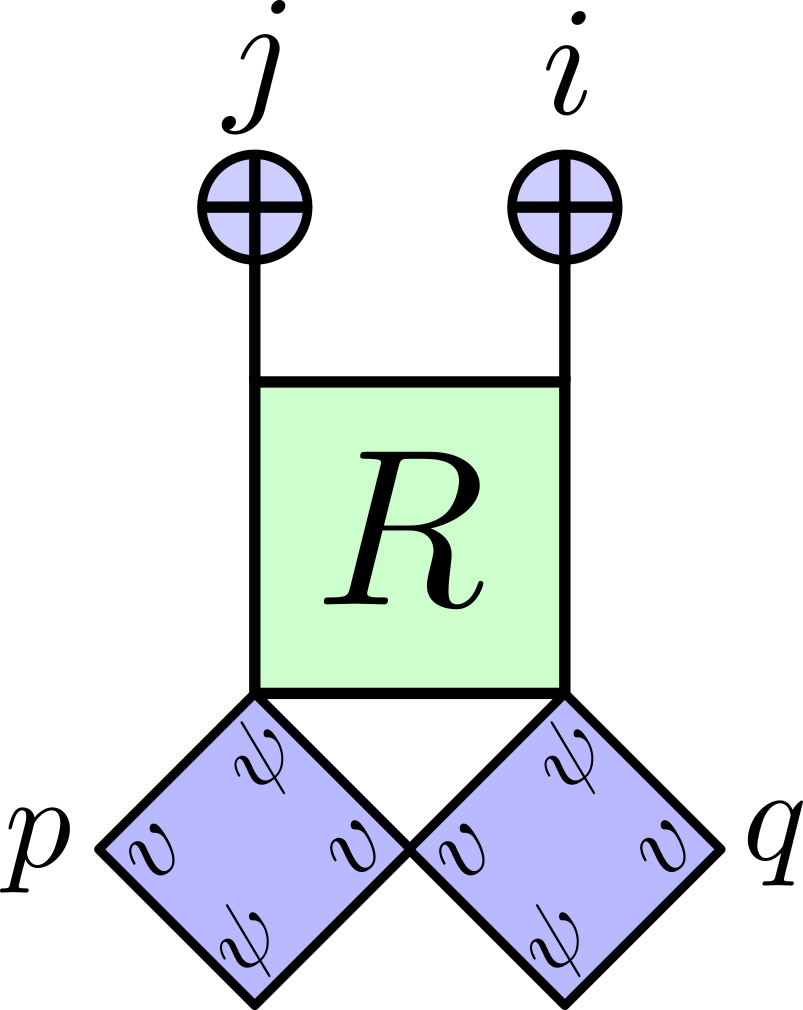}~,
\end{eqnarray}
where we used Eqs.~(\ref{eq:Antipode-graphical},\ref{eq:Adj-action-graphical}), while the RHS gives
\begin{equation}\label{eq:HA-Rpp-R}
	\adjincludegraphics[height=9ex,valign=c]{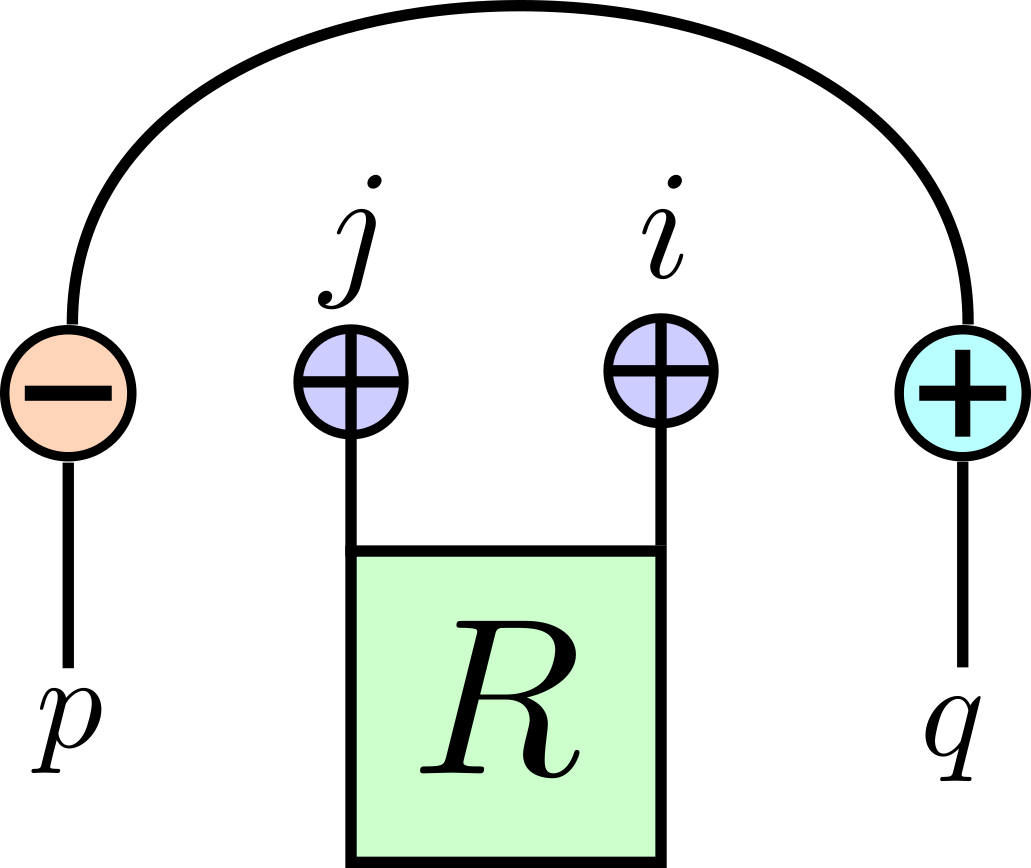}
    =
    \adjincludegraphics[height=10ex,valign=c]{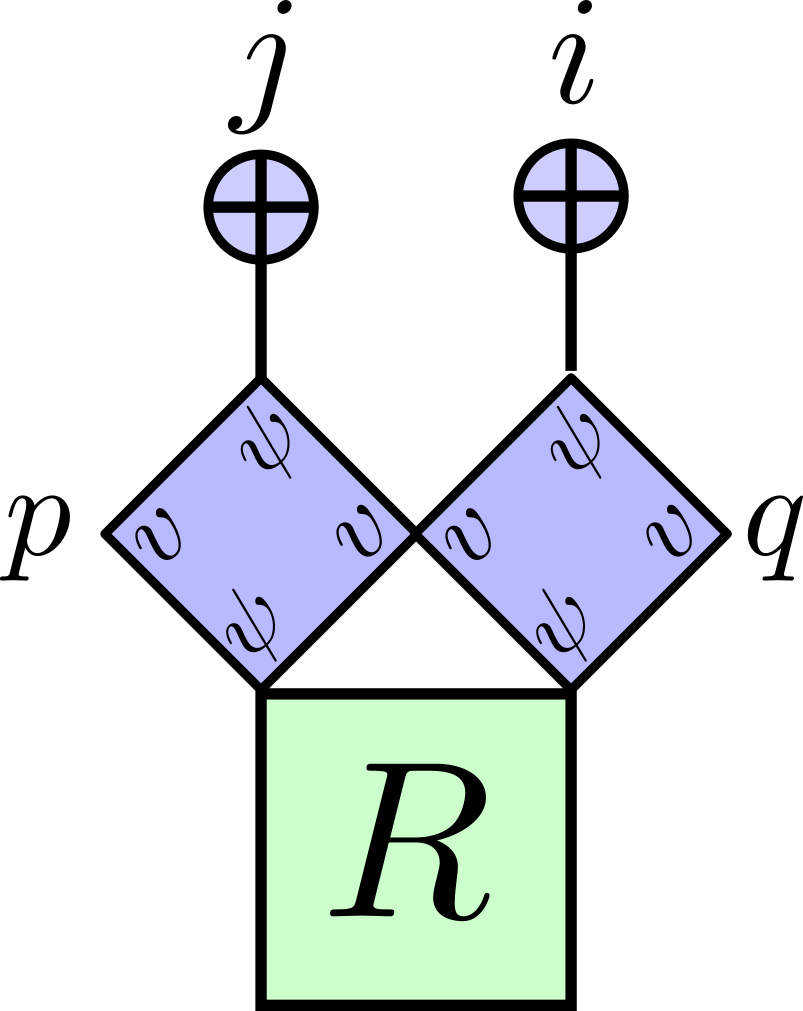}~.
\end{equation}
The consistency between Eq.~\eqref{eq:HA-Rpp} and Eq.~\eqref{eq:HA-Rpp-R} requires the following condition on the $R$-matrix
\begin{equation}\label{eq:Rpp-condition}
	\adjincludegraphics[height=7ex,valign=c]{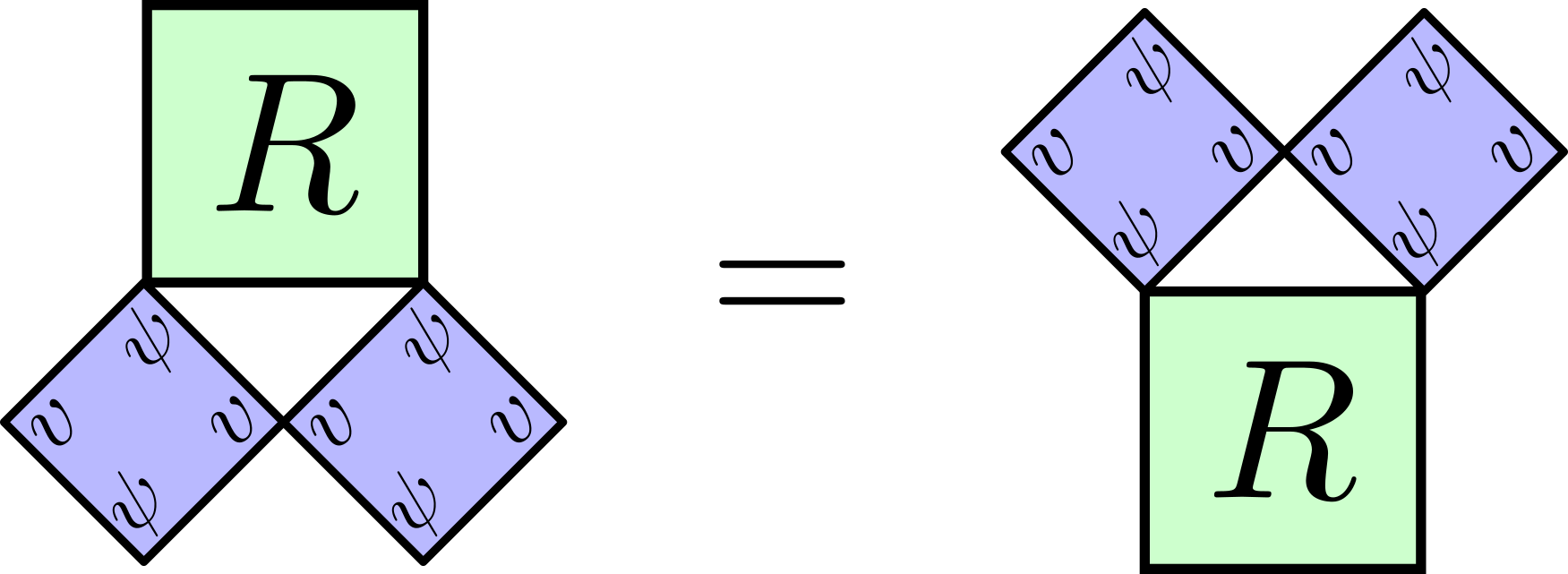}, 
\end{equation}
which is exactly the graphical version of Eq.~\eqref{eq:HA-sym-R-condition} with $z=v_{pq}$. Similarly, the invariance of the first and third lines of Eq.~\eqref{eq:fundamental_Rcommu} under $\mathrm{ad}_z$ %
requires the following two conditions on $R$
\begin{equation}
	\adjincludegraphics[height=7ex,valign=c]{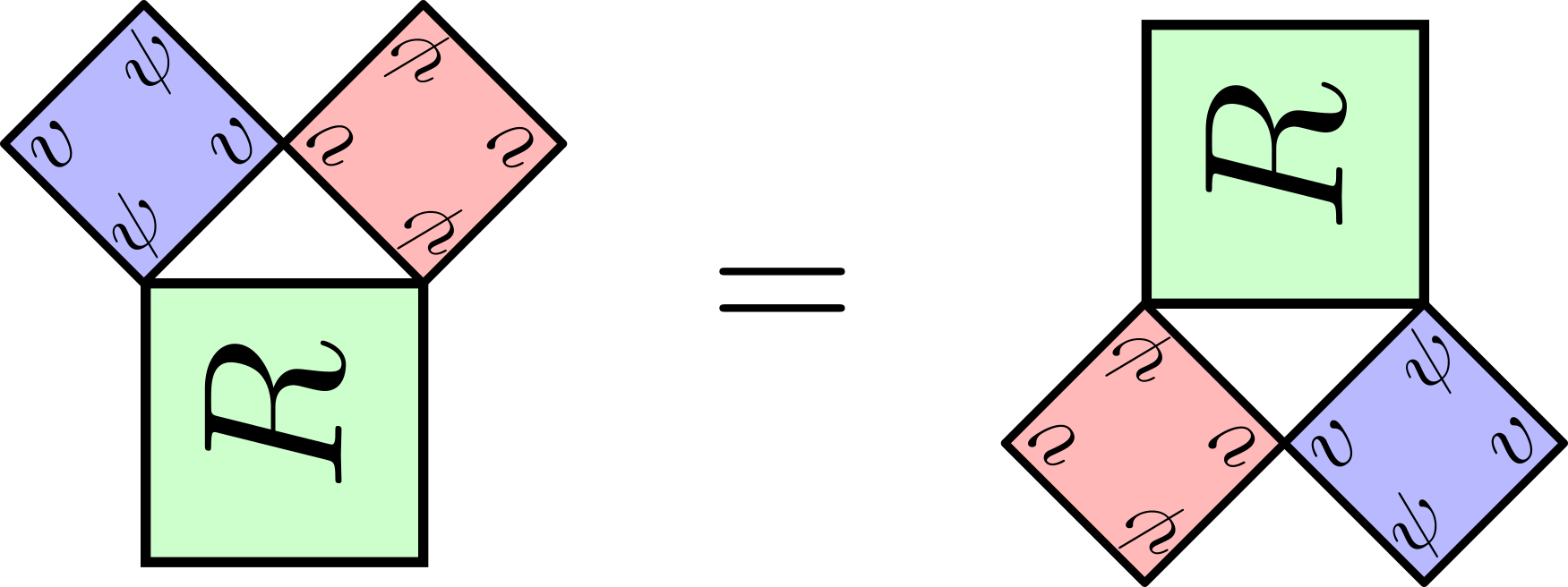},
    \quad \adjincludegraphics[height=7ex,valign=c]{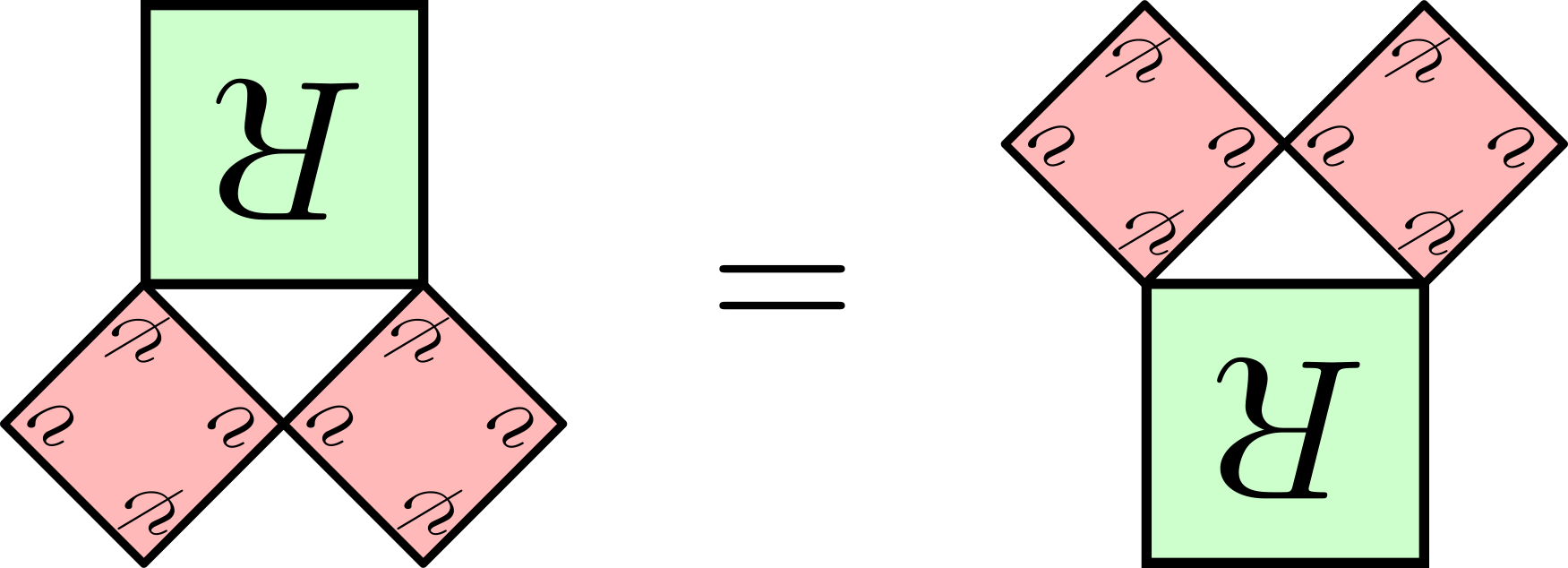},
\end{equation}
both of which are equivalent to Eq.~\eqref{eq:Rpp-condition} due to Eq.~\eqref{eq:Antipode-graphical}. Therefore the only consistency condition here is Eq.~\eqref{eq:Rpp-condition} for any corepresentation $v_{pq}$. 

Eq.~\eqref{eq:HAaction-Fock} can be  derived graphically as
\begin{eqnarray}\label{eq:vpqactionstatespace}
	\hat{\Theta}(v_{pq})\ket{\psi}%
	&=&
    \adjincludegraphics[height=9ex,valign=c]{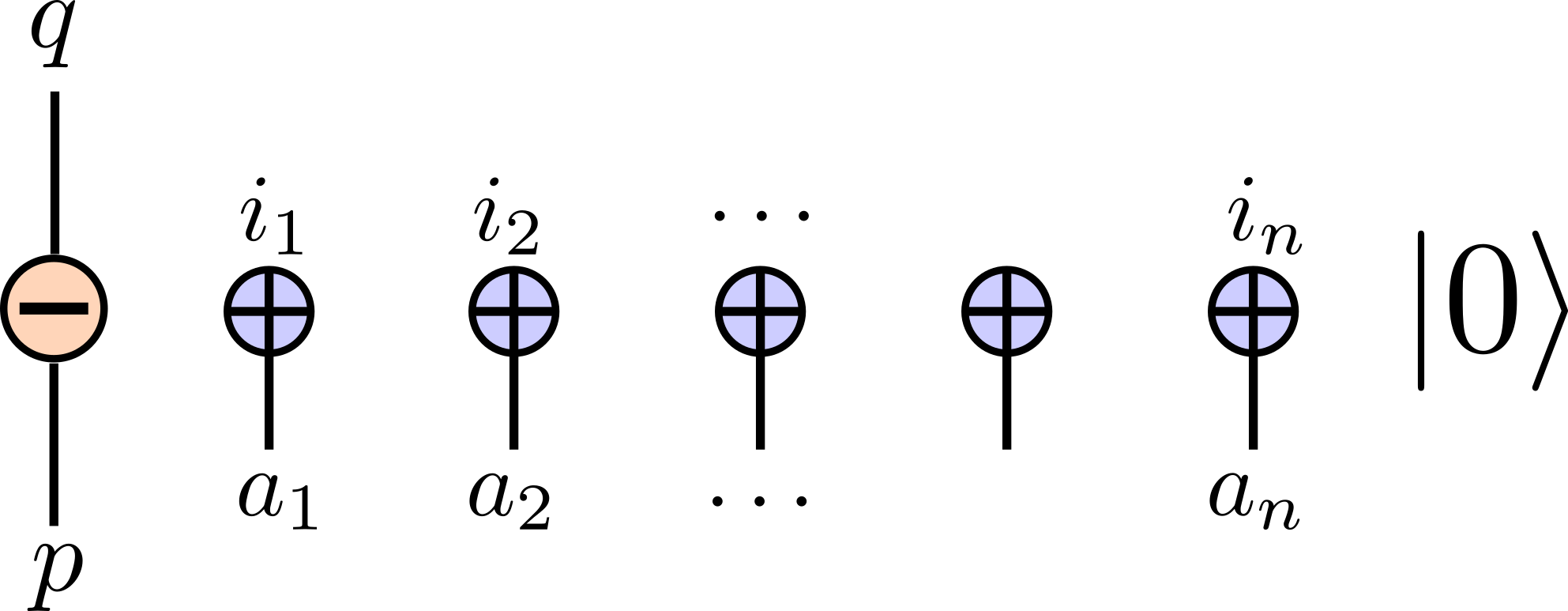}\\
	&=&
    \adjincludegraphics[height=11.5ex,valign=c]{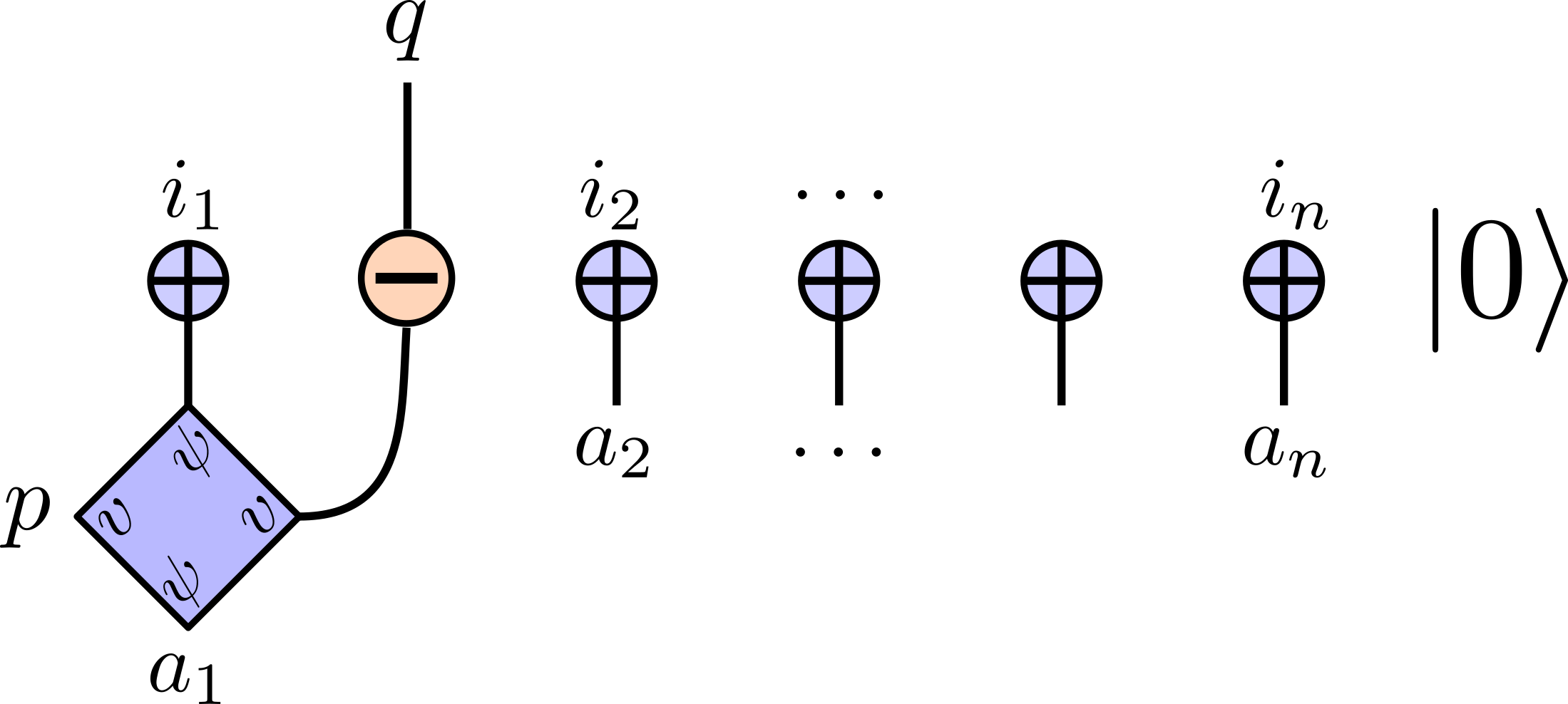}\nonumber\\
	&=&
    \adjincludegraphics[height=11.5ex,valign=c]{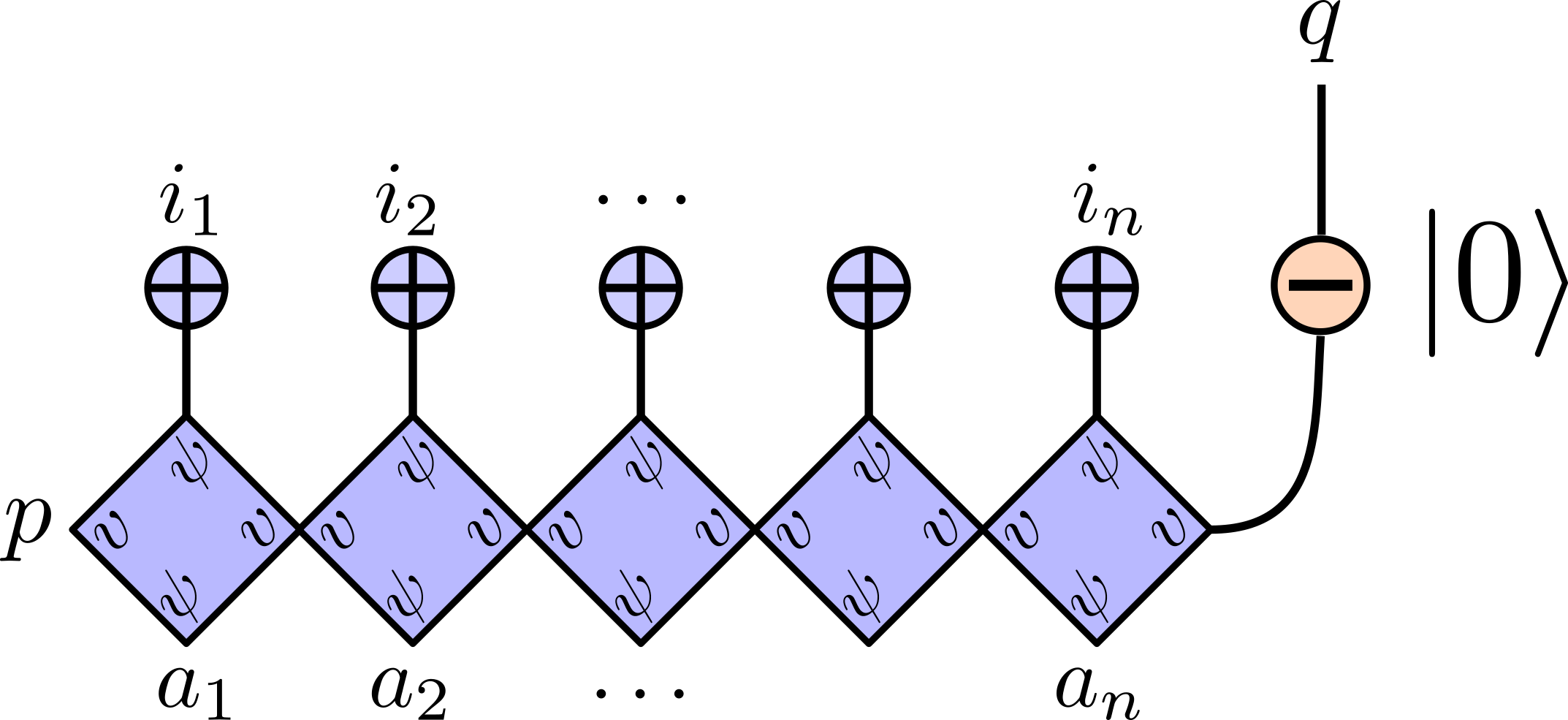}\nonumber\\
	&=&
    \adjincludegraphics[height=9ex,valign=c]{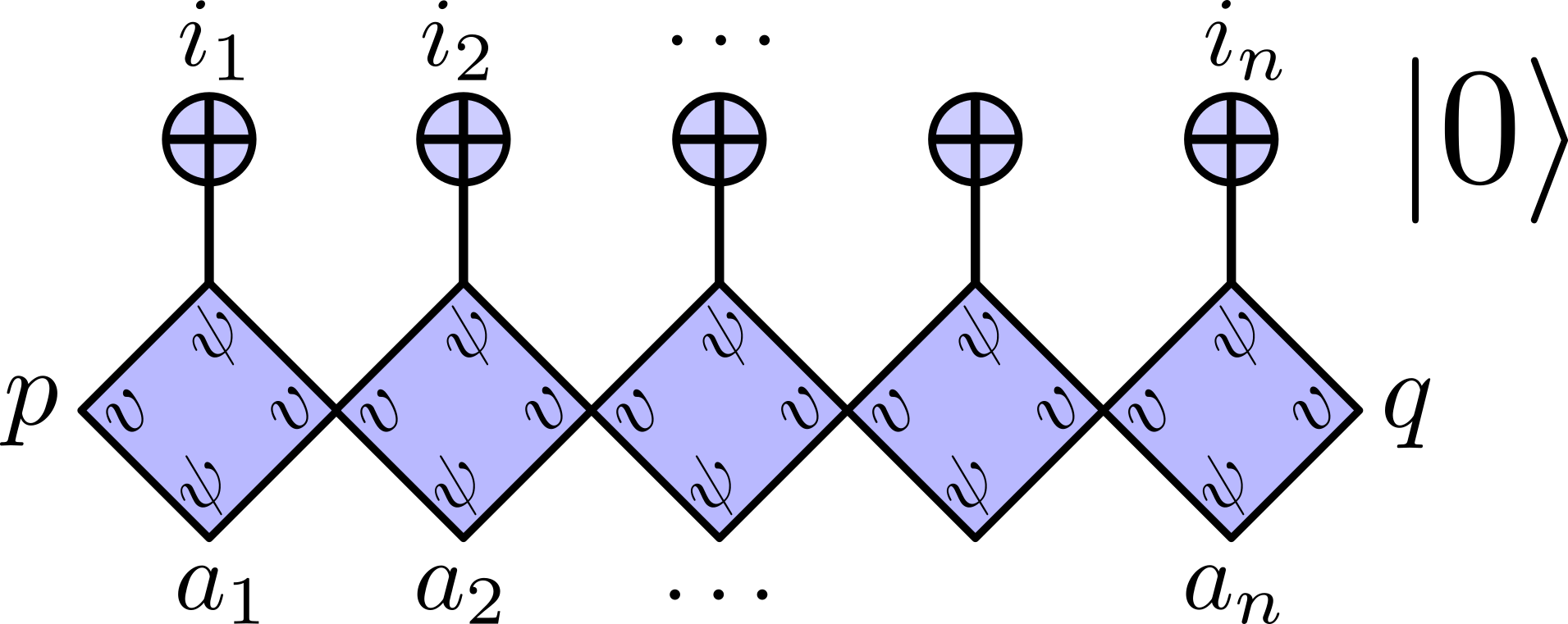}, \nonumber
\end{eqnarray}
where $\ket{\psi}=\hat{\psi}_{i_1,a_1}^{+} \cdots \hat{\psi}_{i_n,a_n}^{+}\ket{0}$, and we used a variant of Eq.~\eqref{eq:Adj-action-graphical} along with $	\hat{\Theta}(v_{pq})\ket{0}=\delta_{pq}\ket{0}$. 
Eq.~\eqref{eq:HA-eij-commute} is derived graphically as:
\begin{equation}\label{eq:HA-eij}
	\adjincludegraphics[height=8ex,valign=c]{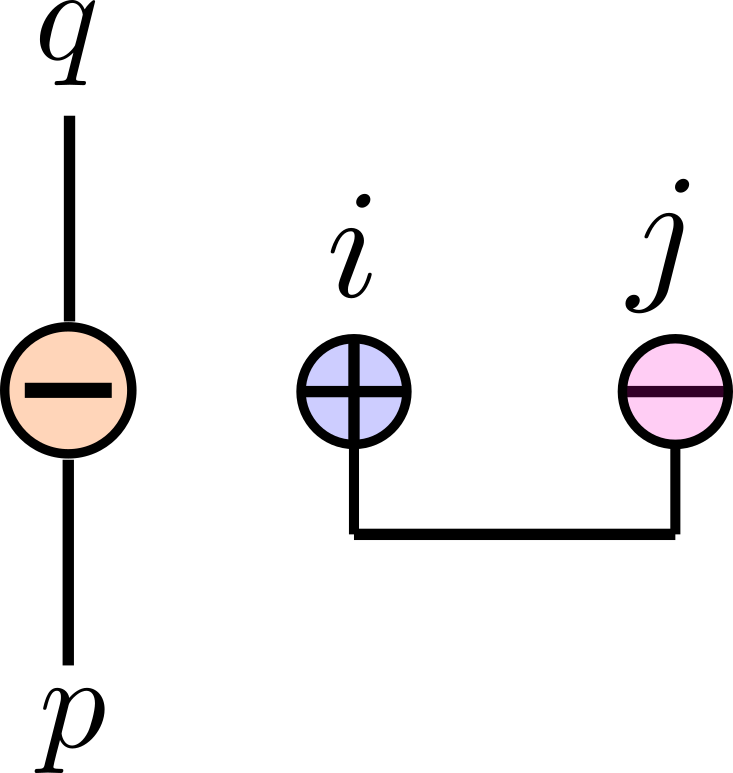}
	~=~
    \adjincludegraphics[height=9ex,valign=c]{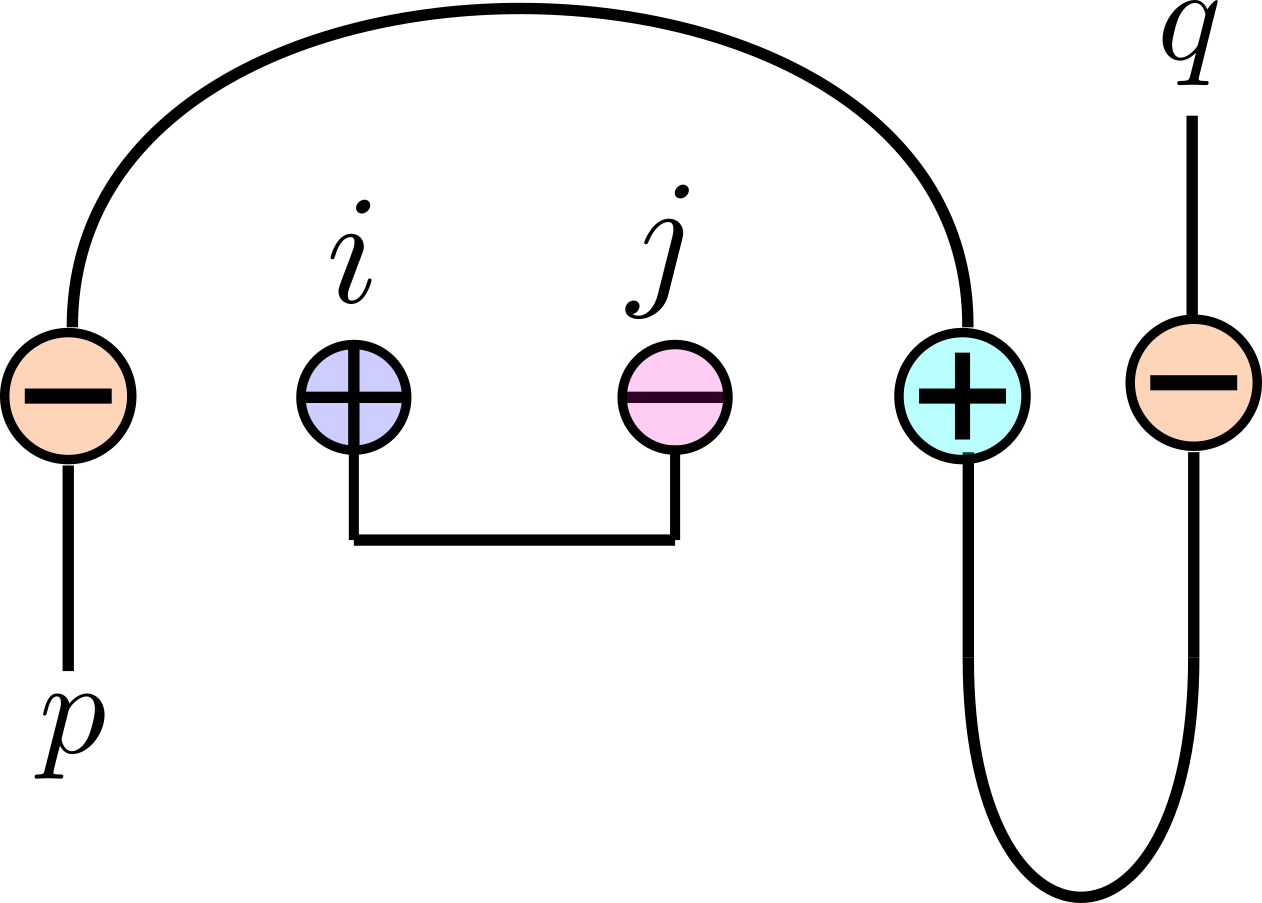}
	~=
    \adjincludegraphics[height=9ex,valign=c]{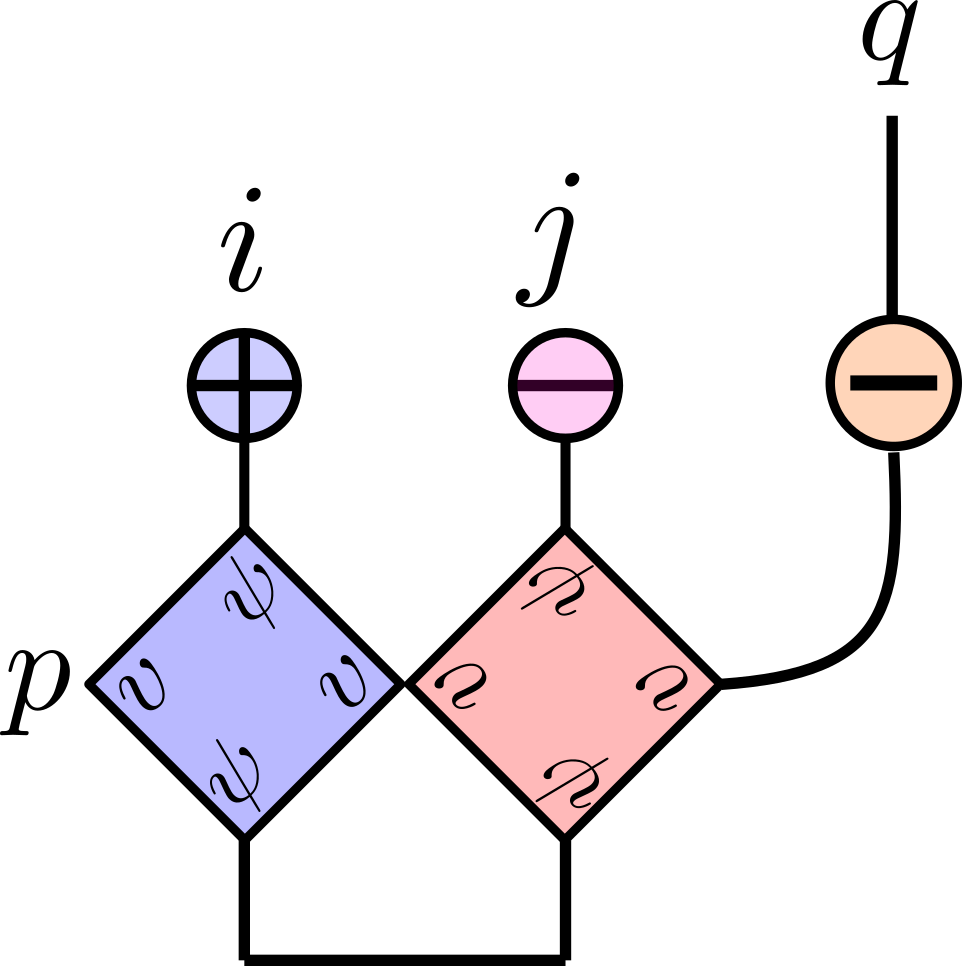}
	~=~
    \adjincludegraphics[height=8ex,valign=c]{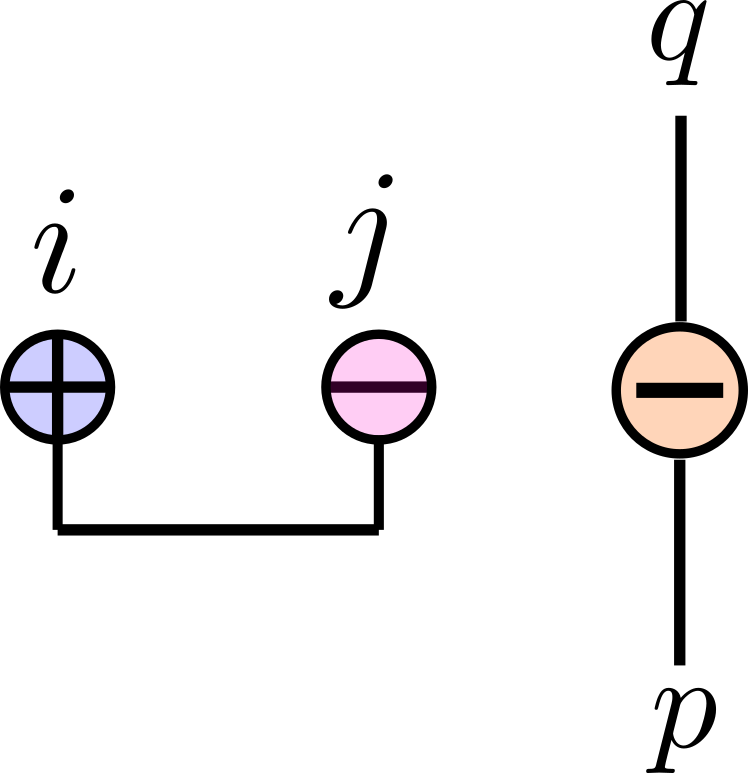}~,
\end{equation}
where we used Eqs.~(\ref{eq:Antipode-graphical}, \ref{eq:Adj-action-graphical}). Eq.~\eqref{eq:HAsymmetry_alpha-cond} is rewritten as
\begin{equation}\label{eq:HA-alpha-cond}
	\adjincludegraphics[width=18ex,valign=c]{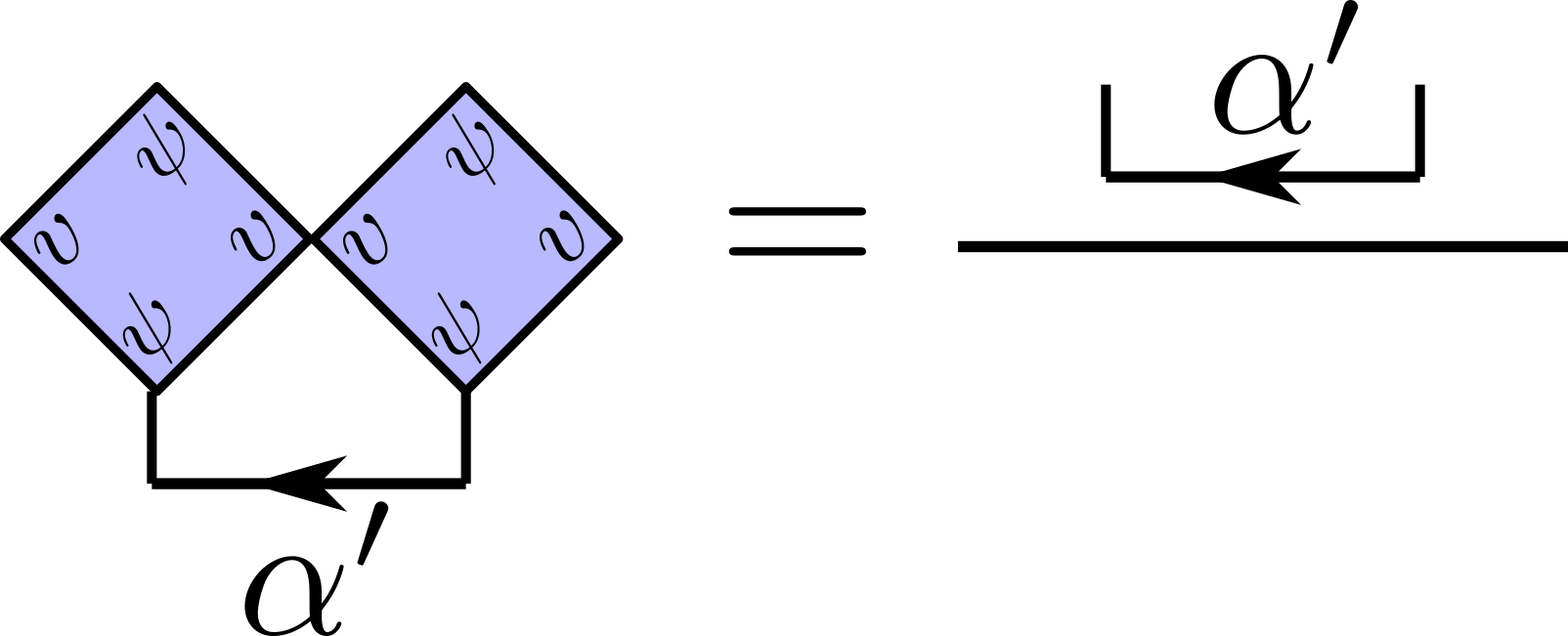}~,
    \quad	 \adjincludegraphics[width=18ex,valign=c]{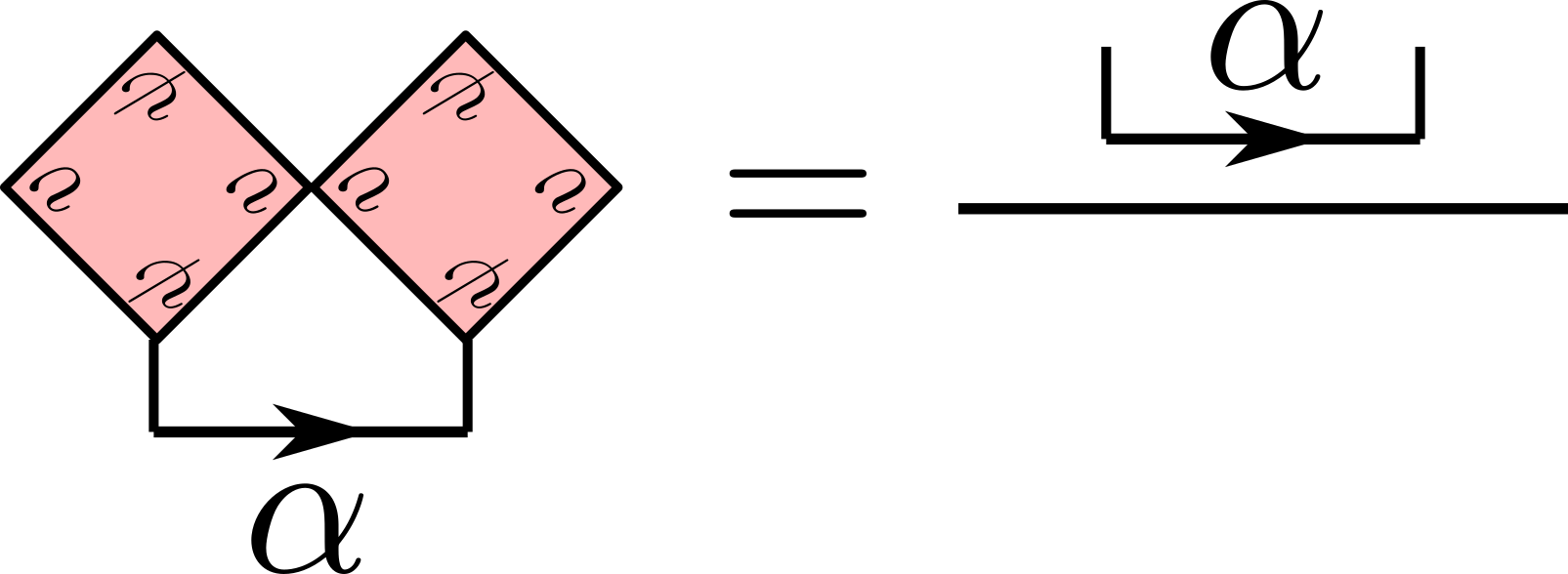}~.
\end{equation}

\begin{remark}\label{rmk:unitary-hidden-symmetry}
	If $R$ is unitary, then it is natural to assume that $\calA$ is a $\C^*$-Hopf algebra, and $\zetapsi$ is a $*$-representation of $\calA$, i.e. $\zetapsi(z)^\dagger=\zetapsi(z^*) ~ \forall z\in \calA$, and $v_{pq}$ is a unitary corepresentation, i.e., $S v_{pq}=v_{qp}^*$. It is then straightforward to show that under this condition $\bar{\zetapsi}$ and $\hat{\Theta}$ are also $*$-representations, and the two lines of Eq.~\eqref{eq:HA-adjoint-action} are equivalent by taking Hermitian conjugate. 
    Indeed, the Hermitian conjugate of the first line of Eq.~\eqref{eq:HA-adjoint-action} gives
    \begin{eqnarray}\label{eq:HA-adj-action-conjugation}
     \hat{\Theta}[(S z_{(2)})^*]\hat{\psi}_{i,a}^{-} \hat{\Theta}[(z_{(1)})^*]
     &=&\sum_b\zetapsi^*_{ba}(z) \hat{\psi}^{-}_{i,b}\nonumber\\
     &=&\sum_b\zetapsi_{ab}(z^*) \hat{\psi}^{-}_{i,b}.
    \end{eqnarray}
    We now consider the substitution $z=(Sy)^*$, and 
    use the simplification~(and that $S$ is a coalgebra antihomomorphism)
    \begin{equation}
    [S(Sy_{(1)})^*\otimes (Sy_{(2)})^*]^*=y_{(1)}\otimes Sy_{(2)},
    \end{equation}
    which is derived from the general identity $*\circ S\circ *\circ S=\id$  valid for all Hopf $*$-algebras~\cite{klimyk1997book}.
    Eq.~\eqref{eq:HA-adj-action-conjugation} now becomes
      \begin{eqnarray}
	\hat{\Theta}[y_{(1)}] \hat{\psi}_{i,a}^{-} \hat{\Theta}[S y_{(2)}]&=& \sum_b\zetapsi_{ab}(Sy) \hat{\psi}^{-}_{i,b},
    \end{eqnarray}
    which is equivalent to the second line of Eq.~\eqref{eq:HA-adjoint-action}.
    A caveat, however, is that $\bar{\zetapsi}(z)$ is generally not the complex conjugate of $\zetapsi(z)$, although this is true in the case of group algebras.  
    Moreover,  Eq.~\eqref{eq:Antipode-graphical} implies that the tensor $\adjincludegraphics[height=6ex,valign=c]{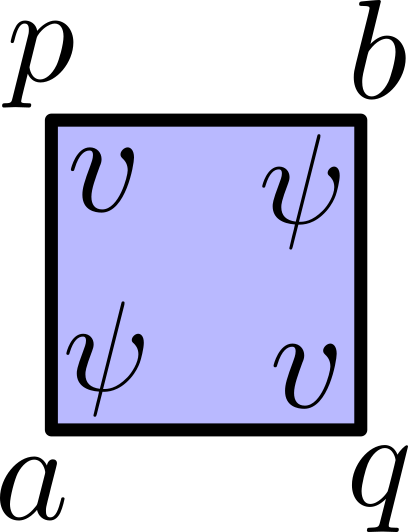}$ is unitary when viewed as a matrix~(where $a,q$ are input indices). If $\calA$ is in addition finite-dimensional, Ref.~\cite{wang2024hopf} shows that this tensor is always a dual unitary tensor. 
	
	Without the $\C^*$-condition on $\calA$, it may happen that the symmetry transformation in Eq.~\eqref{eq:HA-adjoint-action} is not consistent with the Hermitian conjugation of the paraparticle operators $ (\hat{\psi}_{i,a}^{-})^\dagger=\hat{\psi}_{i,a}^{+}$. We do not discuss this possibility in this paper. 
\end{remark}

We finally mention a basic fact that we will use later in Sec.~\ref{sec:local-indist-revisit}: 
\begin{equation}\label{eq:Thetavpqonbra}
\bra{0}\adjincludegraphics[height=7ex,valign=c]{Figures/GeneralizedSymmetry-1/HAaction_tensors-2-1.png}=\delta_{pq}\bra{0}.
\end{equation}
To show this, we show that for any state $\ket{\Psi}$, we have
\begin{equation}\label{eq:bravac}
(\bra{0}\adjincludegraphics[height=7ex,valign=c]{Figures/GeneralizedSymmetry-1/HAaction_tensors-2-1.png}-\delta_{pq}\bra{0})\ket{\Psi}=0.
\end{equation}
Indeed, for $\ket{\Psi}=\ket{0}$, Eq.~\eqref{eq:bravac} follows from the action of $\hat{\Theta}(v_{pq})$ on the vacuum, Eq.~\eqref{eq:HAaction-vac}. For $\ket{\Psi}=\hat{\psi}_{i_1,a_1}^+\ldots \hat{\psi}_{i_n,a_n}^+\ket{0}$ with $n\geq 1$, we note from Eq.~\eqref{eq:vpqactionstatespace} that the action of $\hat{\Theta}(v_{pq})$ preserves total particle number, hence
$\braket{0|\hat{\Theta}(v_{pq})|\Psi}=\braket{0|\Psi}=0$, concluding the proof. %

\subsection{Examples}\label{sec:HAsymexample}
We have seen that the classification of generalized hidden symmetries of a given paraparticle model   boils down to finding solutions to the tensor equation~\eqref{eq:Rpp-condition}~[and also Eq.~\eqref{eq:HA-alpha-cond} in case there are pair-creation terms]. Below we discuss a few specific examples for the $R$-matrices given in Tab.~\ref{tab:Hilbert_series}.
\paragraph{Ordinary fermions and bosons}
For the $R=\pm X$ %
describing ordinary fermions or  bosons~(where $X$ is the swap tensor), 
it is clear from Eq.~\eqref{eq:G-sym-R-condition} that any $m$-dimensional representation of any group $G$ defines a hidden symmetry of the system. 
Below we prove that any $\C^*$-Hopf-algebra hidden symmetry of a system of free fermions/bosons defined in Sec.~\ref{sec:HAsymmetry} is essentially equivalent to a group-like symmetry. 

Let $M$ be the algebra generated by the set of matrices $\{W^{b}_a|1\leq a,b\leq m\}$, where
$[W^{b}_a]_{pq}=
\adjincludegraphics[height=6ex,valign=c]{Figures/GeneralizedSymmetry-1/HAaction_tensors-1-sq.png}$. According to Remark~\ref{rmk:unitary-hidden-symmetry}, 
$\adjincludegraphics[height=6ex,valign=c]{Figures/GeneralizedSymmetry-1/HAaction_tensors-1-sq.png}$ is unitary as a matrix, therefore, $M$ is semisimple~\footnote{Here $M$ can be viewed as a generalization of the algebra $\Dse[R]$ in Definition~\ref{def:simpleR-matrix} to the four-index tensor $W$, and semisimplicity of $M$ follows from the fact that its defining representation can be decomposed as a direct sum of irreducibles, the proof of which is almost identical to the proof of the first statement of Thm.~\ref{thm:directsum-decomposition}, which only uses the unitarity of $R$.
}. 
Meanwhile, with $R=\pm X$, Eq.~\eqref{eq:Rpp-condition} implies that $M$ is Abelian. Therefore, $M$ is the direct sum of 1-dimensional algebras.  %
Without loss of generality, we can assume that the tensor $\adjincludegraphics[height=6ex,valign=c]{Figures/GeneralizedSymmetry-1/HAaction_tensors-1-sq.png}$ is already diagonal in the $\searrow$ direction, i.e.,
\begin{equation}
\adjincludegraphics[height=6ex,valign=c]{Figures/GeneralizedSymmetry-1/HAaction_tensors-1-sq.png}=[U_q]^b_a \delta_{p,q},
\end{equation}
where $U_q$ is unitary for each $q=1,2,\ldots,d$. Let $G$ be the subgroup of U$(m)$ generated by $\{U_q\}^d_{q=1}$, and let $\zetapsi_0$ be the defining $m$-dimensional representation of $G$. Then $(G,\zetapsi_0)$ essentially describes the same hidden symmetry as $(\calA,\zetapsi)$ since they transform the paraparticle operators in the same way. 

Therefore, ordinary free particle systems can only display ordinary group-like symmetries. In the following we show that non-trivial paraparticle systems usually exhibit 
non-trivial Hopf algebra symmetries~(even in free paraparticle systems). 
\begin{remark}
	Systems of interacting fermions and bosons can display emergent generalized symmetries that are not present at the level of elementary field operators. Examples are topologically-ordered systems of interacting fermions and bosons, which display generalized emergent symmetry described by~(higher) fusion category theory~\cite{kong2014braided,Kong2020Algebraichigher}. 	
\end{remark}

\paragraph{Ex.~\ref{ex:1m} in Tab.~\ref{tab:Hilbert_series}}
Since the $R$-matrix in Ex.~\ref{ex:1m} is proportional to the identity matrix, Eq.~\eqref{eq:G-sym-R-condition} is trivially satisfied. Therefore any $m$-dimensional representation $\zetapsi$ of any Hopf algebra $\calA$ defines a generalized hidden symmetry of this type of paraparticles. 

\paragraph{$R$-matrices from quasitriangular Hopf algebras}
If the $R$-matrix is constructed from the universal $\mathcal{R}$-matrix of a quasitriangular Hopf algebra $\calA$, then $\calA$ itself defines a hidden symmetry of the paraparticle system.  Specifically, let $\calA$ be a quasitriangular Hopf algebra with universal $\mathcal{R}$-matrix $\mathcal{R}$, and let $\zetapsi$ be a representation of $\calA$. Define an $R$-matrix as
\begin{equation}\label{eq:RmatFromQHA}
	R=X[(\zetapsi\otimes \zetapsi)\mathcal{R}],%
\end{equation}
where $X$ is the swap tensor. Such an $R$ always satisfies the YBE~(see App.~\ref{app:RMatTHA}). Assume in addition that $R$ is involutive~(this is always satisfied if $\calA$ is triangular; if $\calA$ is only quasitriangular, then $R$ is involutive only for a subset of representations of $\calA$).  Then Eq.~\eqref{eq:HA-sym-R-condition} is always satisfied due to the following property of the universal $\mathcal{R}$-matrix:
\begin{equation}
	\tau[\Delta%
    (z)]=\mathcal{R} \Delta(z) \mathcal{R}^{-1},\quad\forall z\in \calA,
\end{equation}
where $\tau:\calA\otimes \calA\to\calA\otimes \calA$ swaps the tensor factors,  i.e., $\tau(y\otimes z)=z\otimes y$ for all $y,z\in\calA$. 

To give more specific examples, each of the $R$-matrices in Ex.~\ref{ex:decoupled}, Ex.~\ref{ex:Green}, Ex.~\ref{ex:setth}, and Ex.~\ref{ex:setth-ext} of Tab.~\ref{tab:Hilbert_series} can be constructed from a finite dimensional triangular Hopf algebra $\calA$, so $\calA$ itself is a hidden symmetry for the corresponding $R$-paraparticle system. %
A caveat, however, is that different quasitriangular Hopf algebras can give the same $R$-matrix via Eq.~\eqref{eq:RmatFromQHA}. Therefore, for a given Hopf algebra $R$-matrix, the underlying generalized symmetry algebra $\calA$ is generally not unique. In particular, all the $R$-paraparticles in Tab.~\ref{tab:Hilbert_series} have infinite dimensional Hopf algebra symmetries as well. 

\subsection{Mutual parastatistics as a hidden symmetry}\label{sec:statistics-as-symmetry}
In this section, we provide an alternative way to understand the theory of local observables in Sec.~\ref{sec:locality} and the theory of mutual parastatistics in Sec.~\ref{sec:mutual_para} 
using the viewpoint of Hopf algebra hidden symmetry, and %
later we use this viewpoint to study local indistinguishability of $R$-paraparticles and formulate the categorical description of $R$-parastatistics.  
Comparing Eq.~\eqref{eq:Rpp-condition} with Eq.~\eqref{eq:Rpsiphiconstraint}, we notice that 
for a given paraparticle $\psi$ described by an $R$-matrix $R$, any type of paraparticle $\varphi$ in the same system defines a hidden symmetry transformation of $\hat{\psi}^\pm_{i,a}$ via
\begin{eqnarray}\label{eq:zeta_tensor_def-mutual}
	\adjincludegraphics[height=7ex,valign=c]{Figures/GeneralizedSymmetry-1/HAaction_tensors-1-sq.png}=
    \zetapsi_{ba}[v_{pq}]=
    \adjincludegraphics[height=6ex,valign=c]{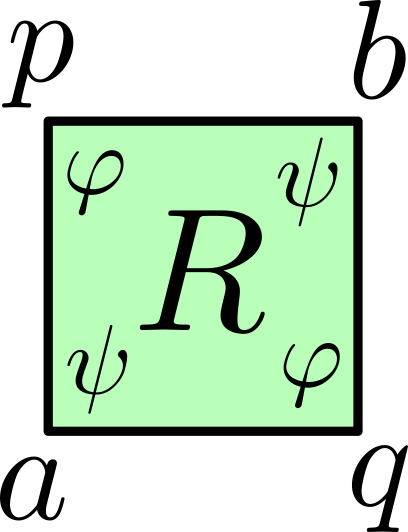},
\end{eqnarray}
Indeed, this is not totally surprising because we can rewrite Eq.~\eqref{eq:Adj-action-graphical} as 
\begin{equation}\label{eq:Adj-action-graphical-2}
\adjincludegraphics[width=6ex,valign=c]{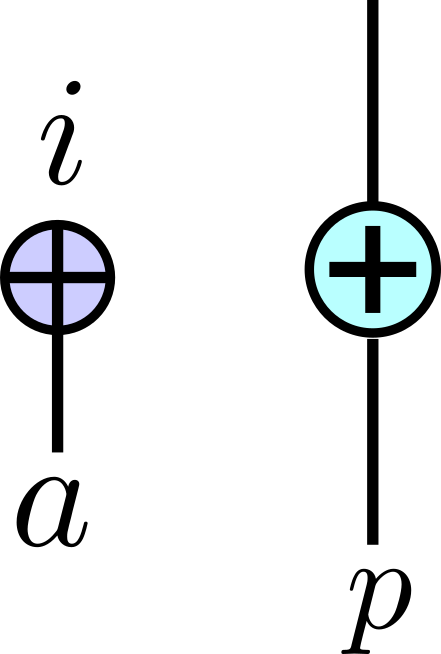}=
\adjincludegraphics[width=6ex,valign=c]{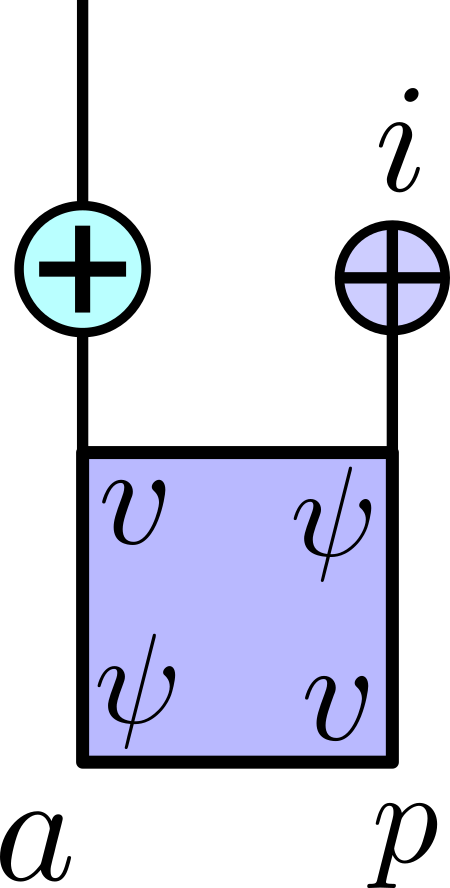},
\end{equation}
which makes it similar to the CR between paraparticle operators $\hat{\psi}^+_{i,a}$ and $\hat{\varphi}^+_{j,b}$ in Fig.~\ref{fig:RMQA-rela}. We now give a formal result supporting this viewpoint:

\begin{restatable}{thm}{basicHA} \label{thm:basicHA}
Let $\frakF$ be an $R$-paraparticle system, 
let $\T$ be the set of particle types and let  $\mathbf{R}$ be the composite $R$-matrix defined in Eq.~\eqref{eq:compositeR-def}. 
Then $\frakF$ has a hidden bialgebra symmetry, denoted by $\calA_\mathbf{R}$, with a representation $\psi$ and a corepresentation $v^\psi$ associated to each particle type $\psi\in\T$, such that for any $\psi,\varphi\in \T$, we have %
\begin{eqnarray}\label{eq:zeta_tensor_def-basicHA}
\adjincludegraphics[height=7ex,valign=c]{Figures/GeneralizedSymmetry-1/HAaction_tensors-1-sq.png}=\zetapsi_{ba}[v^\varphi_{pq}]=
\adjincludegraphics[height=6ex,valign=c]{Figures/GeneralizedSymmetry-1/Mutual-R-nc},
\end{eqnarray}
and $\Rep(\calA_\mathbf{R})$ is generated by $\T$ as a monoidal $\C$-linear additive category.  
\end{restatable}
In simple terms, the statement that $\Rep(\calA_\mathbf{R})$ is generated by $\T$  as a monoidal $\C$-linear additive category means that all representations of $\calA_\mathbf{R}$ can be obtained by recursively taking tensor product of representations in $\T$ and decomposing the result into indecomposable ones, as we describe more explicitly in Sec.~\ref{sec:cat-description}. 
The proof of Thm.~\ref{thm:basicHA} is given in App.~\ref{app:basichopfsymmetry}, where $\calA_\mathbf{R}$ is constructed explicitly.  
We call $\calA_\mathbf{R}$ the fundamental bialgebra symmetry associated to $\mathbf{R}$. 
\begin{remark}\label{rmk:DUHopf}
Under the condition of Thm.~\ref{thm:basicHA},
if $\mathbf{R}$ is dual unitary, then $\calA_\mathbf{R}$ can be extended to a Hopf algebra hidden symmetry, which we call the fundamental Hopf algebra hidden symmetry. %
If, in addition, $\mathbf{R}$ is weakly equivalent to $\pm  X$, then $\calA_\mathbf{R}$ is finite dimensional. Since we do not use these facts in this paper in an essential way, we omit the proof, but in App.~\ref{app:basichopfsymmetry} we briefly explain the intuition why dual unitarity of $\mathbf R$ leads to an antipode on $\calA_\mathbf{R}$.
\end{remark}

\begin{remark}\label{rmk:hiddenBAsymmetry}
For a  hidden bialgebra symmetry $\calA$, some of the equations in Sec.~\ref{sec:generalized_symmetry} need to be rewritten in a form that does not involve the antipode. For example, Eq.~\eqref{eq:HA-adjoint-action} should be written as 
\begin{eqnarray}\label{eq:BA-adjoint-action}
	\hat{\Theta}(z) \hat{\psi}_{i,a}^{+} &=&\sum_b \hat{\psi}^{+}_{i,b}\zetapsi_{ba}(z_{(1)})\hat{\Theta}[ z_{(2)}],\nonumber\\
	\hat{\psi}_{i,a}^{-} \hat{\Theta}(z)&=& \sum_b\hat{\Theta}[ z_{(2)}]\zetapsi_{ab}(z_{(1)}) \hat{\psi}^{-}_{i,b},
\end{eqnarray}
with tensor network representation
\begin{equation}\label{eq:Adj-action-graphical-BA}
\adjincludegraphics[width=6ex,valign=c]{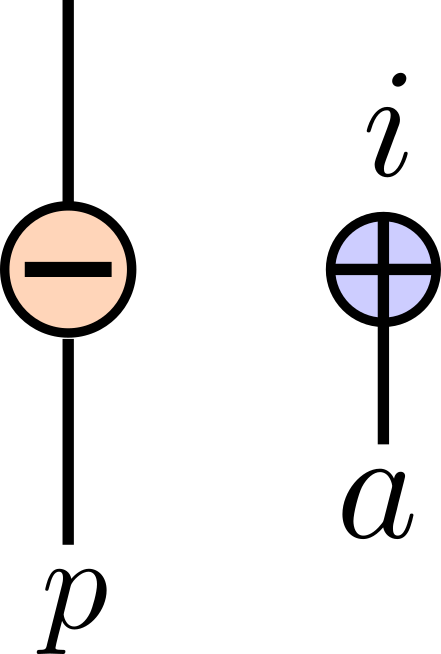}=
\adjincludegraphics[width=6ex,valign=c]{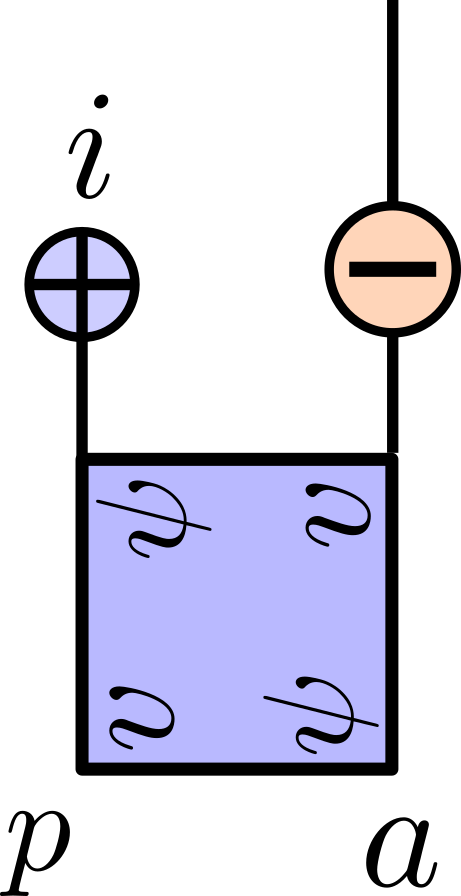},
\quad 
\adjincludegraphics[width=6ex,valign=c]{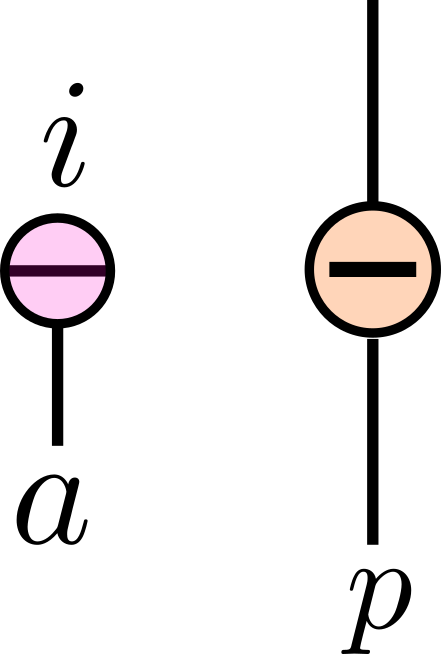}=
\adjincludegraphics[width=6ex,valign=c]{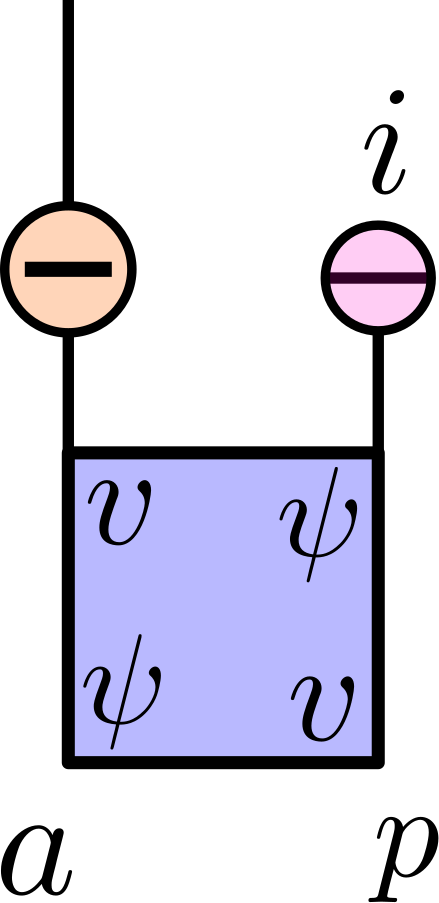},
\end{equation}
in place of Eq.~\eqref{eq:Adj-action-graphical}. Other results in Sec.~\ref{sec:generalized_symmetry} generalize in a similar way. 
Eq.~\eqref{eq:BA-adjoint-action} is enough to guarantee the $\calA$-invariance of $\hat{e}_{ij}$ in Eq.~\eqref{eq:HA-eij-commute}, 
since we have 
\begin{equation}
\hat{\Theta}(z)\hat{e}_{ij}=\sum_{a,b}\hat{\psi}^+_{i,b}\zetapsi_{ba}(z_{(1)})\hat{\Theta}(z_{(2)})\hat{\psi}^-_{j,a}=\hat{e}_{ij}\hat{\Theta}(z).
\end{equation}
\end{remark}

Viewing  parastatistics as a hidden symmetry transformation is useful because it leads to an equivalent~(but conceptually more convenient) definition of general local observables, as we note in a simple yet fundamental fact:
\begin{fact}\label{fact:locality-and-A-invariance}
Let $\calA$ be the fundamental bialgebra symmetry of a paraparticle system. 
Then in the Definition~\ref{def:general_LO} of local observables~(see also the last paragraph in Sec.~\ref{sec:mutual_para}), Eq.~\eqref{eq:generalLOcondition} is equivalent to 
\begin{equation}\label{eq:generalLO-HAsyn}
[\hat{O},\hat{\Theta}(v_{pq}^\varphi)]=0,~\forall p,q,\varphi.
\end{equation}
Therefore, a local observable supported on region $S$ is exactly an $\calA$-invariant Hermitian element of $\mathfrak{F}_S$. More formally, we can write $\mathfrak{A}_S=\mathfrak{F}_S^\calA$, %
where $\mathfrak{F}_S$ is the algebra generated by $\{\hat{\psi}^\pm_{i,a}|i\in S,1\leq a\leq m,\psi\in\T\}$~(roughly, the algebra of quantum fields on $S$), and the notation $\mathfrak{F}_S^\calA$ means the subalgebra of $\calA$-invariant elements of $\mathfrak{F}_S$. 
\end{fact}

\begin{remark}
Here we remark that there may be a chicken-egg dilemma about whether the definition of local observables should come before or after hidden symmetries. In this paper we first defined local observables  in Sec.~\ref{sec:locality} and then define hidden symmetry as those linear transformations of $\hat{\psi}^{\pm}_{i,a}$ that leave all local observables invariant. An alternative approach is to first specify a hidden symmetry $\calA$ and then construct the local observable algebra $\mathfrak{A}$ simply as the algebra of $\calA$-invariant local observables. We expect the first approach to be more general.  %
\end{remark}

We finally mention that in 2D topological phases, %
$R$-paraparticles generally have mutual anyonic statistics with other particle types
in the same system, including genuine anyons~(those excitations with $R^2\neq \mathds{1}$). Although our second quantization framework cannot include creation/annihilation operators for genuine anyons, we can treat the mutual anyonic statistics between paraparticles and anyons as a Hopf algebra hidden symmetry transformation on paraparticle operators.   Specifically, consider a 2D topological phase described by a braided fusion category $\calC=\Rep(\calA)$, where $\calA$ is a quasitriangular Hopf algebra with universal $R$-matrix $\calR$. 
Then it follows from the basic axioms of quasitriangular Hopf algebras
\begin{equation}\label{def:QtHA}
		\left(\Delta\otimes \id \right) \mathcal{R}=\mathcal{R}_{13} \mathcal{R}_{23},%
\end{equation}
that every representation $\varphi$ of $\calA$ gives corepresentation  $v_{ij}$ of $\calA$ via
\begin{equation}\label{eq:corep_from_rep}
v_{pq}=(\id\otimes \varphi_{pq})\calR.
\end{equation}
Inserting Eq.~\eqref{eq:corep_from_rep} into Eq.~\eqref{eq:zeta_tensor_def}, we obtain
\begin{eqnarray}\label{def:Urhovbialg-2}
	\psi_{ab}(v_{pq})=(\psi_{ab}\otimes \varphi_{pq})\calR\equiv [R^{(\psi\varphi)}]^{pa}_{bq}\equiv\adjincludegraphics[width=6ex,valign=c]{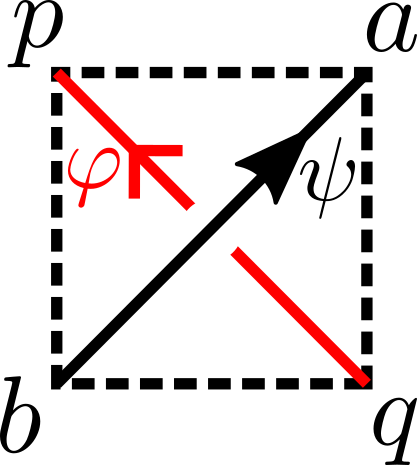}
\end{eqnarray}
where $R^{(\psi\varphi)}$ is exactly the matrix form of the braiding natural isomorphism $R^{(\psi\varphi)}:\psi\otimes \varphi\to \varphi\otimes \psi$ of the braided fusion category $\Rep(\calA)$. Physically, $R^{(\psi\varphi)}$ gives the braiding statistics between anyons $\psi$ and $\varphi$. Furthermore, it is a basic fact that if $\calR$ is a universal $\calR$-matrix of $\calA$, then so is $\calR_{21}^{-1}$. If one replaces $\calR$ by $\calR_{21}^{-1}$ in Eq.~\eqref{eq:corep_from_rep} then Eq.~\eqref{def:Urhovbialg-2} would produce 
\begin{eqnarray}\label{def:Urhovbialg-3}
[{R^{(\varphi\psi)}}^{-1}]^{pa}_{bq}\equiv\adjincludegraphics[width=6ex,valign=c]{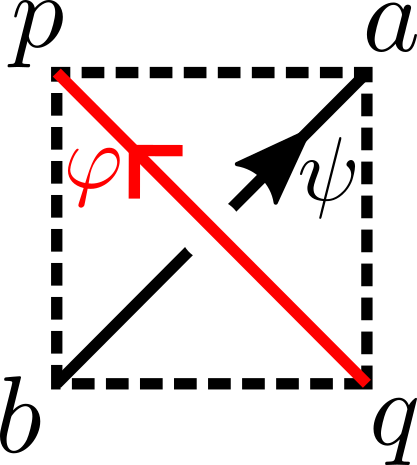},
\end{eqnarray}
which is different from Eq.~\eqref{def:Urhovbialg-2} if $\psi$ and $\varphi$ have mutual anyonic statistics, and we use upper and lower crossing to distinguish clockwise and counterclockwise braiding. In summary, each type of anyon in the physical system defines a hidden symmetry transformation on the $R$-paraparticle system via mutual statistics. 

\subsection{The local indistinguishability criterion revisited}\label{sec:local-indist-revisit}
We now use hidden symmetry as a tool to answer the important physical question about which type of $R$-paraparticles satisfy the local indistinguishability criterion we introduced back in Sec.~\ref{sec:indistcrit}, and we will eventually see that local indistinguishability imposes a nontrivial condition on the $R$-matrix. 

\subsubsection{Equivalence between simplicity and local indistinguishability}\label{sec:simplevsindist}
We first introduce the notion of a simple $R$-paraparticle:
\begin{definition}\label{def:simplicity}
Let $R$ be a unitary involutive $R$-matrix, and let $\psi$ be the associated type of $R$-paraparticle. Let $\FA$ be an algebra of local observables.  %
We say that $\psi$ is simple if for any distinct set of points $\{i_1,i_2,\ldots,i_n\}$,
\begin{equation}\label{eq:TPdegenerate}
	\!\braket{0;i_1^{b_1} \ldots i_n^{b_n}|\hat{O}|0;i_1^{a_1} \ldots i_n^{a_n}}=C^{O}_{i_1\ldots i_n}\!\prod_{j=1}^n \delta_{a_j b_j}
\end{equation}
for any local observable $\hat{O}\in \FA$ whose spatial support contains no more than one point %
of $\{i_1,i_2,\ldots,i_n\}$, 
where $C^O_{i_1\ldots i_n}$ is a constant independent of $a_1,\ldots,a_n$, and  %
$\ket{0;i_1^{a_1} \ldots i_n^{a_n}}=\hat{\psi}^+_{i_1,a_1}\ldots\hat{\psi}^+_{i_n,a_n}\ket{0}$. %
\end{definition}
Physically, being simple means that $\psi$ is stable against any local perturbation. 
We emphasize that the definition of simplicity and local indistinguishability both depend on the algebra of local observables $\FA$, which is considered as a part of the definition of a paraparticle theory.
Here  elements of $\FA$ are required to be local operators in the sense of Definition~\ref{def:general_LO}; however, $\FA$ need not include all possible local operators. 
We note the following useful fact: %
\begin{proposition}\label{prop:simplen=1case}
Under the conditions of Definition~\ref{def:simplicity}, if Eq.~\eqref{eq:TPdegenerate} holds for the case $n=1$ for all $\hat{O}\in \FA$, then it holds for all positive integers $n$. %
\end{proposition}
\begin{proof}
For any $\hat{O}\in \FA$, let $S$ be its spatial support, %
which contains no more than one point of $\{i_1,i_2,\ldots,i_n\}$. We treat the following two cases separately:\\
(i) if $S\cap\{i_1,\ldots,i_{n-1}\}=\emptyset$, then
$\hat{O}$ commutes with $\hat{\psi}^-_{i_{n-1},b_{n-1}}\ldots\hat{\psi}^-_{i_1,b_1}$ due to 
Definition~\ref{def:general_LO}, and the LHS of Eq.~\ref{eq:TPdegenerate} becomes
\begin{eqnarray}\label{eq:TPdegenerate-proof}
    &&\bra{0}\hat{\psi}^-_{i_{n},b_{n}}\hat{O}\hat{\psi}^-_{i_{n-1},b_{n-1}}\ldots\hat{\psi}^-_{i_1,b_1}\hat{\psi}^+_{i_1,a_1}\ldots \hat{\psi}^+_{i_n,a_n}|0\rangle\nonumber\\
    &=&\braket{0|\hat{\psi}^-_{i_{n},b_{n}}\hat{O}\hat{\psi}^+_{i_n,a_n}|0}\!\prod_{j=1}^{n-1} \delta_{a_j b_j},
\end{eqnarray}
where in the second line we use Eq.~\eqref{eq:fundamental_Rcommu} and the fact that $\{i_1,i_2,\ldots,i_n\}$ are distinct. This proves Eq.~\eqref{eq:TPdegenerate} for all $n$  for this case, assuming it is true for $n=1$;\\ %
(ii) if $S$ contains $i_k$ for some $k\in\{1,\ldots,n-1\}$, then in evaluating the LHS of Eq.~\ref{eq:TPdegenerate} we first use the fundamental CRs Eq.~\eqref{eq:fundamental_Rcommu} to move $\hat{\psi}^+_{i_k,a_k}$ to the rightmost position, and move $\hat{\psi}^-_{i_k,b_k}$ to the leftmost position. Then, since $S\cap \{i_1,\ldots,i_{k-1},i_{k+1},\ldots,i_n\}=\emptyset$, one can apply the result of case (i) and prove Eq.~\eqref{eq:TPdegenerate} for this case as well. 
\end{proof}

The following lemma plays a central role in this section, and is a nice application of hidden symmetry:
\begin{lemma}\label{lemma:kappainDCR}
Let $\calA$ be any bialgebra hidden symmetry %
of a paraparticle system. 
We use $\End_\calA(\psi)$ to denote the space of $m\times m$ matrices  commuting with the $\calA$-action, i.e., 
\begin{equation}
 \End_\calA(\psi)\equiv\{\kappa\in M_m(\C)|\zetapsi(z)\kappa=\kappa\zetapsi(z),~\forall z\in\calA\}.
\end{equation}
For any %
$\calA$-invariant observable $\hat{O}$
and any $i,j$, define
\begin{equation}
\kappa^O_{ab}\equiv \braket{0|\hat{\psi}^-_{i,a}\hat{O}\hat{\psi}^+_{j,b}|0}.   
\end{equation}
Then we have $\kappa^O\in \End_{\calA}(\psi)$. %

\end{lemma}
\begin{proof}
For any $z\in \calA$, we have 
\begin{align}\label{eq:kappaOproof}
\sum_b\kappa^O_{ab}\zetapsi_{bc}(z)&=\sum_b\braket{0|\hat{\psi}^-_{i,a}\hat{O}\hat{\psi}^+_{j,b}|0}\zetapsi_{bc}(z)\nonumber\\
&=\sum_b\braket{0|\hat{\psi}^-_{i,a}\hat{O}\hat{\psi}^+_{j,b}\zetapsi_{bc}(z_{(1)})\hat{\Theta}(z_{(2)})|0}\nonumber\\
&=\braket{0|\hat{\psi}^-_{i,a}\hat{O}\hat{\Theta}(z)\hat{\psi}^+_{j,c}|0}\nonumber\\
&=\braket{0|\hat{\psi}^-_{i,a}\hat{\Theta}(z)\hat{O}\hat{\psi}^+_{j,c}|0}\nonumber\\
&=\sum_b\braket{0|\hat{\Theta}(z_{(2)})\zetapsi_{ab}(z_{(1)})\hat{\psi}^-_{i,b}\hat{O}\hat{\psi}^+_{j,c}|0}\nonumber\\
&=\sum_b\zetapsi_{ab}(z)\kappa^O_{bc}, %
\end{align}
where in the second line we use Eq.~\eqref{eq:HAaction-vac} and the counit axiom, 
in the third and the fifth lines we use Eq.~\eqref{eq:BA-adjoint-action}~[which is a variant of Eq.~\eqref{eq:HA-adjoint-action}], in the fourth line we use $[\hat{\Theta}(z),\hat{O}]=0$, and in the last line we use Eq.~\eqref{eq:Thetavpqonbra}. Therefore, $\kappa^O\in \End_\calA(\psi)$. 
\end{proof}
We note that if $\calA$ is the fundamental bialgebra hidden symmetry defined in Thm.~\ref{thm:basicHA}, then taking $z=v_{pq}^\psi$ in Eq.~\eqref{eq:kappaOproof} and recalling that 
$\adjincludegraphics[height=3ex,valign=c]{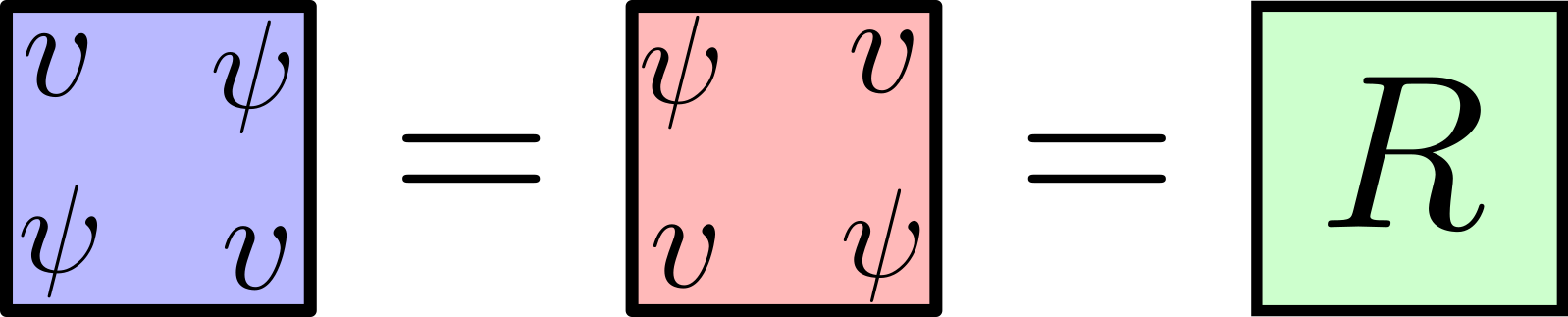}$~[which follows from Eq.~\eqref{eq:zeta_tensor_def-mutual}], 
the above lemma implies that $\kappa^O_{ab}$ satisfies Eq.~\eqref{eq:kappaDC}, i.e., $\kappa^O\in\DC[R]$. 
\begin{corollary}\label{cor:simplemodule-simpletype}
Let $\calA$ be a hidden bialgebra symmetry of a paraparticle system, and let $\FA$ be any $\calA$-invariant local observable algebra. 
If a paraparticle $\psi$ is simple as an $\calA$-module, then it is simple with respect to $\FA$.
\end{corollary}
\begin{proof}
For any $\calA$-invariant local observable $\hat{O}\in \FA$, Lemma~\ref{lemma:kappainDCR} implies $\kappa^O\in \End_\calA(\psi)$. 
If $\zetapsi$ is a simple $\calA$-module, then Schur's lemma implies $\End_\calA(\zetapsi)=\C\mathds{1}$, 
hence $\kappa^O_{ab}\propto\delta_{ab}$, and Proposition~\ref{prop:simplen=1case} implies that $\psi$ is simple with respect to $\FA$. 
\end{proof}

\begin{theorem}\label{thm:eqv-simple-indist}
A paraparticle $\psi$ is simple with respect to a local observable algebra $\FA$ if and only if it satisfies the local indistinguishability criterion with respect to $\FA$.  
\end{theorem}
\begin{proof}
If $\psi$ is simple, then Eq.~\eqref{eq:TPdegenerate} along with the exchange relation in Eq.~\eqref{eq:braiding_derivation} imply that $\psi$ satisfies the local indistinguishability condition in Eq.~\eqref{eq:local_indist}. 

 If $\psi$ satisfies local indistinguishability, for any local observable $\hat{O}\in\FA$ and any position $i$, let $\kappa^O_{a b}=\braket{0;i^{b} |\hat{O}|0; i^{a}}$. Below we show that $\kappa^O_{a b} \propto \delta_{ab}$, and then the simplicity of $\psi$ follows from Proposition~\ref{prop:simplen=1case}. From Lemma~\ref{lemma:kappainDCR}, we know that $\kappa^O_{ab}\in \DC[R]$. Take any $j\notin \{i\}\cup S$, where $S$ is the support of $\hat{O}$. Consider the state 
 \begin{equation}
\ket{\Psi}=\sum_{a,b}\Psi_{ab}\ket{0;i^aj^b},
 \end{equation}
 where $\Psi_{ab}$ is any normalized wavefunction. 
 Inserting $\ket{\Psi}$ into the local indistinguishability condition in Eq.~\eqref{eq:local_indist} and using Eq.~\eqref{eq:TPdegenerate-proof}, we get 
\begin{equation}\label{eq:local_indist-proof}
\adjincludegraphics[width=4ex,valign=c]{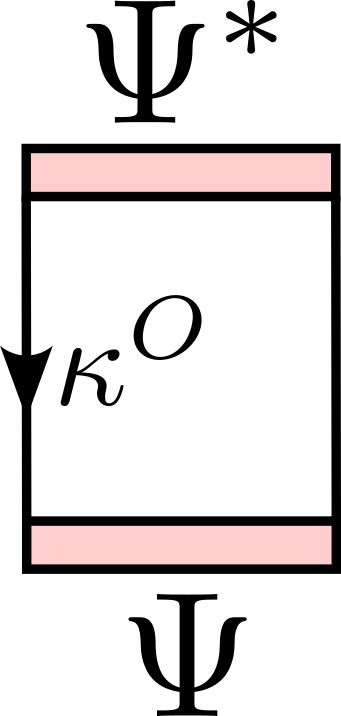}=\braket{\Psi|\hat{O}|\Psi}=\braket{\Psi|\hat{E}^\dagger_{ij}\hat{O}\hat{E}_{ij}|\Psi} =
\adjincludegraphics[width=4ex,valign=c]{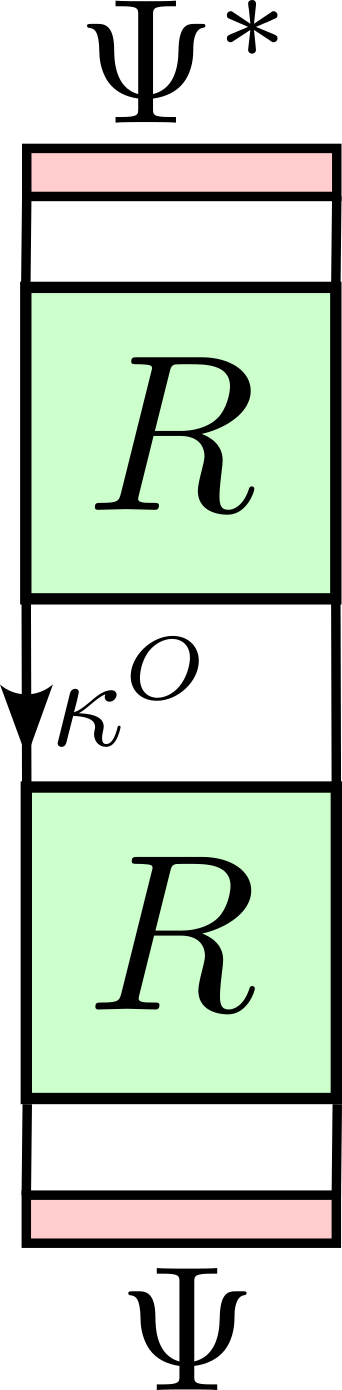}=
\adjincludegraphics[width=4ex,valign=c]{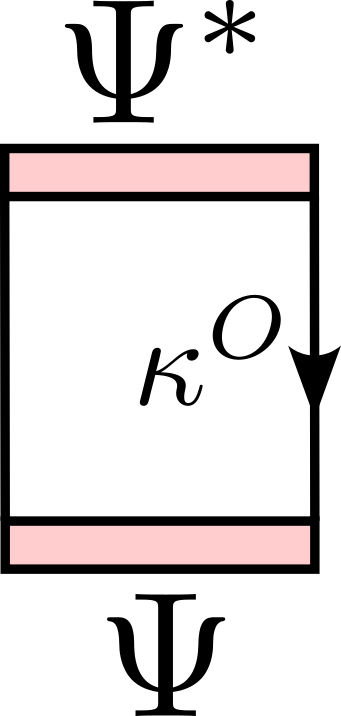}, 
\end{equation}
where we use 
$\adjincludegraphics[width=4ex,valign=c]{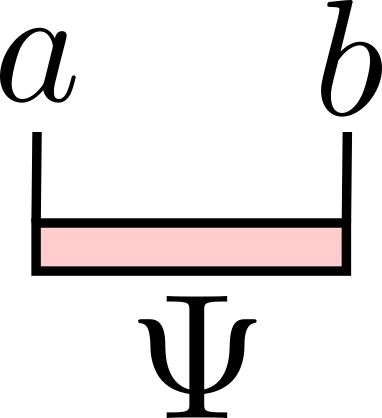}$ to denote $\Psi_{ab}$ and in the last step we used $\kappa^O_{ab}\in \DC[R]$. Local indistinguishability requires that Eq.~\eqref{eq:local_indist-proof} holds for any $\Psi_{ab}$, which is possible only if $\kappa^O\otimes \mathds{1}=\mathds{1}\otimes \kappa^O$, i.e., $\braket{0;i^{b} |\hat{O}|0; i^{a}}=\kappa^O_{ab}\propto\delta_{ab}$. From Proposition~\ref{prop:simplen=1case} we know that $\psi$ is simple.

\end{proof}

\subsubsection{Local indistinguishability condition on the $R$-matrix}\label{sec:hiddensym-localindist}
With Thm.~\ref{thm:eqv-simple-indist}, local indistinguishability of an $R$-paraparticle is equivalent to simplicity of its particle type. 
Our next goal is to translate the simplicity condition to an explicit condition on the $R$-matrix. We begin with a definition:  
\begin{definition}\label{def:simpleR-matrix}
Let $R$ be an involutive $R$-matrix with quantum dimension $m$. Define two sets of $m\times m$ matrices $\{R^{(ad)}_{\swarrow}|1\leq a,d\leq m\}$ and $\{R^{(bc)}_{\searrow}|1\leq b,c\leq m\}$ as follows
\begin{eqnarray}
[R^{(ad)}_{\swarrow}]_{bc}=R^{ab}_{cd}=[R^{(bc)}_{\searrow}]_{ad}.
\end{eqnarray}
Let $\Dsw(R)$ be the matrix algebra generated by $\{R^{(ad)}_{\swarrow}|1\leq a,d\leq m\}$, and similarly let $\Dse(R)$ be the matrix algebra generated by $\{R^{(bc)}_{\searrow}|1\leq b,c\leq m\}$. 
\end{definition}
Let $\DCsw[R]$ be the commutant of $\Dsw(R)$, and similarly, let $\DCse[R]$ be the commutant of $\Dse(R)$. It is clear from definition that $\DCsw[R]$ is exactly the space of $\kappa$ satisfying the first relation in Eq.~\eqref{eq:kappaDC}, and similarly $\DCse[R]$ is the space of $\kappa$ satisfying the second relation in Eq.~\eqref{eq:kappaDC}. Since these two relations are equivalent for involutive $R$-matrices, we have $\DCsw[R]=\DCse[R]$, and we denote them simply by $\DC[R]$, called the diagonal commutant of $R$.

\begin{theorem}\label{thm:directsum-decomposition}
Let $R$ be any unitary involutive $R$-matrix, and let $V$ be the internal space.  Then $V$ has a direct sum decomposition
\begin{equation}\label{eq:directsumdecomposeV}
    V=\bigoplus_{\psi\in\T} V^{(\psi)},
\end{equation}
where $\T$ is some finite indexing set, such that \\
(1) each direct summand $V^{(\psi)}$ is an irrep of both $\Dsw[R]$ and $\Dse[R]$~(i.e., has no proper invariant subspace under the action of $\Dsw[R]$, and similarly for $\Dse[R]$);\\
(2) for any $\psi,\varphi\in\T$, we have
\begin{equation}
    R(V^{(\psi)}\otimes V^{(\varphi)})\equiv R^{(\psi\varphi)}(V^{(\psi)}\otimes V^{(\varphi)})=V^{(\varphi)}\otimes V^{(\psi)},
\end{equation}
where $R^{(\psi\varphi)}$ is the restriction of $R$ to $V^{(\psi)}\otimes V^{(\varphi)}$;\\
(3) For any paraparticle system based on $R$, let $\FA$ be the algebra of all local observables. Then
each $\psi\in\T$ is a simple particle type with respect to $\FA$, and $R^{(\psi\varphi)}$ is the mutual $R$-matrix between $\psi$ and $\varphi$. 
\end{theorem}
\begin{proof}
We first prove (1). 
If $\Dsw[R]$ has no proper invariant subspace, then we can simply take $\T=\{\psi\}$ and the conclusion of (1) and (2) follows trivially. Otherwise,  
 $\Dsw[R]$ has a proper invariant subspace $V_1$. 
Saying that $V_1$ is invariant under $\Dsw[R]$ is equivalent to saying that 
$R(V_1\otimes V)\subseteq V\otimes V_1$,
and, since $R$ is unitary, this is equivalent to $R(V_1\otimes V)=V\otimes V_1$.
Taking orthogonal complement and using unitarity of $R$, we get $R(V_1^\perp\otimes V)=V\otimes V_1^\perp$, which is equivalent to saying that $V^\perp_1$ is also invariant under $\Dsw[R]$.
We thus have $V=V_1\oplus V_1^\perp$. %
We can then use the same method recursively on  $V_1$ and $V_1^\perp$ to decompose each of them into a direct sum of subspaces, in which each summand is invariant under $\Dsw[R]$ and has no proper invariant subspaces,  eventually obtaining Eq.~\eqref{eq:directsumdecomposeV} satisfying (1). 

To prove (2), for each $\psi\in\T$,
let $p^{(\psi)}$ be the projector onto $V^{(\psi)}$, %
and we have $p^{(\psi)}\in\DC[R]$.  
Since $(\Dsw[R])'=\DC[R]=(\Dse[R])'$, $V^{(\psi)}$ is also invariant under $\Dse[R]$.  Then for any $\psi,\varphi\in\T$, we have  
$R(V^{(\psi)}\otimes V^{(\varphi)})\subseteq V^{(\varphi)}\otimes V^{(\psi)}$; but since $R$ is unitary, this is equivalent to  $R(V^{(\psi)}\otimes V^{(\varphi)})= V^{(\varphi)}\otimes V^{(\psi)}$,
and we use $R^{(\psi\varphi)}$ to denote the restriction of $R$ to $V^{(\psi)}\otimes V^{(\varphi)}$, hence proving (2). 

To prove (3), we insert the direct sum decomposition of $V$ and $R$ into the second quantization algebra in Eq.~\eqref{eq:fundamental_Rcommu}, and we obtain Eq.~\eqref{eq:fundamental_Rcommu-rela}, where each $\psi\in\T$ becomes a particle type in the system with its own creation/annihilation operators $\hat{\psi}^{\pm}_{i,a}$. Consider the fundamental bialgebra hidden symmetry $\calA_R$ defined in App.~\ref{app:basichopfsymmetry}, and we also use $\psi$ to denote its defining representation on $\psi$. From Eq.~\eqref{eq:BA-adjoint-action} we see that saying that $V^{(\psi)}$ is irreducible under $\Dse[R]$ is equivalent to saying that $\psi$ is an irrep of 
$\calA_R$. Applying Corollary~\ref{cor:simplemodule-simpletype} to the $\calA_R$-invariant local observable algebra $\FA$, we conclude that $\psi$ is a simple particle type. 
\end{proof}
\begin{corollary}
We say that $R$ is simple if both $\Dsw(R)$ and  $\Dse(R)$ are simple as associative algebras, i.e., if both are equal to $M_m(\C)$, the algebra of $m\times m$
matrices. For a unitary $R$-matrix, simplicity is equivalent to indecomposability, and also equivalent to $\DC[R]=\C\mathds{1}$. 

Hence from Thm.~\ref{thm:eqv-simple-indist} and Thm.~\ref{thm:directsum-decomposition}, for a paraparticle system generated by a single particle type $\psi$, local indistinguishability of $\psi$, simplicity of $\psi$, and simplicity of its $R$-matrix are all equivalent. 
\end{corollary}
A practical way to decide if a unitary $R$-matrix is simple is to directly compute $\DC[R]$. %
Examples of unitary simple $R$-matrices include Exs.~\ref{ex:1m},~\ref{ex:setth}, and \ref{ex:setth-ext} in Tab.~\ref{tab:Hilbert_series}, where a direct computation gives $\DC[R]=\C \mathds{1}$.  %
By contrast, the $R$-matrix in Ex.~\ref{ex:Green}  can be decomposed as a direct sum of $m$ fermions with mutual bosonic statistics.  

We emphasize again that the definition of simplicity depends on the algebra of local observables $\FA$, and an $R$-paraparticle with a non-simple $R$-matrix may become simple if we introduce some extra hidden symmetry $\calA$. In this case  we can use the following theorem to decide if the $R$-paraparticle is simple with respect to the subalgebra of $\calA$-invariant observables: 
\begin{theorem}{}
Let $\calA$ be any bialgebra hidden symmetry
of a paraparticle system, and let $\mathfrak{A}=\mathfrak{A}_0^\calA$ be the algebra of $\calA$-invariant local observables. 
Then a paraparticle $\psi$ is simple with respect to $\FA$ if and only if $\DC[R]\cap \End_\calA[\psi]
=\C \mathds{1}$. 
\end{theorem}
\begin{proof}
To prove the if direction, note that for any local observable $\hat{O}$, we always have $\kappa^O\in \DC[R]$, see the remark below Lemma~\ref{lemma:kappainDCR}. If $\hat{O}$ is $\calA$-invariant, then Lemma~\ref{lemma:kappainDCR} shows that  $\kappa^O\in \End_\calA[\psi]$.
Therefore 
$\kappa^O\in \DC[R]\cap \End_\calA[\psi]
=\C \mathds{1}$. Since this applies to any $\calA$-invariant observable $\hat{O}$,  Prop.~\ref{prop:simplen=1case} implies that $\psi$ is simple with respect to $\FA$.  

For the only if direction, 
assume there exists some nontrivial $\kappa\in 
\DC[R]\cap \End_\calA[\psi]$
not proportional to identity matrix. Then the operator $\hat{e}^\kappa_{ij}$ is local since $\kappa\in\DC[R]$~[see the argument below Eq.~\eqref{eq:def_e_ijkappa}], and is $\calA$-invariant, since $[\kappa,\zetapsi(z)]=0$ for any $z\in\calA$. 
Then a straightforward computation shows that $ \braket{0|\hat{\psi}^-_{i,a}\hat{e}^\kappa_{ij}\hat{\psi}^+_{j,b}|0}=\kappa_{ab}$, which is not proportional to $\delta_{ab}$, hence by Proposition~\ref{prop:simplen=1case}, $\psi$ is not simple with respect to $\FA$. 
\end{proof}

For example, for the $R$-matrix in Ex.~\ref{ex:Green}, the associated $R$-paraparticle becomes simple if we add a suitable hidden group symmetry. The choice of such hidden symmetry may not be unique, and some simple examples can be found in App.G.2 of Ref.~\cite{wang2025secret}. For example, when $m=3$, we can demand local observables to be invariant under an $A_4$ hidden symmetry~(i.e., take $\calA=\C[A_4]$ to be the group algebra), where the paraparticle transforms in the 3-dimensional irrep of $A_4$~[see Eq.~(G19) of Ref.~\cite{wang2025secret}], and hence  $\End_\calA[\psi]=\C\mathds{1}$ and $\psi$ is simple with respect to the $\calA$-invariant local observable algebra.

\subsection{The categorical description of $R$-paraparticles}\label{sec:cat-description}
In this section we show that the generalized hidden symmetry viewpoint directly leads to a categorical description of the universal topological properties of $R$-parastatistics---
those properties that are stable against smooth deformations~(local perturbations), %
such as particle fusion, exchange statistics, topological twist factor, and Frobenius-Schur indicator. Although these topological properties are also encoded in the second quantization formalism, the latter contains too much information about microscopic details of the system, making it inconvenient for describing universal properties. By contrast, the categorical formulation focuses on the universal properties
and is more convenient in situations where microscopic details are not important.  

The categorical formulation also highlights the analogy to the theory of anyons, and helps us understand the relation to DHR no-go theorem~\cite{doplicher1971local,doplicher1974local}. It is by now a  well-established fact that the universal topological properties of anyons that emerge in 2D topological phases are described by unitary braided fusion categories~\cite{kitaev2006anyons,Kong2020Algebraichigher,simon2020topological_protobook}, a mathematical structure that describes the fusion, braiding, and antiparticles of an anyon model. In this section we will see that $R$-paraparticles are naturally described by symmetric tensor categories~\cite{TenCat_EGNO} that are not necessarily rigid, while by contrast, under the assumptions of DHR no-go theorem, superselection sectors are always described by a rigid symmetric tensor category. 

We begin by sketching the main results. 
Let $\calA$ be the fundamental bialgebra symmetry of this system, as defined in Thm.~\ref{thm:basicHA}, 
and let $\FA=\mathfrak{F}^\calA$ be the algebra of $\calA$-invariant local observables. Below we show that the symmetric tensor category $\calC=\Rep(\calA)$ controls all universal topological properties of the $R$-paraparticles in the following way:
\begin{enumerate}
\item A simple object in $\calC$ corresponds to a simple particle type  in the system, satisfying the local indistinguishability criterion;
\item For simple particle types $\psi,\varphi\in \calC$, their tensor product representation $\psi\otimes\varphi$ describes the fusion of particles $\psi$ and $\varphi$;
\item Exchange symmetry is described by the  $R$-matrix; %
\item For a simple $\psi\in\calC$, if the dual representation $\bar{\psi}\in \calC$ exists, then $\bar \psi$ corresponds to the antiparticle of $\psi$.
\end{enumerate}
The tensor  category $\calC$ is called rigid if the dual representation $\bar \psi$ exists for all $\psi\in \calC$, and this happens whenever $\calA$ is a Hopf algebra. Physically, this means that every particle type in the system has a corresponding antiparticle type, and this happens for all examples in Tab.~\ref{tab:Hilbert_series} except Ex.~\ref{ex:1m}. For $R=-\mathds{1}$ in Ex.~\ref{ex:1m}, $\calA$ is only a bialgebra, and the corresponding $R$-paraparticle does not have an antiparticle, as we have already seen in Sec.~\ref{sec:pair-creation-AP}. 

Since the construction of the fundamental hidden symmetry $\calA$ from a given $R$-matrix is a bit complicated and implicit in general~(see App.~\ref{app:basichopfsymmetry}), it turns out convenient in practice to have a more explicit description of the category $\calC$ directly in terms of the $R$-matrix.  
To begin, consider an arbitrary unitary involutive $R$-matrix $\mathbf{R}$, which can be either simple or composite in the sense of  Definition~\ref{def:simpleR-matrix}, and let $\Psi$ be the associated particle type, with creation and annihilation operators denoted by $\hat{\Psi}_{i,A}^\pm$ satisfying the CRs in Eq.~\eqref{eq:fundamental_Rcommu-rela-full}. 
Below we describe a procedure to construct a symmetric tensor category $\calC$ from this given $\mathbf{R}$, which involves recursively constructing all other particle types in $\calC$ by fusing known particle types and decomposing the fusion product, and we call $\Psi$ a generating object of $\calC$. This procedure is analogous to constructing new representations of a group or Lie algebra by taking tensor product of known representations and decomposing the result into irreducibles, e.g., constructing all irreps of $SU(2)$ from the spin-$1/2$ representation. We begin by defining objects and morphisms of $\calC$. 
\begin{definition}\label{def:object}{(Object)}
Any object $\psi\in\calC$ is labeled by a four-index tensor %
\begin{equation}\label{eq:4-index-tensor_obj}	%
\mathbf{R}^{Bb}_{aA}=\adjincludegraphics[height=7ex,valign=c]{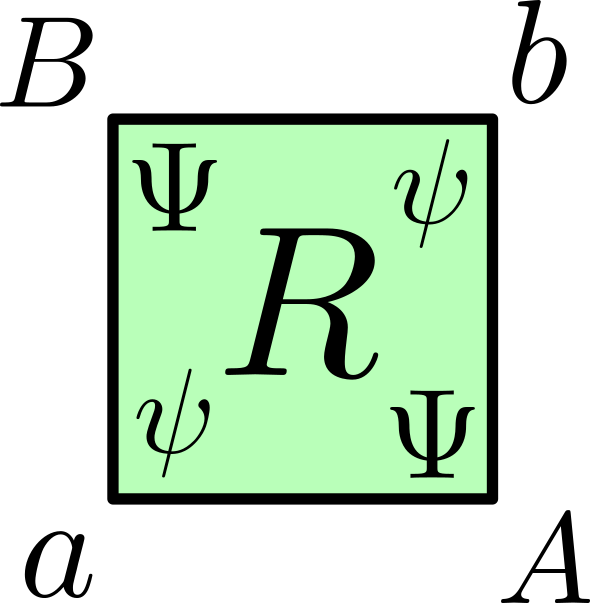}, 
\end{equation}
describing the exchange statistics of $\psi$ with the generating object $\Psi$: 
\begin{eqnarray}\label{eq:definingCR_object}
\hat{\psi}^+_{i,a}\hat{\Psi}^+_{j,A}=\sum_{B,b} \hat{\Psi}^+_{j,B}  \hat{\psi}^+_{i,b}\mathbf{R}^{Bb}_{aA}.%
\end{eqnarray}
\end{definition}

\begin{definition}\label{def:morphisms}{(Morphisms)}
Let $\psi,\varphi$ be objects.  A morphism $f: \psi\to\varphi$ is a $m_\varphi\times m_\psi$-dimensional matrix, denoted by an arrow $f_{ab}=\adjincludegraphics[height=5ex,valign=c]{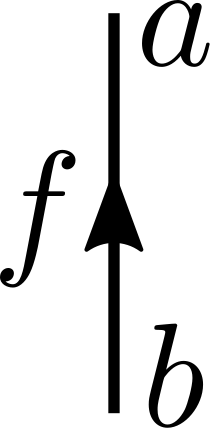}$, satisfying 
\begin{equation}\label{eq:def_morphisms}
	\adjincludegraphics[height=9ex,valign=c]{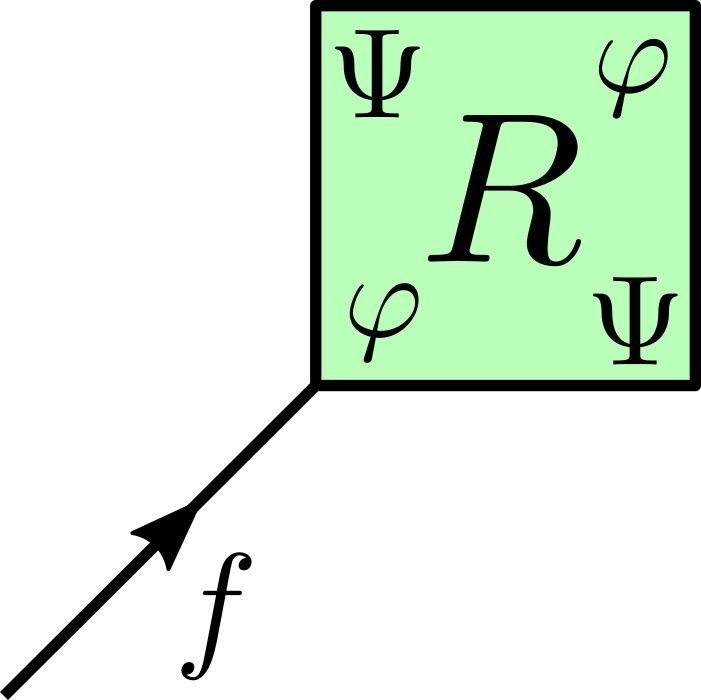}~~=~~
    \adjincludegraphics[height=9ex,valign=c]{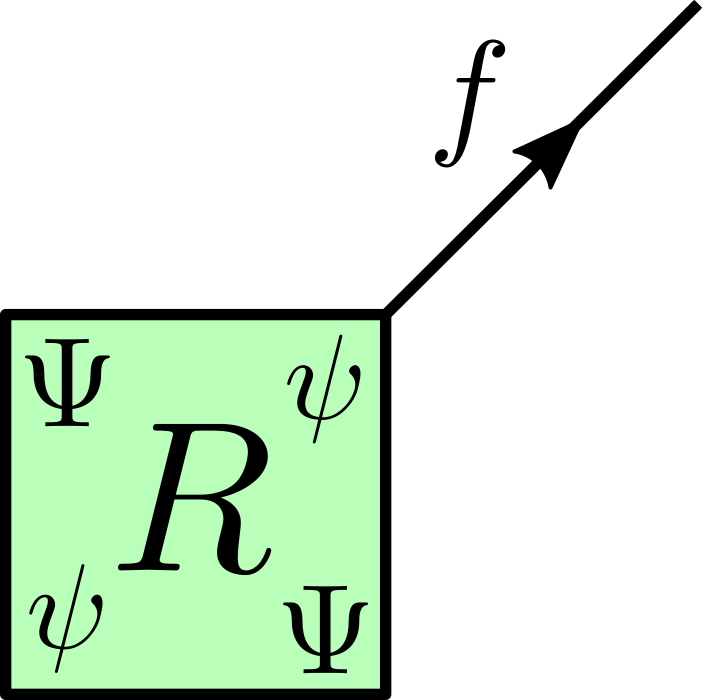}~~.
\end{equation}
It is clear that all such morphisms form a vector space, which we denote by $\Hom_\calC(\psi,\varphi)$. 
For each object $\psi$, the $m_\psi\times m_\psi$-dimensional identity matrix is a morphism, called the identity morphism, denoted as $\id_{\psi}\in\End_{\calC}(\psi)\equiv\Hom_\calC(\psi,\psi)$.
The morphism $f$ is called a monomorphism/epimorphism/isomorphism if the matrix $f_{ab}$ is injective/surjective/bijective, respectively. Two objects are isomorphic if there is an isomorphism $f\in \Hom_\calC(\psi,\varphi)$. A particle type is an isomorphism class of objects in $\calC$. 
\end{definition}

\begin{remark}
A particle type $\psi$ defines a superselection sector of quantum states $S_\psi=\{\hat{O}\hat{\psi}^+_{i}\ket{0}| i\in\Lambda, \hat{O}\in \FA\}$, where $\Lambda$ is the underlying lattice. 
Here we omit the internal index $a$ in $\hat{O}\hat{\psi}^+_{i}\ket{0}$, to indicate that it denotes the space of quantum states 
spanned by $\{\hat{O}\hat{\psi}^+_{i,a}\ket{0}\}_{a=1}^m$,
which are locally indistinguishable if $\psi$ is simple, %
and is hence treated as a whole. 

For example, if all particles in $\calC$ have antiparticles, then the superselection sector $S_\psi$ includes quantum states like $\hat{\psi}^+_{i}\ket{0},\hat{e}^{+}_{jk}\hat{\psi}^+_{i}\ket{0},\hat{e}^{+}_{jk}\hat{e}^{+}_{lp}\hat{\psi}^+_{i}\ket{0},\ldots$, where $\hat{e}^+_{jk}$ is the pair-creation operator defined in Eq.~\eqref{eq:e_pm_ab-psipsib}. By construction, $S_\psi$ is invariant under any local unitary transformations of the form $e^{i t\hat{O}}$ for any local observable $\hat{O}\in\FA$, hence the name ``superselection sector''. This connects our definition of  particle type to the %
literature of topologically-ordered systems~\cite{kitaev2006anyons,kong2014braided,kong2022invitation}, and algebraic quantum field theories~\cite{doplicher1971local,doplicher1974local,haag2012book}, i.e. a particle type is a superselection sector of quantum states invariant under local operations~(equivalently,  a module over the local observable algebra $\FA$). 
\end{remark}

\begin{definition}\label{def:direct_sum_cat}{(Subobjects and direct sum)}
An object $\varphi\in\calC$ is called a subobject of $\Phi\in\calC$ if there is a monomorphism $\iota\in \Hom_\calC(\varphi,\Phi)$. 
Let $\T$ be any finite set of objects. A direct sum of objects in $\T$ is an object $\Phi$ equipped with epimorphisms $\pi_\psi\in\Hom(\Phi,\psi)$ and monomorphisms $\iota_\psi\in\Hom(\psi,\Phi)$ for each $\psi\in \T$, such that 
\begin{eqnarray}\label{eq:direct_sum_def}
\pi_\varphi\circ \iota_\psi=\delta_{\psi\varphi}\id_{\psi},\quad \sum_{\psi\in\mathcal{T}}\iota_\psi\circ \pi_\varphi =\id_\Psi.
\end{eqnarray}
In this case, each $\psi\in \T$ is a subobject of $\Psi$.  %
\end{definition}
It is not hard to show that the explicit description of direct sum in Sec.~\ref{sec:mutual_para} agrees with this categorical definition. 

\begin{definition}\label{def:TensorProduct}{(Tensor product)}
Given any objects $\psi,\varphi\in\calC$, we define their tensor product $\psi\otimes \varphi$ as the object labeled by the four-index tensor
\begin{eqnarray}\label{eq:def_TensorProductObj}
	\adjincludegraphics[height=6ex,valign=c]{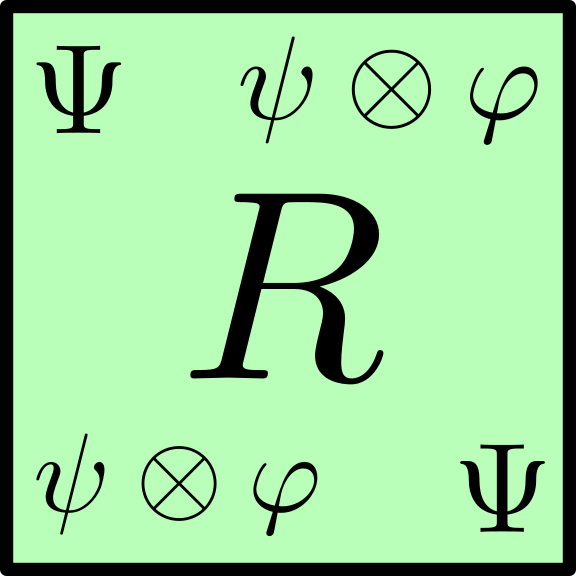}~&=&~
    \adjincludegraphics[height=9ex,valign=c]{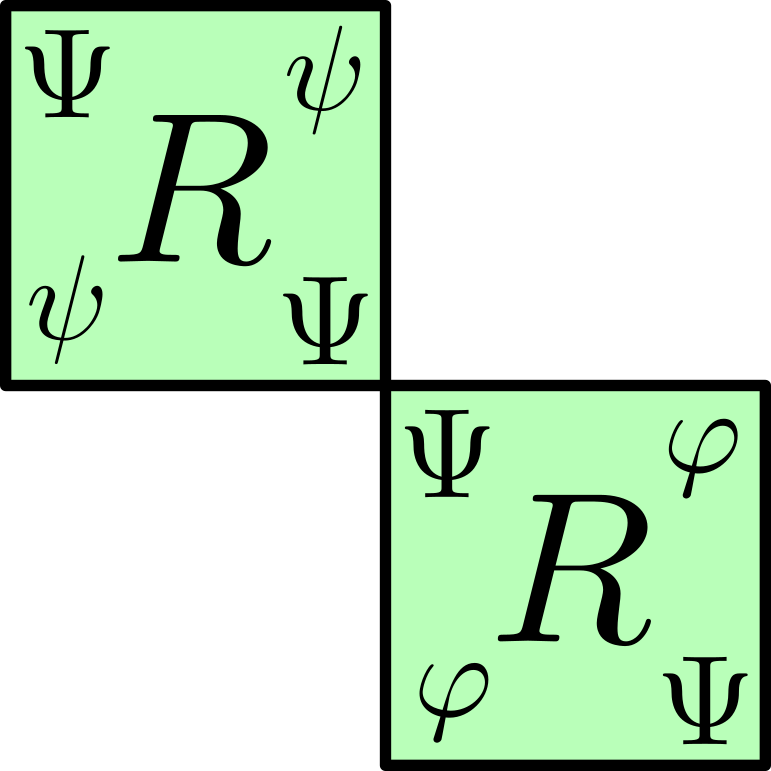}~,%
\end{eqnarray}
describing the exchange statistics of $\hat{\psi}^+_a\hat{\varphi}^+_b$ with the generating object $\Psi$.  
\end{definition}
The tensor product defined this way is associative
\begin{equation}\label{eq:trivial_associator}
(\psi\otimes\varphi)\otimes \rho=\psi\otimes(\varphi\otimes \rho),
\end{equation}
since it is derived from the associative operator product. 
For convenience, we also define the unit object, denoted simply by $1$, as the object labeled by the tensor
\begin{equation}\label{eq:unitObj}
	\adjincludegraphics[height=5ex,valign=c]{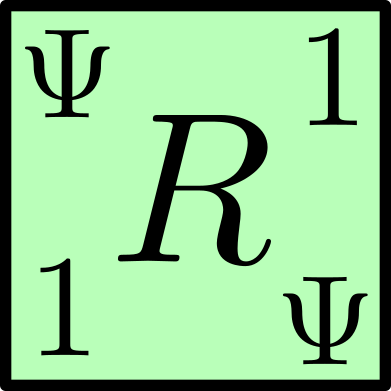}~~=~~
	\adjincludegraphics[height=5ex,valign=c]{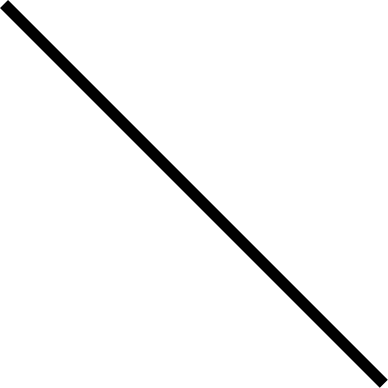}.
\end{equation}
It is then apparent that $1\otimes \psi\cong\psi\otimes 1\cong\psi$ for any object $\psi$, giving the category $\calC$ the structure of a monoidal category. 

As we mentioned before, starting from a given generating object $\Psi$, we can 
recursively construct all particle types in $\calC$ by fusing existing particle types and decomposing the fusion product. The fusion morphisms play an important role in this procedure, which generalizes Clebsch-Gordan coefficients in Lie algebra representation theory:
\begin{definition}(Fusion morphisms)
Let $\psi,\varphi,\gamma$ be simple objects. By specializing the definition of morphisms in Eq.~\eqref{eq:def_morphisms} to the tensor product object defined in Eq.~\eqref{eq:def_TensorProductObj}, we see that $\Hom(\gamma,\psi\otimes \varphi)$~(called the fusion space) is the linear space of all linear maps from $\gamma$ to $\psi\otimes \varphi$ satisfying 
\begin{equation}\label{eq:fusionmorphism-decomposition}
	\adjincludegraphics[height=11ex,valign=c]{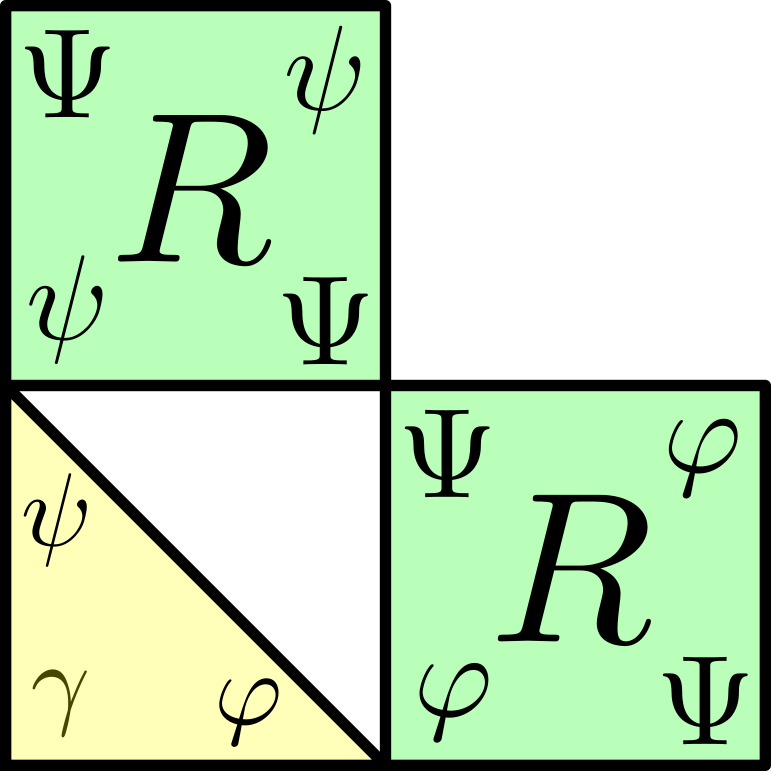}~~=~~
    \adjincludegraphics[height=11ex,valign=c]{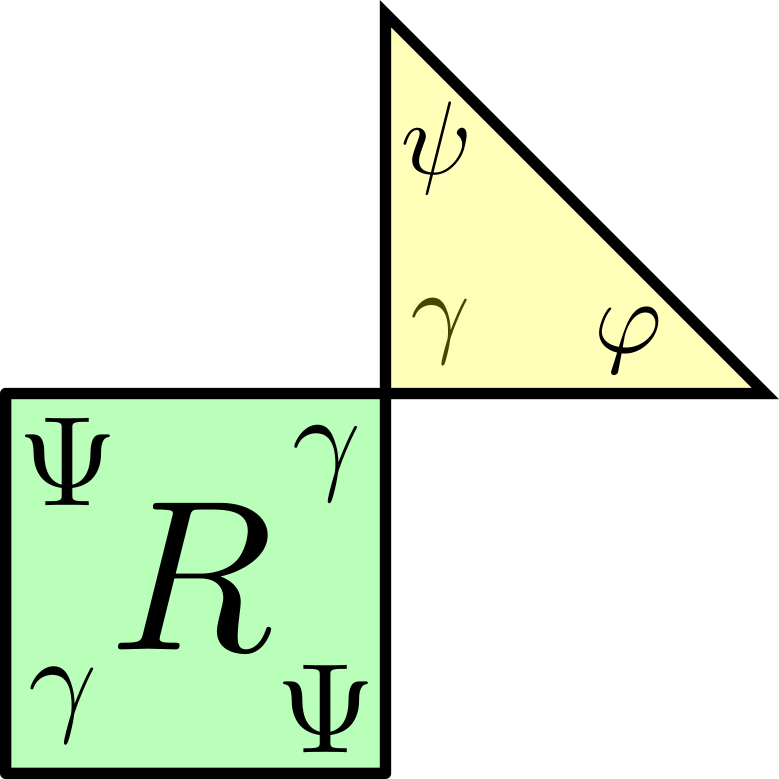}~~. 
\end{equation}
\end{definition}
The space $\Hom(\psi\otimes \varphi,\gamma)$ is defined similarly. For convenience, we choose a basis for the space $\Hom(\gamma,\psi\otimes \varphi)$ and use the corresponding dual basis for $\Hom(\psi\otimes \varphi,\gamma)$, i.e., we have
\begin{equation}\label{eq:orthonormal_fusion_basis}
	\sum_\gamma\adjincludegraphics[height=8ex,valign=c]{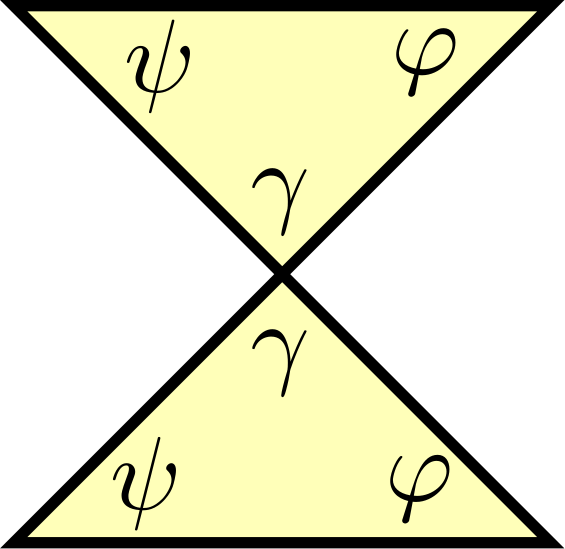}~=~
    \adjincludegraphics[height=8ex,valign=c]{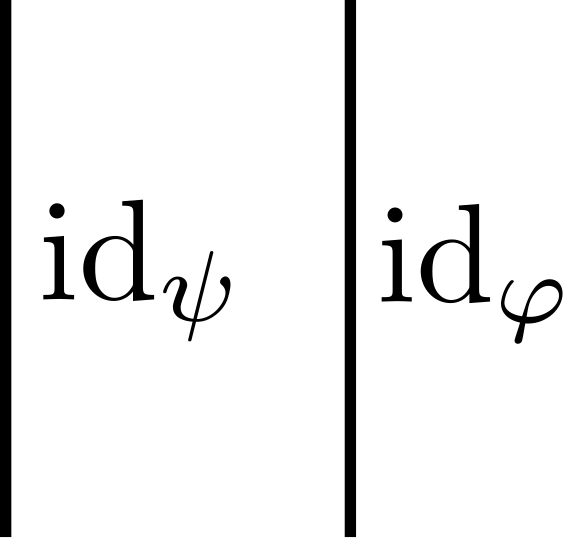}~,
    \quad
	\adjincludegraphics[height=8ex,valign=c]{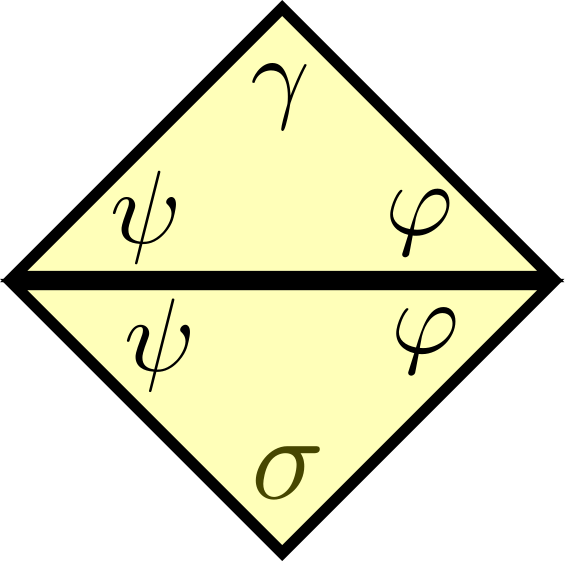}=
    \delta_{\gamma\sigma}\adjincludegraphics[height=8ex,valign=c]{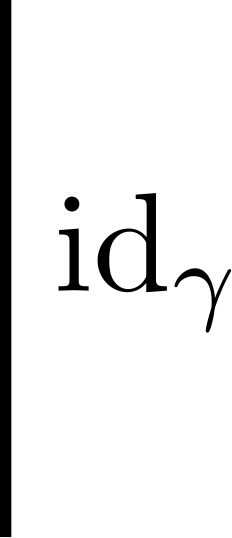}~, 
\end{equation}
where $\gamma$ and $\sigma$ runs over all isomorphic classes of subobjects of $\psi\otimes \varphi$. 
Note that if $R$ is unitary, then the dual basis elements can always be obtained as the Hermitian conjugate of the corresponding basis elements, and Eq.~\eqref{eq:orthonormal_fusion_basis} simply expresses the orthonormality of the basis. 
According to Definition~\ref{def:direct_sum_cat}, the fusion morphisms decompose the tensor product $\psi\otimes\varphi$ into direct sum of simple objects. Note that Eq.~\eqref{eq:fusionmorphism-decomposition} gives a method to construct new particle types $\gamma$ in the system out of existing particle types $\psi,\varphi$ by decomposing the tensor product $\psi\otimes\varphi$, and doing this recursively we can in principle construct all distinct particle types in the system.   

\begin{definition}(Braiding morphisms)
For any objects $\psi, \varphi$, the mutual $R$-matrix satisfies
\begin{equation}\label{eq:def_braiding_morphism}
	\adjincludegraphics[height=12ex,valign=c]{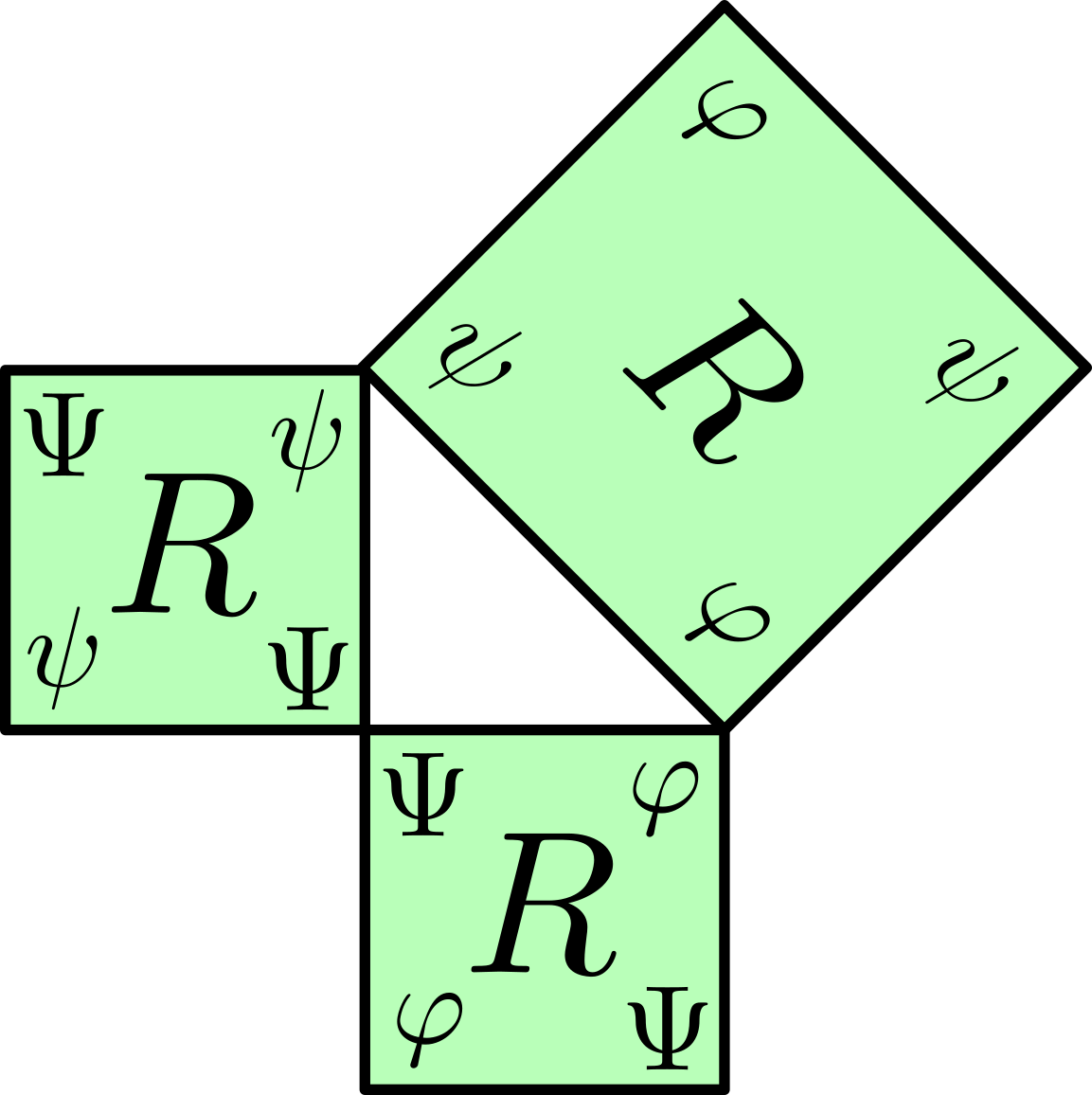}~=~
    \adjincludegraphics[height=12ex,valign=c]{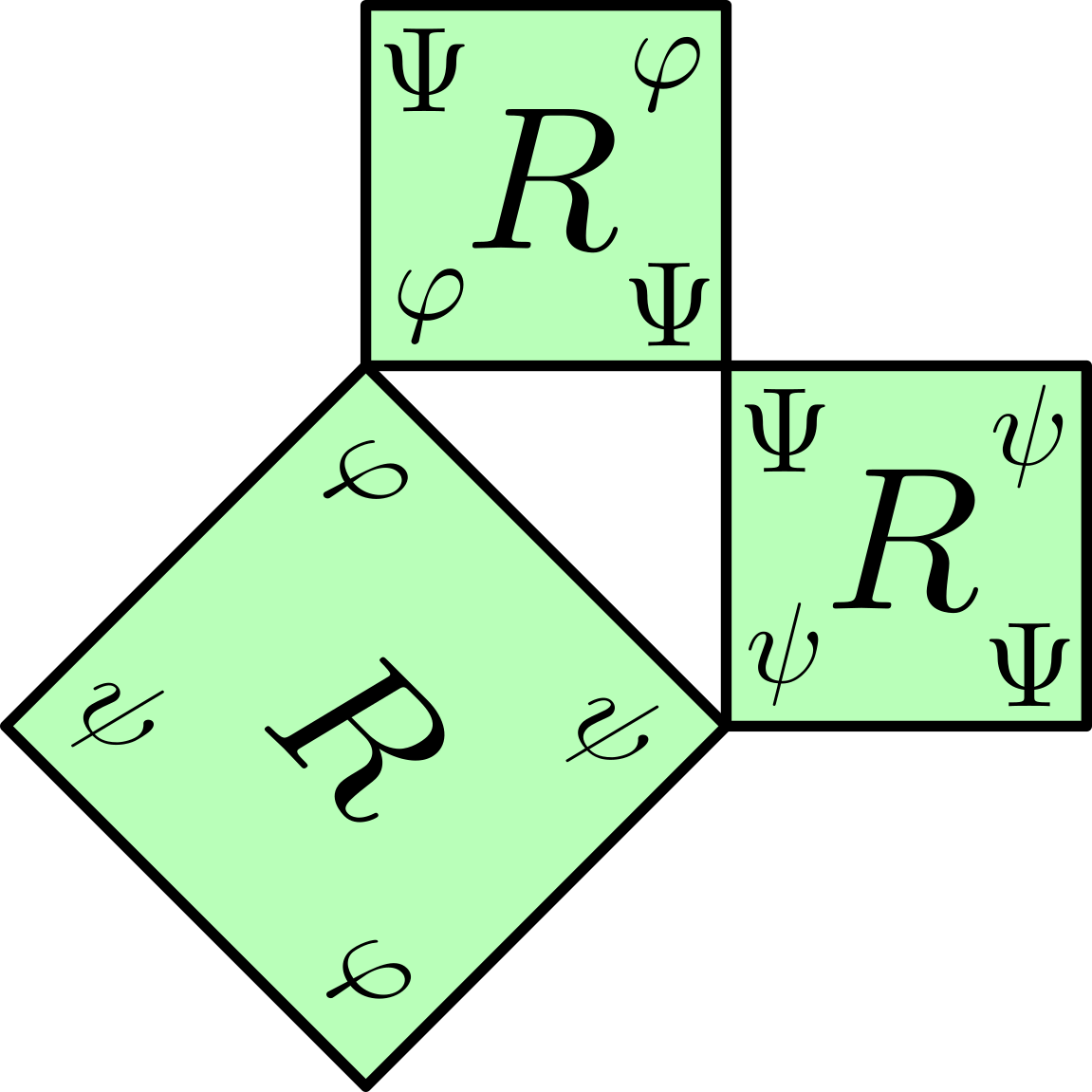}~, 	
\end{equation}
which is the YBE for mutual $R$-matrices in Eq.~\eqref{eq:YBEgraphical-rela},  
hence defines an isomorphism $R^{(\psi\varphi)}\in \Hom(\psi\otimes \varphi,\varphi\otimes \psi)$, called the braiding morphism. Since the mutual $R$-matrices satisfy $R^{(\psi\varphi)}R^{(\varphi\psi)}=\mathds{1}$, 
they define a symmetric braiding on the monoidal category $\mathcal{C}$, making it a 
$\C$-linear symmetric monoidal category. 
We note that Eq.~\eqref{eq:fusionmorphism-decomposition} along with Eq.~\eqref{eq:def_braiding_morphism} gives a method to construct new $R$-matrices out of existing ones, 
as one can take $\varphi=\psi=\gamma$ to be the new particle type generated by the aforementioned recursive algorithm~[via Eq.~\eqref{eq:fusionmorphism-decomposition}]. 
\end{definition}
\begin{remark}
For readers' convenience, we give the definition of the $F$ and $R$ symbols which are widely used in physics literature~\cite{levin2005string,kitaev2006anyons,simon2020topological_protobook}. These data can all be extracted from the fusion morphisms. The $F$-symbol is defined by \begin{equation}\label{def:Fsymbol}
	\adjincludegraphics[height=8ex,valign=c]{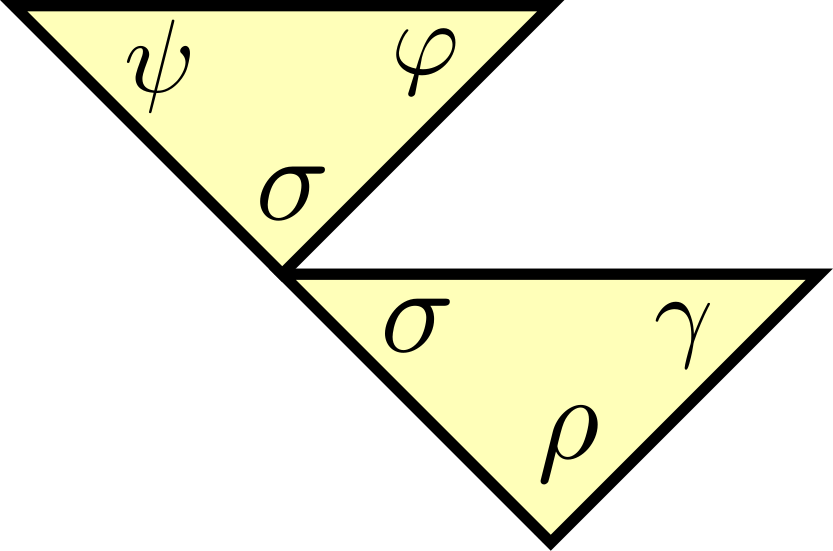}=
    \sum_\tau [F^{\psi\varphi\gamma}_\rho]_{\sigma\tau}\adjincludegraphics[height=8ex,valign=c]{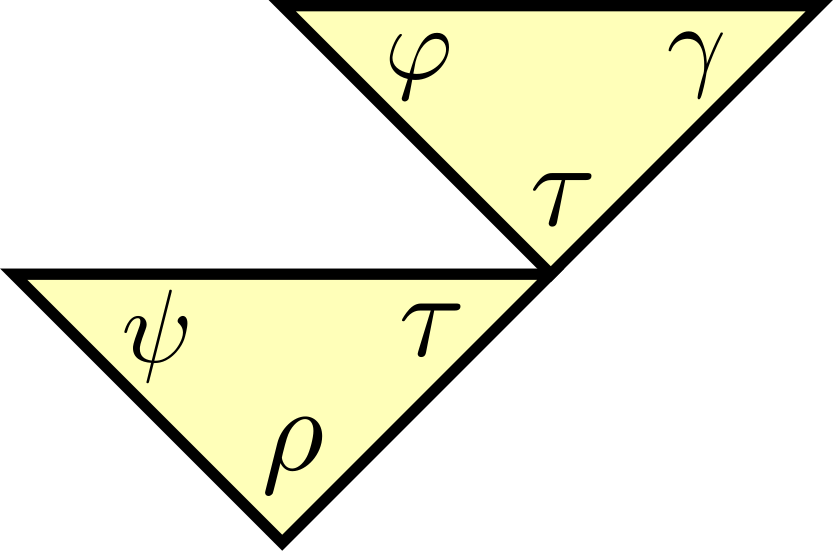}~~, 
\end{equation}
which is analogous to the $6j$-symbol in Lie algebra representation theory. The braiding $R$-symbol is related to the mutual $R$-matrix via
\begin{equation}\label{def:Rsymbol}
	\adjincludegraphics[width=7ex,valign=c]{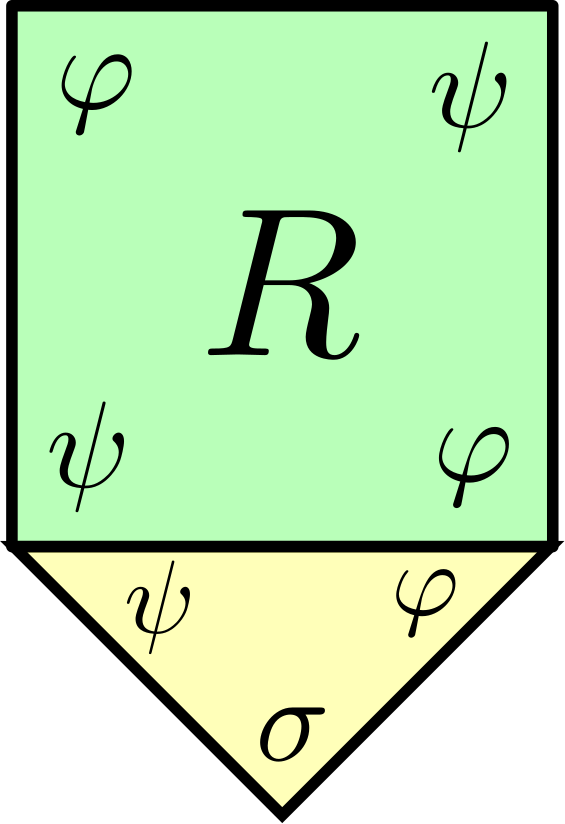}
    =
    R^{\psi\varphi}_\sigma~ 
    \adjincludegraphics[width=7ex,valign=c]{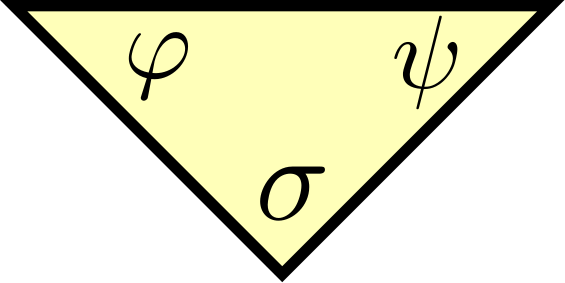}~.	
\end{equation}
One can verify that they satisfy the pentagon and hexagon equations of symmetric tensor categories. 
Since we do not use them in this paper, we omit the detailed derivations. 
\end{remark}

\begin{definition}(Dual objects)
For any object $\psi\in\calC$,   
its dual object is an object $\bar\psi\in\calC$ equipped with
$\alpha^{(\bar{\psi}\psi)}$ and $\alpha^{\prime(\psi\bar{\psi})}$
satisfying Eqs.~\eqref{eq:evcoevdef} and \eqref{eq:snakelemma} in Sec.~\ref{sec:pair-creation-AP}. Note that Eq.~\eqref{eq:evcoevdef} is equivalent to saying that $\alpha^{(\bar{\psi}\psi)}\in\Hom_\calC(\bar\psi\otimes\psi,1)$ and $\alpha^{\prime(\psi\bar{\psi})}\in \Hom_\calC(1,\psi\otimes\bar\psi)$, and are called, respectively, evaluation and coevaluation morphisms in category theory. The monoidal category $\calC$ is called rigid if every object in $\calC$ has a dual object.
\end{definition}
This completes the definition and construction of $\calC$ from a given $R$-matrix $\mathbf{R}$. Below we give a few examples from Tab.~\ref{tab:Hilbert_series}.

\begin{example}(Ex.~\ref{ex:Green} in Tab.~\ref{tab:Hilbert_series})
 It is clear from the fundamental CRs~\eqref{eq:fundamental_Rcommu} that this (decomposable) $R$-matrix describes $m$-flavors of fermions that are mutually bosonic. We use $\psi_a$ to denote the simple particle type created by $\hat\psi^+_{j,a}$.   There are $2^m$ simple objects in total, denoted as $\psi^{n_1}_1\psi^{n_2}_2\ldots \psi^{n_m}_m$, with $n_a\in\{0,1\}$ for $a=1,2,\ldots,m$, and the fusion rules are induced by $\psi_a\times\psi_a=1$.  
Hence this system is described by the SFC $\mathcal{C}=\mathrm{sRep}(Z_2)^{\boxtimes m}$, the $m$-fold Deligne tensor product of $\mathrm{sRep}(Z_2)$.
\end{example}

\begin{example}(Infinite non-rigid symmetric tensor category from Ex.~\ref{ex:1m})
    For $R=-\mathds{1}_{m}\otimes \mathds{1}_{m}$, $\mathcal{C}_{R}$ contains infinitely many simple objects.  %
Each simple object is indexed by a non-negative integer $n$, denoted by $\psi^{(n)}$, 
describing the particle type created by the operator
\begin{equation}
\hat{\psi}^{(n)+}_{I,A}=\hat{\psi}^+_{i_1,a_1}\hat{\psi}^+_{i_2,a_2}\ldots \hat{\psi}^+_{i_n,a_n},
\end{equation}
where $I$ collectively labels  the positions %
$i_1,i_2,\ldots,i_n$ that are assumed to be mutually different~(otherwise $\hat{\psi}^{(n)+}_{I,A}$ is zero), and $A$ collectively labels the internal states  %
$a_1,a_2,\ldots,a_n$.
The unit object is $\psi^{(0)}=1$.  %
The fusion rules are generated from the product of creation operators, and in this case we have
\begin{equation}
\hat{\psi}^{(n)+}_{I,A}\hat{\psi}^{(l)+}_{J,B}=\hat{\psi}^{(n+l)+}_{IJ,AB}, \text{ if } I\cap J=0,
\end{equation}
where $IJ$ denotes the concatenation of $I$ and $J$ and similarly for $AB$. This gives the fusion rule
\begin{equation}
\psi^{(n)}\otimes\psi^{(l)}=\psi^{(n+l)}.
\end{equation}
The fusion $F$-symbol is trivial, and the $R$-symbol is given by $R^{n,l}_{n+l}=(-1)^{nl}$. Note that in this category $\psi^{(n)}$ does not have an  antiparticle for $n\geq 1$. 
\end{example}

\begin{example}(Symmetric fusion categories from Ex.~\ref{ex:setth} and Ex.~\ref{ex:setth-ext})
Since both are Hopf algebra $R$-matrices, the algorithm in this section returns  symmetric fusion categories. For  Ex.~\ref{ex:setth}, we get $\mathcal{C}=\mathrm{sRep}(D_4\ltimes Z_2^{\times 3})$ with 16 simple objects, while for Ex.~\ref{ex:setth-ext}, we get $\mathcal{C}=\mathrm{sRep}(A_4\times Z_3)$ with 12 simple objects. These categories are described more explicitly in Ref.~\cite{wang2025secret}. 
\end{example}

\section{Summary of the first part}\label{sec:summary_part1}
In this paper we developed the theoretical foundation for $R$-parastatistics.  
While part of the paper reviews and systematizes earlier results in Ref.~\cite{wang2023para}, the central goal of the present work is to 
develop a general theory of local observables beyond the earlier theory, based on which we investigate 
what aspects of $R$-paraparticles are physically observable, what aspects remain locally hidden, and how this structure leads to physical phenomena beyond ordinary fermions and bosons. 

For the foundational part, we  detailed the basics of the $R$-parastatistics theory based on the second quantization framework proposed in Ref.~\cite{wang2023para}, the construction of the Fock space, how the second quantization algebra leads to the nontrivial exclusion and exchange statistics of the $R$-paraparticles, and how the Lie algebra of the bilinear operators lead to the exact solvability of the free paraparticle Hamiltonians. This provides a solid ground on which the later extensions are formulated. We also proposed the  notion of weak equivalence between $R$-paraparticles, which will play an important role in understanding the condensed matter realizations and the relation to no-go theorems~\cite{doplicher1971local,doplicher1974local}. 

The main extension of the local observable theory proceeds in three directions. First, we introduced local observables capable of distinguishing particle types, constructed from the bilinear local operators $\hat{e}^{\kappa}_{ij}$ in Eq.~\eqref{eq:def_e_ijkappa}. They play a  fundamental role in the discussion of local indistinguishability of paraparticles~(which was not  discussed in depth in Ref.~\cite{wang2023para}), and  naturally leads to the concept of mutual parastatistics--nontrivial exchange statistics between different types of $R$-paraparticles, which can be physically observed in the exchange experiments discussed in Sec.~\ref{sec:observe_parastatistics}.  

Second, we described observables localized near special point defects that can probe otherwise hidden internal indices, thereby giving an operational meaning to the exchange $R$-matrix and an experimental scheme to observe $R$-parastatistics. 
Using this general framework, we discussed two unique physical consequences and applications of $R$-paraparticles--noise-robust secret communication and entanglement generation over  long distance. We also proposed a conjectured upper bound on the maximal capacity of information transfer and entanglement generation of a given type of $R$-paraparticles, which has important implications on the realizability of beyond-SFC $R$-paraparticles in higher dimensional condensed matter systems. 

Third, we introduced local pair-creation and pair-annihilation observables, which leads to the antiparticle theory of $R$-paraparticles. We proved a classification theorem for unitary $R$-matrices that admit local pair-creation, where important physical  concepts such as topological twist factor and Frobenius-Schur indicators emerge naturally, and connect to the topological order literature. Using this framework, we then exactly solved general bilinear paraparticle Hamiltonians with pair creation and annihilation terms, which leads to the $R$-paraparticle analog of Bogoliubov-de Gennes mean-field theory. Several familiar concepts and results, such as particle-hole symmetry, self-dual Majorana particles, and Wick's theorem for Gaussian ground states, also have their $R$-paraparticle analog.
These results will also play a fundamental role in  formulating relativistic quantum field theories of $R$-paraparticles, which we do in the second part of the series. 

Finally, we identified generalized hidden symmetries acting on the internal indices of $R$-paraparticles while preserving the local observable algebra. These symmetries, naturally described by Hopf algebraic structures, provide a systematic way to study universal topological properties of $R$-paraparticles.
Using hidden symmetry as a tool, we formulated a simple condition on the $R$-matrix, under which we prove that the associated type of $R$-paraparticle satisfies the local indistinguishability criterion we introduced back in Sec.~\ref{sec:indistcrit}. The category of representations of the Hopf algebra hidden symmetry, naturally a symmetric tensor category, describes fusion and exchange symmetry of $R$-paraparticles. Such symmetric tensor categories are more general than those encountered in 3D  topological order, in that here antiparticles do not always exist~(failure of rigidity), and there may be infinitely many simple objects.  

The formalism we developed in this paper will provide a solid ground for studying the realization of $R$-paraparticles in nature, which will be the subject of the second part of this series. Specifically, in part II, we will study emergent $R$-paraparticles in condensed matter systems in all spatial dimensions, and formulate relativistic quantum field theories of $R$-paraparticles. We will also clarify the relation to no-go theorems, and discuss several open questions related to $R$-paraparticles.

\acknowledgments
We thank Alexei Kitaev, Xiao-Gang Wen, J. Ignacio Cirac, Meng Cheng, Cenke Xu, Zhi-Hao Zhang, Liang Kong,  Yu-An Chen, Norbert Schuch, Chong Wang, Dominic Else, Sung-Sik Lee, Tao Shi,  Xiaoqi Sun, Nicol\'{a}s Medina S\'{a}nchez, and Borivoje Daki\'{c} for discussions. 
Part of this work was carried out while ZW was at the Max Planck Institute of Quantum Optics, supported by 
the Munich Center for Quantum Science and Technology~(MCQST), funded by the Deutsche Forschungsgemeinschaft~(DFG) under Germany’s Excellence Strategy~(EXC2111-390814868). 
ZW's research at Perimeter Institute is supported in part by the Government of Canada through the Department of
Innovation, Science and Economic Development and by the Province of Ontario through the Ministry of Colleges and Universities.
KH acknowledge support from the National Science Foundation~(PHY-1848304), the Office of Naval Research~(N00014-20-1-2695), and the W. M. Keck
Foundation~(Grant No. 995764). 

\appendix
\section{Involutive $R$-matrices and triangular Hopf algebras}\label{appen:other_Rmat}
\subsection{Constructing $R$-matrices from  triangular Hopf algebras}\label{app:RMatTHA} %
A common way to construct solutions to the constant Yang-Baxter equation~\eqref{eq:YBE} is through quasitriangular Hopf algebras~\cite{drinfeld1986quantum,Majid1990,Majid1995BookFoundationQG,klimyk1997book,kasselQuantumGroups1995}. 
A quasitriangular Hopf algebra $\calA$ is a Hopf algebra equipped with a universal $R$-matrix $\mathcal{R}\in \calA\otimes \calA$ satisfying the algebraic Yang-Baxter equation~(along with some other axioms we omit here)
\begin{equation}\label{eq:algebraicYBE}
\calR_{12}\calR_{13}\calR_{23}=\calR_{23}\calR_{13}\calR_{12},
\end{equation}
where we use the notation introduced in Definition~\ref{def:abstractTNindexing}, and $\calA$ is called triangular~\cite{drinfeld1986quantum,Majid1995BookFoundationQG,etingof1998THAconstruction,etingof2000semisimpleTHAclassification,TenCat_EGNO} if $\calR$ satisfies in addition
\begin{equation}\label{eq:algebraic-involutive}
\calR_{21}\calR=1\otimes 1,
\end{equation}
where $1$ is the unit of $\calA$. %
To construct an $R$-matrix, let $\zetapsi$ be a representation of $\calA$, and define
\begin{equation}\label{eq:RmatFromTHA-0}
R=X[(\zetapsi\otimes \zetapsi)\mathcal{R}],%
\end{equation}
where $X$ is the swap matrix. Then it is straightforward to verify that Eq.~\eqref{eq:algebraic-involutive} and Eq.~\eqref{eq:algebraicYBE} leads to the first and the second line of Eq.~\eqref{eq:YBE}, respectively. This method of constructing $R$-matrices is available in the Mathematica package for $R$-parastatistics, where several examples of triangular Hopf algebras are given. A closely related method to construct involutive $R$-matrices is using a special class of symmetric fusion categories, %
see, e.g., App.G of Ref.~\cite{wang2025secret}, which is equivalent to the above technique due to the relationship between triangular Hopf algebras and symmetric fusion categories~\cite{etingof2000semisimpleTHAclassification,TenCat_EGNO}. 

\subsection{Fundamental Hopf symmetry of an $R$-paraparticle system}\label{app:basichopfsymmetry}
In the following we prove Thm.~\ref{thm:basicHA} on the existence of a basic bialgebra symmetry of an $R$-paraparticle system. The basic idea is simple: we construct the basic bialgebra symmetry  $\calA$ as the algebra generated by $\{v_{pq}|1\leq p,q\leq m\}$, where we define $v_{pq}$ directly via its action on the state space, using Eq.~\eqref{eq:vpqactionstatespace},
with $\adjincludegraphics[height=3ex,valign=c]{Figures/GeneralizedSymmetry-1/kappaOproof-last}$. 
This defines the algebra $\calA$. We then use Eq.~\eqref{def:corep} to define the coalgebra structure.  The tricky part is to show that $\Delta:\calA\to\calA\otimes \calA$ and $\epsilon:\calA\to \C$ are algebra homomorphisms. We begin with the following lemma as a preparation: 
\begin{lemma}\label{lemma:algebrahomext}
Let $S$ be a finite set, $V,W$ be vector spaces~(potentially infinite dimensional), and suppose we are given maps $f:S\to \End(V)$ and $g:S\to \End(W)$. Let $\calA_f$ be the subalgebra of $\End(V)$ generated by the image of $f$,  and similarly for $\calA_g$.  
Suppose the left $\calA_f$-module $V$ and the left $\calA_g$-module $W$ have direct sum decompositions 
\begin{align}
    V\cong \bigoplus_{i\in I}V_i,\quad W\cong \bigoplus_{j\in J}W_j,
\end{align}
for some indexing sets $I$ and $J$, such that for any $j\in J$, there exists some $i\in I$ with a vector space isomorphism $\gamma: W_j\xrightarrow{\sim} V_i$ such that 
\begin{equation}\label{eq:gammafg}
\gamma [g(s).w]=f(s).\gamma(w),\quad \forall s\in S,~w\in W_j. 
\end{equation}
Then there exists a unique algebra homomorphism $\Gamma:\calA_f\to \calA_g$ satisfying
$\Gamma[f(s)]=g(s)$ for all $s\in S$.
\end{lemma}
\begin{proof}
Let $\C\braket{S}$ be the free algebra generated by $S$. Then the map $f:S\to \End(V)$ uniquely extends to an algebra homomorphism $\tilde{f}:\C\braket{S} \to \End(V) $, and similarly  $g:S\to \End(V)$ uniquely extends to an algebra homomorphism $\tilde{g}:\C\braket{S} \to \End(W) $. We then have $\calA_f=\mathrm{im}(\tilde{f})\cong \C\braket{S}/\ker{(\tilde{f})}$, where $\ker{(\tilde{f})}$ is a two-sided ideal of $\C\braket{S}$ containing all relations satisfied by $\{f(s)|s\in S\}$, and similarly  $\calA_g=\mathrm{im}(\tilde{g})\cong \C\braket{S}/\ker{(\tilde{g})}$. Eq.~\eqref{eq:gammafg} is equivalent to 
\begin{equation}\label{eq:gammafg-1}
\gamma [\tilde g(a).w]=\tilde{f}(a).\gamma(w),\quad \forall a\in \C\braket{S},~w\in W_j.
\end{equation}
This then implies that $\tilde{g}(\ker(\tilde{f}))=0$, i.e., $\ker(\tilde{f})\subseteq\ker(\tilde{g})$, hence $\calA_g$ is a quotient algebra of $\calA_f$, and the canonical projection map $\Gamma:\calA_f\to \calA_g$ satisfies
$\Gamma\circ  \tilde{f}=\tilde{g}$, and the conclusion follows by restricting the relation $\Gamma\circ  \tilde{f}=\tilde{g}$ to $S$. Finally, $\Gamma$ is unique since an algebra homomorphism is uniquely determined by its action on generators. 
\end{proof}
We now use this lemma to prove Thm.~\ref{thm:basicHA}. 
\basicHA*
\begin{proof}
Without loss of generality, we assume that the paraparticle system has an infinite number of modes, i.e., the index $i$ in $\hat{\psi}_{i,a}$ takes value in an infinite set. This is a harmless assumption since hidden symmetry only acts on the internal space of paraparticles, %
and a finite paraparticle system can always be embedded as a subsystem of an infinite paraparticle system, such that any hidden symmetry of the latter becomes a hidden symmetry of the former upon restricting to the subsystem. Below we first prove the theorem for a paraparticle system with one single particle type $\psi$ with $R$-matrix $R$~(i.e., the case $|\T|=1$),  and later we generalize to systems with any number of particle types using the trick in Sec.~\ref{sec:mutual_para}, by taking $\psi$ to be the composite particle type $\Psi=\bigoplus_{\varphi\in \T}\varphi$.

We use $V$ to denote the Fock space of the system, and let  $S=\{v_{pq}|1\leq p,q\leq m\}$, where each $v_{pq}$ is an abstract symbol,  %
and we define $f:S\to\End(V)$ via Eq.~\eqref{eq:vpqactionstatespace},
with $f(v_{pq})\equiv \hat{\Theta}(v_{pq})$ and  $\adjincludegraphics[height=3ex,valign=c]{Figures/GeneralizedSymmetry-1/kappaOproof-last}$. We define $\calA\equiv \calA_f$ as the subalgebra of $\End(V)$ generated by $\mathrm{im}(f)$, as in Lemma~\ref{lemma:algebrahomext}.
We then use Eq.~\eqref{def:corep} to define  $g:S\to \End(W)$ with  $W\equiv V\otimes V$ via $g(v_{pq})=%
\sum^m_{r=1} \hat\Theta (v_{pr})\otimes \hat\Theta (v_{rq})$.

Under the action of $\hat{\Theta}({v}_{pq})$, the Fock space $V$ decomposes as
\begin{equation}\label{eq:decomposeV}
V=\bigoplus_{n\geq 0}\bigoplus_{i_1,\ldots,i_n\in I}V_{i_1,\ldots,i_n},
\end{equation}
where $I$ is the infinite indexing set of modes, and each $V_{i_1,\ldots,i_n}$ is an invariant subspace~(not necessarily irreducible), spanned by all states of the form $\ket{\psi}=\hat{\psi}_{i_1,a_1}\ldots \hat{\psi}_{i_n,a_n}\ket{0}$,  with the action of $\hat{\Theta}({v}_{pq}) $ given by Eq.~\eqref{eq:vpqactionstatespace}. 
Eq.~\eqref{eq:decomposeV} induces a decomposition of $W\equiv V\otimes V$ under the action of $g(v_{pq})$, where each invariant submodule is of the form $V_{i_1,\ldots,i_n}\otimes V_{j_1,\ldots,j_k}$.

Since $I$ is infinite, we can choose $j'_1,\ldots,j'_k\in I$ such that $
|\{j'_1,\ldots,j'_k\}|=|\{j_1,\ldots,j_k\}|$ and $\{j'_1,\ldots,j'_k\}\cap\{i_1,\ldots,i_n\}=\emptyset$. We define a linear map 
$\gamma:V_{i_1,\ldots,i_n}\otimes V_{j_1,\ldots,j_k}\xrightarrow{} V_{i_1,\ldots,i_n,j'_1,\ldots,j'_k}$ as follows: for any  
$\ket{\psi}=\hat{\psi}_{i_1,a_1}\ldots \hat{\psi}_{i_n,a_n}\ket{0}\in V_{i_1,\ldots,i_n}$ and any $\ket{\varphi}=\hat{\psi}_{j_1,b_1}\ldots \hat{\psi}_{j_k,b_k}\ket{0}\in V_{j_1,\ldots,j_k}$, we define 
$\gamma(\ket{\psi}\otimes \ket{\varphi})= \ket{\Psi}$, where $\ket{\Psi}=\hat{\psi}_{i_1,a_1}\ldots \hat{\psi}_{i_n,a_n}\hat{\psi}_{j'_1,b_1}\ldots \hat{\psi}_{j'_k,b_k}\ket{0}\in V_{i_1,\ldots,i_n,j'_1,\ldots,j'_k}$. Since $\{j'_1,\ldots,j'_k\}\cap\{i_1,\ldots,i_n\}=\emptyset$, 
from our earlier discussion on the structure of the Fock space in Sec.~\ref{sec:state_space} and exclusion statistics in Sec.~\ref{sec:exclusion_statistics_calc}, we see that $\gamma$ is a vector space isomorphism~(physically, different modes are statistically independent).

From the action of $g(v_{pq})$ on $\ket{\psi}\otimes \ket{\varphi}$ 
\begin{equation}\label{eq:Deltavpqactionstatespace}
g({v}_{pq})(\ket{\psi}\otimes \ket{\varphi})=
\adjincludegraphics[height=9ex,valign=c]{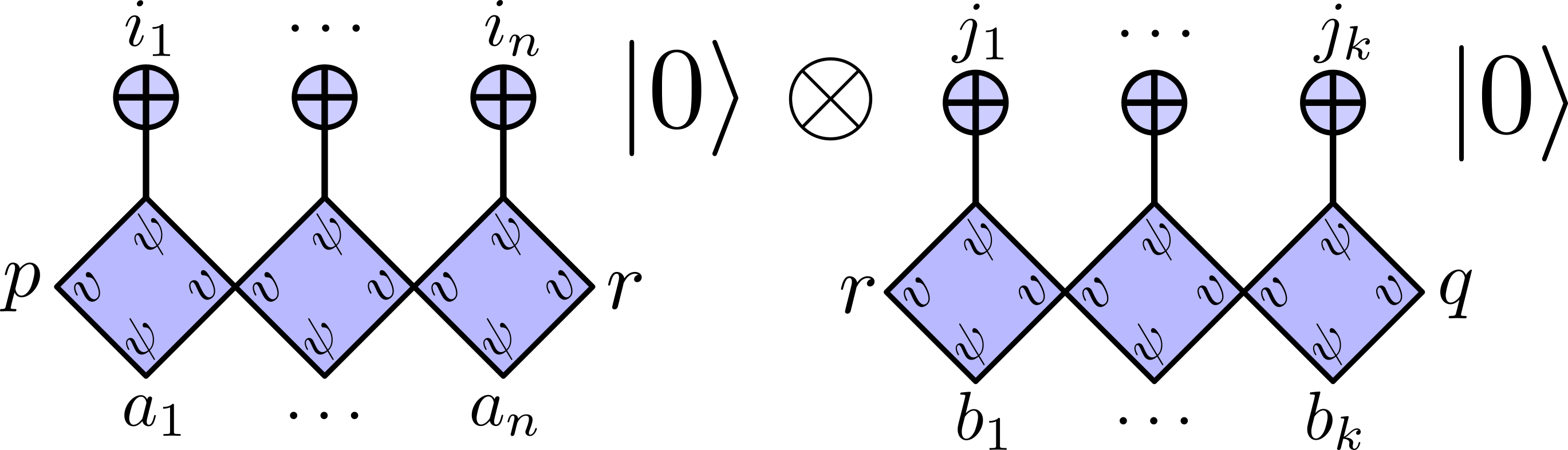},
\end{equation}
and the action of $f(v_{pq})\equiv \hat{\Theta}(v_{pq})$ on $\ket{\Psi}$ 
\begin{equation}\label{eq:vpqactionstatespace2}
\hat{\Theta}({v}_{pq})\ket{\Psi}=\adjincludegraphics[height=9ex,valign=c]{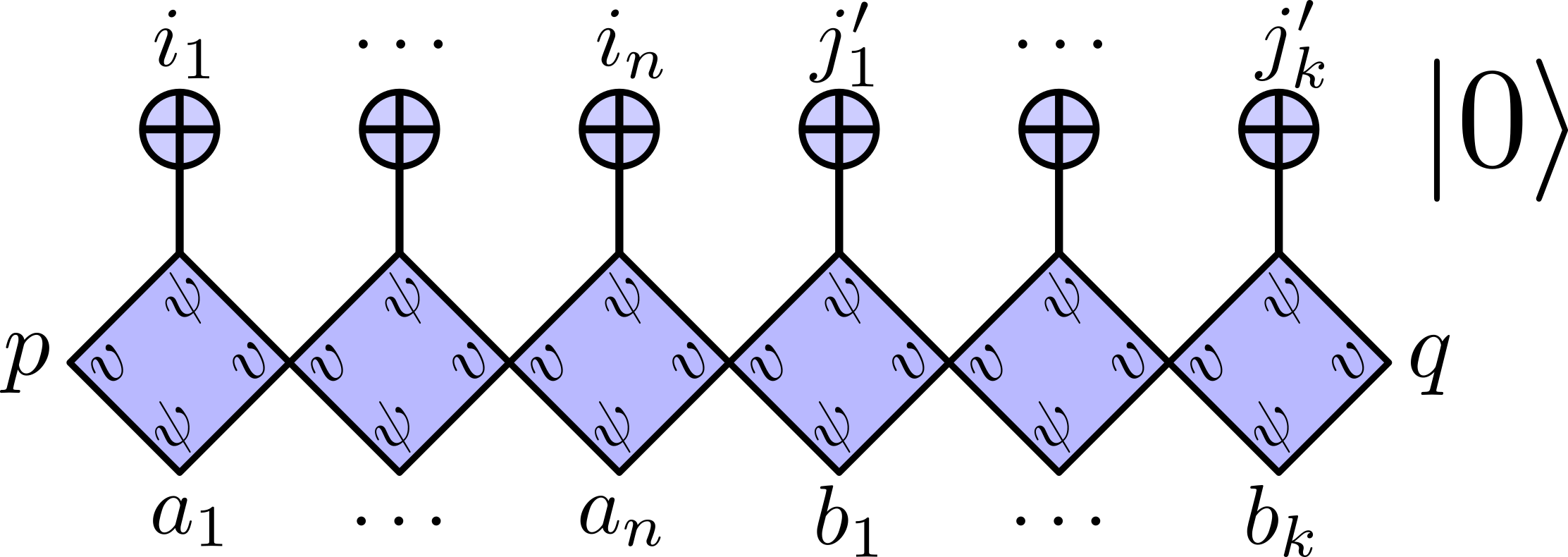},
\end{equation}
we see that $\gamma$ satisfies %
\begin{equation}
\gamma[g({v}_{pq})(\ket{\psi}\otimes \ket{\varphi})]=\hat\Theta({v}_{pq})\gamma(\ket{\psi}\otimes \ket{\varphi}),
\end{equation}
which is Eq.~\eqref{eq:gammafg} in the current situation. 
By Lemma~\ref{lemma:algebrahomext}~ (with $\mathcal{A}_f = \mathcal{A}$), there exists an algebra homomorphism $\Gamma : \mathcal{A} \to \mathcal{A}_g$ satisfying $\Gamma[\hat\Theta(v_{pq})] = g(v_{pq}) = \sum_r \hat\Theta(v_{pr}) \otimes \hat\Theta(v_{rq})$ for all $p,q$. Since every generator $g(v_{pq})$ already lies in $\mathcal{A} \otimes \mathcal{A} \subseteq \mathrm{End}(V \otimes V)$, we have $\mathcal{A}_g \subseteq \mathcal{A} \otimes \mathcal{A}$, and setting $\Delta := \Gamma$ shows that $\Delta : \mathcal{A} \to \mathcal{A} \otimes \mathcal{A}$ is a well-defined algebra homomorphism. This completes the construction of comultiplication $\Delta$. 

The counit of $\calA$ is constructed in exactly the same way: we still take $V$ in Lemma~\ref{lemma:algebrahomext} to be the Fock space, and take $f(v_{pq})=\hat\Theta(v_{pq})\in \End(V)$, but this time we take $W=\C$ and $g(v_{pq})=\delta_{pq}\in \End(W)\cong \C$. Let $V_\emptyset$ be the $n=0$ component in the RHS of Eq.~\eqref{eq:decomposeV}, a 1-dimensional subspace spanned by the vacuum $\ket{0}$, and we define a vector space isomorphism $\gamma: \C\xrightarrow{\sim} V_\emptyset $  by $\gamma(1)=\ket{0}$. Then $\gamma$ satisfies Eq.~\eqref{eq:gammafg} in Lemma~\ref{lemma:algebrahomext}, as it translates into
$\delta_{pq}\ket{0}=\hat\Theta(v_{pq})\ket{0}$, which is the $n=0$ special case of Eq.~\eqref{eq:vpqactionstatespace} that we use to define the action of $\hat\Theta(v_{pq})$. Hence the condition of Lemma~\ref{lemma:algebrahomext} is satisfied again, and we get an algebra homomorphism $\epsilon:\calA\to \C$ satisfying $\epsilon(v_{pq})=\delta_{pq}$, which we take as the counit for $\calA$. 
 
By now we get algebra homomorphisms $\Delta:\calA\to\calA\otimes \calA$ and $\epsilon:\calA\to \C$ satisfying Eq.~\eqref{def:corep}. It is then straightforward to check that $(\calA,\Delta,\epsilon)$ satisfies all axioms of a bialgebra. 

We now define the representation $\psi$ and the corepresentation $v^{\psi}$ associated to $\psi$. We set $v^{\psi}=v$ to be the defining corepresentation of $\calA$. To define $\psi$, let $V_i$ be any one particle subspace appearing in the decomposition in Eq.~\eqref{eq:decomposeV}.
Since the action of $\hat{\Theta}(v_{pq})$ on $V_i$ is independent of $i$, as we see in Eq.~\eqref{eq:vpqactionstatespace},  all $V_i$ for different $i$ are isomorphic as $\calA$-modules, and we set $\psi$ to be the isomorphism class of $\calA$-modules represented by any $V_i$.   
Then Eq.~\eqref{eq:zeta_tensor_def-basicHA}
follows from the definition of $\hat{\Theta}(v_{pq})$ and $\psi$. 
We now verify Eq.~\eqref{eq:HA-sym-R-condition}:
\begin{equation}\label{eq:HA-sym-R-condition-1}
	R[\zetapsi(z_{(1)})\otimes \zetapsi(z_{(2)})]=[\zetapsi(z_{(1)})\otimes \zetapsi(z_{(2)})]R.
\end{equation}
It is straightforward to see that if Eq.~\eqref{eq:HA-sym-R-condition-1} is satisfied for $z=z_1$ and $z=z_2$, then it is satisfied for $z=z_1 z_2$, since $(\psi\otimes \psi)\circ\Delta$ is an algebra homomorphism, so it remains to verify Eq.~\eqref{eq:HA-sym-R-condition-1} is satisfied by the generators $v_{pq}$ of $\calA$. 
But for $z=v_{pq}$, Eq.~\eqref{eq:HA-sym-R-condition-1} simply becomes the YBE for $R$~[due to Eq.~\eqref{eq:zeta_tensor_def-basicHA}]. Hence Eq.~\eqref{eq:HA-sym-R-condition-1} is satisfied for all $z\in\calA$. 
Using a similar trick, and using the definition of $\hat\Theta$ and $\psi$, we can verify Eq.~\eqref{eq:BA-adjoint-action}. 
This completes the proof of the first statement of the theorem~(bialgebra case) for the case of a single particle type $\psi$.

To generalize to a paraparticle system with any number of particle types, apply the construction above to the  composite particle type $\Psi=\bigoplus_{\varphi\in \T}\varphi$ defined in Sec.~\ref{sec:mutual_para}. Then it is straightforward to see that the representation $\Psi$ and corepresentation $v^\Psi$ decompose as a direct sum of subrepresentations: $\Psi=\bigoplus_{\varphi\in \T}\varphi$, and $v^\Psi=\bigoplus_{\varphi\in T}v^\varphi$, such that  Eq.~\eqref{eq:zeta_tensor_def-basicHA} is satisfied. Thus $\calA_\mathbf{R}=\calA$ is the hidden bialgebra symmetry we are looking for. 

Since $\calA_{\mathbf R}$ is defined via its action on the Fock space $V$, i.e., as a subalgebra of $\End(V)$, $V$ contains all indecomposable representations of $\calA_{\mathbf R}$, and from Eq.~\eqref{eq:vpqactionstatespace} it is clear that
$\Rep(\calA_\mathbf{R})$ is generated by $\Psi$~(hence also by $\T$) as a monoidal $\C$-linear additive category.  
\end{proof}
We now briefly explain the intuition why $\calA_\mathbf{R}$ can be extended to a Hopf algebra hidden symmetry if $\mathbf{R}$ is dual unitary, as we mentioned in 
Remark~\ref{rmk:DUHopf}.
Indeed, if $R$ is dual unitary, then we let $\calA$ be the algebra generated by $\{v_{pq},\bar{v}_{pq}|1\leq p,q\leq m\}$, where $v_{pq}$ is same as above, and we define $\bar{v}_{pq}$ via its action on the state space
\begin{eqnarray}\label{eq:barvpqactionstatespace}
\hat{\Theta}(\bar{v}_{pq})\ket{\psi}=\adjincludegraphics[height=9ex,valign=c]{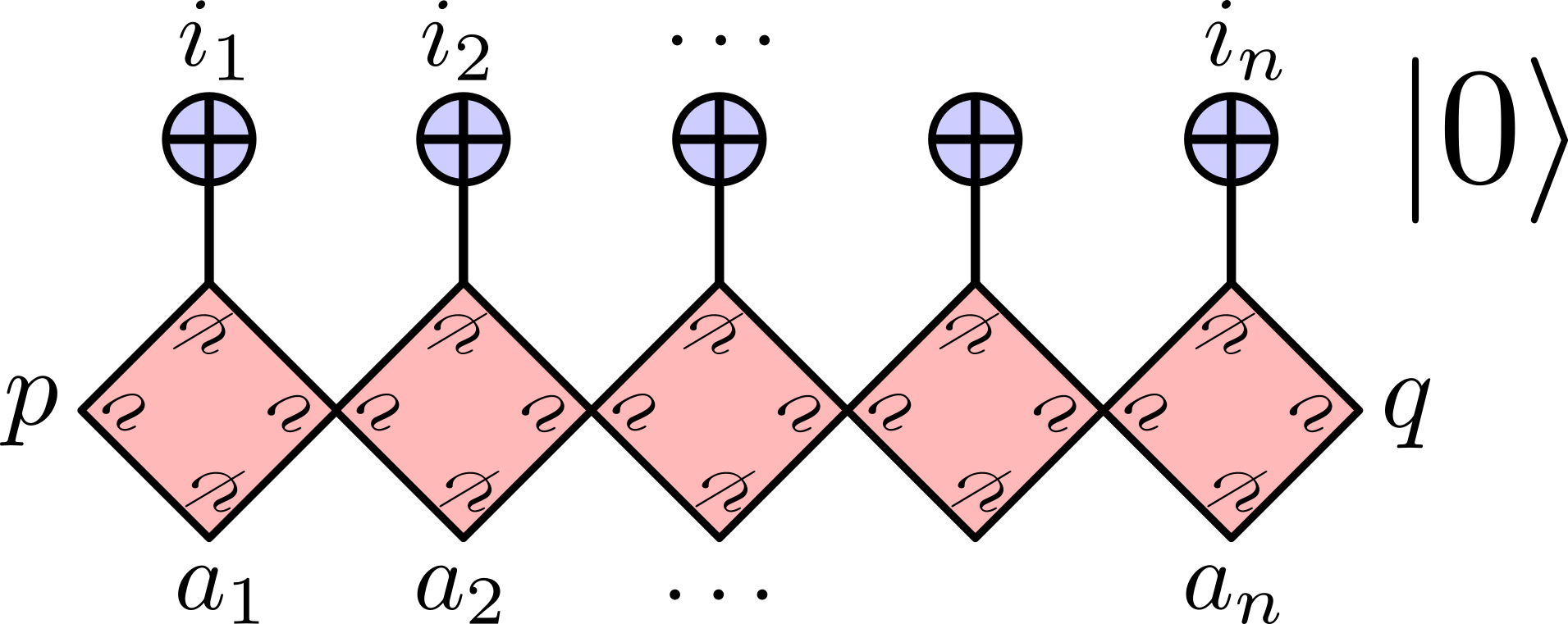},
\end{eqnarray}
with $\adjincludegraphics[height=3ex,valign=c]{Figures/GeneralizedSymmetry-1/kappaOproof-last}$ as before. 
Then by applying Lemma~\ref{lemma:algebrahomext}~(using a similar trick as above), 
we can define algebra homomorphisms $\Delta:\calA\to\calA\otimes\calA$,$\epsilon:\calA\to\C$, and antihomomorphism $S:\calA\to\calA$ via their action on the generators as follows
\begin{align}\label{eq:HAactionGen}
\Delta(v_{pq})&=\sum^d_{r=1} v_{pr}\otimes v_{rq},\quad \epsilon(v_{pq})=\delta_{pq},\nonumber\\
\Delta(\bar v_{pq})&=\sum^d_{r=1} \bar v_{pr}\otimes \bar v_{rq},\quad \epsilon(\bar v_{pq})=\delta_{pq},\nonumber\\
S(v_{pq})&=\bar{v}_{qp},\quad S(\bar{v}_{pq})=v_{qp},
\end{align}
It is then straightforward to verify that $(\calA,\Delta,\epsilon,S)$ satisfy all axioms of 
a Hopf algebra, and indeed define a hidden symmetry for the paraparticle system. Importantly, the antipode axiom
\begin{equation}
   S(z_{(1)})z_{(2)}=z_{(1)}S(z_{(2)})=\epsilon(z)1 
\end{equation}
and the assignment in Eq.~\eqref{eq:HAactionGen} forces $R$ to be dual unitary. Indeed, taking $z=v_{pq}$ in the antipode axiom 
forces Eq.~\eqref{eq:Antipode-graphical}~($R^2=1$ in the current case), while $z=\bar v_{pq}$ in the antipode axiom forces dual unitarity condition in Eq.~\eqref{eq:DUcondition}. 

\section{Classification of unitary $R$-matrices with pair creation}\label{sec:unitary-R-with-alpha}
In this appendix we prove  theorem~\ref{thm:unitary_simple_self-dualR} that reveals several important properties of unitary $R$-matrices of self-dual simple paraparticles, %
which defines their topological twist factor $\thetaR$ and Frobenius-Schur indicator $\nu$. It also shows that the Hilbert series of this class of paraparticles must be trivial. %

We first recall a classification theorem of unitary involutive $R$-matrices in terms of their Hilbert series, proved in Ref.~\cite{LECHNER2019106769}: 
\begin{theorem}{(Eq.~5.7 of Ref.~\cite{LECHNER2019106769})}\label{thm:classification_Herm_R}
	The Hilbert series of any unitary involutive $R$-matrix must be of the form
	\begin{equation}\label{eq:classification_Herm_R}
		z_R(x)=\frac{\prod^l_{j=1}(1+b_j x)}{\prod^k_{i=1}(1-a_i x)},
	\end{equation}
	where $k$ and $l$ are non-negative integers, and $\{a_i\}^k_{i=1},\{b_j\}^l_{j=1}$ are positive integers. Note that $z_R(x)$ is finite if and only if $k=0$.
\end{theorem}
Notice that any Hilbert series of the form in Eq.~\eqref{eq:classification_Herm_R} can be realized by the following $R$-matrix
\begin{equation}
R=(\boxplus^k_{i=1}\mathds{1}_{a_i})\boxplus [\boxplus^l_{j=1}(-\mathds{1}_{b_j})],
\end{equation}
where $\boxplus$ is the direct sum construction of $R$-matrices defined in Sec.~\ref{sec:mutual_para}. 

We now prove a few lemmas which we will use in the proof of Thm.~\ref{thm:unitary_simple_self-dualR}. 
\begin{lemma}\label{lemma:alphaunitary}
	Let $R$ be a unitary involutive $R$-matrix such that $z_R(x)$ is finite 
    and there exists an invertible $\alpha$ satisfying  Eq.~\eqref{eq:Ralpha_def-graphical}, Eq.~\eqref{eq:eta_R} and 
	$\alpha^*\alpha=\nu \mathds{1}$, where $\nu=\pm 1$. Then $\alpha$ is unitary and we have $\nu=-\eta_\alpha$ and $z_R(x)=(1+x)^{m}$ where $m$ is the quantum dimension of the paraparticle.
\end{lemma}
\begin{proof}

Under the conditions of the lemma, Eq.~\eqref{eq:Tralphaalpha'} holds. 
Using $\alpha'=\alpha^{-1}=\alpha^*/\nu$, Eq.~\eqref{eq:Tralphaalpha'} becomes 
\begin{equation}\label{eq:Tralphaalphadagger}
\Tr(\alpha\alpha^\dagger)=-\nu\eta_\alpha n_\mathrm{max}. %
\end{equation}
In the following we use Eq.~\eqref{eq:Tralphaalphadagger} to prove that $n_\mathrm{max}\geq m$.
	
	Consider the matrix version of Cauchy-Schwarz inequality: for any nonzero $m\times m$ matrices $A$ and $B$, we have
	\begin{equation}\label{eq:matrixCSIneq}
		\Tr(A^\dagger A)\Tr(B^\dagger B)\geq [\Tr(A^\dagger B)]^2,
	\end{equation}
	where the equality holds if and only if $A=\lambda B$ for some $\lambda\in\C$. Taking $A=\alpha$ and $B=\alpha^T$ in Eq.~\eqref{eq:matrixCSIneq}, we obtain
	\begin{equation}\label{eq:matrixCSIneq-alpha}
		\Tr(\alpha^\dagger \alpha)\geq |\Tr(\alpha^*\alpha)|=|\nu| |\text{Tr}(\mathds{1})|=m.
	\end{equation}
	 Inserting Eq.~\eqref{eq:Tralphaalphadagger} into Eq.~\eqref{eq:matrixCSIneq-alpha}, 
     we have 
\begin{equation}
-\nu \eta_\alpha n_{\text{max}} \ge m.
\label{eq:eta-alpha-n-m}
\end{equation}
 Noting that $m$ and $n_\text{max}$ are positive, this implies $\nu \eta_\alpha<0$; therefore, since $\nu, \eta_\alpha\in\{-1,1\}$, we have $\nu \eta_\alpha=-1$. Then Eq.~\eqref{eq:eta-alpha-n-m} gives $n_{\text{max}}\geq m$.

	On the other hand, it is not hard to see that the classification theorem~(Theorem~\ref{thm:classification_Herm_R}) for unitary $R$-matrices implies that $n_\mathrm{max}\leq m$ if $z_R(x)$ is finite~\footnote{Indeed, $n_\mathrm{max}\leq m$ holds for any involutive $R$-matrix with a finite Hilbert series, even without the unitarity condition. 
    Ref.~\cite{davydov2000totally} proves that for any involutive  $R$-matrix~(or more generally, any Heckian $R$-matrix), the Hilbert series $z_R(x)$ must be a so-called totally-positive sequence. 
    Using the classification theory of totally-positive sequences~\cite{Edrei1953TotallyPositiveSequence} it is straightforward to show that $n_\mathrm{max}\leq m$ if $z_R(x)$ is finite.}~(in which case we have $k=0$), since the theorem implies that 
    $m = \sum_{j=1}^l b_j\geq l=n_{\text{max}}$. 
	Therefore, we must have $n_\mathrm{max}= m$ and $z_R(x)=(1+x)^m$. Furthermore, since the inequality in Eq.~\eqref{eq:matrixCSIneq-alpha} becomes an equality, we have $\alpha^T=\lambda\alpha$ for some $\lambda\in\C$. 
    Note that together with $\alpha^* \alpha=\nu \mathds 1$, this implies %
\begin{equation}
\alpha^\dagger \alpha= \lambda \nu\mathds{1}. \label{eq:alph-dag-alpha-1}
\end{equation}
 Since Eq.~\eqref{eq:matrixCSIneq-alpha}
 becomes an equality, we have $m=\mathrm{Tr}(\alpha^\dagger \alpha) = \lambda \nu \mathrm{Tr}(\mathds{1}) = \lambda \nu m$. Hence $\lambda\nu=1$ and Eq.~\eqref{eq:alph-dag-alpha-1} says $\alpha$ is unitary.
\end{proof}

\begin{lemma}\label{lemma:OmegaLemma-Ralpha}
	Let $R$ be a unitary involutive $R$-matrix. Assume that there exists a positive semi-definite $m\times m$ matrix $\Omega\geq 0$ 
    satisfying the following conditions: \\
	(1)  $(\Omega\otimes \Omega)  R=R(\Omega\otimes \Omega) $;\\
	(2)  $\Tr_2[(\mathds{1}_m\otimes \Omega)  R]=-\mathds{1}_m$,
	where $\Tr_2$ takes the partial trace over the second tensor factor. \\
	Then $A\equiv\Tr[\Omega]$ %
    is a positive integer and $z_R(x)$ is a finite polynomial with degree $n_\mathrm{max}=A$.
\end{lemma}
\begin{proof}
	Let $\rho_{R,n}$ denote the representation of the symmetric group $S_n$ generated by the $R$-matrix, as defined in Sec.~\ref{sec:YBE}. %
	Define a function $\chi_{\Omega,n}:S_n\to \mathbb{R}$ as follows
	\begin{equation}
		\chi_{\Omega,n}(\sigma):=\Tr[\Omega^{\otimes n} \rho_{R,n}(\sigma)],\quad \forall\sigma\in S_n.
	\end{equation}
	The condition (1) of the Lemma implies that $\chi_{\Omega,n}$ is a class function on the group $S_n$, i.e., $\chi_{\Omega,n}$ is constant on the conjugacy class of $S_n$: $\chi_{\Omega,n}(\tau\sigma\tau^{-1})=\chi_{\Omega,n}(\sigma)$ for all $\sigma,\tau\in S_n$. We now apply  the basic fact of the symmetric group $S_n$ that each element $\sigma$ of $S_n$ 
	is conjugate to a product of disjoint cycles in which each cycle acts on a consecutive subsequence of $1,2,\ldots,n$. (This is the usual cycle decomposition of a permutation, followed by relabeling the elements which gives a conjugate permutation.) 
    Using condition (2) above, we can evaluate $\Tr[\Omega^{\otimes n} \rho_{R,n}(\sigma)]$ for such a product of disjoint cycles. For example, for $\sigma=(123\ldots n)$, we have
	$\rho_{R,n}(\sigma)=R_{12}R_{23}\ldots R_{n-1,n}$, and therefore 
    \begin{eqnarray}
    \chi_{\Omega,n}(\sigma)&=&
    \Tr[\Omega^{\otimes n} R_{12}R_{23}\ldots R_{n-1,n}]\nonumber\\
    &=&
    -\Tr[\Omega^{\otimes (n-1)} R_{12}R_{23}\ldots R_{n-2,n-1}]\nonumber\\
    &=&\ldots\nonumber\\
    &=&(-1)^{n-2}\Tr[(\Omega\otimes \Omega) R_{12}]\nonumber\\
    &=&(-1)^{n-1}\Tr[\Omega]\nonumber\\
    &=&(-1)^{n-1}A,
    \end{eqnarray}
    where we have used the above condition (2) for $n-1$ times. 
    More generally, we have
	\begin{equation}
		\chi_{\Omega,n}(\sigma)=(-1)^n (-A)^{l(\sigma)},\quad\forall \sigma\in S_n,
	\end{equation}
	where $l(\sigma)$ is the total number of disjoint cycles in $\sigma$. 
	
	Now consider the projection operator to the trivial subrepresentation of $\rho_{R,n}$:
\begin{equation}\label{eq:PnProjTrivialSubrep}
    P_n=\frac{1}{n!}\sum_{\sigma\in S_n}\rho_{R,n}(\sigma),
\end{equation}
introduced in the proof of Proposition~\ref{prop:dn-char-rel}, which relates to the Hilbert series via $d_n=\Tr[P_n]$. 
We have 
	\begin{eqnarray}\label{eq:OmegaPnderivation}
		\Tr[\Omega^{\otimes n} P_n]&=&\frac{1}{n!}\sum_{\sigma\in S_n} \chi_{\Omega,n}(\sigma)\nonumber\\
		&=&\frac{1}{n!}\sum_{\sigma\in S_n} (-1)^n (-A)^{l(\sigma)}\nonumber\\
		&=&\frac{1}{n!}A(A-1)\ldots(A-n+1),
	\end{eqnarray}
	where in the last step we used the identity~\footnote{
   This identity can be proved by evaluating $\Tr[\Omega^{\otimes n} P_n]$ for the trivial $R$-matrix $R=-X$
   ~(with $\Omega=\mathds{1}_m$ and $A=m$) in two different ways. For convenience, we denote $g(x)\equiv x(x+1)\ldots (x+n-1)$. On one side, the first two lines of Eq.~\eqref{eq:OmegaPnderivation} show that $\Tr[\Omega^{\otimes n} P_n]=(-1)^n f(-m)/n!$. On the other side, with $\Omega=\mathds{1}_m$, $\Tr[\Omega^{\otimes n} P_n]=\Tr[P_n]=\mathrm{rank}[P_n]$ 
		can be directly computed for this trivial $R$-matrix, and is equal to the dimensional of the totally anti-symmetric subspace of $\mathfrak{I}^{\otimes n}$ for an $m$-dimensional vector space $\mathfrak{I}$, which is $\binom{m}{n}=m(m-1)\ldots(m-n+1)/n!=(-1)^n g(-m)/n!$. Therefore $f(-m)=g(-m)$ for all integers $m\geq 1$. Since both $f(x)$ and $g(x)$ are polynomials of degree $n$, we conclude that $f(x)=g(x)$. }
	\begin{equation}\label{eq:SnSumRule}
		f(x)\equiv\sum_{\sigma\in S_n} x^{l(\sigma)}=x(x+1)\ldots (x+n-1).
	\end{equation}
	Since both $\Omega$ and $P_n$ are positive semi-definite, we must have $\Tr[\Omega^{\otimes n} P_n]\geq 0$ 
	for any integer $n\geq 2$. 
    If $A$ is not an integer, then the RHS of Eq.~\eqref{eq:OmegaPnderivation} would be negative for some $n>A+1$; therefore $A$ must be a positive integer~(note that we already have $A>0$). %
    This then implies that $P_n=0$ for any $n\geq A+1$, meaning that $\rho_{R,n}$ does not contain the trivial representation of $S_n$ for $n\geq A+1$. 
    Meanwhile, we have $P_A\neq 0$, which follows from the $n=A$ case of Eq.~\eqref{eq:OmegaPnderivation}. 
    Therefore $z_R(x)$ is a finite polynomial with degree $n_\mathrm{max}=A$. 
\end{proof}

\begin{lemma}\label{lemma:alpha_unique_inv}
	Let $R$ be a simple involutive $R$-matrix. If there exists a nonzero $\alpha$ satisfying Eq.~\eqref{eq:Ralpha_def-graphical}, then it must be unique~(up to a multiplicative constant) and invertible.
\end{lemma}
\begin{proof}
	Take any $w\in \mathrm{Ker}(\alpha)$, the kernel of the matrix $\alpha$. From Eq.~\eqref{eq:Ralpha_def-graphical} we obtain
	$0=\adjincludegraphics[height=4ex,valign=c]{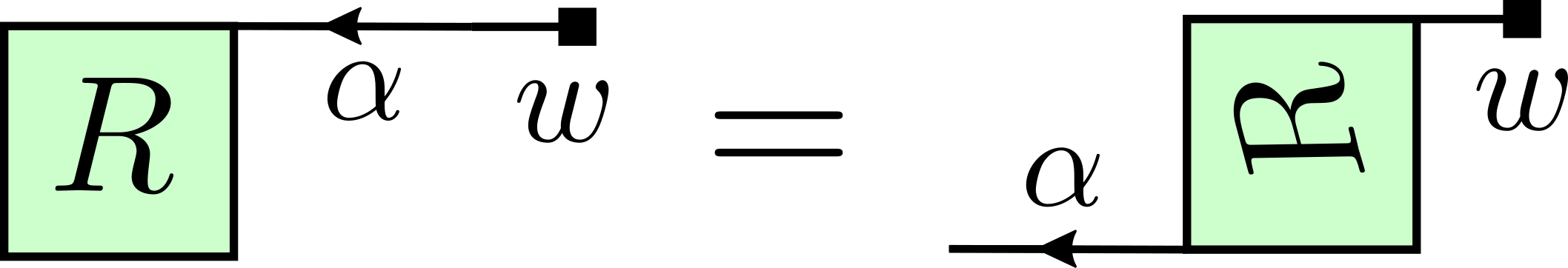}$, which means that $\mathrm{Ker}(\alpha)$ is an invariant subspace of $\Dse[R]$. %
	But since we assume that $R$ is simple, meaning that $\Dse[R]$ is simple as an algebra,  $\Dse[R]$ has no proper invariant subspace. 
    Since we assume $\alpha\neq 0$, it follows that $\text{Ker}(\alpha)=0$, 
    i.e., $\alpha$ is invertible. Furthermore, for any two invertible $\alpha_1,\alpha_2$ satisfying Eq.~\eqref{eq:Ralpha_def-graphical}, a straightforward tensor network manipulation shows that $\alpha_2^{-1}\alpha_1\in \DC[R]$. Since $R$ is simple, we have $\DC[R]\cong \C$, leading to $\alpha_2=\lambda\alpha_1$ for some $\lambda\in \C$, i.e., $\alpha$ is unique up to a multiplicative constant. 
\end{proof}

\begin{remark}\label{rmk:example-noninvalpha}
Here is an example of a non-simple $R$-matrix where a nonzero $\alpha$ satisfying Eq.~\eqref{eq:Ralpha_def-graphical} exists, but is not invertible. Consider the $R$-matrix $R=R_1\boxplus R_3$, the direct sum of $R_1$ and $R_3$ with trivial mutual statistics, as defined in Sec.~\ref{sec:mutual_para}, %
where $R_1$ and $R_3$ are taken from Ex.~\ref{ex:decoupled} and Ex.~\ref{ex:1m} of Tab.~\ref{tab:Hilbert_series}, with quantum dimensions $m_1$ and $m_3$, respectively. One can check that for this $R$-matrix, $\alpha$ must have the form $\alpha=i_1 \alpha_1 p_1$, where $i_1$ and $p_1$ are the inclusion map and projector for the subobject $R_1$, respectively,  and $\alpha_1$ is an arbitrary $m_1\times m_1$ matrix. Such an $\alpha$ is clearly not invertible. 

\end{remark}

\unitarysimpleselfdualR*
\begin{proof}
	(1). This has already been proved in Lemma~\ref{lemma:alpha_unique_inv}. \\
	(2). Taking the complex conjugate of Eq.~\eqref{eq:Ralpha_def-graphical} and using the fact that $R$ is unitary and involutive -- which gives $ R^\dagger = R$ and therefore $R^*=R^T$ -- we have
	\begin{equation}\label{eq:unit-inv-alpha-R}
		\adjincludegraphics[height=4ex,valign=c]{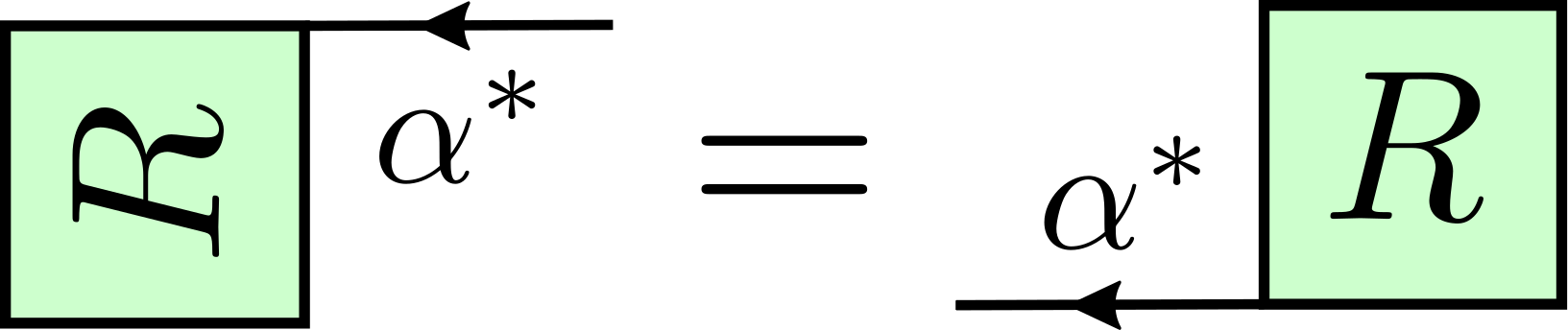}.
	\end{equation}
	Therefore
	\begin{equation}\label{eq:derivalphaalpha}
		\adjincludegraphics[height=4ex,valign=c]{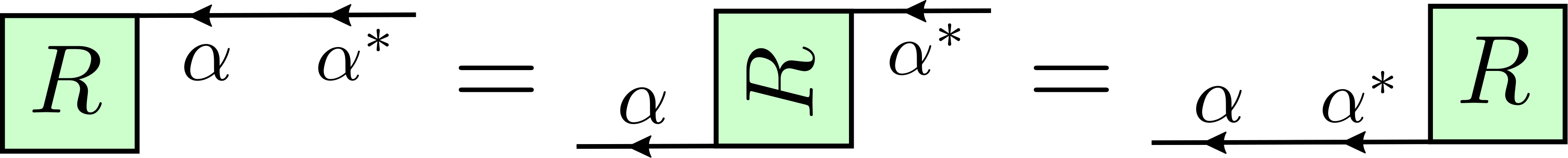},
	\end{equation}
	where the first equality follows from the relation in Fig.~\ref{fig:Ralphaeqn}, and the second follows from Eq.~\eqref{eq:unit-inv-alpha-R}. Eq.~\eqref{eq:derivalphaalpha} says that $\alpha^*\alpha\in \DC[R]\cong \C$, leading to $\alpha^*\alpha=\nu \mathds{1}$ for some $\nu\in \C$. 
    From $\nu^* \mathds{1}=\alpha\alpha^*=\alpha^*\alpha=\nu\mathds{1}$~(recall that every finite invertible matrix commutes with its inverse), 
    we see that $\nu\in\mathbb{R}$. 
    We can then normalize $\alpha$ such that $|\nu|=1$, i.e., $\nu=\pm 1$. \\
	(3). Using the relations in Fig.~\ref{fig:Ralphaeqn} along with $\alpha'=\nu\alpha^*$~[which follows from (2)], we have 
	\begin{equation}\label{eq:Tr1aaR}
		\adjincludegraphics[height=4ex,valign=c]{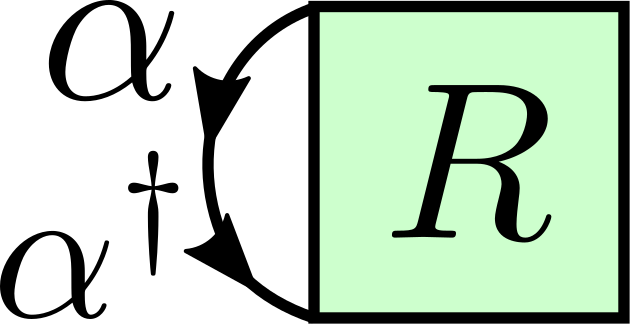}
        =
        \adjincludegraphics[height=4ex,valign=c]{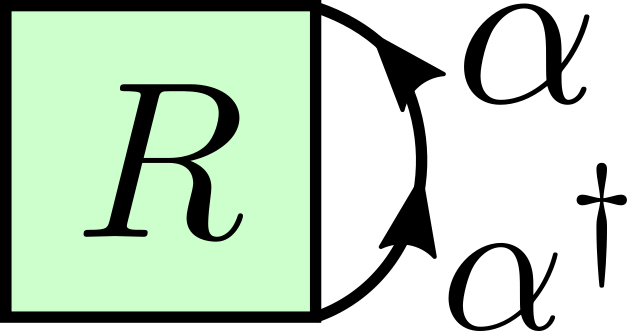}=\nu\eta_\alpha\mathds{1},	
	\end{equation}
	where $\nu\eta_\alpha=\pm 1$. Below we first prove the theorem for the case $\nu\eta_\alpha=-1$, and then apply the same argument to $-R$ to prove the other case. 
	
	If $\nu\eta_\alpha=-1$, then we apply Lemma~\ref{lemma:OmegaLemma-Ralpha} with $\Omega=\alpha^\dagger\alpha$.  Eq.~\eqref{eq:Tr1aaR} implies  condition (2) of the Lemma, and condition (1) is proved as follows
	\begin{equation}
		\adjincludegraphics[height=9ex,valign=c]{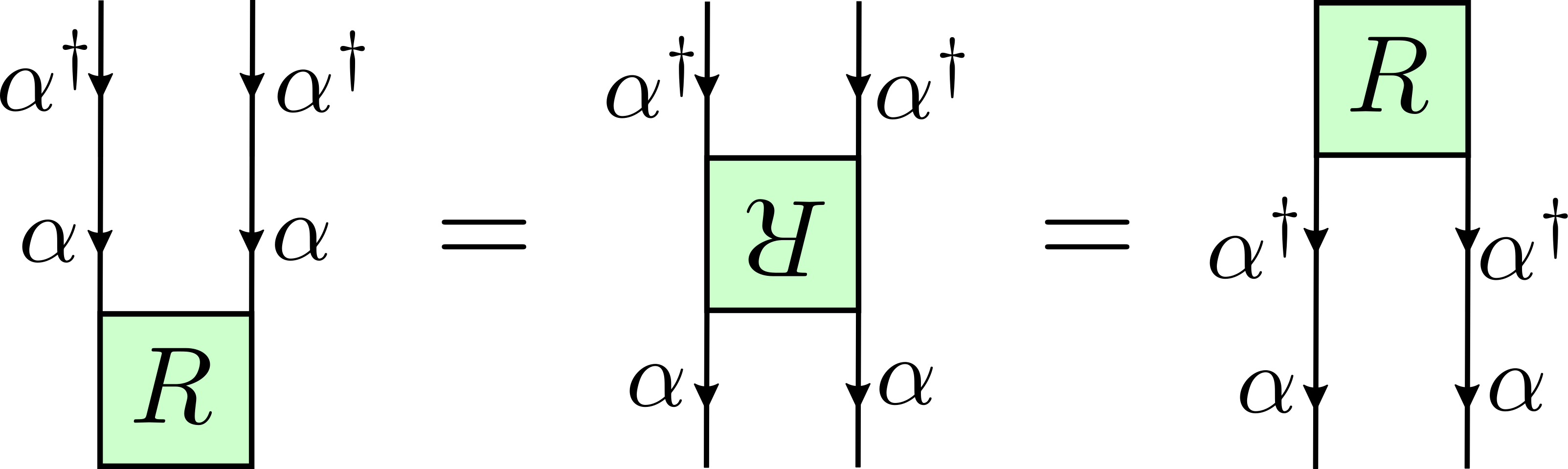},
	\end{equation}
	where we used the relations in Fig.~\ref{fig:Ralphaeqn} and the unitarity of $R$. Therefore Lemma~\ref{lemma:OmegaLemma-Ralpha} applies, which shows that $z_R(x)$ is finite. It follows that all the conditions of Lemma~\ref{lemma:alphaunitary} are satisfied, which implies that $\alpha$ is unitary. Combining with 
	$\alpha^*\alpha=\nu \mathds{1}$, we have $\alpha^T=\nu \alpha$. \\
	(4).	Since both $R$ and $\alpha$ are unitary, Eq.~\eqref{eq:Ralpha_def-graphical} implies that $\Rmp$ is also unitary, i.e. $R$ is a dual-unitary tensor.\\
	(5). The first part follows from Eq.~\eqref{eq:Tr1aaR} and the unitarity of $\alpha$, while the second part follows from Lemma~\ref{lemma:alphaunitary}. 

	At this point, we have proved  Theorem~\eqref{thm:unitary_simple_self-dualR} for the case $\nu\eta_\alpha=-1$. If we have $\nu\eta_\alpha=+1$ instead, we can apply the same arguments above to the $R$-matrix $-R$, which also satisfies all the conditions of the theorem, but for which $\nu\eta_\alpha=-1$. This proves (3)-(5) of the theorem for the case $\nu\eta_\alpha=+1$, except the last claim on $z_R(x)$. But since $z_{-R}(x)=(1+x)^m$, Eq.~\eqref{eq:hRh_mRrelation} shows that $z_{R}(x)=(1-x)^{-m}$, thereby completing the proof.
\end{proof}

\section{Comparison between $R$-parastatistics and other known types of particle statistics}\label{sec:relation_to_others}
In  this appendix, we briefly mention previous theories of parastatistics~(in particular, Green's theory~\cite{Green1952}) and other known types of particle statistics, and explain their relation and difference to  $R$-parastatistics. %

\subsection{Green's theory of parastatistics}\label{app:Green-theory}
Green's  parastatistics~\cite{Green1952,Araki1961,Greenberg1965,LANDSHOFF196772,druhl1970parastatistics,Taylor1970b} is defined by a set of \textit{trilinear} CRs~(rather than \textit{bilinear} CRs used to define $R$-parastatistics) between paraparticle creation and annihilation operators 
\begin{equation}\label{eq:CRGreen}
	\begin{aligned}
		& {\left[\left[\hat{\psi}^{\dagger}_i, \hat{\psi}_j\right]_{ \pm}, \hat{\psi}_l\right]_-=-2 \delta_{i l} \hat{\psi}_j,} \\
		& {\left[\left[\hat{\psi}_i, \hat{\psi}_j\right]_{ \pm}, \hat{\psi}_l\right]_-=0,}
	\end{aligned}
\end{equation}
where $[\hat{A},\hat{B}]_\pm=\hat{A}\hat{B}\pm \hat{B}\hat{A}$, and $i,j,l$ are mode indices. This theory is also well defined in any spatial dimension.
Unlike $R$-parastatistics, the correspondence between first and second quantization for Green's parastatistics is less direct. Ref.~\cite{Taylor1970b} showed that Green's paraparticles also realize higher dimensional representations of %
$S_N$ in the sense of Eq.~\eqref{eq:wavefuntion_exchange_para}, %
but the derivation is much longer than the $R$-parastatistics case presented in Sec.~\ref{sec:relation_first_quantization}. %

Green's parastatistics and $R$-parastatistics have a small overlap, %
which we explain below. To define the Fock space for Green's paraparticles, Green introduced the following decomposition~\cite{Green1952}
\begin{eqnarray}\label{eq:GreenDecomposition}
\hat{\psi}^{\dagger}_j=\sum^m_{a=1}\hat{\psi}^{+}_{j,a},\quad %
\hat{\psi}^{\phantom{\dagger}}_j=\sum^m_{a=1}\hat{\psi}^{-}_{j,a},
\end{eqnarray}
where $m\in\Z_{\geq 1}$ is called the order of the paraparticle~\cite{Green1952},  $\{\hat{\psi}^{\pm}_{j,a}\}$ satisfy the CRs in Eq.~\eqref{eq:fundamental_Rcommu} with
$R^{ab}_{cd}=\pm\delta_{ad}\delta_{bc}(-1)^{\delta_{ab}}$~(Ex.~\ref{ex:Green}), and one can verify that this decomposition %
satisfies Eq.~\eqref{eq:CRGreen}. 
Thus, Green's parastatistics  is equivalent to  $R$-parastatistics in Ex.~\ref{ex:Green} after the decomposition in Eq.~\ref{eq:GreenDecomposition}. The main technical subtlety here is that the decomposition in  %
Eq.~\eqref{eq:GreenDecomposition}
generally leads to a \textit{reducible} representation of Green's trilinear algebra in Eq.~\eqref{eq:CRGreen}, %
even if one uses an irrep of 
$\{\hat{\psi}^{\pm}_{j,a}\}$. 
To find irreps of Green's trilinear field algebra~\eqref{eq:CRGreen}, one needs to find subspaces of the Fock space in Sec.~\ref{sec:state_space} that are left invariant by all $\{\hat{\psi}^{\phantom{\dagger}}_j,\hat{\psi}^{\dagger}_j\}$, which is mathematically hard~\cite{Green1952,Araki1961,Greenberg1965,LANDSHOFF196772,druhl1970parastatistics,Stoilova2020}.  
Consequently, the exclusion statistics of Green's paraparticles is hard to compute and is still not fully solved to date~\cite{Stoilova2020}. 
In particular, unlike $R$-parastatistics, the grand partition function of a system of Green's paraparticles  generally does not factorize as a product of single mode partition functions~\cite{Stoilova2020}, making it difficult to study thermodynamics. 

We now briefly discuss %
whether Green's paraparticles can emerge %
in condensed matter systems. %
First, we should distinguish between trivial and nontrivial realizations. Example of a trivial realization is to simply stack $m$ identical layers of a bosonic topological phase with emergent fermions, for example the 2D toric code~\cite{kitaev2003fault}, and use $\hat{\psi}^\pm_{j,a}$ to denote the fermion creation and annihilation operators in the $a$-th layer. Since operators in different layers commute, $\{\hat{\psi}^{\pm}_{j,a}\}$ satisfy  Eq.~\eqref{eq:fundamental_Rcommu} with
$R$ in Ex.~\ref{ex:Green}~[with $(+)$ sign], and consequently, this realizes Green's trilinear algebra via Eq.~\eqref{eq:GreenDecomposition}. This ``realization'' is  trivial, since it  does not satisfy the local indistinguishability criterion in Definition~\ref{def:weak_indistinguishable_particles}, as the index $a$ can be accessed by local measurements even when the particle is deep in the bulk. Consequently, this realization is useless for the quantum information applications in Sec.~\ref{sec:QI_statistics}, as it is vulnerable to local noise and eavesdropping due to being a composite particle. 
Nevertheless, we know Ex.~\ref{ex:Green} for any $m>1$ can be realized as simple particles in a special class of deconfined finite group gauge theories in 3D, as discussed in App.~G.2 of Ref.~\cite{wang2025secret}. Still, this only realizes a reducible representation of Green's trilinear algebra. Existence of a solvable condensed matter model with emergent Green's paraparticles realizing an irrep of Green's trilinear algebra appears open, and we refer to Ref.~\cite{zhou2026green} for recent progress in this direction.

\subsection{Infinite statistics and quons}\label{app:infinitestatistics}
Quons were introduced by Greenberg~\cite{greenbergParticlesSmallViolations1991} as a deformation of Bose and Fermi statistics, with annihilation and creation operators satisfying
\begin{equation}\label{eq:quonalgebra}
\hat{\psi}_k \hat{\psi}_l^{\dagger}-q \hat{\psi}_l^{\dagger} \hat{\psi}_k=\delta_{k l},
\end{equation}
where $-1\leq q\leq 1$, and $q=\pm 1$ recover ordinary bosons~($+$) and fermions~($-$), respectively, while $q=0$ gives infinite statistics~\cite{greenbergExampleInfiniteStatistics1990}. Note that
no relation is imposed on products of creation operators, in contrast to $R$-parastatistics.
Consequently, for $n\geq 1$,
the quon Fock space %
contains $n!$ linearly independent $n$-particle states with distinct mode labels, and all representations of $S_n$ occur, and there is no upper bound on the particle number $n$. 
Quon theories also admit an analog of $\hat e_{ij}$ operators %
satisfying the $\mathfrak{gl}_N$ relations %
in Eqs.~\eqref{eq:commu_Eab_psi_p} and \eqref{eq:commu_Eab_Ecd}, but in quon theories  
$\hat{e}_{ij}$ is given by an infinite series involving all modes~\cite{greenbergParticlesSmallViolations1991}:
\begin{align}\label{eq:eijQuon}
\hat e_{ij}
={}&\hat\psi_i^\dagger\hat\psi_j
+\sum_k \hat\psi_k^\dagger\hat\psi_i^\dagger
\hat\psi_j\hat\psi_k  \\
&+\sum_{k_1,k_2}
\hat\psi_{k_1}^\dagger\hat\psi_{k_2}^\dagger
\hat\psi_i^\dagger\hat\psi_j
\hat\psi_{k_2}\hat\psi_{k_1}
+\cdots . \nonumber
\end{align}
To our knowledge, quons %
are not known to emerge %
in condensed matter systems. An exact realization of the quon algebra would require an infinite local Hilbert space, which rules out quantum spin systems. Relativistic quon field theories appear to suffer from locality~(causality) problems~\cite{greenbergParticlesSmallViolations1991},   although certain formal properties still hold, such as cluster decomposition and a CPT theorem. 
\subsection{Transtatistics}
Ref.~\cite{Dakic2024reconstructionof} developed a   theory of generalized quantum statistics under the name of transtatistics, which is closely related to $R$-parastatistics, but using a fundamentally different approach. They started with a different set of assumptions: instead of postulating a second-quantized
operator algebra, they assume that single-particle dynamics is described by the unitary group $U(N)$ acting on $N$ modes, subject to a certain locality condition in the Fock space. Under these assumptions, they classify the
possible single-mode partition function $z(x)$ with $x=e^{-\beta(\epsilon-\mu)}$, and showed that $z(x)$ must 
form a so-called \textit{totally positive sequence}~\cite{Dakic2024reconstructionof}. 
This classification captures the
same kind of generalized exclusion statistics of $R$-paraparticles measured by
$z_R(x)$, and is more general in that there exist exotic totally positive sequences that cannot be realized by $R$-paraparticles~(see discussion at the end). 

A more recent work~\cite{sanchez2025reconstruction} proposed a
second-quantized realization of transtatistics, defined by the quadratic algebraic relations analogous to %
Eq.~\eqref{eq:fundamental_Rcommu}:
\begin{align}
[\hat\psi^+_{i,a},\hat\psi^+_{j,b}]_{A,B}&=0,\nonumber\\
    [\hat\psi^-_{i,a},\hat\psi^-_{j,b}]_{A,B}&=0,\nonumber\\
    [\hat\psi^-_{i,a},\hat\psi^+_{j,b}]_{A,B}
    &=g_{ab}\delta_{ij}\mathbf{1},
\end{align}
where $g_{ab}$ is a Hermitian matrix, and $[~,~]_{A,B}$ is the so-called 
$(A,B)$-bracket defined by
\begin{equation}
\begin{aligned}
    [\hat P_{i,a},\hat Q_{j,b}]_{A,B}
    &=
    \sum_{c,d} A^{cd}_{ab}
    \{\hat P_{i,c},\hat Q_{j,d}\}
    +
    \sum_{c,d} B^{cd}_{ab}
    [\hat P_{i,c},\hat Q_{j,d}],
\end{aligned}
\end{equation}
for $P,Q\in\{\psi^+,\psi^-\}$, where $A,B$ are four-index tensors satisfying certain algebraic constraints analogous to the Yang-Baxter equation~\eqref{eq:YBE}. Ref.~\cite{sanchez2025reconstruction} also claims that the $R$-commutation relations in Eq.~\eqref{eq:fundamental_Rcommu} are recovered via a special choice of $A$ and $B$, however, the current version of Ref.~\cite{sanchez2025reconstruction} seems to have some subtle issues regarding the orientation of the $R$-matrix. 

Therefore, the main difference between transtatistics and $R$-parastatistics is their starting point and the overall philosophy. 
The former starts from a set of operational axioms, classify possible exclusion statistics, and derive operator algebra of quantum fields, while
the latter directly starts from 
$R$-matrix second quantization and investigates the unique physical consequences of the theory as well as realization in physical systems. 
In our opinion, an interesting open question here is which class of transtatistics can be realized in physical systems, including in condensed matter and high energy physics, which was only briefly discussed as a possible future direction in Refs.~\cite{Dakic2024reconstructionof,sanchez2025reconstruction}. In particular, there are exotic types of transtatistics with degenerate vacuum states~(i.e., $d_0>1$), or even with states with negative occupation numbers~(i.e., $d_{n}>0$ for some negative integer $n$), in which case the single mode partition function $z(x)$ is a Laurent series instead of a Taylor series. These unusual behaviors are not present in any $R$-parastatistics theory, and deviate from our general expectations of  physical quantum field theories, making them hard to realize in physical systems. It is possible that by imposing extra axioms on transtatistics, one recovers the theory of $R$-parastatistics.

\subsection{Generalized statistics in 1D}\label{sec:1Dstatistics}
Up to now we have been discussing particle statistics that are consistently defined in any spatial dimension. In this section we mention some interesting families of particle statistics that are limited to 1D, including $q$-deformed bosons~(Sec.~\ref{app:q-deformation}), parafermions~~(Sec.~\ref{sec:parafermions}), and 1D Abelian anyons~(Sec.~\ref{sec:1DAnyons}). Then in Sec.~\ref{sec:non-involutive-R-para} we mention a potential extension of $R$-paraparticles to non-involutive $R$-matrices that are limited to 1D, which potentially generalize other 1D particle statistics to the non-Abelian regime. 
\subsubsection{$q$-deformed  bosons}\label{app:q-deformation}
The second quantized theory of $q$-deformed bosons is defined by a $q$-deformation of the CCR of ordinary bosons~\cite{Biedenharn1989qboson,macfarlaneQanaloguesQuantumHarmonic1989,Sun1989qboson}
\begin{align}
\hat{\psi}_i \hat{\psi}_i^\dagger&=q \hat{\psi}_i^\dagger \hat{\psi}_i+q^{-\hat{n}_i},\nonumber\\
[\hat{n}_i,\hat{\psi}_j^\dagger]&=\delta_{ij}\hat{\psi}_j^\dagger,
\end{align}
and different modes commute. Note the difference from quons in Eq.~\eqref{eq:quonalgebra}. Here $\hat{n}_i$ is the analog of particle number operator in Eq.~\eqref{eq:def_n_a}, and also have non-negative integer eigenvalues. The analog of the bilinear operator $\hat{e}_{ij}$ in Eq.~\eqref{eq:def_e_ab} are
\begin{align}
\hat{E}_{i,i+1}&=\hat{\psi}_i^\dagger \hat{\psi}_{i+1},\nonumber\\
\hat{E}_{i+1,i}&=\hat{\psi}_{i+1}^\dagger \hat{\psi}_i,\nonumber\\
\hat{K}_{i,i+1}&=q^{(\hat{n}_i-\hat{n}_{i+1})/2}. 
\end{align}
These operators $\{\hat{E}_{i,i+1},\hat{E}_{i+1,i},\hat{K}_{i,i+1}\}_{i=1}^{N-1}$ 
realize the quantum group $U_q(\mathfrak{sl}_N)$
\cite{Sun1989qboson} 
instead of the ordinary Lie algebra $\mathfrak{sl}_N$. 
In general, $q$-deformed bosons lack a natural free-particle theory that can be solved simply by diagonalizing
a one-particle Hamiltonian, although $q$-bosons still lead to Bethe-ansatz-solvable models~\cite{bogoliubovExactSolutionQboson1993,bogoliubovCorrelationFunctionsStrongly1998,tsilevichQuantumInverseScattering2006}.
We also mention that there is a different version of $q$-deformed bosons~\cite{PuszWoronowicz1989} where different modes $q$-commute~(i.e., $\hat{\psi}_i \hat{\psi}_j=q \hat{\psi}_j \hat{\psi}_i$ for $i<j$ ) instead of commute, realizing the dual quantum group $SU_q(N)$.

\subsubsection{Parafermions}\label{sec:parafermions}
Parafermions~\cite{fendleyParafermionicEdgeZero2012,Fendley2014Freepara,Fendley2016review} are generalizations of Majorana fermions~\cite{Kitaev2001Majorana}, defined by the following commutation relations
\begin{eqnarray}\label{eq:def_parafermions}	\hat{\psi}_{i}\hat{\psi}_{j}&=&\omega\hat{\psi}_{j}\hat{\psi}_{i},\text{ for } i<j,\nonumber\\
	\hat{\psi}_{i}^n&=&1, ~\forall i,
\end{eqnarray}
where $n\geq 2$ is an integer and $\omega$ is a primitive $n$-th root of unity. 
Parafermions emerge in 1D $\mathbb{Z}_n$-symmetric quantum spin chains~\cite{fendleyParafermionicEdgeZero2012} and the 1D boundary of certain 2D topological phases~\cite{Fendley2016review}. Parafermions also admit an exactly solvable free particle theory ~\cite{Fendley2014Freepara} %
, however,  the free  Hamiltonian is non-Hermitian, and a Hermitian free parafermion theory is not known to date.  
Furthermore, the solution technique for free parafermions is significantly more involved than other free particle theories we mentioned above.  

\subsubsection{Abelian anyons in 1D}\label{sec:1DAnyons}
An analog of Abelian anyons in one-dimensional continuum were proposed in 
\cite{kunduExactSolutionDouble1999}, defined by the commutation relations
\begin{align}\label{eq:1Danyons}
\hat\psi(x)\hat\psi(y)
&=e^{-i\kappa\,\mathrm{sgn}(x-y)}\hat\psi(y)\hat\psi(x),
\\
\hat\psi(x)\hat\psi^\dagger(y)
&=e^{i\kappa\,\mathrm{sgn}(x-y)}
\hat\psi^\dagger(y)\hat\psi(x)+\delta(x-y),\nonumber
\end{align}
where $\kappa=0$ and $\kappa=\pi$ recover bosons and fermions, respectively. Bilinear Hamiltonians of 1D anyons generally do not have a free-particle spectrum, but can nevertheless be solved by Bethe-ansatz techniques~\cite{kunduExactSolutionDouble1999,batchelorOneDimensionalInteractingAnyon2006,girardeauAnyonFermionMappingApplications2006,
patuCorrelationFunctionsOnedimensional2007}. 
A lattice version of Eq.~\eqref{eq:1Danyons} were proposed in Ref.~\cite{keilmannStatisticallyInducedPhase2011}, where fractional Jordan--Wigner
transformations were introduced to show that 1D Abelian anyons can emerge in a bosonic system. %
1D Abelian anyons have been observed in recent cold-atom experiments~\cite{kwanRealizationOnedimensionalAnyons2024,dharObservingAnyonizationBosons2025}.

\subsubsection{Non-involutive $R$-parastatistics in 1D}\label{sec:non-involutive-R-para}
Below we mention an interesting open direction to study $R$-parastatistics in 1D with non-involutive $R$-matrices, which can potentially generalize all other 1D particle statistics mentioned above. Let $R$ be a non-involutive $R$-matrix, and consider quadratic relations of the following form
\begin{equation}\label{eq:noninvolR-para}
    \hat\psi_{i,a}\hat\psi_{j,b}    =\sum_{c,d}R^{cd}_{ab}\hat\psi_{j,c}\hat\psi_{i,d},
    \quad i<j. 
\end{equation}
This is the analog of the second and the third lines of Eq.~\eqref{eq:fundamental_Rcommu} for a non-involutive $R$-matrix, where the extra constraint $i<j$ makes the relation consistent. %
Note that the relation in Eq.~\eqref{eq:noninvolR-para} only makes sense in a 1D geometry, where particle positions have a well-defined ordering. 
This relation can be viewed as a non-Abelian
generalization of the commutation relations of  parafermions and 1D Abelian anyons~[first lines of Eq.~\eqref{eq:def_parafermions}	 and Eq.~\eqref{eq:1Danyons}, respectively].
The second lines of Eq.~\eqref{eq:def_parafermions}	 and Eq.~\eqref{eq:1Danyons} can be generalized similarly.
The main challenge is not formally writing down such algebras, but finding useful
choices for which the resulting many-body theory is solvable and realizable in physical systems. In
particular, unlike the involutive case studied in this paper, a general
non-involutive \(R\)-matrix often does not produce a free-particle theory. Solvability would likely require extra integrability structures, which remains an interesting open problem. Furthermore, to show that non-involutive $R$-paraparticles can emerge in condensed matter systems, one probably also needs to generalize the MPO Jordan-Wigner transformation~\cite{wang2023para} to the non-involutive case, and construct a 1D quantum spin chains that exactly maps to non-involutive $R$-paraparticle theories.

\subsection{Other theories}\label{sec:othertheories}
Up to now we have discussed various second quantized
theories of quantum particle statistics that generalize ordinary
fermions and bosons. We now mention some notable families of
quantum statistics that are studied without using a second
quantization formalism. 

In spatial dimensions higher than two, topological phases can contain extended excitations such as  strings and membranes~\cite{wangBraidingStatisticsLoop2014,wangNonAbelianStringParticle2015,
LinLevin2015Loop,LanKongWen3DAB,kongDefects3dimensionalToric2020,kobayashiGeneralizedStatisticsLattices2026}, whose statistics are described by
particle-loop braiding, loop-loop braiding, three-loop braiding, or
more general linking processes rather than by exchange of point
particles. 
Some extrinsic defects in topological phases can also exhibit nontrivial quantum statistics
~\cite{bombinTopologicalOrderTwist2010,teoMajoranaFermionsNonAbelian2010,freedmanProjectiveRibbonPermutation2011,barkeshliTwistDefectsProjective2013}.

We finally mention some other generalizations of the Pauli exclusion principle in the literature, including various types of generalized exclusion statistics in the first quantization formulation~\cite{Gentile1940,Haldane1991Fractional,wuStatisticalDistributionGeneralized1994}, and exotic exclusion statistics of quasiparticles in conformal field theory spectra~\cite{Schoutens1997CFT,Bouwknegt1999CFT}. There are also various theories of generalized particle exchange statistics proposed in the first quantization formulation, such as $Z_2\times Z_2$ parastatistics~\cite{Toppan2021Inequivalent}, and immanons~\cite{tichyExtendingExchangeSymmetry2017}. 
\bibliography{library1}
\end{document}